\documentclass[a4paper, 11pt]{article}
\usepackage[utf8]{inputenc}
\usepackage[ngermanb, english]{babel}
\usepackage{amsmath}
\usepackage{amssymb}
\usepackage{amsthm}
\usepackage{tikz-cd}
\usepackage{commath}
\usepackage{mathrsfs}
\usepackage{mathtools}
\usepackage{slashed}
\usepackage{tensor}
\usepackage{parskip}
\usepackage{epstopdf}
\usepackage{graphicx}
\usepackage[nottoc]{tocbibind}
\usepackage{babelbib}
\usepackage{hyperref}
\usepackage[top=2.5cm, left=2.5cm, right=2.5cm, bottom=2.5cm]{geometry}

\theoremstyle{plain}
\newtheorem{thm}{Theorem}[section]
\newtheorem{lem}[thm]{Lemma}
\newtheorem{prop}[thm]{Proposition}
\newtheorem{cor}[thm]{Corollary}

\theoremstyle{definition}
\newtheorem{defn}[thm]{Definition}
\newtheorem{exmp}[thm]{Example}
\newtheorem{sol}[thm]{Solution}
\theoremstyle{remark}
\newtheorem{rem}[thm]{Remark}

\newtheorem{pro}[thm]{Problem}

\providecommand{\sectionref}[1]{Section~\ref{#1}}
\providecommand{\ssecref}[1]{Subsection~\ref{#1}}

\providecommand{\sssecaref}[2]{Subsubsections~\ref{#1}~and~\ref{#2}}
\providecommand{\eqnref}[1]{Equation~\eqref{#1}}
\providecommand{\eqnsref}[1]{Equations~\eqref{#1}}
\providecommand{\eqnsaref}[2]{Equations~\eqref{#1}~and~\eqref{#2}}
\providecommand{\eqnssaref}[3]{Equations~\eqref{#1},~\eqref{#2}~and~\eqref{#3}}

\providecommand{\proref}[1]{Problem~\ref{#1}}
\providecommand{\solref}[1]{Solution~\ref{#1}}
\providecommand{\solsaref}[4]{Solutions~\ref{#1},~\ref{#2},~\ref{#3}~and~\ref{#4}}
\providecommand{\solsoref}[4]{Solutions~\ref{#1},~\ref{#2},~\ref{#3}~or~\ref{#4}}
\providecommand{\ftnref}[1]{Footnote~\ref{#1}}

\newlength{\graphlength}
\newlength{\cgreenlength}
\newcommand{\order}[1]{\mathcal O \left ( #1 \right )}
\newcommand{\one}{\mathbb{I}}
\newcommand{\coone}{\hat{\mathbb{I}}}
\newcommand{\Q}{\mathcal{Q}}
\newcommand{\HQ}{\mathcal{H}_{\mathcal{Q}}}
\newcommand{\RQ}{\mathcal{R}_{\mathcal{Q}}}

\newcommand{\GQ}{\mathcal{G}_{\mathcal{Q}}}
\newcommand{\sym}[1]{\operatorname{Sym} \left ( #1 \right )}
\newcommand{\res}[1]{\operatorname{Res} \left ( #1 \right )}
\newcommand{\cpl}[1]{\operatorname{Cpl} \left ( #1 \right )}
\newcommand{\dt}[1]{\operatorname{Det} \left ( #1 \right )}
\newcommand{\tr}[1]{\operatorname{Tr} \left ( #1 \right )}

\newcommand{\sdd}[1]{\omega \left ( #1 \right )}
\newcommand{\D}[1]{\Delta \left ( #1 \right )}
\newcommand{\antipode}[1]{S \left ( #1 \right )}
\newcommand{\mult}{m}
\newcommand{\ring}{k}
\newcommand{\leh}{- \frac{1}{2 \gcoupling^2} R}
\newcommand{\lmd}{\frac{1}{4 \ecoupling^2} g^{\mu \rho} g^{\nu \sigma} F_{\mu \nu} F_{\rho \sigma} + \overline{\Psi} \left ( \imaginary \slashed{\nabla}^{U(1) \times_\rho \Sigma M} - m \right ) \Psi}
\newcommand{\lqgrqed}{\left ( \leh + \lmd \right ) \dif V_g + \mathcal{L}_{\text{GF}} + \mathcal{L}_{\text{Ghost}}}
\newcommand{\graph}[1]{\vcenter{\hbox{\includegraphics[width=\graphlength]{#1-eps-converted-to.pdf}}}}
\newcommand{\tgraph}[1]{\vcenter{\hbox{\includegraphics[width=0.75\graphlength]{#1-eps-converted-to.pdf}}}}
\newcommand{\ttgraph}[1]{\vcenter{\hbox{\includegraphics[width=0.5\graphlength]{#1-eps-converted-to.pdf}}}}
\newcommand{\ggraph}[1]{\vcenter{\hbox{\includegraphics[width=1.5\graphlength]{#1-eps-converted-to.pdf}}}}
\newcommand{\cgreen}[1]{\vcenter{\hbox{\includegraphics[width=\cgreenlength]{#1-eps-converted-to.pdf}}}}
\newcommand{\ograph}[1]{\vcenter{\hbox{\includegraphics[width=0.833333333\graphlength]{#1-eps-converted-to.pdf}}}}
\newcommand{\tcgreen}[1]{\vcenter{\hbox{\includegraphics[width=0.75\cgreenlength]{#1-eps-converted-to.pdf}}}}
\newcommand{\residue}[1]{\vcenter{\hbox{\includegraphics[width=\cgreenlength]{#1-eps-converted-to.pdf}}}}
\newcommand{\tresidue}[1]{\vcenter{\hbox{\includegraphics[width=0.75\cgreenlength]{#1-eps-converted-to.pdf}}}}
\newcommand{\FR}[1]{\Phi \left ( #1 \right )}
\newcommand{\FRC}[1]{\Phi^{\mathscr{C}} \left ( #1 \right )}
\newcommand{\FRCRO}[1]{\Phi^{\mathscr{C}_{\text{R}} \mathscr{C}_{\text{O}}} \left ( #1 \right )}
\newcommand{\FRCOR}[1]{\Phi^{\mathscr{C}_{\text{O}} \mathscr{C}_{\text{R}}} \left ( #1 \right )}
\newcommand{\FRCR}[1]{\Phi^{\mathscr{C}_{\text{R}}} \left ( #1 \right )}

\newcommand{\regFR}{{\Phi_\mathscr{E}^\varepsilon}}
\newcommand{\renFR}{{\Phi_\mathscr{R}}}
\newcommand{\counterterm}{S_\mathscr{R}^\regFR}

\newcommand{\deDonder}{{d \negmedspace D}}
\newcommand{\Lorenz}{L}
\newcommand{\gravitonghost}{\chi}
\newcommand{\photonghost}{\theta}
\newcommand{\textfrac}[2]{#1 / #2}
\newcommand{\ecoupling}{\mathrm{e}}
\newcommand{\gcoupling}{\varkappa}
\newcommand{\imaginary}{\mathrm{i}}
\newcommand{\id}{\operatorname{Id}}
\newcommand{\diff}{\operatorname{Diff}_0 \left ( M \right )}
\newcommand{\autu}{\operatorname{Aut} \left ( M \times U(1) \right )}
\newcommand{\phixt}{\phi_X^\tau}
\newcommand{\phixtt}{\phi_{X \left ( \tau \right )}^\tau}
\newcommand{\vphif}{\varphi_{\mathfrak{f}}}
\newcommand{\precombgreen}{\mathfrak{x}}
\newcommand{\combgreen}{\mathfrak{X}}
\newcommand{\rescombgreen}{\mathfrak{X}}
\newcommand{\Lie}{\pounds}
\newcommand{\surject}{\to \!\!\!\!\! \to}

\makeatletter
\newcommand{\subalign}[1]{%
  \vcenter{%
    \Let@ \restore@math@cr \default@tag
    \baselineskip\fontdimen10 \scriptfont\tw@
    \advance\baselineskip\fontdimen12 \scriptfont\tw@
    \lineskip\thr@@\fontdimen8 \scriptfont\thr@@
    \lineskiplimit\lineskip
    \ialign{\hfil$\m@th\scriptstyle##$&$\m@th\scriptstyle{}##$\crcr
      #1\crcr
    }%
  }
}
\makeatother

\title{\textsc{Algebraic Structures in the Coupling of Gravity to Gauge Theories}}
\author{David Prinz\footnote{Department of Mathematics and Department of Physics at Humboldt University of Berlin and Max Planck Institute for Gravitational Physics (Albert Einstein Institute) in Potsdam-Golm; prinz@\{math.hu-berlin.de, physik.hu-berlin.de, aei.mpg.de\}}}
\date{April 30, 2020}

\begin{document}

\maketitle

\begin{abstract}
This article is an extension of the author's second master thesis \cite{Prinz_2}. It aims to introduce to the theory of perturbatively quantized General Relativity coupled to Spinor Electrodynamics, provide the results thereof and set the notation to serve as a starting point for further research in this direction. It includes the differential geometric and Hopf algebraic background, as well as the corresponding Lagrange density and some renormalization theory. Then, a particular problem in the renormalization of Quantum General Relativity coupled to Quantum Electrodynamics is addressed and solved by a generalization of Furry's Theorem. Next, the restricted combinatorial Green's functions for all two-loop propagators and all one-loop divergent subgraphs thereof are presented. Finally, relations between these one-loop restricted combinatorial Green's functions necessary for multiplicative renormalization are discussed.

\end{abstract}

\section{Introduction}

The theory of General Relativity (GR) and Quantum Field Theory (QFT) are the two great achievements of physics in the 20\textsuperscript{th} century. GR, on the one side, describes nature on very large scales when huge masses are involved. QFT, on the other side, describes nature on very small scales when tiny masses are involved. Being very successful in their regimes, there are situations when both conditions appear at the same time, i.e.\ when huge masses are compressed into small scales. For example, this situation occurs in models of the big bang and in models of black holes. For these situations a theory of Quantum Gravity (QG) is needed to understand nature. In particular, a theory of QG should be able to clarify how the universe emerged, i.e.\ through a big bang or otherwise. Therefore, it was soon tried to apply the usual techniques of perturbative QFT to the dynamical part of the metric in spacetimes of GR \cite{Rovelli}. These works, in \cite{Rovelli} called ``The covariant line of research'', were started by M. Fierz, W. Pauli and L. Rosenfeld in the 1930s. Then, R. Feynman \cite{Feynman_Hatfield_Morinigo_Wagner} and B. DeWitt \cite{DeWitt_I,DeWitt_II,DeWitt_III,DeWitt_IV} calculated Feynman rules of GR in the 1960s. Next, D. Boulware, S. Deser, P. van Nieuwenhuizen \cite{Boulware_Deser_Nieuwenhuizen} and G. 't Hooft \cite{tHooft_QG} and M. Veltman \cite{Veltman} found evidence of the non-renormalizability of Quantum General Relativity (QGR) in the 1970s. We stress that by QGR we mean a quantization of GR using QFT methods, whereas by QG we mean any theory of quantized gravitation, such as e.g.\ Loop Quantum Gravity, String Theory or Supergravity. In this article we continue the work on perturbative QGR as started by D. Kreimer in the 2000s \cite{Kreimer_QG1,Kreimer_QG2}. D. Kreimer used the modern techniques of Hopf algebraic renormalization developed by A. Connes and himself in the 1990s and 2000s \cite{Kreimer_Hopf_Algebra,Connes_Kreimer_0,Connes_Kreimer_1,Connes_Kreimer_2,Kreimer_Anatomy}. Similar situations were studied in the so-called core Hopf algebra by D. Kreimer and W.D. van Suijlekom in the 2000s \cite{Kreimer_Core,Kreimer_vSuijlekom}.

We start this article in \sectionref{sec:differential_geometric_notions_and_the_lagrange_density_of_qgr_qed} with the differential geometric background needed to understand the Lagrange density of Quantum General Relativity coupled to Quantum Electrodynamics (QGR-QED). Then, we introduce the Lagrange density of QGR-QED,
\begin{equation}
\! \! \mathcal{L}_{\text{QGR-QED}} = \lqgrqed \, ,
\end{equation}
which consists of the usual Einstein-Hilbert Lagrange density, the canonical generalization of the Maxwell-Dirac Lagrange density to curved spacetimes and the gauge fixing and ghost Lagrange densities. Finally, the Lagrange density of QGR-QED is discussed in detail. Then, in \sectionref{sec:hopf_algebras_and_the_connes-kreimer_renormalization_hopf_algebra} we introduce Hopf algebras in general and the Connes-Kreimer renormalization Hopf algebra in particular. Next, we discuss a problem which can occur when associating the Connes-Kreimer renormalization Hopf algebra to a given local QFT. As we remark, this problem occurs already in QED and in particular in QGR-QED. More concretely, there can exist divergent Feynman graphs whose residue is not in the residue set of the given local QFT. Thus, there are no vertex residues in the theory present which are able to absorb the corresponding divergences. Then, we present four different solutions to this problem in \solsaref{sol:solution_1}{sol:solution_2}{sol:solution_3}{sol:solution_4} and discuss their physical interpretation. We proceed by analyzing the structure of Hopf ideals in the renormalization Hopf algebra which represent the symmetries compatible with renormalization. Then, we define the Hopf algebra of QGR-QED. Therefore, we formulate and prove a generalization of Furry's Theorem in \thmref{thm:generalized_furry_theorem} which holds also for amplitudes with an arbitrary number of external gravitons. This is in particular useful, since at least for the calculations done in the realm of this article, these are the only Feynman graphs which need to be set to zero when constructing the renormalization Hopf algebra of QGR-QED, besides from pure self-loop Feynman graphs, which vanish for kinematic renormalization schemes. Finally, in \sectionref{sec:coproduct_structure_of_the_greens_function} we present all one- and two-loop propagator graphs and all one-loop three-point graphs. Then, we present their coproduct structure, for which the coproduct of 155 Feynman graphs has been computed. Using these relations, we discuss the obstructions to multiplicative renormalizability of QGR-QED, which results in a generalization of Ward-Takahashi and Slavnov-Taylor identities \cite{Ward,Takahashi,tHooft_YM,Taylor,Slavnov}.\footnote{Actually, Slavnov-Taylor identities were first discovered by Gerard 't Hooft in \cite{tHooft_YM}.} Finally, we give an outlook in \sectionref{sec:conclusion}.

We remark that the Feynman graphs were generated in two steps: First, the corresponding scalar (unlabeled) Feynman graphs were created using \cite{Borinsky_Feyngen} and then the Feynman graphs were labeled with particle labels from QGR-QED using several python programs written by the author \cite{Python}, cf.\ \cite[Appendix A]{Prinz_2}. Then, the Feynman diagrams were drawn with JaxoDraw \cite{Binosi_Theussl,Binosi_Collins_Kaufhold_Theussl}.

\section{Differential geometry and the Lagrange density of QGR-QED} \label{sec:differential_geometric_notions_and_the_lagrange_density_of_qgr_qed}

We start with the differential geometric background and the Lagrange density of QGR-QED. Throughout this article we use the Einstein summation convention, if not stated otherwise. Furthermore, we use the sign-convention \((-++)\), as classified by \cite{Misner_Thorne_Wheeler}. Moreover, we require sections to be smooth, i.e.\ \(\Gamma \left ( U, \cdot \right ) := \Gamma^\infty \left ( U, \cdot \right )\) for any \(U \subseteq M\) open, where \(M\) is a manifold. Additionally, we write the electric charge \(\ecoupling\) in an upright font, in order to avoid confusion with Euler's number \(e\), vielbeins \(e_\mu^m\) and inverse vielbeins \(e^\mu_m\).

\subsection{Differential geometry} \label{ssec:differential_geometric_notions}

\begin{defn}[Spacetime] \label{def:spacetime}
Let \((M,g)\) be a Lorentzian manifold. \((M,g)\) is called a spacetime if it is smooth, connected, 4-dimensional and time-orientable. Furthermore, we choose the West coast (``mostly minus'') signature signature for \(g\). Moreover, we denote the constant Minkowski metric by \(\eta\) and the unit matrix by \(\delta\), i.e.\
\begin{equation}
	\eta := \begin{pmatrix} 1 & 0 & 0 & 0 \\ 0 & - 1 & 0 & 0 \\ 0 & 0 & - 1 & 0 \\ 0 & 0 & 0 & - 1 \end{pmatrix} \qquad \text{and} \qquad \delta := \begin{pmatrix} 1 & 0 & 0 & 0 \\ 0 & 1 & 0 & 0 \\ 0 & 0 & 1 & 0 \\ 0 & 0 & 0 & 1 \end{pmatrix}\, .
\end{equation}
\end{defn}

\vspace{\baselineskip}

\begin{rem}
We assume spacetimes to be 4-dimensional as e.g.\ the gravity-matter Feynman rules depend directly on the spacetime dimension, cf.\ \remref{rem:Feynman_rules_1} and \cite{Prinz_4}. However, the corresponding calculations could also be carried out in dimensions different than 4.
\end{rem}

\vspace{\baselineskip}

\begin{defn}[Matter-compatible spacetime] \label{defn:matter-compatible_spacetime}
Let \((M,g)\) be a spacetime. We call \((M,g)\) a matter-compatible spacetime if it is diffeomorphic to the Minkowski spacetime \((\mathbb{M}, \eta)\).
\end{defn}

\vspace{\baselineskip}

\begin{rem}
The motivation for \defnref{defn:matter-compatible_spacetime} comes from the definition of the graviton field \(h_{\mu \nu}\) relative to the Minkowski metric \(\eta_{\mu \nu}\), cf.\ \defnref{defn:electron-positron_photon_graviton_field}.\footnote{It is in principle also possible to choose a different background metric, which might require a different topology of \(M\). This is for example necessary when a non-vanishing cosmological constant is included, as the background metric should be a solution of the vacuum Einstein field equations. Then, the cosmological constant behaves in the quantum theory like a mass term for the graviton propagator.} In a more general setting \defnref{defn:matter-compatible_spacetime} could be weakened to: \emph{Let \((M,g)\) be a spacetime. We call \((M,g)\) a matter-compatible spacetime if it is globally hyperbolic, oriented, time-oriented and such that its second de Rham cohomology vanishes.} Then, this is then motivated by the following facts: We want to consider spacetimes with well-defined Cauchy problems, which requires \((M,g)\) to be globally hyperbolic. Furthermore, we want \((M,g)\) to admit a spin structure in order to define matter via a spinor bundle, which is equivalent to \((M,g)\) being globally hyperbolic, 4-dimensional, oriented and time-oriented \cite{Geroch_1,Geroch_2}. Finally, we want to consider spacetimes which allow the exclusion of magnetic monopoles, which is ensured if the second de Rham cohomology vanishes, as this allows for a global definition of the a priori local connection form \(\imaginary \ecoupling A_\mu\) on the base-manifold \(M\).
\end{rem}

\vspace{\baselineskip}

\begin{defn}[Spacetime-matter bundle] \label{defn:spacetime-matter_bundle}
Let \((M,g)\) be a matter-compatible spacetime. Then we define the spacetime-matter bundle to be the globally trivial super bundle
\begin{equation}
	\mathcal{S} := M \times TM \times E \times \left ( U(1) \times_\rho \Sigma M \right ) \times \Pi \left ( TM \oplus T^*M \right ) \times \Pi \left ( \mathfrak{u}(1) \oplus \mathfrak{u}(1)^* \right ) \, .
\end{equation}
\(T M\) is the tangent bundle and \(E\) a real \(4\)-dimensional vector bundle used for the definition of vielbeins and inverse vielbeins (then considered as the tensor product bundle \(TM \otimes_\mathbb{R} E\), to stress that vielbeins and inverse vielbeins are multilinear maps), cf.\ \defnref{defn:vielbeins_inverse_vielbeins}. \(U(1)\) is the principle bundle modeling electrodynamics and acting via the representation \(\rho\) on the spinor bundle \(\Sigma M\), which is a 4-dimensional complex vector bundle modeling fermions, cf.\ \defnref{defn:electron-positron_photon_graviton_field}. \(U(1) \times_\rho \Sigma M\) denotes the corresponding fiber product bundle, to which we refer to as twisted spinor bundle. Finally, \(\Pi \left ( TM \oplus T^*M \right )\) and \(\Pi \left ( \mathfrak{u}(1) \oplus \mathfrak{u}(1)^* \right )\) denote the parity shifted Lie algebra and dual Lie algebra bundles of the gauge groups which model the graviton ghost and the photon ghost, respectively, cf.\ \remref{rem:gauge_transformations}. We equip the spacetime-matter bundle with metrics in \defnref{defn:metrics_spacetime-matter_bundle} which, in turn, naturally include the corresponding dual bundles.\footnote{We denote the dual bundles via the asterisk, \(*\), except for the spinor bundle, twisted spinor bundle and the parity shifted Lie algebra and dual Lie algebra bundles of the gauge groups for which we use the overline, \(\overline{\phantom{*}}\).} Additionally, we also equip it with connections in \defnref{defn:connections_spacetime-matter_bundle} such that we have a notion of curvature, cf.\ \defnref{defn:curvatures_spacetime-matter_bundle}.
\end{defn}

\vspace{\baselineskip}

\begin{rem}
The global triviality of the spacetime-matter bundle in \defnref{defn:spacetime-matter_bundle} is motivated by the following facts: The tangent bundle \(TM\) and the spinor bundle \(\Sigma M\) are globally trivial since matter-compatible spacetimes are defined to be diffeomorphic to the Minkowski spacetime and are thus in particular parallelizable, cf.\ \defnref{defn:matter-compatible_spacetime}. Furthermore, the vector bundle \(E\) is chosen to be trivial for the definition of vielbeins and inverse vielbeins, cf.\ \defnref{defn:vielbeins_inverse_vielbeins}. Moreover, the \(U(1)\) principle bundle is globally trivial since matter-compatible spacetimes are defined to be diffeomorphic to the Minkowski spacetime and thus in particular have vanishing second de Rham cohomology which implies the first Chern class of the associated line bundle to vanish, cf.\ \defnref{defn:matter-compatible_spacetime}. Finally, the parity shifted Lie algebra and dual Lie algebra bundles of the gauge groups \(\Pi \left ( TM \oplus T^*M \right )\) and \(\Pi \left ( \mathfrak{u}(1) \oplus \mathfrak{u}(1)^* \right )\) are globally trivial because of the global triviality of the bundles \(TM\) and \(U(1)\).
\end{rem}

\vspace{\baselineskip}

\begin{defn}[Metrics on the spacetime-matter bundle] \label{defn:metrics_spacetime-matter_bundle}
We consider the following metrics on the spacetime-matter bundle \(\mathcal{S}\): On the tangent bundle \(TM\) we consider the Lorentzian metric \(g\) with West coast (``mostly minus'') signature, mapping vector fields \(X_1, X_2 \in \Gamma \left ( M, TM \right )\) to the real number
\begin{equation}
	\left \langle X_1 , X_2 \right \rangle_{TM} := g_{\mu \nu} X_1^\mu X_2^\nu \in \mathbb{R} \, .
\end{equation}
Furthermore, on the real vector bundle \(E\) we consider the constant Minkowski metric \(\eta\) with West coast (``mostly minus'') signature, mapping vector fields \(Y_1, Y_2 \in \Gamma \left ( M, E \right )\) to the real number
\begin{equation}
	\left \langle Y_1 , Y_2 \right \rangle_E := \eta_{m n} Y_1^m Y_2^n \in \mathbb{R} \, .
\end{equation}
Moreover, on the \(U(1)\) principle bundle we use the Hermitian metric, mapping sections \(s_1, s_2 \in \Gamma \left ( M, U(1) \right )\) to the complex number\footnote{Where we denote complex conjugation via the asterisk, \(*\).}
\begin{equation}
	\left \langle s_1 , s_2 \right \rangle_{U(1)} := s_1^* s_2 \in \mathbb{C} \, .
\end{equation}
Additionally, on the spinor bundle \(\Sigma M\) we use the Hermitian metric together with the Clifford multiplication by the unit timelike vector field, mapping spinor fields \(\psi_1, \psi_2 \in \Gamma \left ( M, \Sigma M \right )\) to the complex number
\begin{equation}
	\left \langle \psi_1 , \psi_2 \right \rangle_{\Sigma M} := \overline{\psi_1} \psi_2 \in \mathbb{C} \, ,
\end{equation}
where we have set \(\overline{\psi} := e_0^m \psi^\dagger \gamma_m\),\footnote{Where we denote Hermitian conjugation via the dagger, \(\dagger\).} with the curved unit timelike vector field components of a vielbein \(e_0^m\),\footnote{We remark that since we consider matter-compatible spacetimes to be diffeomorphic to the Minkowski spacetime and thus in particular diffeomorphic to a globally hyperbolic spacetime, it is also possible to consider charts in which \(e^m_0 \equiv \delta^m_0\) such that in particular \(e^m_0 \gamma_m \equiv \gamma_0\), and some references use this implicitly, e.g.\ \cite{Choi_Shim_Song}. However it should be noted that then the theory is only invariant under diffeomorphisms which preserve global hyperbolicity. \label{foot:choi_shim_song_gh-manifold}} being introduced in \defnref{defn:vielbeins_inverse_vielbeins}, and the Clifford multiplication \(\gamma_m\), being introduced in \defnref{defn:clifford_multiplication}. Thus, on the twisted spinor bundle \(U(1) \times_\rho \Sigma M\) we consider the induced fiber product metric, mapping twisted spinor fields \(\Psi_1, \Psi_2 \in \Gamma \left ( M, U(1) \times_\rho \Sigma M \right )\) to the complex number
\begin{equation}
	\left \langle \Psi_1 , \Psi_2 \right \rangle_{U(1) \times_\rho \Sigma M} := \overline{\Psi_1} \Psi_2 \in \mathbb{C} \, .
\end{equation}
where, again, we have set \(\overline{\Psi} := e_0^m \Psi^\dagger \gamma_m\). Finally, on the parity shifted Lie algebra and dual Lie algebra bundles of the gauge groups \(\Pi \left ( TM \oplus T^*M \right )\) and \(\Pi \left ( \mathfrak{u}(1) \oplus \mathfrak{u}(1)^* \right )\) we consider the metrics induced by the metrics on \(TM\) and \(U(1)\).
\end{defn}

\vspace{\baselineskip}

\begin{defn}[Vielbeins and inverse vielbeins] \label{defn:vielbeins_inverse_vielbeins}
Let \(\mathcal{S}\) be the spacetime-matter bundle. Then we can define global vector bundle isomorphisms \(e \in \Gamma \left ( M, T^*M \otimes_\mathbb{R} E \right )\), called vielbeins, such that
\begin{subequations}
\begin{equation}
	g_{\mu \nu} = \eta_{m n} e^m_\mu e^n_\nu \, . \label{eqn:vielbeins}
\end{equation}
Furthermore, we can define global inverse vector bundle isomorphisms \(e^* \in \Gamma \left ( M, TM \otimes_\mathbb{R} E^* \right )\), called inverse vielbeins, such that\footnote{We omit the asterisk, \(*\), for inverse vielbeins \(e^*\) when the abstract index notation is used because of \eqnref{eqn:transformation_vielbeins_inverse_vielbeins}.}
\begin{equation}
	\eta_{m n} = g_{\mu \nu} e^\mu_m e^\nu_n \, . \label{eqn:inverse_vielbeins}
\end{equation}
\end{subequations}
Greek indices, here \(\mu\) and \(\nu\), belong to the tangent bundle \(TM\) and are referred to as curved indices. Thus, they are raised and lowered using the usual metric \(g_{\mu \nu}\) and its inverse \(g^{\mu \nu}\). Latin indices, here \(m\) and \(n\), belong to the vector bundle \(E\) and are referred to as flat indices. Thus, they are raised and lowered using the Minkowski metric \(\eta_{m n}\) and its inverse \(\eta^{m n}\). Therefore, inverse vielbeins are related to vielbeins via
\begin{equation}
	e^\mu_m = g^{\mu \nu} \eta_{m n} e^n_\nu \, . \label{eqn:transformation_vielbeins_inverse_vielbeins}
\end{equation}
Moreover, notice that \eqnsaref{eqn:vielbeins}{eqn:inverse_vielbeins} are equivalent to
\begin{equation}
	g_{\mu \nu} e^\nu_n = \eta_{m n} e^m_\mu = e_{\mu n} \, , \label{eqn:inverse_vielbeins_eigenvalue_equation}
\end{equation}
which states that, when suppressing the Einstein summation convention on the flat indices and viewing them as numbers, inverse vielbeins \(\left \{ e^\mu_m \right \}_{m \in \left \{ 1,2,3,4 \right \}}\) are the set of the \(4\) eigenvector fields of the metric \(g_{\mu \nu}\) with eigenvalues \(\eta_{m m} \in \left \{ \pm 1 \right \}\), which are normalized to unit length
\begin{equation}
	\left \| e_m \right \|_g := \sqrt{\left \vert g_{\mu \nu} e^\mu_m e^\nu_m \right \vert} = \sqrt{\left \vert \eta_{m m} \right \vert} = 1 \, .
\end{equation}
Finally, we remark that the global definition of vielbeins and inverse vielbeins is only possible for parallelizable manifolds, such as the matter-compatible spacetimes of \defnref{defn:matter-compatible_spacetime}. This is because the tangent frame bundle \(FM\) allows for a global section if and only if the manifold is parallelizable, i.e.\ it is then possible to choose a global coordinate system on \(TM\) that consists of normalized eigenvector fields of the metric \(g_{\mu \nu}\).
\end{defn}

\vspace{\baselineskip}

\begin{rem} \label{rem:ambiguity_vielbeins_inverse_vielbeins}
Given the situation of \defnref{defn:vielbeins_inverse_vielbeins}, notice that the definitions of vielbeins and inverse vielbeins are not unique, since for any local Lorentz transformation acting on \(E^*\), i.e.\ a section of the corresponding orthogonal frame bundle \(\Lambda \in \Gamma \left ( U, F_O E^* \right )
\) acting via the standard representation on \(E^*\), we have
\begin{equation}
\begin{split}
	g_{\mu \nu} e^\mu_m e^\nu_n & = \eta_{m n}\\
	& = \eta_{r s} \tensor{\Lambda}{^r _m} \tensor{\Lambda}{^s _n}\\
	& = g_{\mu \nu} e^\mu_r e^\nu_s \tensor{\Lambda}{^r _m} \tensor{\Lambda}{^s _n}\\
	& = g_{\mu \nu} \tilde{e}^\mu_m \tilde{e}^\nu_n \, .
\end{split}
\end{equation}
Here, we denoted the transformed inverse vielbeins via \(\tilde{e}^\mu_m := e^\mu_r \tensor{\Lambda}{^r _m}\). Obviously the same calculation also holds for vielbeins instead of inverse vielbeins with local Lorentz transformations acting on \(E\). This ambiguity will lead to the first term in the spin connection, \(e_\nu^n \left ( \partial_\mu e^\nu_l \right )\), cf.\ \eqnref{eqn:spin_connection}. In fact, the first term in the spin connection can be viewed as the gauge field associated to local Lorentz transformations.
\end{rem}

\vspace{\baselineskip}

\begin{defn}[Connections on the spacetime-matter bundle] \label{defn:connections_spacetime-matter_bundle}
We use the following connections on the spacetime-matter bundle \(\mathcal{S}\): For the tangent bundle \(TM\) of the manifold \(M\) we use the Levi-Civita connection \(\nabla^{TM}_\mu\), acting on a vector field \(X \in \Gamma \left ( M, TM \right )\) via
\begin{subequations}
\begin{align}
	\nabla^{TM}_\mu X^\nu & := \partial_\mu X^\nu + \Gamma^\nu_{\mu \lambda} X^\lambda
\intertext{and on a corresponding covector field via}
	\nabla^{TM}_\mu X_\nu & := \partial_\mu X_\nu - \Gamma^\lambda_{\mu \nu} X_\lambda \, ,
\end{align}
\end{subequations}
with the Christoffel symbol \(\Gamma^\nu_{\mu \lambda}\), given in the case of the Levi-Civita connection via
\begin{equation}
	\Gamma^\nu_{\mu \lambda} := \frac{1}{2} g^{\nu \tau} \left ( \partial_\mu g_{\tau \lambda} + \partial_\lambda g_{\mu \tau} - \partial_\tau g_{\mu \lambda} \right ) \, .
\end{equation}
Furthermore, for the vector bundle \(E\) we use the covariant derivative \(\nabla^{E}_\mu\), induced via the connection on the tangent bundle using vielbeins and inverse vielbeins and acting on a vector field \(Y \in \Gamma \left ( M, E \right )\) via
\begin{subequations}
\begin{align}
	\nabla^{E}_\mu Y^n & := \partial_\mu Y^n + \omega_{\mu l}^n Y^l
\intertext{and on a corresponding covector field via}
	\nabla^{E}_\mu Y_n & := \partial_\mu Y_n - \omega_{\mu n}^l Y_l \, ,
\end{align}
\end{subequations}
with the spin connection
\begin{equation}
\begin{split}
	\omega_{\mu l}^n & := e_\nu^n \left ( \nabla^{TM}_\mu e^\nu_l \right )\\ & \phantom{:} = e_\nu^n \left ( \partial_\mu e^\nu_l \right ) + e_\nu^n \Gamma^\nu_{\mu \lambda} e^\lambda_l \, . \label{eqn:spin_connection}
\end{split}
\end{equation}
We remark that the first term in the spin connection, \(e_\nu^n \left ( \partial_\mu e^\nu_l \right )\), is due to the ambiguity in the definition of vielbeins and inverse vielbeins, as was discussed in \remref{rem:ambiguity_vielbeins_inverse_vielbeins}. Moreover, for the \(U(1)\)-principle bundle, we use the covariant derivative \(\nabla^{U(1)}_\mu\), acting on a section \(s \in \Gamma \left ( M, U(1) \right )\) via\footnote{The existence of global sections in the \(U(1)\) principle bundle is ensured by its global triviality, cf.\ \defnref{defn:spacetime-matter_bundle}, contrary to vector bundles, such as \(TM\), \(E\) or \(\Sigma M\), which always allow for the global zero section.}
\begin{subequations}
\begin{align}
	\nabla^{U(1)}_\mu s & := \partial_\mu s + \imaginary \ecoupling A_\mu s
\intertext{and on a corresponding dual section via}
	\nabla^{U(1)}_\mu s^* & := \partial_\mu s^* - \imaginary \ecoupling A_\mu s^* \, ,
\end{align}
\end{subequations}
with the \(\mathfrak{u}(1) \cong \imaginary \mathbb{R}\)-valued connection form \(\imaginary \ecoupling A_\mu\).\footnote{The imaginary unit \(\imaginary\) is included into the definition, such that \(A_\mu\) is real-valued. Moreover, the coupling constant \(\ecoupling\) is included to introduce the right proportionality constant in the interaction term in the Lagrange density \(\mathcal{L}_\text{GR-ED}\) when taking covariant derivatives of fermion fields, cf.\ \defnref{defn:electron-positron_photon_graviton_field}.} Additionally, for the spinor bundle we use the covariant derivative \(\nabla^{\Sigma M}_\mu\), induced via the connection on the real vector bundle \(E\) using a suitable representation and acting on a spinor field \(\psi \in \Gamma \left ( M, \Sigma M \right )\) via
\begin{subequations}
\begin{align}
	\nabla^{\Sigma M}_\mu \psi & := \partial_\mu \psi + \varpi_\mu \psi
\intertext{and on a corresponding cospinor field via}
	\nabla^{\Sigma M}_\mu \overline{\psi} & := \partial_\mu \overline{\psi} - \overline{\psi} \varpi_\mu \, ,
\end{align}
\end{subequations}
with the spinor bundle spin connection
\begin{equation}
	\varpi_\mu := - \frac{\imaginary}{4} \omega_\mu^{r s} \sigma_{r s} \, ,
\end{equation}
where \(\sigma_{r s} := \frac{\imaginary}{2} \left [ \gamma_r , \gamma_s \right ]\) is a representation of \(\mathfrak{spin} (1,3)\) on \(\Sigma M\). Thus, on the twisted spinor bundle \(U(1) \times_\rho \Sigma M\) we use the covariant derivative \(\nabla^{U(1) \times_\rho \Sigma M}_\mu\), induced via the connection on the \(U(1)\) principle bundle and the connection on the spinor bundle \(\Sigma M\) and acting on a twisted spinor field \(\Psi \in \Gamma \left ( M, U(1) \times_\rho \Sigma M \right )\) via\footnote{Being pedantic, \eqnsref{eqn:u1-sm-covariant_derivative} should actually read \(\nabla^{U(1) \times_\rho \Sigma M}_\mu \Psi = \partial_\mu \Psi + \left ( \id \times_\rho \varpi_\mu \right ) \Psi + \left ( \imaginary \ecoupling A_\mu \times_\rho \id \right ) \Psi\) and \(\nabla^{U(1) \times_\rho \Sigma M}_\mu \overline{\Psi} = \partial_\mu \overline{\Psi} - \overline{\Psi} \left ( \id \times_\rho \varpi_\mu \right ) - \overline{\Psi} \left ( \imaginary \ecoupling A_\mu \times_\rho \id \right )\).}
\begin{subequations} \label{eqn:u1-sm-covariant_derivative}
\begin{align}
	\nabla^{U(1) \times_\rho \Sigma M}_\mu \Psi & := \partial_\mu \Psi + \varpi_\mu \Psi + \imaginary \ecoupling A_\mu \Psi
\intertext{and on a corresponding twisted cospinor field via}
	\nabla^{U(1) \times_\rho \Sigma M}_\mu \overline{\Psi} & := \partial_\mu \overline{\Psi} - \overline{\Psi} \varpi_\mu - \imaginary \ecoupling A_\mu \overline{\Psi} \, .
\end{align}
\end{subequations}
Finally, on the parity shifted Lie algebra and dual Lie algebra bundles of the gauge groups \(\Pi \left ( TM \oplus T^*M \right )\) and \(\Pi \left ( \mathfrak{u}(1) \oplus \mathfrak{u}(1)^* \right )\) we use the connections induced via the connections on \(TM\) and \(U(1)\), respectively. We remark that the Levi-Civita connection is metric and torsion-free and the other connections are metric.
\end{defn}

\vspace{\baselineskip}

\begin{rem}[Tetrad postulate] \label{rem:tetrad_postulate}
Given the situation of \defnref{defn:vielbeins_inverse_vielbeins}, the tetrad postulate states that vielbeins and inverse vielbeins are parallel sections in \(\Gamma \left ( U, T^*M \otimes_\mathbb{R} E \right )\) and \(\Gamma \left ( U, TM \otimes_\mathbb{R} E^* \right )\), respectively, with respect to the corresponding tensor product bundle connection, cf.\ \defnref{defn:connections_spacetime-matter_bundle}, i.e.\ we have
\begin{subequations}
\begin{align}
	\nabla^{TM \otimes_\mathbb{R} E}_\mu e_\nu^n & = \partial_\mu e_\nu^n - \Gamma_{\mu \nu}^\lambda e_\lambda^n + \sigma_{\mu l}^n e_\nu^l \equiv 0
\intertext{and}
	\nabla^{TM \otimes_\mathbb{R} E}_\mu e^\nu_n & = \partial_\mu e^\nu_n + \Gamma_{\mu \lambda}^\nu e^\lambda_n - \sigma_{\mu n}^l e^\nu_l \equiv 0 \, .
\end{align}
\end{subequations}
In particular, this implies that it is irrelevant whether vielbeins and inverse vielbeins are placed before or after the covariant derivative \(\nabla^{TM \otimes_\mathbb{R} E}_\mu\) on the product bundle \(TM \otimes_\mathbb{R} E\). Finally, we remark that despite its name the tetrad postulate is not a postulate but a consequence of the definition of the connection on \(E\) via the connection on \(TM\) using vielbeins and inverse vielbeins.
\end{rem}

\vspace{\baselineskip}

\begin{defn}[Clifford multiplication] \label{defn:clifford_multiplication}
Given the spacetime-matter bundle \(\mathcal{S}\), we define the Clifford multiplication \(\gamma \in \Gamma \left ( M, E^* \otimes_\mathbb{R} \operatorname{End} \left ( \Sigma M \right ) \right )\) for vector fields \(Y \in \Gamma \left ( M, E \right )\) and spinor fields \(\psi \in \Gamma \left ( M, \Sigma M \right )\) via
\begin{align}
	Y \cdot \psi & := Y^m \gamma_m \psi \, .
\intertext{Furthermore, using vielbeins \(e \in \Gamma \left ( M, T^*M \otimes_\mathbb{R} E \right )\), we can define the Clifford multiplication \(\gamma \circ e \in \Gamma \left ( M, T^*M \otimes_\mathbb{R} \operatorname{End} \left ( \Sigma M \right ) \right )\) for vector fields \(X \in \Gamma \left ( M, TM \right )\) via}
	X \cdot \psi & := X^\mu e_\mu^m \gamma_m \psi \, .
\end{align}
Moreover, we extend this definition to the twisted spinor bundle \(U(1) \times_\rho \Sigma M\) via its action on \(\Sigma M\) and denote it for simplicity via the same symbol \(\gamma \in \Gamma \left ( M, E^* \otimes_\mathbb{R} \operatorname{End} \left ( U(1) \times_\rho \Sigma M \right ) \right )\). Then, we obtain for vector fields \(Y \in \Gamma \left ( M, E \right )\) and twisted spinor fields \(\Psi \in \Gamma \left ( M, U(1) \times_\rho \Sigma M \right )\)
\begin{align}
	Y \cdot \Psi & := Y^m \gamma_m \Psi \, .
\intertext{Additionally, using vielbeins \(e \in \Gamma \left ( M, T^*M \otimes_\mathbb{R} E \right )\), we can extend this definition to obtain \(\gamma \circ e \in \Gamma \left ( M, T^*M \otimes_\mathbb{R} \operatorname{End} \left ( U(1) \times_\rho \Sigma M \right ) \right )\) for vector fields \(X \in \Gamma \left ( M, TM \right )\)}
	X \cdot \Psi & := X^\mu e_\mu^m \gamma_m \Psi \, .
\end{align}
In the following, we typically use the latter definitions for the twisted spinor bundle.
\end{defn}

\vspace{\baselineskip}

\begin{rem}
We remark that the Clifford multiplication \(\gamma\) or \(\gamma \circ e\) turns the spaces of spinor fields \(\Gamma \left ( M, \Sigma M \right )\) and twisted spinor fields \(\Gamma \left ( M, U(1) \times_\rho \Sigma M \right )\) into modules over the space of vector fields \(\Gamma \left ( M, E \right )\) or \(\Gamma \left ( M, TM \right )\), respectively.\footnote{This is similar to the fact that all spaces of (super) vector fields are modules over spaces of functions, i.e.\ \(\Gamma \left ( M, TM \right )\), \(\Gamma \left ( M, E \right )\), \(\Gamma \left ( M, \Pi \left ( TM \oplus T^*M \right ) \right )\) and \(\Gamma \left ( M, \Pi \left ( \mathfrak{u}(1) \oplus \mathfrak{u}(1)^* \right ) \right )\) are modules over the space \(\Gamma \left ( M, \mathbb{R} \right )\) and \(\Gamma \left ( M, \Sigma M \right )\) and \(\Gamma \left ( M, U(1) \times_\rho \Sigma M \right )\) are modules over the space \(\Gamma \left ( M, \mathbb{C} \right )\).} Furthermore, it induces an automorphism if and only if the corresponding vector field \(Y \in \Gamma \left ( M, E \right )\) or \(X \in \Gamma \left ( M, TM \right )\) has nowhere vanishing seminorm with respect to the metric \(\eta_{m n}\) or \(g_{\mu \nu}\), respectively, i.e.\ \(\left \| Y \right \|_\eta := \sqrt{\left \vert \eta_{m n} Y^m Y^n \right \vert} \not \equiv 0\) or \(\left \| X \right \|_g := \sqrt{\vert g_{\mu \nu} X^\mu X^\nu \vert} \not \equiv 0\).\footnote{Again, this is similar to the fact that the multiplication of a (super) vector field with a function induces an automorphism if and only if the corresponding function \(f \in \Gamma \left ( M, \mathbb{K} \right )\) for \(\mathbb{K} \in \left \{ \mathbb{R}, \mathbb{C} \right \}\) has nowhere vanishing absolute value, i.e.\ \(\left | f \right | \not \equiv 0\). The only difference is that the Clifford multiplication depends also on the metric, which only induces a norm if the metric is positive or negative definite.}
\end{rem}

\vspace{\baselineskip}

\begin{defn}[Clifford relation] \label{defn:clifford_relation}
We set the Clifford relation for the vector bundle \(E\) as
\begin{align}
	\left \{ \gamma_m , \gamma_n \right \} & = 2 \eta_{m n} \id \, ,
\intertext{or equivalently, using vielbeins \(e \in \Gamma \left ( M, T^*M \otimes_\mathbb{R} E \right )\), for the tangent bundle \(TM\) as}
	e_\mu^m e_\nu^n \left \{ \gamma_m , \gamma_n \right \} & = 2 g_{\mu \nu} \id \, ,
\end{align}
where the identity automorphism \(\id\) is either considered on the spinor bundle \(\Sigma M\) or extended to the twisted spinor bundle \(U(1) \times_\rho \Sigma M\), cf.\ \defnref{defn:clifford_multiplication}.
\end{defn}

\vspace{\baselineskip}

\begin{rem}
We remark that the West coast (``mostly minus'') signature for the metrics \(g\) and \(\eta\) together with the ``plus signed'' Clifford relation induces a quaternionic representation for the Clifford algebra \(\operatorname{Cliff} \left ( 1, 3 \right )\) as the matrix algebra \(\operatorname{Mat} \left (2, \mathbb{H} \right )\). Choosing the Pauli matrices as a representation for the quaternions, we obtain the usual complex Dirac representation as the matrix algebra \(\operatorname{Mat} \left (4, \mathbb{C} \right )\), whose generators are Hermitian.
\end{rem}

\vspace{\baselineskip}

\begin{defn}[(Twisted) Dirac operator] \label{defn:dirac_operator}
Let \(\mathcal{S}\) be the spacetime-matter bundle from \defnref{defn:spacetime-matter_bundle}. Then we define the Dirac operator \(\slashed{\nabla}^{\Sigma M}\) on spinors \(\psi \in \Gamma \left ( M, \Sigma M \right )\) such that the following diagram commutes:
\begin{subequations}
\begin{equation}
\begin{tikzcd}[row sep=huge]
	\Gamma \left ( M, \Sigma M \right ) \arrow{r}{\slashed{\nabla}^{\Sigma M}} \arrow[swap]{d}{\nabla^{\Sigma M}_\cdot} & \Gamma \left ( M, \Sigma M \right )\\
\Gamma \left ( M, T^* M \otimes_\mathbb{R} \Sigma M \right ) \arrow[swap]{r}{g^{-1} \otimes_{\mathbb{R}} \id} & \Gamma \left ( M, T M \otimes_\mathbb{R} \Sigma M \right ) \arrow[swap]{u}{\gamma \circ e}
\end{tikzcd}
\end{equation}
Here, \(\nabla^{\Sigma M}_\mu = \partial_\mu + \varpi_\mu\) is the covariant derivative on the spinor bundle \(\Sigma M\), and \(e_\nu^n \gamma_n\) is the local representation of the Clifford-multiplication, with \(\gamma_n\) being the usual Minkowski spacetime Dirac matrices. Thus, the local description of the Dirac-operator on the spinor bundle \(\Sigma M\) is given via\footnote{Notice that by the tetrad postulate given in \remref{rem:tetrad_postulate} it does not matter whether we place the inverse vielbeins \(e^{\mu m}\) before or after the covariant derivative \(\nabla^{\Sigma M}_\mu\) (if we place it after, however, we need to consider the covariant derivative \(\nabla^{TM \otimes_\mathbb{R} E \times \Sigma M}_\mu\)).}
\begin{equation}
\begin{split}
	\slashed{\nabla}^{\Sigma M} & := e^{\mu m} \gamma_m \nabla^{\Sigma M}_{\mu} \\ & \phantom{:} = e^{\mu m} \gamma_m \left ( \partial_\mu + \varpi_\mu \right ) \, .
\end{split}
\end{equation}
\end{subequations}
Furthermore, we extend this definition to the twisted spinor bundle \(U(1) \times_\rho \Sigma M\) using the covariant derivative \(\nabla^{U(1) \times_\rho \Sigma M}_\mu\) instead of \(\nabla^{\Sigma M}_\mu\). Then we define the twisted Dirac operator \(\slashed{\nabla}^{U(1) \times_\rho \Sigma M}\) on twisted spinors \(\Psi \in \Gamma \left ( M, U(1) \times_\rho \Sigma M \right )\) such that the following diagram commutes:
\begin{subequations}
	\begin{equation}
	\begin{tikzcd}[row sep=huge]
	\Gamma \left ( M, U(1) \times_\rho \Sigma M \right ) \arrow{r}{\slashed{\nabla}^{U(1) \times_\rho \Sigma M}} \arrow[swap]{d}{\nabla^{U(1) \times_\rho \Sigma M}_\cdot} & \Gamma \left ( M, U(1) \times_\rho \Sigma M \right )\\
	\Gamma \left ( M, T^* M \otimes_\mathbb{R} \left ( U(1) \times_\rho \Sigma M \right ) \right ) \arrow[swap]{r}{g^{-1} \otimes_{\mathbb{R}} \id} & \Gamma \left ( M, T M \otimes_\mathbb{R} \left ( U(1) \times_\rho \Sigma M \right ) \right ) \arrow[swap]{u}{\gamma \circ e}
	\end{tikzcd}
	\end{equation}
Here, \(\nabla^{U(1) \times_\rho \Sigma M}_\mu = \partial_\mu + \varpi_\mu + \imaginary \ecoupling A_\mu\) is the covariant derivative on the twisted spinor bundle \(U(1) \times_\rho \Sigma M\), and the rest as above. Thus, the local description of the twisted Dirac-operator on the twisted spinor bundle \(U(1) \times_\rho \Sigma M\) is given via\footnote{Notice again that by the tetrad postulate given in \remref{rem:tetrad_postulate} it does not matter whether we place the inverse vielbeins \(e^{\mu m}\) before or after the covariant derivative \(\nabla^{U(1) \times_\rho \Sigma M}_\mu\) (if we place it after, however, we need to consider the covariant derivative \(\nabla^{TM \otimes_\mathbb{R} E \times U(1) \times_\rho \Sigma M}_\mu\)).}
	\begin{equation}
	\begin{split}
	\slashed{\nabla}^{U(1) \times_\rho \Sigma M} & := e^{\mu m} \gamma_m \nabla^{U(1) \times_\rho \Sigma M}_{\mu} \\ & \phantom{:} = e^{\mu m} \gamma_m \left ( \partial_\mu + \varpi_\mu + \imaginary \ecoupling A_\mu \right ) \, .
	\end{split}
	\end{equation}
\end{subequations}
\end{defn}

\vspace{\baselineskip}

\begin{defn}[Curvatures of the spacetime-matter bundle] \label{defn:curvatures_spacetime-matter_bundle}
Using the connections from \defnref{defn:connections_spacetime-matter_bundle} on the spacetime-matter bundle \(\mathcal{S}\), we can construct curvature tensors as commutators of the corresponding covariant derivatives. We start with the Riemann tensor of the tangent bundle, which acts on a vector field \(X \in \Gamma \left ( M, TM \right )\) via
\begin{subequations}
\begin{align}
	\tensor{R}{^\rho _\sigma _\mu _\nu} X^\sigma & := \left [ \nabla^{TM}_\mu , \nabla^{TM}_\nu \right ] X^\rho
\intertext{and reads}
	\tensor{R}{^\rho _\sigma _\mu _\nu} & = \partial_\mu \Gamma^\rho_{\nu \sigma} - \partial_\nu \Gamma^\rho_{\mu \sigma} + \Gamma^\rho_{\mu \lambda} \Gamma^\lambda_{\nu \sigma} - \Gamma^\rho_{\nu \lambda} \Gamma^\lambda_{\mu \sigma} \, .
\end{align}
\end{subequations}
Starting from the Riemann tensor, we can define the Ricci tensor as the following contraction
\begin{align}
	R_{\mu \nu} & := \tensor{R}{^\rho _\mu _\rho _\nu} \, ,
\intertext{and finally the Ricci scalar as}
	R & := g^{\mu \nu} R_{\mu \nu} \, .
\end{align}
Furthermore, we can construct the curvature two-form on the \(U(1)\) principle bundle, which acts on a section \(s \in \Gamma \left ( M, TM \right )\) via\footnote{We stress that \(F_{\mu \nu}\) is defined as the covariant exterior derivative of the connection form \(\imaginary \ecoupling A_\rho\), or for abelian gauge groups via the ordinary exterior derivative as given in \eqnref{eqn:fmunu}. Thus it has nothing to do with the connection on the tangent bundle. Nevertheless, for torsion free connections the expressions coincide, as the Christoffel symbols are then symmetric in their lower two indices, i.e.\ \(\Gamma_{\mu \nu}^\rho \equiv \Gamma_{\nu \mu}^\rho\). \label{foot:fmunu_exterior_derivative}}
\begin{subequations}
\begin{align}
	F_{\mu \nu} s & := \left [ \nabla^{U(1)}_\mu , \nabla^{U(1)}_\nu \right ] s
\intertext{and reads}
	F_{\mu \nu} & = \imaginary \ecoupling \left ( \partial_\mu A_\nu - \partial_\nu A_\mu \right ) \, . \label{eqn:fmunu}
\end{align}
\end{subequations}
Finally, we remark that being forms, the connection form \(\imaginary \ecoupling A_\mu\) and the curvature form \(F_{\mu \nu}\), both have naturally lower indices.
\end{defn}

\vspace{\baselineskip}

\begin{prop}[Ricci scalar for the Levi-Civita connection] \label{prop:ricci_scalar_for_the_levi_civita_connection}
Using the Levi-Civita connection, the Ricci scalar is given via partial derivatives of the metric and its inverse as follows:
\begin{equation} \label{eqn:ricci_scalar_metric}
\begin{split}
	R & = g^{\mu \rho} g^{\nu \sigma} \left ( \partial_\mu \partial_\nu g_{\rho \sigma} - \partial_\mu \partial_\rho g_{\nu \sigma} \right )\\
	& \phantom{ = } + g^{\mu \rho} g^{\nu \sigma} g^{\kappa \lambda} \left ( \left ( \partial_\mu g_{\kappa \lambda} \right ) \left ( \partial_\nu g_{\rho \sigma} - \frac{1}{4} \partial_\rho g_{\nu \sigma} \right ) + \left ( \partial_\nu g_{\rho \kappa} \right ) \left ( \frac{3}{4} \partial_\sigma g_{\mu \lambda} - \frac{1}{2} \partial_\mu g_{\sigma \lambda} \right ) \right .\\
	& \phantom{ \phantom{ = } + g^{\mu \rho} g^{\nu \sigma} g^{\kappa \lambda} \{ } \left . \vphantom{\left ( \frac{1}{2} \right )} - \left ( \partial_\mu g_{\rho \kappa} \right ) \left ( \partial_\nu g_{\sigma \lambda} \right ) \right )
\end{split}
\end{equation}
\end{prop}

\begin{proof}
The claim is verified by the calculation
\begin{equation}
\begin{split}
	R & = g^{\nu \sigma} \tensor{R}{^\mu _\sigma _\mu _\nu}\\
	& = g^{\nu \sigma} \left ( \partial_\mu \Gamma^\mu_{\nu \sigma} - \partial_\nu \Gamma^\mu_{\mu \sigma} + \Gamma^\mu_{\mu \kappa} \Gamma^\kappa_{\nu \sigma} - \Gamma^\mu_{\nu \kappa} \Gamma^\kappa_{\mu \sigma} \right )\\
	& = g^{\nu \sigma} \left ( \left ( \partial_\mu g^{\mu \rho} \right ) \left ( \partial_\nu g_{\rho \sigma} - \frac{1}{2} \partial_\rho g_{\nu \sigma} \right ) - \frac{1}{2} \left ( \partial_\nu g^{\mu \rho} \right ) \left ( \partial_\sigma g_{\mu \rho} \right ) + g^{\mu \rho} \left ( \partial_\mu \partial_\nu g_{\rho \sigma} - \partial_\mu \partial_\rho g_{\nu \sigma} \right ) \right ) \\
& \phantom{ = } + g^{\mu \rho} g^{\nu \sigma} g^{\kappa \lambda} \left ( \left ( \partial_\mu g_{\kappa \lambda} \right ) \left ( \frac{1}{2} \partial_\nu g_{\rho \sigma} - \frac{1}{4} \partial_\rho g_{\nu \sigma} \right ) + \left ( \partial_\nu g_{\rho \kappa} \right ) \left ( \frac{1}{4} \partial_\sigma g_{\mu \lambda} - \frac{1}{2} \partial_\mu g_{\sigma \lambda} \right ) \right )\\
	& = g^{\mu \rho} g^{\nu \sigma} \left ( \partial_\mu \partial_\nu g_{\rho \sigma} - \partial_\mu \partial_\rho g_{\nu \sigma} \right )\\
	& \phantom{ = } + g^{\mu \rho} g^{\nu \sigma} g^{\kappa \lambda} \left ( \left ( \partial_\mu g_{\kappa \lambda} \right ) \left ( \partial_\nu g_{\rho \sigma} - \frac{1}{4} \partial_\rho g_{\nu \sigma} \right ) + \left ( \partial_\nu g_{\rho \kappa} \right ) \left ( \frac{3}{4} \partial_\sigma g_{\mu \lambda} - \frac{1}{2} \partial_\mu g_{\sigma \lambda} \right ) \right .\\
	& \phantom{ \phantom{ = } + g^{\mu \rho} g^{\nu \sigma} g^{\kappa \lambda} ( } \left . \vphantom{\left ( \frac{1}{2} \right )} - \left ( \partial_\mu g_{\rho \kappa} \right ) \left ( \partial_\nu g_{\sigma \lambda} \right ) \right ) \, ,
\end{split}
\end{equation}
where we have used \(\left ( \partial_\rho g^{\nu \sigma} \right ) g_{\mu \sigma} = - g^{\nu \sigma} \left ( \partial_\rho g_{\mu \sigma} \right )\), which results from
\begin{equation}
\begin{split}
	0 & = \nabla^{TM}_\rho \delta_\mu^\nu\\
	& = \partial_\rho \delta_\mu^\nu + \Gamma_{\rho \sigma}^\nu \delta_\mu^\sigma - \Gamma_{\rho \mu}^\sigma \delta_\sigma^\nu\\
	& = \partial_\rho \delta_\mu^\nu + \Gamma_{\rho \mu}^\nu - \Gamma_{\rho \mu}^\nu\\
	& = \partial_\rho \delta_\mu^\nu\\
	& = \partial_\rho \left ( g_{\mu \sigma} g^{\nu \sigma} \right )\\
	& = \left ( \partial_\rho g_{\mu \sigma} \right ) g^{\nu \sigma} + g_{\mu \sigma} \left ( \partial_\rho g^{\nu \sigma} \right ) \, .
\end{split}
\end{equation}
\end{proof}

\vspace{\baselineskip}

\begin{rem}
Assuming the Levi-Civita connection, the Riemann tensor \(\tensor{R}{^\rho _\sigma _\mu _\nu}\) and the Ricci tensor \(R_{\mu \nu}\) are not sensitive to the choice of the signature of the metric \(g_{\mu \nu}\), whereas the lowered Riemann tensor \(R_{\rho \sigma \mu \nu} := g_{\rho \lambda} \tensor{R}{^\lambda _\sigma _\mu _\nu}\) and the Ricci scalar \(R\) are. Therefore, the Einstein-Hilbert Lagrange density is sensitive to this choice as well, which is the reason for the minus sign in front of the Einstein-Hilbert Lagrange density in \eqnref{eqn:lagrange_gr-ed} in our West coast (``mostly minus'') signature convention, cf.\ \defnref{def:spacetime}.
\end{rem}

\vspace{\baselineskip}

\begin{defn}[Riemannian volume form] \label{def:riemannian_volume_form}
We define the Riemannian volume form for the spacetime \((M,g)\) as
\begin{equation}
	\dif V_g := \sqrt{- \dt{g}} \dif t \wedge \dif x \wedge \dif y \wedge \dif z \, .
\end{equation}
In particular, the Riemannian volume form for the Minkowski spacetime \((\mathbb{M},\eta)\) is given via
\begin{equation}
	\dif V_\eta = \dif t \wedge \dif x \wedge \dif y \wedge \dif z \, .
\end{equation}
\end{defn}

\vspace{\baselineskip}

\begin{defn}[Fermion, photon and graviton field] \label{defn:electron-positron_photon_graviton_field}
The mathematical objects correspond in the following way to physical particles: The deviation of the metric \(g_{\mu \nu}\) from the Minkowski metric \(\eta_{\mu \nu}\) is defined to be proportional to the graviton field \(h_{\mu \nu}\), with the proportionality factor given by the gravitational coupling constant \(\gcoupling := \sqrt{8 \pi G}\) to obtain the right interaction couplings,\footnote{Where \(G\) is Newton's constant and we have furthermore \(\gcoupling = \sqrt{\kappa}\), where \(\kappa := 8 \pi G\) is Einstein's constant.} i.e.\
\begin{equation}
	h_{\mu \nu} := \frac{1}{\gcoupling} \left ( g_{\mu \nu} - \eta_{\mu \nu} \right ) \iff g_{\mu \nu} \equiv \eta_{\mu \nu} + \gcoupling h_{\mu \nu} \, .
\end{equation}
The graviton field can be thought of as a symmetric \((0,2)\)-tensor field living on the flat background Minkowski spacetime \((\mathbb{M},\eta)\), i.e.\ \(h \in \Gamma \left ( \mathbb{M}, \operatorname{Sym}^2 \left ( T^*M \right ) \right )\). Moreover, the connection form \(\imaginary \ecoupling A_\mu\) is defined to be proportional to the photon field \(A_\mu\), with the proportionality factor given by the imaginary unit \(\imaginary\) and the electromagnetic coupling constant \(\ecoupling\), such that the photon field is real and induces the right interaction couplings with fermions. It induces the Faraday or electromagnetic field strength tensor \(F_{\mu \nu}\), cf.\ \defnref{defn:curvatures_spacetime-matter_bundle}. Moreover, a section in the spinor bundle \(\psi \in \Gamma \left ( M, \Sigma M \right )\), corresponds to a fermion field. In particular, in our case it is a linear combination of the electron and the positron field with charge \(\pm \ecoupling\), respectively. In the following, we are interested in the coupling of the fermion field \(\psi\) to the photon field \(A_\mu\). Mathematically, this is obtained by viewing fermions as sections in the twisted spinor bundle \(\Psi \in \Gamma \left ( M, U(1) \times_\rho \Sigma M \right )\). Therefore, in the following we use twisted spinor fields, i.e.\ sections in the twisted bundle \(U(1) \times_\rho \Sigma M\), if not stated otherwise, which are denoted by \(\Psi\), compared to \(\psi\), which denote sections in \(\Sigma M\). Spinors and twisted spinors induce the four-current \(j^\mu\), defined via
\begin{equation}
	j^\mu := e^{\mu m} \overline{\psi} \gamma_m \psi \equiv e^{\mu m} \overline{\Psi} \gamma_m \Psi \, .
\end{equation}
Finally, the ghosts are sections in the parity shifted Lie algebra and dual Lie algebra bundles of the gauge groups \(\Pi \left ( TM \oplus T^*M \right )\) and \(\Pi \left ( \mathfrak{u}(1) \oplus \mathfrak{u}(1)^* \right )\): We have the graviton ghost \(\gravitonghost \in \Gamma \left ( M, \Pi \left ( TM \right ) \right )\), the graviton antighost \(\overline{\gravitonghost} \in \Gamma \left ( M, \Pi \left ( T^*M \right ) \right )\), the photon ghost \(\photonghost \in \Gamma \left ( M, \Pi \left ( \mathfrak{u}(1) \right ) \right )\) and the photon antighost \(\overline{\photonghost} \in \Gamma \left ( M, \Pi \left ( \mathfrak{u}(1)^* \right ) \right )\).
\end{defn}

\vspace{\baselineskip}

\begin{rem}[Physical interpretation of the connections on the spacetime-matter bundle] \label{rem:physical_interpretation_connections}
We stress that even though the connections on the spacetime-matter bundle are mathematically similarly defined, cf.\ \defnref{defn:connections_spacetime-matter_bundle}, their physical interpretation is rather different, cf.\ \defnref{defn:electron-positron_photon_graviton_field}: The Christoffel symbols \(\Gamma_{\mu \lambda}^\nu\) and the spin connection \(\omega_{\mu l}^n\) are proportional to a series in the graviton field, i.e.\ to a sum of arbitrary many particles. Contrary, the connection form \(\imaginary \ecoupling A_\mu\) on the \(U(1)\) principle bundle corresponds directly to the photon field, i.e.\ to a single particle.
\end{rem}

\vspace{\baselineskip}

\begin{rem}[Gauge transformations] \label{rem:gauge_transformations}
The Lagrange density of GR-ED, \eqnref{eqn:lagrange_gr-ed}, is invariant under two different symmetry transformations, in the following called gauge transformations:\footnote{We remark that the invariance is only strict for the bundle automorphism \(\varphi \in \autu\), because diffeomorphisms homotopic to the identity \(\phi \in \diff\) add a total derivative to the Einstein-Hilbert Lagrange density.} The first are diffeomorphisms of spacetime \((M,g)\) homotopic to the identity \(\phi \in \diff\) and the second are \(U(1)\) principle bundle automorphisms \(\varphi \in \autu\). While the first affects the whole spacetime-matter bundle \(\mathcal{S}\), the second affects only the \(U(1)\) principle bundle together with the spinor bundle \(\Sigma M\), due to the action \(\rho\). We remark that in contrast to finite dimensional Lie groups, the exponential map
\begin{equation}
	\operatorname{Exp} \, : \; \mathfrak{g} \to G
\end{equation}
is not locally surjective for infinite dimensional Lie groups, such as \(\diff\); rather they are regular Lie groups in the sense that there exists a bijective evolution operator from smooth curves in the Lie algebra to smooth curves in the Lie group
\begin{equation}
	\operatorname{Evol} \, : \; \Gamma \left ( \mathcal{I}, \mathfrak{g} \right ) \to \Gamma \left ( \mathcal{I}, G \right ) \, ,
\end{equation}
where \(\mathcal{I} \subset \mathbb{R}\) is a compact interval, cf.\ \cite{Milnor,Grabowski,Kriegl_Michor,Schmeding}. Concretely, we have
\begin{subequations}
\begin{align}
	\mathfrak{diff} \left ( M \right ) & \cong \Gamma_\text{c} \left ( M, TM \right ) \, ,
	\intertext{where \(\Gamma_\text{c} \left ( M, TM \right )\) denotes compactly supported vector fields, and thus}
	\Gamma \left ( \mathcal{I}, \mathfrak{diff} \left ( M \right ) \right ) & \cong \Gamma_\text{c} \left ( \mathcal{I} \times M, TM \right ) \, ,
\end{align}
\end{subequations}
i.e.\ compactly supported one-parameter vector fields \(X \left ( \tau \right )\) for \(\tau \in \mathcal{I}\),\footnote{We remark that diffeomorphisms homotopic to the identity, \(\phi \in \diff\), differ from the identity also only on compactly supported domains.} and
\begin{equation}
	\mathfrak{u}(1) \cong \imaginary \mathbb{R} \, .
\end{equation}
Thus, we can associate elements in the corresponding Lie algebras with elements in the corresponding Lie groups. This will be useful when we consider infinitesimal gauge transformations in order to fix the gauge and calculate the Faddeev-Popov ghosts in \sssecaref{sssec:gauge_fixing_lagrange_density}{subsubsec:ghost_lagrange_density}, respectively. Using the notations
\begin{equation}
	\phixtt := \operatorname{Evol} \left ( X \left ( \tau \right ) \right ) \left ( \tau \right ) \, : \; \Gamma \left ( \mathcal{I}, \mathfrak{diff} \left ( M \right ) \right ) \to \Gamma \left ( \mathcal{I}, \operatorname{Diff}_0 \left ( M \right ) \right )
\end{equation}
with \(\tau \in \mathcal{I}\) and \(X \left ( \tau \right ) \in \Gamma_\text{c} \left ( \mathcal{I} \times M, TM \right )\), and
\begin{equation}
	\vphif := \operatorname{Exp} \left ( \mathfrak{f} \right ) \, : \; \mathfrak{u}(1) \to U(1)
\end{equation}
with \(\mathfrak{f} \in \Gamma \left ( M, \mathfrak{u}(1) \right )\), the transformations of sections in the spacetime-matter bundle \(\Gamma \left ( M, \mathcal{S} \right )\) with respect to \(\phixtt\) and \(\vphif\) are given via corresponding pullbacks and pushforewards.\footnote{In the following we will not make a distinction between pullbacks and pushforewards, as for diffeomorphisms homotopic to the identity \(\phi \in \diff\) the pushforeward is the pullback of the inverse, i.e.\ \(\phi_* = \big ( \phi^{-1} \big )^*\), and vice versa. This applies in particular also to bundle automorphisms \(\vphif \in \autu\). Furthermore, when regarding evolutions we need to take the sign of the parameter \(\tau\) into account, as \(\big ( \phixtt \big )^{-1} \equiv \phi_{X ( \tau )}^{-\tau}\), i.e.\ contravariant indices transform via \(\big ( \phixtt \big )_* \equiv \big ( \phi_{X ( \tau )}^{- \tau} \big )^*\), whereas covariant indices transform via \(\big ( \phixtt \big )^* \equiv \big ( \phi_{X ( \tau )}^{- \tau} \big )_*\). \label{ftnt:minus_sign}} Furthermore, we also remark the identities
\begin{align}
	\eval{\frac{\dif}{\dif \tau} \left ( \phixtt \right )}_{\tau = 0} & = X(0) \label{eqn:differential_evol-operator}
\intertext{and}
	\eval{\frac{\dif}{\dif \tau} \left ( \varphi_{\tau \mathfrak{f}} \right )}_{\tau = 0} & = \mathfrak{f} \, .
\end{align}
Moreover, given an arbitrary tensor field of type \((r,s)\), \(\mathcal{T} \in \Gamma \left ( M, T^r_s M \right )\), such that the tensor field \(\left ( \phixtt \right )^* \mathcal{T}\) is representable via a Taylor series around \(\tau = 0\), we can consider the Taylor expansion of its pullback \(\left ( \phixtt \right )^* \mathcal{T}\) in the parameter \(\tau\) around \(\tau = 0\), which converges then for some \(\tau \in \mathcal{J} \subset \mathcal{I}\) open,
\begin{subequations}
\begin{equation}
\begin{split}
	\left ( \phixtt \right )^* \mathcal{T} & = \sum_{k = 0}^\infty \frac{1}{k!} \eval{\left ( \frac{\mathrm{d}^k}{\dif \sigma^k} \left ( \left ( \phi_{X \left ( \sigma \right )}^\sigma \right )^* \mathcal{T} \right ) \right )}_{\sigma = 0} \tau^k \\
	& = \sum_{k = 0}^\infty \frac{1}{k!} \left ( \Lie_{X \left ( \tau \right )}^k \mathcal{T} \right ) \tau^k \, ,
\end{split} \label{eqn:taylor_expansion_flow_vector_field}
\end{equation}
where we have set
\begin{equation}
	\left ( \Lie_{X \left ( \tau \right )}^k \mathcal{T} \right )_p := \eval{\left ( \frac{\mathrm{d}^k}{\dif \sigma^k} \left ( \left ( \phi_{X \left ( \sigma \right )}^\sigma \right )^* \mathcal{T}_p \right ) \right )}_{\sigma = 0} \label{eqn:taylor_coefficients_flow_vector_field}
\end{equation}
\end{subequations}
for all \(p \in M\). We remark that \(\Lie_{X \left ( \tau \right )}^0 \mathcal{T} \equiv \mathcal{T}\) is the identity and \(\Lie_{X \left ( \tau \right )}^1 \mathcal{T} \equiv \Lie_{X \left ( 0 \right )} \mathcal{T}\) is the Lie derivative with respect to the constant vector field \(X \left ( 0 \right )\).
\end{rem}

\vspace{\baselineskip}

\begin{lem}
Given the situation of \remref{rem:gauge_transformations} and a constant vector field \(X \left ( \tau \right ) \equiv X\),\footnote{In general, the formulas will also contain derivatives of the vector field \(X \left ( \tau \right )\) in the parameter \(\tau\).} then the coefficients of \eqnsaref{eqn:taylor_expansion_flow_vector_field}{eqn:taylor_coefficients_flow_vector_field} are given via iterated Lie derivatives, i.e.\
\begin{equation}
	\left ( \Lie_{X}^{k+1} \mathcal{T} \right )_p \equiv \left ( \Lie_{X} \left ( \Lie_{X}^k \mathcal{T} \right ) \right )_p \, .
\end{equation}
\end{lem}

\begin{proof}
This follows from the following calculation:
\begin{equation}
\begin{split}
	\left ( \Lie_{X}^{k+1} \mathcal{T} \right )_p & = \eval{\frac{\mathrm{d}^{k+1}}{\dif \sigma^{k+1}} \left ( \left ( \phi_{X}^{\sigma} \right )^* \mathcal{T}_p \right )}_{\sigma = 0}\\
	& = \eval{\frac{\mathrm{d}}{\dif \sigma_1} \frac{\mathrm{d}^k}{\dif \sigma_2^k} \left ( \left ( \phi_{X}^{{\sigma_1} + {\sigma_2}} \right )^* \mathcal{T}_p \right )}_{\substack{{\sigma_1} = 0\\{\sigma_2} = 0}}\\
	& = \eval{\frac{\mathrm{d}}{\dif \sigma_1} \frac{\mathrm{d}^k}{\dif \sigma_2^k} \left ( \left ( \phi_{X}^{\sigma_1} \right )^* \left ( \phi_{X}^{\sigma_2} \right )^* \mathcal{T}_p \right )}_{\substack{{\sigma_1} = 0\\{\sigma_2} = 0}}\\
	& = \eval{\frac{\mathrm{d}}{\dif \sigma_1} \frac{\mathrm{d}^k}{\dif \sigma_2^k} \left ( \left ( \phi_{X}^{\sigma_1} \right )^* \left ( \left ( \phi_{X}^{\sigma_2} \right )^* \mathcal{T}_p \right )_p \right )}_{\substack{{\sigma_1} = 0\\{\sigma_2} = 0}}\\
	& = \eval{\frac{\mathrm{d}}{\dif \sigma_1} \left ( \eval{\frac{\mathrm{d}^k}{\dif \sigma_2^k} \left ( \left ( \phi_{X}^{\sigma_1} \right )^* \left ( \left ( \phi_{X}^{\sigma_2} \right )^* \mathcal{T}_p \right )_p \right )}_{{\sigma_2} = 0} \right )}_{{\sigma_1} = 0}\\
	& = \eval{\frac{\mathrm{d}}{\dif \sigma_1} \left ( \left ( \phi_{X}^{\sigma_1} \right )^* \left ( \eval{\frac{\mathrm{d}^k}{\dif \sigma_2^k} \left ( \left ( \phi_{X}^{\sigma_2} \right )^* \mathcal{T}_p \right )}_{{\sigma_2} = 0} \right )_p \right )}_{{\sigma_1} = 0}\\
	& = \eval{\frac{\mathrm{d}}{\dif \sigma_1} \left ( \left ( \phi_{X}^{\sigma_1} \right )^* \left ( \Lie_{X}^k \mathcal{T} \right )_p \right )}_{{\sigma_1} = 0}\\
	& = \left ( \Lie_{X} \left ( \Lie_{X}^k \mathcal{T} \right ) \right )_p
\end{split}
\end{equation}
\end{proof}

\vspace{\baselineskip}

\begin{rem}[Infinitesimal transformation properties] \label{rem:transformation_properties}
Given the situation of \remref{rem:gauge_transformations}, we are now interested in infinitesimal transformation properties of the graviton field and the photon field. This is needed in order to fix the gauge and calculate the Faddeev-Popov ghosts in \sssecaref{sssec:gauge_fixing_lagrange_density}{subsubsec:ghost_lagrange_density}, respectively. To this end we can restrict to diffeomorphisms that are given as the flow of a constant vector field \(X \left ( \tau \right ) \equiv X\), i.e.\ \(\phixt \in \diff\), as they generate the corresponding Lie algebra \(\mathfrak{diff} \left ( M \right )\), cf.\ \eqnref{eqn:differential_evol-operator}. To obtain the right interaction behaviour for the corresponding ghosts with the other fields, we rescale the vector field \(X \in \Gamma \left ( M, TM \right )\) by \(\gcoupling\) and the Lie algebra section \(\mathfrak{f} \in \Gamma \left ( M, \mathfrak{u}(1) \right )\) by \(\imaginary \ecoupling\),\footnote{Given that \(\mathfrak{u}(1) \cong \imaginary \mathbb{R}\), the convention is such that the connection form \(\imaginary \ecoupling A \in \Gamma \left ( M, T^*M \otimes_\mathbb{R} \mathfrak{u}(1) \right )\) is purely imaginary and thus \(A \in \Gamma \left ( M, T^*M \otimes_\mathbb{R} \mathbb{R} \right ) \cong \Gamma \left ( M, T^*M \right )\) is real (and in particular Hermitian, as is the case when this definition is used for non-abelian gauge theories). Thus, in particular \(\mathfrak{f} \in \Gamma \left ( M, \mathfrak{u}(1) \right ) \cong \Gamma \left ( M, \imaginary \mathbb{R} \right )\) with \(f \in \Gamma \left ( M, \mathbb{R} \right )\).} i.e.\ set
\begin{align}
	\mathcal{X}^\mu := \frac{1}{\gcoupling} X^\mu & \iff X^\mu \equiv \gcoupling \mathcal{X}^\mu
	\intertext{and}
	f := - \frac{\imaginary}{\ecoupling} \mathfrak{f} & \iff \mathfrak{f} \equiv \imaginary \ecoupling f \, .
\end{align}
Then, the infinitesimal transformations of the graviton field \(h_{\mu \nu}\) and the photon field \(A_\mu\) with respect to \(\phixt\) are given via Lie derivatives with respect to \(X^\mu\) and read (\(\order{\cdot}\) denotes the Landau symbol):
\begin{equation}
	\left ( \phixt \right )_* h_{\mu \nu} = h_{\mu \nu} + \left ( \partial_\mu \mathcal{X}_\nu + \partial_\nu \mathcal{X}_\mu - 2 \Gamma_{\mu \nu}^\lambda \mathcal{X}_\lambda \right ) \tau + \order{\tau^2} \label{eqn:lie_derivative_metric}
\end{equation}
and
\begin{equation}
	\left ( \phixt \right )_* A_\mu = A_\mu + \gcoupling \left ( \left ( \partial_\mu g^{\rho \sigma} \right ) A_\rho \mathcal{X}_\sigma + g^{\rho \sigma} \left ( \partial_\rho A_\mu \right ) \mathcal{X}_\sigma + g^{\rho \sigma} A_\rho \left ( \partial_\mu \mathcal{X}_\sigma \right ) \right ) \tau + \order{\tau^2} \, . \label{eqn:lie_derivative_connection_form}
\end{equation}
Furthermore, the photon field \(A_\mu\) transforms in the following way under the principle bundle automorphism \(\vphif \in \autu\)
\begin{equation}
	\left ( \vphif \right )_* A_\mu = A_\mu + \partial_\mu f \, , \label{eqn:gauge_transformation_connection_form}
\end{equation}
where \(f \in \Gamma \left ( M, \mathbb{R} \right )\) is the real-valued function on the spacetime \((M,g)\) associated to the bundle automorphism \(\vphif\). Finally, we remark that only the linearized Riemann tensor is invariant under the gauge transformation of \eqnref{eqn:lie_derivative_metric}. This is similar to non-abelian gauge theories where the field strength tensor is also not invariant under gauge transformations. However, being a Lorentz scalar, the Ricci scalar is invariant.
\end{rem}

\vspace{\baselineskip}

\begin{rem}[Feynman rules 1] \label{rem:Feynman_rules_1}
To calculate the corresponding gravity-matter Feynman rules, we need to express the inverse metric, vielbeins and inverse vielbeins and the prefactor of the Riemannian volume form in terms of the graviton field. Since we need only general properties of the Feynman rules in this work, we just motivate their derivation here and refer for a detailed treatment to \cite{Prinz_4}, cf.\ also \cite{Choi_Shim_Song,Hamber}.\footnote{Where we refer again to \ftnref{foot:choi_shim_song_gh-manifold} for a comment on the spinor Feynman rules in \cite{Choi_Shim_Song}.} Furthermore, we remark that all following series converge if and only if \(\left | \gcoupling \right | \left \| h \right \|_{\max} := \left | \gcoupling \right | \max_{\lambda \in \operatorname{EW} \left ( h \right )} \left | \lambda \right | < 1\), where \(\operatorname{EW} \left ( h \right )\) denotes the set of eigenvalues of \(h\). The inverse metric in terms of the graviton field is given by the corresponding Neumann series, i.e.\
\begin{equation}
	g^{\mu \nu} = \sum_{k = 0}^\infty \left ( - \gcoupling \right )^k \left ( h^k \right )^{\mu \nu} \, ,
\end{equation}
where
\begin{subequations}
\begin{align}
	h^{\mu \nu} & := \eta^{\mu \rho} \eta^{\nu \sigma} h_{\rho \sigma} \, , \\
	\left ( h^0 \vphantom{h^k} \right )^{\mu \nu} & := \eta^{\mu \nu} \\
\intertext{and}
	\left ( h^k \right )^{\mu \nu} & := \underbrace{h^\mu_{\kappa_1} h^{\kappa_1}_{\kappa_2} \cdots h^{\kappa_{k-1} \nu}}_{\text{\(k\)-times}} \, , \; k \in \mathbb{N} \, .
\end{align}
\end{subequations}
Furthermore, the expressions for vielbeins and inverse vielbeins read
\begin{subequations}
\begin{align}
	e_\mu^m = \sum_{k = 0}^\infty \gcoupling^k \binom{\frac{1}{2}}{k} \left ( h^k \right )_\mu^m \, ,
\intertext{with \(h_\mu^m := \eta^{m \nu} h_{\mu \nu}\), and}
	e_m^\mu = \sum_{k = 0}^\infty \gcoupling^k \binom{- \frac{1}{2}}{k} \left ( h^k \right )_m^\mu \, ,
\end{align}
\end{subequations}
with \(h^\mu_m := \eta^{\mu \nu} \delta^\rho_m h_{\nu \rho}\). Moreover, the prefactor of the Riemannian volume form,
\begin{equation}
\begin{split}
	\sqrt{- \dt{g}} & = \sqrt{- \dt{\eta + \gcoupling h}} \\
	& = \sqrt{- \dt{\eta} \dt{\delta + \gcoupling \eta^{-1} h}} \\
	& = \sqrt{\dt{\delta + \gcoupling \eta h}} \, ,
\end{split}
\end{equation}
can be obtained by first expressing the determinant in terms of traces and then plug the result into the Taylor series expansion of the square-root. More precisely, the determinant of a \(4 \times 4\) matrix\footnote{This works for any \(d \times d\) matrix with the obvious generalizations of the following formulas.} \(\mathbf{M}\) can be expressed via Newton's identities in terms of its trace as
\begin{equation}
	\dt{\mathbf{M}} = \frac{1}{4!} \operatorname{Det} \begin{pmatrix} \tr{\mathbf{M}} & 1 & 0 & 0 \\ \tr{\mathbf{M}^2} & \tr{\mathbf{M}} & 2 & 0 \\ \tr{\mathbf{M}^3} & \tr{\mathbf{M}^2} & \tr{\mathbf{M}} & 3 \\ \tr{\mathbf{M}^4} & \tr{\mathbf{M}^3} & \tr{\mathbf{M}^2} & \tr{\mathbf{M}} \end{pmatrix} \, .
\end{equation}
Now, we set
\begin{equation}
	\mathbf{M} := \delta + \gcoupling \eta h \, .
\end{equation}
Since the trace is linear, we have
\begin{equation}
\begin{split}
	\tr{\delta + \gcoupling \eta h} & = \tr{\delta} + \gcoupling \tr{\eta h}\\
	& = 4 + \gcoupling \eta^{\mu \nu} h_{\mu \nu} \, ,
\end{split}
\end{equation}
and similar expressions for the higher powers \(\tr{\left ( \delta + \gcoupling \eta h \right )^k}\) for \(k \in \left \{ 2, 3, 4 \right \}\). Thus, the determinant is a polynomial in traces of powers of \(\gcoupling \eta h\), i.e.\
\begin{equation}
	\dt{\delta + \gcoupling \eta h} \in \mathbb{R} \left [ \gcoupling \tr{\eta h} , \gcoupling^2 \tr{\left ( \eta h \right ) ^2} , \gcoupling^3 \tr{\left ( \eta h \right ) ^3} , \gcoupling^4 \tr{\left ( \eta h \right ) ^4} \right ] \, .
\end{equation}
We separate the constant term, which is \(1\), and summarize the non-constant terms as the polynomial \(\mathbf{P} \left ( \gcoupling \eta h \right )\), i.e.\
\begin{equation}
	\mathbf{P} \left ( \gcoupling \eta h \right ) := \dt{\delta + \gcoupling \eta h} - 1 \iff \dt{\delta + \gcoupling \eta h} \equiv \mathbf{P} \left ( \gcoupling \eta h \right ) + 1 \, .
\end{equation}
Then, we plug this expression into the Taylor series expansion of the square root around \(1\), i.e.\
\begin{equation}
	\sqrt{x + 1} = \sum_{k = 0}^\infty \binom{\frac{1}{2}}{k} x^k \, ,
\end{equation}
which converges for \(\left \vert x \right \vert < 1\). Finally, we to obtain
\begin{equation}
\begin{split}
	\sqrt{- \dt{g}}
	& = \sqrt{\dt{\delta + \gcoupling \eta h}}\\
	& = \sqrt{\mathbf{P} \left ( \gcoupling \eta h \right ) + 1}\\
	& = \sum_{k = 0}^\infty \binom{\frac{1}{2}}{k} \left ( \mathbf{P} \left ( \gcoupling \eta h \right ) \right )^k \, .
\end{split}
\end{equation}
However, for the realm of this article it suffices to know that the inverse metric, vielbeins, inverse vielbeins and the prefactor of the Riemannian volume form are all power series in the graviton field \(h_{\mu \nu}\).
\end{rem}

\vspace{\baselineskip}

\begin{rem}[Feynman rules 2] \label{rem:Feynman_rules_order_2}
In this article we consider two-loop propagator Feynman graphs and one-loop propagator and three-point Feynman graphs. Therefore, it is actually sufficient to consider the Lagrange density \(\mathcal{L}_{\text{QGR-QED}}\) only up to order \(\gcoupling^2\), since the higher valent graviton vertices would only contribute to graphs with self-loops. Thus, the formulas from \remref{rem:Feynman_rules_1} read as follows (\(\order{\cdot}\) denotes the Landau symbol): The inverse metric is given by
\begin{equation}
	g^{\mu \nu} = \eta^{\mu \nu} - \gcoupling \eta^{\mu \rho} \eta^{\nu \sigma} h_{\rho \sigma} + \gcoupling^2 \eta^{\mu \rho} \eta^{\sigma \kappa} \eta^{\tau \nu} h_{\rho \sigma} h_{\kappa \tau} + \order{\gcoupling^3} \, .
\end{equation}
Furthermore, vielbeins and inverse vielbeins, defined in \defnref{defn:vielbeins_inverse_vielbeins}, are given by
\begin{subequations}
\begin{align}
	e^m_\mu & = \delta^m_\mu + \frac{1}{2} \gcoupling \eta^{\nu m} h_{\mu \nu} - \frac{1}{8} \gcoupling^2 \eta^{\rho \sigma} \eta^{\kappa m} h_{\mu \rho} h_{\sigma \kappa} + \order{\gcoupling^3}
\intertext{and}
	e^\mu_m & = \delta^\mu_m - \frac{1}{2} \gcoupling \eta^{\mu \nu} h_{\nu m} + \frac{3}{8} \gcoupling^2 \eta^{\mu \rho} \eta^{\sigma \kappa} h_{\rho \sigma} h_{\kappa m} + \order{\gcoupling^3} \, .
\end{align}
\end{subequations}
And finally, the prefactor of the Riemannian volume form, defined in \defnref{def:riemannian_volume_form}, is given by
\begin{equation}
	\sqrt{- \dt{g}} = 1 + \frac{1}{2} \gcoupling \eta^{\mu \nu} h_{\mu \nu} + \frac{1}{8} \gcoupling^2 \left ( \eta^{\mu \nu} \eta^{\rho \sigma} h_{\mu \nu} h_{\rho \sigma} - 2 \eta^{\mu \sigma} \eta^{\nu \rho} h_{\mu \nu} h_{\rho \sigma} \right ) + \order{\gcoupling^3} \, .
\end{equation}
\end{rem}

\subsection{The Lagrange density of QGR-QED} \label{ssec:Lagrange_Density_of_QGR-QED}
Having introduced the necessary differential geometric background in \ssecref{ssec:differential_geometric_notions}, we can now turn our attention to physics: In this article we consider QGR-QED, described via the following Lagrange density as a functional on the spacetime-matter bundle \(\mathcal{S}\):
\begin{subequations} \label{eqn:lagrange_qgr-qed_complete}
\begin{align}
	\mathcal{L}_{\text{QGR-QED}} & = \mathcal{L}_{\text{GR-ED}} + \mathcal{L}_{\text{GF}} + \mathcal{L}_{\text{Ghost}} \, , \label{eqn:lagrange_qgr-qed}
\intertext{where \(\mathcal{L}_{\text{GR-ED}}\) is the classical Lagrange density of General Relativity (GR) coupled with Spinor Electrodynamics (ED)}
	\mathcal{L}_\text{GR-ED} & = \left ( \leh + \lmd \right ) \dif V_g \, , \label{eqn:lagrange_gr-ed}
\intertext{\(\mathcal{L}_{\text{GF}}\) is the gauge fixing Lagrange density}
	\mathcal{L}_{\text{GF}} & = - \left ( \frac{1}{4 \gcoupling^2 \zeta} g_{\mu \nu} \deDonder^\mu \deDonder^\nu + \frac{1}{2 \ecoupling^2 \xi} \Lorenz^2 \right ) \dif V_g \, , \label{eqn:lagrange_gauge-fixing}
\intertext{with \(\deDonder^\mu := g^{\rho \sigma} \Gamma^\mu_{\rho \sigma}\) and \(\Lorenz := g^{\mu \nu} \nabla^{TM}_\mu \imaginary \ecoupling A_\nu\), and \(\mathcal{L}_{\text{Ghost}}\) is the ghost Lagrange density}
\begin{split}
\mathcal{L}_{\text{Ghost}} & = - \left ( \frac{1}{2} \overline{\gravitonghost}_\mu g^{\mu \nu} g^{\rho \sigma} \left ( \partial_\rho \partial_\sigma \gravitonghost_\nu + \partial_\nu \left ( \Gamma_{\rho \sigma}^\lambda \gravitonghost_\lambda \right ) - \partial_\rho \left ( \Gamma_{\sigma \nu}^\lambda \gravitonghost_\lambda \right ) - \partial_\sigma \left ( \Gamma_{\nu \rho}^\lambda \gravitonghost_\lambda \right ) \right ) \right . \\
& \phantom{= - (} \! \! \! \! \! \! \! \! \! \! \! \! + \frac{1}{2} \gcoupling \overline{\photonghost} \left ( g^{\mu \nu} \partial_\mu \left ( \left ( \partial_\nu g^{\rho \sigma} \right ) A_\rho \gravitonghost_\sigma + g^{\rho \sigma} \left ( \partial_\rho A_\nu \right ) \gravitonghost_\sigma + g^{\rho \sigma} A_\rho \left ( \partial_\nu \gravitonghost_\sigma \right ) \right ) \right . \\
& \phantom{= - ( + \frac{1}{2} \gcoupling \overline{\photonghost} (} \! \! \! \! \! \! \! \! \! \! \! \! \left . - g^{\mu \nu} g^{\rho \sigma} \Gamma_{\mu \nu}^\lambda \left ( \left ( \partial_\lambda g^{\rho \sigma} \right ) A_\rho \gravitonghost_\sigma + g^{\rho \sigma} \left ( \partial_\rho A_\lambda \right ) \gravitonghost_\sigma + g^{\rho \sigma} A_\rho \left ( \partial_\lambda \gravitonghost_\sigma \right ) \right ) \right ) + \text{g.c.} \\
& \phantom{= - (} \! \! \! \! \! \! \! \! \! \! \! \! \left . + \overline{\photonghost} g^{\mu \nu} \left ( \partial_\mu \left ( \partial_\nu \photonghost \right ) - \Gamma_{\mu \nu}^\lambda \left ( \partial_\lambda \photonghost \right ) \right ) \vphantom{\frac{1}{2} g^{\mu \nu} g^{\rho \sigma} \overline{\gravitonghost}_\mu \left ( \partial_\rho \partial_\sigma \gravitonghost_\nu + \partial_\nu \left ( \Gamma_{\rho \sigma}^\lambda \gravitonghost_\lambda \right ) - \partial_\rho \left ( \Gamma_{\sigma \nu}^\lambda \gravitonghost_\lambda \right ) - \partial_\sigma \left ( \Gamma_{\nu \rho}^\lambda \gravitonghost_\lambda \right ) \right )} \right ) \dif V_g \, ,
\end{split} \label{eqn:lagrange_ghost}
\end{align}
\end{subequations}
with the graviton ghost and graviton antighost \(\gravitonghost_\mu\) and \(\overline{\gravitonghost}_\mu\), respectively, and the photon ghost and photon antighost \(\photonghost\) and \(\overline{\photonghost}\), respectively, and g.c.\ denotes `ghost conjugate', as explained below. We compute the ghost Lagrange density, \eqnref{eqn:lagrange_ghost}, via the Faddeev-Popov method \cite{Faddeev_Popov} and remark that this can be embedded into the more general frameworks of BRST cohomology \cite{Becchi_Rouet_Stora_1,Becchi_Rouet_Stora_2,Becchi_Rouet_Stora_3,Tyutin,Kugo_Ojima_1,Kugo_Ojima_2} and BV formalism \cite{Batalin_Vilkovisky_1,Batalin_Vilkovisky_2,Getzler}. We refer the interested reader to \cite{Barnich_Brandt_Henneaux,Mnev,Wernli} for general introductions to BRST cohomology and BV formalism and to \cite{Upadhyay} for the particular discussion of perturbative Quantum General Relativity. We remark that in these formulations the Lagrange multiplier is viewed as a vector in the corresponding Lie algebra; this is equivalent to our formulation where the Lagrange multiplier is a scalar, if a non-degenerate metric is present on the Lie algebra: More precisely, in this case the product of the Lagrange multiplier, viewed as a vector in the Lie algebra, with the gauge-fixing map, described via a momentum map, depends only on its length, since only the vector components parallel to the gauge-fixing map contribute. Thus, the vector-valued Lagrange multiplier can be replaced via a scalar multiple of the gauge-fixing function, as in our gauge-fixing Lagrange density, \eqnref{eqn:lagrange_gauge-fixing}. Now, we discuss the parts of the QGR-QED Lagrange density \(\mathcal{L}_{\text{QGR-QED}}\) in detail:

\subsubsection{The Lagrange density of GR-ED}

The classical Lagrange density of General Relativity coupled with Spinor Electrodynamics is the sum of the Einstein-Hilbert, Maxwell and Dirac Lagrange densities:
\begin{equation}
	\mathcal{L}_\text{GR-ED} = \left ( \leh + \lmd \right ) \dif V_g \, ,
\end{equation}
where we have rescaled the Einstein-Hilbert term by \(\textfrac{1}{\gcoupling^2}\) and the Maxwell term by \(\textfrac{1}{\ecoupling^2}\) such that the graviton and photon propagators are of order \(\order{\gcoupling^0}\) and \(\order{\ecoupling^0}\), respectively. Here, \(R = g^{\nu \sigma} \tensor{R}{^\mu _\sigma _\mu _\nu}\) is the Ricci scalar of the tangent bundle, cf.\ \defnref{defn:curvatures_spacetime-matter_bundle}, \(F_{\mu \nu}\) is the curvature two-form on the \(U(1)\) principle bundle, cf.\ \defnref{defn:curvatures_spacetime-matter_bundle}, \(\slashed{\nabla}^{U(1) \times_\rho \Sigma M}\) the twisted Dirac operator, cf.\ \defnref{defn:dirac_operator} and \(\dif V_g\) is the Riemannian volume form, cf.\ \defnref{def:riemannian_volume_form}. The classical Lagrange density \(\mathcal{L}_{\text{GR-ED}}\) leads to the Einstein field equations:
\begin{subequations}
\begin{align}
	G_{\mu \nu} & \phantom{:} = \gcoupling^2 T_{\mu \nu} \, ,
\intertext{where}
	G_{\mu \nu} & = R_{\mu \nu} - \frac{1}{2} R g_{\mu \nu}
\intertext{is the Einstein tensor and}
\begin{split}
	T_{\mu \nu} & \phantom{:} = \frac{1}{\ecoupling^2} \left ( g^{\rho \sigma} F_{\mu \rho} F_{\sigma \nu} - \frac{1}{4} g_{\mu \nu} g^{\rho \sigma} g^{\kappa \lambda} F_{\rho \kappa} F_{\sigma \lambda} \right ) \\
	& \phantom{:= (}+ \frac{\imaginary}{4} e_\mu^m \left ( \overline{\Psi} \gamma_m \left ( \nabla^{U(1) \times_\rho \Sigma M}_\nu \Psi \right ) - \left ( \nabla^{U(1) \times_\rho \Sigma M}_\nu \overline{\Psi} \right ) \gamma_m \Psi \right ) \\
	& \phantom{:= (} + \frac{\imaginary}{4} e_\nu^n \left ( \overline{\Psi} \gamma_n \left ( \nabla^{U(1) \times_\rho \Sigma M}_\mu \Psi \right ) - \left ( \nabla^{U(1) \times_\rho \Sigma M}_\mu \overline{\Psi} \right ) \gamma_n \Psi \right )
\end{split}
\end{align}
\end{subequations}
is the corresponding (generalized) Hilbert energy-momentum tensor for Spinor Electrodynamics.\footnote{Where we have already used the equations of motion for the fermionic energy-momentum contribution.} We remark that considering linearized GR is equivalent to the vanishing of the first bracket in \eqnref{eqn:ricci_scalar_metric} from \propref{prop:ricci_scalar_for_the_levi_civita_connection}, i.e.\ considering locally geometries such that
\begin{equation}
	g^{\mu \rho} g^{\nu \sigma} \left ( \partial_\mu \partial_\nu g_{\rho \sigma} - \partial_\mu \partial_\rho g_{\nu \sigma} \right ) \equiv 0 \, , \label{eqn:linearized_GR}
\end{equation}
which could otherwise be interpreted as a source term for the graviton field \(h_{\mu \nu}\). Thus, when considering Feynman rules it would correspond to a graviton half-edge --- however, the corresponding Feynman rule vanishes if momentum conservation is considered. Nevertheless, it produces additional contributions to the higher valent pure graviton vertex Feynman rules. Observe that the coupling of the graviton field to the matter fields is given via inverse metrics, inverse vielbeins, the spin-connection in the Dirac-operator and the prefactor of the Riemannian volume form.

\subsubsection{Gauge fixing Lagrange density} \label{sssec:gauge_fixing_lagrange_density}

In QGR-QED there are two gauge symmetries present: One coming from General Relativity and affecting the whole spacetime-matter-bundle \(\mathcal{S}\) and the other one coming from electrodynamics and affecting the \(U(1)\) principle bundle only. These gauge transformations are described in \remref{rem:gauge_transformations} and lead to the transformations of the graviton and photon fields as given in \remref{rem:transformation_properties}. Thus, the gauge fixing part of the Lagrange density consists also of two parts: One for transformations due to infinitesimal diffeomorphisms, where we choose the de Donder gauge
\begin{equation}
	g^{\rho \sigma} \Gamma^\mu_{\rho \sigma} \overset{!}{=} 0 \iff g^{\rho \sigma} \left ( \partial_\rho g_{\sigma \mu} \right ) \overset{!}{=} \frac{1}{2} g^{\rho \sigma} \left ( \partial_\mu g_{\rho \sigma} \right ) \, , \label{eqn:de_donder_gauge}
\end{equation}
and one for infinitesimal principle bundle automorphisms, where we choose the Lorenz gauge
\begin{equation}
	g^{\mu \nu} \nabla^{TM}_\mu \imaginary \ecoupling A_\nu \overset{!}{=} 0 \, . \label{eqn:lorenz_gauge}
\end{equation}
The de Donder gauge is motivated by the fact that in this gauge the divergence of a vector or covector field reduces to the following simpler expression involving only a partial derivative
\begin{equation} \label{eqn:motivation_de_donder_divergence_vectorfield}
\begin{split}
	\nabla^{TM}_\mu X^\mu & = \partial_\mu X^\mu + \Gamma^\mu_{\mu \rho} X^\rho\\
	& = g^{\mu \nu} \left ( \partial_\mu X_\nu - \Gamma^\sigma_{\mu \nu} X_\sigma \right )\\
	& = g^{\mu \nu} \partial_\mu X_\nu\\
	& = \partial_\mu X^\mu \, .
\end{split}
\end{equation}
Thus, the Beltrami-Laplace operator takes on the simple form
\begin{equation}
	\Delta^{TM}_{\text{Beltrami}} = g^{\mu \nu} \partial_\mu \partial_\nu
\end{equation}
and in particular coordinate functions \(x^\alpha \in M\) are harmonic
\begin{equation}
\begin{split}
	\Delta^{TM}_{\text{Beltrami}} x^\alpha & = g^{\mu \nu} \partial_\mu \partial_\nu x^\alpha \\
	& = g^{\mu \nu} \partial_\mu \delta^\alpha_\nu \\
	& = 0 \, .
\end{split}
\end{equation}
Therefore this gauge is also called harmonic coordinate condition. Furthermore, we remark that in general there does not exist a gauge such that the evolution of the graviton field is governed by a wave equation, because the graviton field is in general non-linear. This comes from the first term in \eqnref{eqn:ricci_scalar_metric}, cf. the comment before \eqnref{eqn:linearized_GR}. If, however, the linearized Einstein-Hilbert Lagrange density is considered together with the de Donder gauge, then the evolution of the graviton field is given via a wave equation --- a fact which is used in gravitational wave analysis, cf.\ e.g.\ \cite{Carroll}. On the other hand, the Lorenz gauge is motivated by the fact that then the photon field \(A_\rho\) satisfies a wave equation with a source term given by the dual four-current \(g_{\rho \sigma} j^\sigma\) and the Ricci curvature tensor applied to it \(R^\kappa_\rho A_\kappa\),\footnote{Where we have used that the Levi-Civita connection is torsion-free, i.e.\ \(\Gamma_{\mu \nu}^\rho \equiv \Gamma_{\nu \mu}^\rho\), to match the definition of \(F_{\mu \nu}\), cf.\ \defnref{defn:curvatures_spacetime-matter_bundle} and \ftnref{foot:fmunu_exterior_derivative}.}
\begin{equation}
\begin{split}
	\Delta^{TM}_{\text{Bochner}} A_\rho & = g^{\mu \nu} \nabla^{TM}_\mu \nabla^{TM}_\nu A_\rho \\
& = g^{\mu \nu} \nabla^{TM}_\mu \left ( \nabla^{TM}_\nu A_\rho - \nabla^{TM}_\rho A_\nu + \nabla^{TM}_\rho A_\nu \right ) \\
& = g^{\mu \nu} \nabla^{TM}_\mu \left ( - \frac{\imaginary}{\ecoupling} F_{\nu \rho} + \nabla^{TM}_\rho A_\nu \right ) \\
& = g_{\rho \sigma} j^\sigma + g^{\mu \nu} \left ( \left [ \nabla^{TM}_\mu , \nabla^{TM}_\rho \right ] + \nabla^{TM}_\rho \nabla^{TM}_\mu \right ) A_\nu\\
& = g_{\rho \sigma} j^\sigma + g^{\mu \nu} \tensor{R}{^\kappa _\nu _\rho _\mu} A_\kappa\\
& = g_{\rho \sigma} j^\sigma + R^\kappa_\rho A_\kappa \, ,
\end{split}
\end{equation}
where \(\Delta^{TM}_{\text{Bochner}}\) is the Bochner-Laplace operator. Observe that when we apply the de Donder gauge given in \eqnref{eqn:de_donder_gauge} to the Lorenz gauge given in \eqnref{eqn:lorenz_gauge}, the expression simplifies to
\begin{equation}
\begin{split}
	g^{\mu \nu} \nabla^{TM}_\mu \imaginary \ecoupling A_\nu & = g^{\mu \nu} \left ( \partial_\mu \imaginary \ecoupling A_\nu - \Gamma_{\mu \nu}^\rho \imaginary \ecoupling A_\rho \right )\\
	& = g^{\mu \nu} \partial_\mu \imaginary \ecoupling A_\nu \, ,
\end{split}
\end{equation}
as was shown in general in \eqnref{eqn:motivation_de_donder_divergence_vectorfield}. For the following, we write the de Donder gauge as
\begin{align}
	\deDonder^\mu & := g^{\rho \sigma} \Gamma^\mu_{\rho \sigma} \overset{!}{=} 0
\intertext{and the Lorenz gauge as}
	\Lorenz & := g^{\mu \nu} \nabla^{TM}_\mu \imaginary \ecoupling A_\nu \overset{!}{=} 0 \, .
\end{align}
We implement the de Donder gauge and the Lorenz gauge in the Lagrange density by using the Lagrange multipliers \(\textfrac{1}{\xi}\) and \(\textfrac{1}{\zeta}\), which act in the quantum theory like coupling constants, i.e.\ by adding the following Lagrange density:
\begin{equation}
	\mathcal{L}_{\text{GF}} = - \left ( \frac{1}{4 \gcoupling^2 \zeta} g_{\mu \nu} \deDonder^\mu \deDonder^\nu + \frac{1}{2 \ecoupling^2 \xi} \Lorenz^2 \right ) \dif V_g \, ,
\end{equation}
where \(\dif V_g\) is the Riemannian volume form, cf.\ \defnref{def:riemannian_volume_form}. If \(\textfrac{1}{\xi}\) and \(\textfrac{1}{\zeta}\) are interpreted as parameters rather than Lagrange multipliers, then \(\zeta = 1\) corresponds to the de Donder gauge and \(\xi = 1\) to the Feynman gauge.

\subsubsection{Ghost Lagrange density} \label{subsubsec:ghost_lagrange_density}

In QGR-QED two ghosts together with their antighosts are needed: The graviton ghost and graviton antighost which we denote by \(\gravitonghost_\mu\) and \(\overline{\gravitonghost}_\mu\), respectively, and the photon ghost and photon antighost which we denote by by \(\photonghost\) and \(\overline{\photonghost}\), respectively. The ghost Lagrange density is obtained via Faddeev-Popov's method as the variation of the gauge fixing condition via a gauge transformation and then replacing the transformation fields via the corresponding ghosts multiplied with the antighost corresponding to the gauge fixing. Then, the variation of the Lagrange density with respect to the antighosts gives the equation of motion for the ghosts, which corresponds to the residual gauge transformations of the chosen gauge fixings. Furthermore, we have added the ghost conjugate for the interaction terms involving the graviton ghost and the photon antighost, i.e.\ the corresponding interaction terms with the graviton antighost and the photon ghost with the rest unaltered (second term below). Observe also that the photon ghost couples to gravitons and thus is needed, in contrast to normal QED where it decouples from the other particles:
\begin{equation}
\begin{split}
\mathcal{L}_{\text{Ghost}} & = - \left ( \frac{1}{2 \gcoupling} \overline{\gravitonghost}_\mu g^{\mu \nu} \eval{\left ( \widetilde{\deDonder}_\nu \right )}_{\mathcal{X}_\rho \mapsto \gravitonghost_\rho} - \left ( \frac{\imaginary}{2 \ecoupling} \overline{\photonghost} \eval{\left ( \widetilde{L} \right )}_{\mathcal{X}_\rho \mapsto \gravitonghost_\rho} + \text{g.c.} \right )- \frac{\imaginary}{\ecoupling} \overline{\photonghost} \eval{\left ( \tilde{L} \right )}_{f \mapsto \photonghost} \right ) \dif V_g \\
& = - \left ( \frac{1}{2} \overline{\gravitonghost}_\mu g^{\mu \nu} g^{\rho \sigma} \left ( \partial_\rho \partial_\sigma \gravitonghost_\nu + \partial_\nu \left ( \Gamma_{\rho \sigma}^\lambda \gravitonghost_\lambda \right ) - \partial_\rho \left ( \Gamma_{\sigma \nu}^\lambda \gravitonghost_\lambda \right ) - \partial_\sigma \left ( \Gamma_{\nu \rho}^\lambda \gravitonghost_\lambda \right ) \right ) \right . \\
& \phantom{= - (} + \frac{1}{2} \gcoupling \overline{\photonghost} \left ( g^{\mu \nu} \partial_\mu \left ( \left ( \partial_\nu g^{\rho \sigma} \right ) A_\rho \gravitonghost_\sigma + g^{\rho \sigma} \left ( \partial_\rho A_\nu \right ) \gravitonghost_\sigma + g^{\rho \sigma} A_\rho \left ( \partial_\nu \gravitonghost_\sigma \right ) \right ) \right . \\
& \phantom{= - ( + \frac{1}{2} \gcoupling \overline{\photonghost} (} \left . - g^{\mu \nu} g^{\rho \sigma} \Gamma_{\mu \nu}^\lambda \left ( \left ( \partial_\lambda g^{\rho \sigma} \right ) A_\rho \gravitonghost_\sigma + g^{\rho \sigma} \left ( \partial_\rho A_\lambda \right ) \gravitonghost_\sigma + g^{\rho \sigma} A_\rho \left ( \partial_\lambda \gravitonghost_\sigma \right ) \right ) \right ) + \text{g.c.} \\
& \phantom{= - (} \left . + \overline{\photonghost} g^{\mu \nu} \left ( \partial_\mu \left ( \partial_\nu \photonghost \right ) - \Gamma_{\mu \nu}^\lambda \left ( \partial_\lambda \photonghost \right ) \right ) \vphantom{\frac{1}{2} g^{\mu \nu} g^{\rho \sigma} \overline{\gravitonghost}_\mu \left ( \partial_\rho \partial_\sigma \gravitonghost_\nu + \partial_\nu \left ( \Gamma_{\rho \sigma}^\lambda \gravitonghost_\lambda \right ) - \partial_\rho \left ( \Gamma_{\sigma \nu}^\lambda \gravitonghost_\lambda \right ) - \partial_\sigma \left ( \Gamma_{\nu \rho}^\lambda \gravitonghost_\lambda \right ) \right )} \right ) \dif V_g \, ,
\end{split} \label{eqn:lagrange_ghost_sssec}
\end{equation}
where we have set\footnote{Be aware of the difference between widetilde and tilde for the transformed Lorentz gauge, i.e.\ \(\widetilde{L}\) vs.\ \(\tilde{L}\).}
\begin{align}
	\widetilde{\deDonder}_\nu & := g^{\rho \sigma} \left ( \partial_\rho \partial_\sigma X_\nu + \partial_\nu \left ( \Gamma_{\rho \sigma}^\lambda X_\lambda \right ) - \partial_\rho \left ( \Gamma_{\nu \sigma}^\lambda X_\lambda \right ) - \partial_\sigma \left ( \Gamma_{\nu \rho}^\lambda X_\lambda \right ) \right ) \, , \\
	\begin{split}
	\widetilde{L} & := g^{\mu \nu} \left ( g^{\rho \sigma} \partial_\mu \left ( X_\rho \left ( \partial_\sigma A_\nu \right ) + g^{\rho \sigma} A_\rho \left ( \partial_\nu X_\sigma \right ) + A_\rho X_\sigma \left ( \partial_\nu g^{\rho \sigma} \right ) \right ) \right . \\
	& \phantom{:= g^{\mu \nu} (} \left . - g^{\rho \sigma} \Gamma_{\mu \nu}^\lambda \left ( X_\rho \left ( \partial_\sigma A_\lambda \right ) + g^{\rho \sigma} A_\rho \left ( \partial_\lambda X_\sigma \right ) + A_\rho X_\sigma \left ( \partial_\lambda g^{\rho \sigma} \right ) \right ) \right )
	\end{split}
	\intertext{and}
	\tilde{L} & := g^{\mu \nu} \left ( \partial_\mu \partial_\nu \mathfrak{f} - \Gamma_{\mu \nu}^\lambda \left ( \partial_\lambda \mathfrak{f} \right ) \right ) \, ,
\end{align}
i.e.\ the gauge conditions \(\deDonder_\nu\) and \(\Lorenz\) from \eqnsaref{eqn:de_donder_gauge}{eqn:lorenz_gauge} applied to linearized pure gauge transformations from \eqnssaref{eqn:lie_derivative_metric}{eqn:lie_derivative_connection_form}{eqn:gauge_transformation_connection_form}, and \(\dif V_g\) is the Riemannian volume form, cf.\ \defnref{def:riemannian_volume_form}. Moreover, g.c.\ means ghost conjugate, i.e.\ the simultaneous replacement
\begin{equation}
\begin{Bmatrix}
\overline{\theta} \\ \gravitonghost_\rho
\end{Bmatrix}
\rightsquigarrow
\begin{Bmatrix}
\theta \\ \overline{\gravitonghost}_\rho
\end{Bmatrix}
\end{equation}
with the rest unaltered, but partially integrated such that there is no partial derivative in front of \(\overline{\gravitonghost}_\rho\). We also remark the additional factor of \(\textfrac{1}{2}\) in front of the ghost mixing terms due to the addition of the corresponding ghost conjugate terms. Finally, we notice that the photon ghost is a scalar particle, whereas the graviton ghost is a spin-one particle.

\section{Hopf algebras, the renormalization Hopf algebra and QGR-QED} \label{sec:hopf_algebras_and_the_connes-kreimer_renormalization_hopf_algebra}

In this section, we introduce Hopf algebras in general and the Connes-Kreimer renormalization Hopf algebra in particular. The intention is to review the basic notions and give the relevant definitions. We refer the reader who wishes a more detailed treatment on Hopf algebras in general and its connection to affine group schemes to \cite{Waterhouse,Milne} and the reader who wishes a more detailed treatment on the construction of the Connes-Kreimer renormalization Hopf algebra to \cite{Kreimer_Hopf_Algebra,Connes_Kreimer_0,Connes_Kreimer_1,Connes_Kreimer_2,Kreimer_Anatomy,Manchon}. In this work, we set \(\ring\) to be a commutative ring with one and then later consider the case \(\ring := \mathbb{Q}\) for the Connes-Kreimer renormalization Hopf algebra.\footnote{Actually, the physical needs require \(\ring\) only to be a field with characteristic \(0\). Since \(\mathbb{Q}\) is the smallest such field, it is the canonical choice.} Furthermore, we set the tensor product over the corresponding ring if not stated otherwise, i.e.\ \(\otimes := \otimes_\ring\). Finally, we study Hopf ideals in the renormalization Hopf algebra in order to understand which symmetries are compatible with renormalization.

\subsection{Hopf algebras}

We start by defining Hopf algebras in general. In this article, we set \(\ring\) to be a commutative ring with one. Furthermore, by algebra we mean associative algebra with identity and by coalgebra we mean coassociative coalgebra with coidentity. Moreover, all algebras, coalgebras, bialgebras and Hopf algebras are all considered over the ring \(\ring\).

\vspace{\baselineskip}

\begin{defn}[Algebra]
The triple \(( A , \mult , \one )\) is an algebra, where \(A\) is a \(\ring\)-module, \(\mult \colon A \otimes A \to A\) an associative multiplication map, i.e.\ the following diagram commutes
\begin{equation}
\begin{tikzcd}[row sep=huge]
A \otimes A \otimes A	\arrow{r}{\mult \otimes \id}	\arrow[swap]{d}{\id \otimes \mult}	& A \otimes A \arrow{d}{\mult}\\
A \otimes A	\arrow[swap]{r}{\mult}	& A
\end{tikzcd} \, ,
\end{equation}
and \(\one \colon \ring \to A\) the identity with respect to \(\mult\),\footnote{We denote both the neutral element in the algebra \(A\) with respect to the multiplication map \(\mult\) and the inclusion function \(\ring \to A\) by the symbol \(\one\).} i.e.\ the following diagram commutes
\begin{equation}
\begin{tikzcd}[row sep=huge]
k \otimes A \arrow{r}{\one \otimes \id} \arrow[swap]{rd}{\cong} & A \otimes A \arrow{d}{\mult} & A \otimes k \arrow[swap]{l}{\id \otimes \one} \arrow{dl}{\cong}\\
& A
\end{tikzcd} \, .
\end{equation}
\end{defn}

\vspace{\baselineskip}

\begin{defn}[(Connected) graded algebra]
An algebra \(A\) is called graded, if the \(\ring\)-module \(A\) can be written as a direct sum
\begin{equation}
	A = \bigoplus_{i = 0}^\infty A_i \, ,
\end{equation}
which is respected by the multiplication \(\mult\), i.e.\
\begin{equation}
	\mult \left ( A_i \otimes A_j \right ) \subseteq A_{i + j} \, , \; \forall i , j \in \mathbb{N}_0 \, ,
\end{equation}
and the identity \(\one\), i.e.\
\begin{equation}
	\one \, : \; A_0 \hookrightarrow A \, , \; a_0 \mapsto a_0
\end{equation}
Furthermore, a graded algebra \(A\) is called connected, if the grade zero component is isomorphic to the base ring, i.e.\
\begin{equation}
	A_0 \cong \ring \, .
\end{equation}
\end{defn}

\vspace{\baselineskip}

\begin{defn}[Coalgebra]
The triple \(( C , \Delta , \coone )\) is a coalgebra, where \(C\) is a \(\ring\)-module, \(\Delta \colon C \to C \otimes C\) a coassociative comultiplication map, i.e.\ the following diagram commutes
\begin{equation}
\begin{tikzcd}[row sep=huge]
C	\arrow{r}{\Delta}	\arrow[swap]{d}{\Delta}	& C \otimes C \arrow{d}{\id \otimes \Delta}\\
C \otimes C	\arrow[swap]{r}{\Delta \otimes \id}	& C \otimes C \otimes C
\end{tikzcd} \, ,
\end{equation}
and \(\coone \colon C \to \ring\) the coidentity with respect to \(\Delta\),\footnote{We denote both the neutral element in the coalgebra \(C\) with respect to the comultiplication map \(\Delta\) and the projection function \(C \to \ring\) by the symbol \(\coone\).} i.e.\ the following diagram commutes
\begin{equation}
\begin{tikzcd}[row sep=huge]
& C \arrow[swap]{dl}{\cong} \arrow{d}{\Delta} \arrow{dr}{\cong}\\
\ring \otimes C & C \otimes C \arrow{l}{\coone \otimes \id} \arrow[swap]{r}{\id \otimes \coone} & C \otimes \ring
\end{tikzcd} \, .
\end{equation}
\end{defn}

\vspace{\baselineskip}

\begin{defn}[(Connected) graded coalgebra]
A coalgebra \(C\) is called graded, if the \(\ring\)-module \(C\) can be written as a direct sum
\begin{equation}
	C = \bigoplus_{i = 0}^\infty C_i \, ,
\end{equation}
which is respected by the comultiplication \(\Delta\), i.e.\
\begin{equation}
	\Delta \left ( C_i \right ) \subseteq \sum_{j = 0}^i C_j \otimes C_{i - j} \, , \; \forall i \in \mathbb{N}_0 \, ,
\end{equation}
and the coidentity \(\coone\), i.e.\
\begin{equation}
	\coone \, : \; C \twoheadrightarrow C_0 \, , \; c_i \mapsto \delta_{0 i} c_i \, ,
\end{equation}
where \(\delta_{0 i}\) is the Kronecker delta and \(c_i \in C_i\) are homogeneous elements of degree \(i\). Furthermore, a graded coalgebra \(C\) is called connected, if the grade zero component is isomorphic to the base ring, i.e.\
\begin{equation}
	C_0 \cong \ring \, .
\end{equation}
\end{defn}

\vspace{\baselineskip}

\begin{defn}[Homomorphism of (graded) algebras]
Let \(A_1\) and \(A_2\) be two algebras. Then a map \(f \colon A_1 \to A_2\) is called a homomorphism of algebras, if \(f\) is compatible with the products \(\mult_1\) on \(A_1\) and \(\mult_2\) on \(A_2\), i.e.\ the following diagram commutes
\begin{equation}
\begin{tikzcd}[row sep=huge]
A_1 \otimes A_1 	\arrow{r}{\mult_1}	\arrow[swap]{d}{f \otimes f}	& A_1 \arrow{d}{f}\\
A_2 \otimes A_2	\arrow[swap]{r}{\mult_2}	& A_2
\end{tikzcd} \, ,
\end{equation}
and \(f\) maps the identity \(\one_1\) on \(A_1\) to the identity \(\one_2\) on \(A_2\), i.e.\ the following diagram commutes
\begin{equation}
\begin{tikzcd}[row sep=huge]
& \ring \arrow[swap]{dl}{\one_1} \arrow{dr}{\one_2}\\
A_1 \arrow[swap]{rr}{f} & & A_2
\end{tikzcd} \, .
\end{equation}
If the algebras \(A_1\) and \(A_2\) are both graded, then \(f\) additionally has to respect this structure to be a homomorphism of graded algebras, i.e.\
\begin{equation}
f \left ( \left ( A_1 \right ) _i \right ) \subseteq \left ( A_2 \right ) _i \, , \; \forall i \in \mathbb{N}_0 \, .
\end{equation}
\end{defn}

\vspace{\baselineskip}

\begin{defn}[Homomorphism of (graded) coalgebras]
Let \(C_1\) and \(C_2\) be two coalgebras. Then a map \(g \colon C_1 \to C_2\) is called a homomorphism of coalgebras, if \(g\) is compatible with the two coproducts \(\Delta_1\) on \(C_1\) and \(\Delta_2\) on \(C_2\), i.e.\ the following diagram commutes
\begin{equation}
\begin{tikzcd}[row sep=huge]
C_1 	\arrow{r}{\Delta_1}	\arrow[swap]{d}{g}	& C_1 \otimes C_1 \arrow{d}{g \otimes g}\\
C_2	\arrow[swap]{r}{\Delta_2}	& C_2 \otimes C_2
\end{tikzcd} \, ,
\end{equation}
and \(g\) is compatible with the two coidentities \(\coone_1\) on \(C_1\) and \(\coone_2\) on \(C_2\), i.e.\ the following diagram commutes
\begin{equation}
\begin{tikzcd}[row sep=huge]
C_1 \arrow{rr}{g} \arrow[swap]{dr}{\coone_1} & & C_2 \arrow{dl}{\coone_2}\\
& \ring
\end{tikzcd} \, .
\end{equation}
If the coalgebras \(C_1\) and \(C_2\) are both graded, then \(g\) additionally has to respect this structure to be a homomorphism of graded coalgebras, i.e.\
\begin{equation}
g \left ( \left ( C_1 \right ) _i \right ) \subseteq \left ( C_2 \right ) _i \, , \; \forall i \in \mathbb{N}_0 \, .
\end{equation}
\end{defn}

\vspace{\baselineskip}

\begin{rem}[Relation between algebra and coalgebra]
Finite dimensional ((connected) graded) algebras \(A\) are related to ((connected) graded) coalgebras \(C\) via dualization, i.e.\ applying the functor \(\text{Hom}_{\ring-\mathsf{Alg}} \left ( \cdot , \ring \right )\).
\end{rem}

\vspace{\baselineskip}

\begin{defn}[(Connected graded) bialgebra]
The quintuple \((B, \mult, \one, \Delta, \coone)\) is a bialgebra, where the triple \((B, \mult, \one)\) is an algebra and the triple \((B, \Delta, \coone)\) is a coalgebra. Furthermore, the coproduct \(\Delta\) and the coidentity \(\coone\) are homomorphisms of the graded algebra \((B, \mult, \one)\), or, equivalently, the multiplication \(\mult\) and the identity \(\one\) are homomorphisms of the graded coalgebra \((B, \Delta, \coone)\). Furthermore, a bialgebra \(B\) is called graded, if it is graded as an algebra and as a coalgebra. Moreover, a bialgebra \(B\) is called connected, if it is connected as an algebra, or, equivalently, as a coalgebra.\footnote{It can be shown that a connected graded bialgebra possesses an antipode and thus is already a Hopf algebra, cf.\ \defnref{defn:hopf_algebra}.}
\end{defn}

\vspace{\baselineskip}

\begin{defn}[(Connected graded) Hopf algebra] \label{defn:hopf_algebra}
The sextuple \((H, \mult, \one, \Delta, \coone, S)\) is a Hopf algebra, where the quintuple \((H, \mult, \one, \Delta, \coone)\) is a bialgebra and \(S \colon H \to H\) is an anti-endomorphism,\footnote{Meaning that in the case of non-commutative Hopf algebras the antipode is an order reversing endomorphism, i.e.\ given \(x, y \in H\) then we have \(S \left ( x y \right ) = S \left ( y \right ) S \left ( x \right )\).} called the antipode, and is defined such that the following diagram commutes:
\begin{equation}
\begin{tikzcd}[column sep=tiny,row sep=huge]
& H \otimes H \arrow{rr}{S \otimes \id} & & H \otimes H \arrow{dr}{\mult}\\
H \arrow{ur}{\Delta} \arrow[swap]{dr}{\Delta} \arrow{rr}{\coone} & & k \arrow{rr}{\one} & & H\\
& H \otimes H \arrow[swap]{rr}{\id \otimes S} & & H \otimes H \arrow[swap]{ur}{\mult}
\end{tikzcd}
\end{equation}
Furthermore, a Hopf algebra \(H\) is called graded, if it is graded as a bialgebra.\footnote{Observe that if a Hopf algebra \(H\) is graded as a bialgebra, then the antipode \(S\) is automatically an anti-endomorphism of graded algebras and an anti-endomorphism of graded coalgebras.} Moreover, a Hopf algebra \(H\) is called connected, if it is connected as a bialgebra.
\end{defn}

\vspace{\baselineskip}

\begin{defn}[Hopf ideals and Hopf subalgebras] \label{Hopf_ideals_and_subalgebras}
Let \(H\) be a Hopf algebra and \(\mathfrak{i}\) an ideal in \(H\). Then \(\mathfrak{i}\) is called a Hopf ideal if it satisfies additionally the following three conditions
\begin{subequations}
\begin{align}
\D{\mathfrak{i}} & \subseteq \mathfrak{i} \otimes H + H \otimes \mathfrak{i} \, , \label{eqn:hopf_ideal_1}\\
\coone \left ( \mathfrak{i} \right ) & = 0 \label{eqn:coone_hopf_ideal}
\intertext{and}
\antipode{\mathfrak{i}} & \subseteq \mathfrak{i} \, . \label{eqn:hopf_ideal_3}
\end{align}
\end{subequations}
Then, the quotient \(h := H / \mathfrak{i}\) is a Hopf algebra as well, called a Hopf subalgebra of \(H\). If \(H\) is graded, then \(h\) inherits a grading from \(H\) via
\begin{equation}
	h = H \cap h = \left ( \bigoplus_{i = 0}^\infty H_i \right ) \cap h = \bigoplus_{i = 0}^\infty \left ( H_i \cap h \right ) = \bigoplus_{i = 0}^\infty h_i \, ,
\end{equation}
i.e.\ we define the grade \(i\) subspace of \(h\) as \(h_i := \left ( H_i \cap h \right )\) for all \(i \in \mathbb{N}_0\). Furthermore, if \(H\) is connected, then \(h\) is also connected, since \(\one \neq 0\) in \(h\) and thus \(h_0 = k \one \cong k\).
\end{defn}

\vspace{\baselineskip}

\begin{defn}[Augmentation ideal, quoted from \cite{Prinz_2}] \label{defn:augmentation_ideal}
Given a bi- or a Hopf algebra \(B\), then the kernel of the coidentity
\begin{equation}
	\operatorname{Aug} \left ( \HQ \right ) := \operatorname{Ker} \left ( \coone \right )
\end{equation}
is an ideal, called the augmentation ideal. Furthermore, we denote the projector to it via \(\mathscr{P}\), i.e.\
\begin{equation}
	\mathscr{P} \, : \quad \HQ \surject \operatorname{Aug} \left ( \HQ \right ) \, .
\end{equation}
\end{defn}

\vspace{\baselineskip}

\begin{defn}[Convolution product] \label{defn:convolution_product}
Let \(A\) be an algebra and \(C\) a coalgebra. Then using the multiplication \(\mult_A\) on \(A\) and the comultiplication \(\Delta_C\) on \(C\), we can turn the \(k\)-module \(\text{Hom}_{\ring-\mathsf{Mod}} \left ( C , A \right )\) of \(k\)-linear maps from \(C\) to \(A\) into a \(k\)-algebra by defining the convolution product \(\star\) for given \(f, g \in \text{Hom}_{\ring-\mathsf{Mod}} \left ( C , A \right )\) via
\begin{equation}
	f \star g := \mult_A \circ \left ( f \otimes g \right ) \circ \Delta_C \, .
\end{equation}
Obviously, this definition extends trivially if \(A\) or \(C\) possesses additionally a bi- or Hopf algebra structure, since it only requires a coalgebra structure in the source algebra and an algebra structure in the target algebra.
\end{defn}

\subsection{The Connes-Kreimer renormalization Hopf algebra} \label{ssec:connes-kreimer_renormalization_hopf_algebra}

From now on, we consider the Connes-Kreimer renormalization Hopf algebra, which is a Hopf algebra over \(\ring = \mathbb{Q}\).\footnote{Again, the physical needs require \(\ring\) only to be a field with characteristic \(0\). Since \(\mathbb{Q}\) is the smallest such field, it is the canonical choice.} Furthermore, let in this subsection \(\Q\) be a local QFT such that the residue for each divergent Feynman graph is in the residue set of \(\Q\) --- the general case will be discussed in \ssecref{ssec:associating_the_renormalization_hopf_algebra_to_a_given_local_qft}, cf.\ \proref{pro:problem}.

\vspace{\baselineskip}

\begin{defn}[Weighted residue set of a local QFT] \label{defn:residue_set_of_a_local_qft_and_their_weights}
Let \(\mathcal{Q}\) be a local QFT. Then \(\mathcal{Q}\) is either given via a Lagrange density \(\mathcal{L}_{\mathcal{Q}}\) or via a set of residues \(\mathcal{R}_{\mathcal{Q}}\) together with a weight function \(\omega \colon \mathcal{R}_{\mathcal{Q}} \to \mathbb{N}_0\). The set of residues is a disjoint union of all vertex-types \(\mathcal{R}_{\mathcal{Q}}^{[0]}\) and all edge-types \(\mathcal{R}_{\mathcal{Q}}^{[1]}\) of \(\mathcal{Q}\), i.e.\ \(\mathcal{R}_{\mathcal{Q}} = \mathcal{R}_{\mathcal{Q}}^{[0]} \amalg \mathcal{R}_{\mathcal{Q}}^{[1]}\). If \(\mathcal{Q}\) is given via a Lagrange density, then the set of residues \(\mathcal{R}_{\mathcal{Q}}\) and the weight \(\sdd{R}\) of each residue \(R \in \mathcal{R}_{\mathcal{Q}}\) is given as follows: Each field in the Lagrange density \(\mathcal{L}_{\mathcal{Q}}\) corresponds to a particle type of \(\mathcal{Q}\). Therefore, every monomial in \(\mathcal{L}_{\mathcal{Q}}\) consisting of one field corresponds to a source-term, i.e.\ a vertex-residue in \(\mathcal{R}_{\mathcal{Q}}^{[0]}\) consisting of a half-edge of that particle type. Every monomial in \(\mathcal{L}_{\mathcal{Q}}\) consisting of two equivalent fields corresponds to a propagation term, i.e.\ an edge-residue in \(\mathcal{R}_{\mathcal{Q}}^{[1]}\) consisting of an edge of that particle type. And every monomial in \(\mathcal{L}_{\mathcal{Q}}\) consisting of different fields corresponds to an interaction term, i.e.\ a vertex-residue in \(\mathcal{R}_{\mathcal{Q}}^{[0]}\) consisting of half-edges of that particle types. Finally, the weight \(\sdd{R} \in \mathbb{N}_0\) of a residue \(R \in \mathcal{R}_{\mathcal{Q}}\) is set to be the number of derivative operators involved in the corresponding field monomial in \(\mathcal{L}_{\mathcal{Q}}\).
\end{defn}

\vspace{\baselineskip}

\begin{defn}[Feynman graphs generated by residue sets 1] \label{defn:feynman_graphs_generated_by_residue_sets_1}
Let \(\mathcal{Q}\) be a local QFT with residue set \(\mathcal{R}_{\mathcal{Q}}\). Then we denote by \(\mathscr{G}_{\mathcal{Q}}\) the set of all one-particle irreducible (1PI) Feynman graphs\footnote{The use of 1PI Feynman graphs, rather than connected Feynman graphs, is justified by \thmref{thm:generalized_furry_theorem}, as is discussed in \remref{rem:1pi_generalized_furrys_theorem}.} that can be generated by the residue set \(\mathcal{R}_{\mathcal{Q}}\) of \(\mathcal{Q}\).
\end{defn}

\vspace{\baselineskip}

\begin{defn}[Residue of a Feynman graph]
Let \(\mathcal{Q}\) be a local QFT with residue set \(\RQ\) and Feynman graph set \(\mathscr{G}_{\mathcal{Q}}\). Then the residue of a Feynman graph \(\Gamma \in \mathscr{G}_{\mathcal{Q}}\), denoted by \(\res{\Gamma}\), is the vertex residue or edge residue \(\res{\Gamma}\), not necessary in the residue set \(\RQ\), obtained by shrinking all internal edges of \(\Gamma\) to a single vertex.
\end{defn}

\vspace{\baselineskip}

\begin{defn}[Amplitude set of a local QFT]
Let \(\mathcal{Q}\) be a local QFT with residue set \(\RQ\) and Feynman graph set \(\mathscr{G}_{\mathcal{Q}}\). Then we introduce its amplitude set \(\mathcal{A}_\mathcal{Q}\) as the set of residues of all Feynman graphs
\begin{equation}
	\mathcal{A}_\mathcal{Q} := \left \{ \left . r = \res{\Gamma} \right \vert \Gamma \in \mathscr{G}_\mathcal{Q} \right \} \, .
\end{equation}
Obviously, the residue set is a subset thereof, \(\mathcal{R}_\mathcal{Q} \subset \mathcal{A}_\mathcal{Q}\), and thus the amplitude set splits into a disjoint union of the residue set \(\mathcal{R}_\mathcal{Q}\) and the pure quantum correction set \(\mathcal{C}_\mathcal{Q}\),
\begin{equation}
	\mathcal{A}_\mathcal{Q} = \mathcal{R}_\mathcal{Q} \amalg \mathcal{C}_\mathcal{Q} \, .
\end{equation}
\end{defn}

\vspace{\baselineskip}

\begin{defn}[First Betti number of a Feynman graph, \cite{Hatcher}] \label{defn:first_betti_number}
Let \(\mathcal{Q}\) be a local QFT and \(\mathscr{G}_{\mathcal{Q}}\) the set of its Feynman graphs. Let furthermore \(\Gamma \in \mathscr{G}_{\mathcal{Q}}\) be a Feynman graph. Then we define the first Betti number of \(\Gamma\) as
\begin{equation}
	b_1 \left ( \Gamma \right ) := \# H_1 \left ( \Gamma \right ) \, ,
\end{equation}
where \(\# H_1 \left ( \Gamma \right )\) is the rank of the first singular homology group of \(\Gamma\).
\end{defn}

\vspace{\baselineskip}

\begin{defn}[Superficial degree of divergence] \label{defn:sdd}
Let \(\mathcal Q\) be a local QFT with weighted residue set \(\RQ\) and Feynman graph set \(\mathscr{G}_{\mathcal{Q}}\). We turn \(\mathscr{G}_{\mathcal{Q}}\) into a weighted set as well by declaring the function
\begin{equation}
	\omega \, : \quad \mathscr{G}_{\mathcal{Q}} \to \mathbb{Z} \, , \quad \Gamma \mapsto d b_1 \left ( \Gamma \right ) + \sum_{v \in \Gamma^{[0]}} \omega \left ( v \right ) - \sum_{e \in \Gamma^{[1]}} \omega \left ( e \right ) \, , \label{eqn:superficial_degree_of_divergence}
\end{equation}
where \(d\) is the dimension of spacetime of \(\mathcal{Q}\) and \(b_1 \left ( \Gamma \right )\) the first Betti number of the Feynman graph \(\Gamma \in \mathscr{G}_{\mathcal{Q}}\). Then, the weight \(\omega \left ( \Gamma \right )\) of a Feynman graph \(\Gamma \in \mathscr{G}_{\mathcal{Q}}\) is called the superficial degree of divergence of \(\Gamma\). A Feynman graph \(\Gamma \in \mathcal{G}_{\mathcal{Q}}\) is called superficially divergent if \(\sdd{\Gamma} \geq 0\), otherwise it is called superficially convergent if \(\sdd{\Gamma} < 0\).
\end{defn}

\vspace{\baselineskip}

\begin{rem}
The definition of the superficial degree of divergence of a Feynman graph, \defnref{defn:sdd}, is motivated by the fact that the Feynman integral corresponding to a given Feynman graph via the Feynman rules converges, if the Feynman graph itself and all its subgraphs are superficially convergent \cite{Weinberg}.
\end{rem}

\vspace{\baselineskip}

\begin{defn}[Set of superficially divergent subgraphs of a Feynman graph 1] \label{defn:subset_of_divergent_subgraphs_1}
Let \(\mathcal{Q}\) be a local QFT and \(\Gamma \in \mathscr{G}_{\mathcal{Q}}\) a Feynman graph of \(\mathcal{Q}\). Then we denote by \(\mathscr{D} \left ( \Gamma \right )\) the set of superficially divergent subgraphs of \(\Gamma\), i.e.\
\begin{subequations}
\begin{equation}
	\mathscr{D} \left ( \Gamma \right ) := \left \{ \gamma \subseteq \Gamma \, \left \vert \; \gamma = \coprod_{m = 1}^M \gamma_m \, , \; M \in \mathbb{N} \, : \; \omega \left ( \gamma_m \right ) \geq 0 \, , \; \forall m \right . \right \} \, .
\end{equation}
Furthermore, we define the set \(\mathscr{D}^\prime \left ( \Gamma \right )\) of superficially divergent proper subgraphs of \(\Gamma\), i.e.\
\begin{equation}
	\mathscr{D}^\prime \left ( \Gamma \right ) := \left \{ \gamma \in \mathscr{D} \left ( \Gamma \right ) \, \left \vert \; \emptyset \subsetneq \gamma \subsetneq \Gamma \right . \right \} \, .
\end{equation}
\end{subequations}
\end{defn}

\vspace{\baselineskip}

\begin{defn}[Renormalization Hopf algebra of a local QFT] \label{defn:renormalization_hopf_algebra}
Let \(\mathcal{Q}\) be a local QFT, \(\mathcal{R}_{\mathcal{Q}}\) the set of its weighted residues and \(\mathscr{G}_{\mathcal{Q}}\) the set of its weighted Feynman graphs. We assume \(\mathcal{Q}\) to be such that the residues of all superficially divergent Feynman graphs are in the residue set \(\mathcal{R}_{\mathcal{Q}}\), i.e.\ \(\left \{ \Gamma \, \vert \; \Gamma \in \mathscr{G}_{\mathcal{Q}} \, : \; \sdd{\Gamma} \geq 0 \, : \; \res{\Gamma} \notin \mathcal{R}_{\mathcal{Q}} \right \} = \emptyset\). The general and more involving case is discussed in \ssecref{ssec:associating_the_renormalization_hopf_algebra_to_a_given_local_qft}. Then, the connected graded, cf.\ \defnref{defn:connectedness_gradings_renormalization_hopf_algebra}, renormalization Hopf algebra \((\mathscr{H}_{\mathcal{Q}}, \mult, \one, \Delta, \coone, S)\) is defined as follows: We set \(\mathscr{H}_{\mathcal{Q}}\) to be the vector space over \(\mathbb{Q}\) generated by the set \(\mathscr{G}_{\mathcal{Q}}\). The associative multiplication \(\mult \colon \mathscr{H}_{\mathcal{Q}} \otimes \mathscr{H}_{\mathcal{Q}} \to \mathscr{H}_{\mathcal{Q}}\) is defined as the disjoint union of Feynman graphs, i.e.\
\begin{equation}
	\mult \, : \quad \mathscr{H}_{\mathcal{Q}} \otimes \mathscr{H}_{\mathcal{Q}} \to \mathscr{H}_{\mathcal{Q}} \, , \quad \gamma \otimes \Gamma \mapsto \gamma \Gamma \, .
\end{equation}
Then, the identity \(\one \colon \mathbb{Q} \to \mathscr{H}_{\mathcal{Q}}\) is set to be the empty graph, i.e.\
\begin{equation}
	\one := \emptyset \, .
\end{equation}
Moreover, we define the coproduct of a Feynman graph \(\Gamma\) such that it maps to the following sum over all possible combinations of divergent subgraphs of \(\Gamma\): The left-hand side of the tensor product of each summand is given by a superficially divergent subgraph of the Feynman graph \(\Gamma\), while the right-hand-side of the tensor product is given by returning the Feynman graph \(\Gamma\) with the corresponding subgraph shrunken to zero length, i.e.\
\begin{equation}
	\Delta \, : \quad \mathscr{H}_{\mathcal{Q}} \to \mathscr{H}_{\mathcal{Q}} \otimes \mathscr{H}_{\mathcal{Q}} \, , \quad \Gamma \mapsto \sum_{\gamma \in \mathscr{D} \left ( \Gamma \right )} \gamma \otimes \Gamma / \gamma \, , \label{eqn:coproduct}
\end{equation}
where the quotient \(\Gamma / \gamma\) is defined as follows: If \(\gamma\) is a proper subgraph of \(\Gamma\), then \(\Gamma / \gamma\) is defined by shrinking all internal edges of \(\gamma\) in \(\Gamma\) to a single vertex for each connected component of \(\gamma\). Otherwise, if \(\gamma = \Gamma\) we define the quotient to be the identity, i.e.\ \(\Gamma / \Gamma := \one\). The coidentity \(\coone \colon \mathscr{H}_{\mathcal{Q}} \to \mathbb{Q}\) is set such that its kernel is the Hopf algebra without the subalgebra generated by \(\one\), i.e.\
\begin{equation}
	\coone \, : \quad \mathscr{H}_{\mathcal{Q}} \to \mathbb{Q} \, , \quad \Gamma \mapsto \begin{cases} q & \text{if } \Gamma = q \one \text{ with } q \in \mathbb{Q}\\ 0 & \text{else} \end{cases} \, .
\end{equation}
Finally, we define the antipode recursively via the normalization \(\antipode{\one} = \one\) and on \(\operatorname{Aug} \left ( \mathscr{H}_{\mathcal{Q}} \right )\) via
\begin{equation}
	S \, : \quad \mathscr{H}_{\mathcal{Q}} \to \mathscr{H}_{\mathcal{Q}} \, , \quad \Gamma \mapsto - \left ( S \star \mathscr{P} \right ) \left ( \Gamma \right ) \equiv - \Gamma - \sum_{\gamma \in \mathscr{D}^\prime \left ( \Gamma \right )} \antipode{\gamma} \Gamma/\gamma \, , \label{eqn:antipode}
\end{equation}
where \(\mathscr{P}\) is the projector onto the augmentation ideal from \defnref{defn:augmentation_ideal} and the quotient \(\Gamma / \gamma\) is defined as in the definition of the coproduct after \eqnref{eqn:coproduct}.
\end{defn}

\vspace{\baselineskip}

\begin{defn}[Reduced coproduct] \label{defn:reduced_coproduct}
Let \(\mathcal{Q}\) be a local QFT as in \defnref{defn:renormalization_hopf_algebra}, \(\mathcal{R}_{\mathcal{Q}}\) the set of its weighted residues and \((\mathscr{H}_{\mathcal{Q}}, \mult, \one, \Delta, \coone, S)\) its renormalization Hopf algebra. Then we define the reduced coproduct as the non-trivial part of the coproduct, i.e.\
\begin{equation}
	\Delta^\prime \, : \quad \mathscr{H}_{\mathcal{Q}} \to \mathscr{H}_{\mathcal{Q}} \otimes \mathscr{H}_{\mathcal{Q}} \, , \quad \Gamma \mapsto \sum_{\gamma \in \mathscr{D}^\prime \left ( \Gamma \right )} \gamma \otimes \Gamma / \gamma \, . \label{eqn:reduced_coproduct}
\end{equation}
In particular, the coproduct and the reduced coproduct are related via
\begin{equation}
	\Delta \left ( \Gamma \right ) = \Delta^\prime \left ( \Gamma \right ) + \one \otimes \Gamma + \Gamma \otimes \one \, .
\end{equation}
\end{defn}

\vspace{\baselineskip}

\begin{defn}[Product of coupling constants of a Feynman graph] \label{defn:coupling_constants_feynman_graph}
Let \(\mathcal{Q}\) be a local QFT and \(\mathscr{G}_{\mathcal{Q}}\) the set of its Feynman graphs. Let furthermore \(\Gamma = \prod_{m = 1}^M \Gamma_m\) be a product of \(M \in \mathbb{N}\) connected Feynman graphs \(\Gamma_m \in \mathscr{G}_{\mathcal{Q}}\), \(1 \leq m \leq M\). Then we define the product of coupling constants of \(\Gamma\) as
\begin{equation}
	\cpl{\Gamma} := \prod_{m = 1}^M \left ( \frac{1}{\cpl{\res{\Gamma_m}}} \prod_{v \in \Gamma_m^{[0]}} \cpl{v} \right ) \, ,
\end{equation}
with
\begin{equation}
	\cpl{\res{\Gamma_m}} := \begin{cases} \text{coupling constant of the vertex } \res{\Gamma_m} & \text{if } \res{\Gamma_m} \in \mathcal{R}_{\mathcal{Q}}^{[0]}\\ 1 & \text{else} \end{cases} \, .
\end{equation}
\end{defn}

\vspace{\baselineskip}

\begin{defn}[Multi-index corresponding to the product of coupling constants of a Feynman graph] \label{defn:multi-index}
Let \(\mathcal{Q}\) be a local QFT and \(\mathscr{G}_{\mathcal{Q}}\) the set of its Feynman graphs. Let furthermore \(\Gamma = \prod_{m = 1}^M \Gamma_m\) be a product of \(M \in \mathbb{N}_{\geq 1}\) connected Feynman graphs \(\Gamma_m \in \mathscr{G}_{\mathcal{Q}}\), \(1 \leq m \leq M\), and \(\cpl{\Gamma}\) the product of its coupling constants. Then we define the multi-index \(\mathbf{C} \in \mathbb{Z}^N\), where \(N \in \mathbb{N}_{\geq 1}\) denotes the number of coupling constants of \(\Q\), corresponding to \(\cpl{\Gamma}\) as the vector counting the multiplicities of the several coupling constants in \(\cpl{\Gamma}\).
\end{defn}

\vspace{\baselineskip}

\begin{defn}[Connectedness and gradings of the renormalization Hopf algebra] \label{defn:connectedness_gradings_renormalization_hopf_algebra}
Let \(\mathcal{Q}\) be a local QFT as in \defnref{defn:renormalization_hopf_algebra}, \(\RQ\) the set of its residues and \(\mathscr{H}_{\mathcal{Q}}\) its renormalization Hopf algebra. In the following, sums and direct sums over multi-indices are understood componentwise, e.g.\ let \(\mathbf{G}^+ = ( \mathbf{G}^+_1, \cdots, \mathbf{G}^+_N) \in \mathbb{Z}^N\) and \(\mathbf{G}^- = (\mathbf{G}^-_1, \cdots, \mathbf{G}^-_N) \in \mathbb{Z}^N\) be two multi-indices with \(N \in \mathbb{N}_{\geq 0}\) entries, then we set
\begin{equation}
	\sum_{\mathbf{g} = \mathbf{G}^-}^{\mathbf{G}^+} := \sum_{\mathbf{g}_1 = \mathbf{G}^-_1}^{\mathbf{G}^+_1} \cdots \sum_{\mathbf{g}_N = \mathbf{G}^-_N}^{\mathbf{G}^+_N} \quad \text{and} \quad \bigoplus_{\mathbf{g} = \mathbf{G}^-}^{\mathbf{G}^+} := \bigoplus_{\mathbf{g}_1 = \mathbf{G}^-_1}^{\mathbf{G}^+_1} \cdots \bigoplus_{\mathbf{g}_N = \mathbf{G}^-_N}^{\mathbf{G}^+_N} \, .
\end{equation}
Furthermore, we set \(\mathbf{0} := (0_1 , \cdots , 0_N)\), \(- \mathbf{\infty} := ( - \infty_1 , \cdots , - \infty_N)\) and \(\mathbf{\infty} := (\infty_1 , \cdots , \infty_N)\). Then we consider the following three gradings of \(\mathscr{H}_{\mathcal{Q}}\) as a Hopf algebra which are further refinements of each other: The first grading comes from the first Betti number, i.e.\
\begin{equation}
	\mathscr{H}_{\mathcal{Q}} = \bigoplus_{L = 0}^\infty \left ( \mathscr{H}_{\mathcal{Q}} \right )_L \, .
\end{equation}
The second grading comes from the multi-index corresponding to the coupling constants of a Feynman graph, i.e.\
\begin{equation}
	\mathscr{H}_{\mathcal{Q}} = \bigoplus_{\mathbf{C} = - \mathbf{\infty}}^{\mathbf{\infty}} \left ( \mathscr{H}_{\mathcal{Q}} \right )_{\mathbf{C}} \, .
\end{equation}
Finally, the third grading comes from the multi-index corresponding to the vertex-residues of a Feynman graph, i.e.\
\begin{equation}
	\mathscr{H}_{\mathcal{Q}} = \bigoplus_{\mathbf{R} = - \mathbf{\infty}}^{\mathbf{\infty}} \left ( \mathscr{H}_{\mathcal{Q}} \right )_{\mathbf{R}} \, .
\end{equation}
Clearly, \(\left ( \mathscr{H}_{\mathcal{Q}} \right )_{L = 0} \cong \left ( \mathscr{H}_{\mathcal{Q}} \right )_{\mathbf{C} = \mathbf{0}} \cong \left ( \mathscr{H}_{\mathcal{Q}} \right )_{\mathbf{R} = \mathbf{0}} \cong \mathbb{Q}\), and thus \(\mathscr{H}_{\mathcal{Q}}\) is connected in all three gradings. In this article we use mainly the second grading.
\end{defn}

\vspace{\baselineskip}

\begin{defn}[Symmetry factor of a Feynman graph]
Let \(\mathcal{Q}\) be a local QFT and \(\mathscr{G}_{\mathcal{Q}}\) the set of its Feynman graphs. Let furthermore \(\Gamma \in \mathscr{G}_{\mathcal{Q}}\) be a Feynman graph. Then we define the symmetry factor of \(\Gamma\) as
\begin{equation}
	\sym{\Gamma} := \# \operatorname{Aut} \left ( \Gamma \right ) \, ,
\end{equation}
where \(\# \operatorname{Aut} \left ( \Gamma \right )\) is the rank of the automorphism group of \(\Gamma\), leaving its external leg structure fixed and respecting its vertex and edge types \(v \in \RQ\) and \(e \in \RQ\) for all \(v \in \Gamma^{[0]}\) and \(e \in \Gamma^{[1]}\), respectively.
\end{defn}

\vspace{\baselineskip}

\begin{defn}[Combinatorial Green's functions] \label{def:combinatorial_greens_functions}
Let \(\mathcal{Q}\) be a local QFT, \(\RQ\) the set of its residues and \(\mathscr{G}_{\mathcal{Q}}\) the set of its Feynman graphs. Given \(r \in \RQ\), we set for notational simplicity in this definition
\begin{equation}
	\precombgreen^r := \sum_{\substack{\Gamma \in \mathscr{G}_{\mathcal{Q}}\\\res{\Gamma} = r}} \frac{1}{\sym{\Gamma}} \Gamma \, .
\end{equation}
Then, we define the total combinatorial Green's function with residue \(r\) as the following sums:
\begin{equation}
	\combgreen^r := \begin{cases} \one + \precombgreen^r & \text{if \(r \in \RQ^{[0]}\)} \\ \one - \precombgreen^r & \text{if \(r \in \RQ^{[1]}\)} \\ \precombgreen^r & \text{else, i.e.\ \(r \notin \RQ\)} \end{cases}
\end{equation}
Finally, we denote the restriction of \(\combgreen^r\) to one of the gradings \(\mathbf{g}\) from \defnref{defn:connectedness_gradings_renormalization_hopf_algebra} via
\begin{equation}
	\rescombgreen^r_\mathbf{g} := \eval{\combgreen^r}_{\mathbf{g}} \, .
\end{equation}
\end{defn}

\vspace{\baselineskip}

\begin{rem} \label{rem:different_conventions_restricted_greens_functions}
We remark that restricted combinatorial Green's functions are in the literature often denoted via \(c^r_\mathbf{g}\) and differ by a minus sign from our definition. Our convention is such that they are given as the restriction of the total combinatorial Green's function to the corresponding grading, which requires minus signs for non-empty propagator graphs.
\end{rem}

\vspace{\baselineskip}

\begin{defn}[Hopf subalgebras for multiplicative renormalization] \label{defn:hopf_subalgebras_renormalization_hopf_algebra}
Let \(\mathcal{Q}\) be a local QFT as in \defnref{defn:renormalization_hopf_algebra}, \(\RQ\) its weighted residue set, \(\mathscr{H}_{\mathcal{Q}}\) its renormalization Hopf algebra and \(\rescombgreen^r_\mathbf{G} \in \mathscr{H}_{\mathcal{Q}}\) its restricted Green's functions, where \(\mathbf{G}\) and \(\mathbf{g}\) denotes one of the gradings from \defnref{defn:connectedness_gradings_renormalization_hopf_algebra}. We are interested in Hopf subalgebras which correspond to multiplicative renormalization, i.e.\ Hopf subalgebras of \(\mathscr{H}_{\mathcal{Q}}\) such that the coproduct factors on the restricted combinatorial Green's functions for all multi-indices \(\mathbf{G}\) in the following way:
\begin{equation}
	\Delta \left ( \rescombgreen^r_{\mathbf{G}} \right ) = \sum_{\mathbf{g} = \mathbf{0}}^{\mathbf{G}} \mathbf{P}_{\mathbf{g}} \left ( \rescombgreen^r_{\mathbf{G}} \right ) \otimes \rescombgreen^r_{\mathbf{G} - \mathbf{g}} \, , \label{eqn:hopf_subalgebras_multi-index}
\end{equation}
where \(\mathbf{P}_{\mathbf{g}} \left ( \rescombgreen^r_{\mathbf{G}} \right ) \in \mathscr{H}_{\mathcal{Q}}\) is a polynomial in graphs such that each summand has multi-index \(\mathbf{g}\).\footnote{There exist closed expressions for the polynomials \(\mathbf{P}_{\mathbf{g}} \left ( \rescombgreen^r_{\mathbf{G}} \right )\) as follows:
\begin{equation}
	\mathbf{P}_{\mathbf{g}} \left ( \rescombgreen^r_{\mathbf{G}} \right ) := \eval[2]{\left ( \overline{\combgreen}^r \overline{\mathfrak{Q}}^\mathbf{G} \right )}_{\mathbf{g}} \, ,
\end{equation}
where the overline denotes the restriction to divergent graphs and \(\mathfrak{Q}^v\) denotes so-called combinatorial charges, which were introduced in \cite{Yeats_PhD}. See also \cite{Prinz_3} for a treatment using the same notations and conventions.}
\end{defn}

\vspace{\baselineskip}

\begin{rem}[Hopf subalgebras and multiplicative renormalization]
We shortly remark the connection between Hopf subalgebras in the sense of \defnref{defn:hopf_subalgebras_renormalization_hopf_algebra} and multiplicative renormalization: Let \(\mathcal{Q}\) be a local QFT as in \defnref{defn:renormalization_hopf_algebra}, \(\mathscr{H}_{\mathcal{Q}}\) its renormalization Hopf algebra and \(\mathcal{M}^\varepsilon := \mathbb{C} \big [ \varepsilon^{-1}, \varepsilon \big ] \big ] \) the algebra of meromorphic functions in the variable \(\varepsilon\), called regularization parameter. Consider Feynman diagrams with fixed external momenta outside their Landau singularities. Then, Feynman rules are a character from \(\mathscr{H}_\mathcal{Q}\) to the algebra of formal integral expressions, and we can define integrated regularized Feynman rules \(\regFR\) for a given regularization scheme \(\mathscr{E}\), with regulator \(\varepsilon\), as the map
\begin{equation}
	\regFR \, : \quad \mathscr{H}_\mathcal{Q} \to \mathcal{M}^\varepsilon \, , \quad \Gamma \mapsto f_\Gamma^\varepsilon \, ,
\end{equation}
where \(f_\Gamma^\varepsilon\) is the meromorphic function obtained after regularizing the Feynman integrand of \(\Gamma\) via \(\mathscr{E}\) and then integrating it for a suitable value of \(\varepsilon\). From this we can proceed and define renormalized Feynman rules \(\renFR\) for a given renormalization scheme \(\mathscr{R}\), such that \((\mathcal{M}^\varepsilon, \mathscr{R})\) is a Rota-Baxter algebra of weight \(\lambda = -1\), as the map
\begin{equation}
	\renFR \, : \quad \mathscr{H}_\mathcal{Q} \to \mathbb{C} \subset \mathcal{M}^\varepsilon \, , \quad \Gamma \mapsto \underset{\varepsilon \mapsto 0}{\operatorname{Lim}} \left ( \counterterm \star \regFR \right ) \left ( \Gamma \right ) \, ,
\end{equation}
where \(\star\) is the convolution product from \defnref{defn:convolution_product} and \(\counterterm \left ( \cdot \right )\) is the counterterm map, recursively given via the normalization \(\counterterm \left ( \one \right ) := \mathbf{1}_{\mathcal{M}^\varepsilon} \equiv 1 \in \mathbb{C}\), and on \(\operatorname{Aug} \left ( \mathscr{H}_{\mathcal{Q}} \right )\) via
\begin{equation}
	\counterterm \, : \quad \mathscr{H}_\mathcal{Q} \to \mathcal{M}^\varepsilon \, , \quad \Gamma \mapsto - \mathscr{R} \circ \left ( \counterterm \star \left ( \regFR \circ \mathscr{P} \right ) \right ) \left ( \Gamma \right ) \, ,
\end{equation}
where \(\mathscr{P}\) is the projector onto the augmentation ideal from \defnref{defn:augmentation_ideal}. We remark that the counterterm map as well as the \(Z\)-factors are only well-defined for \(\varepsilon \neq 0\), as they consist of \(\mathscr{R}\)-divergent expressions. If the renormalization Hopf algebra \(\mathscr{H}_\mathcal{Q}\) possesses Hopf subalgebras in the sense of \defnref{defn:hopf_subalgebras_renormalization_hopf_algebra}, we can calculate the \(Z\)-factor for a given residue \(r \in \RQ\) via
\begin{equation}
	Z^r_{\mathscr{E}, \mathscr{R}} \left ( \varepsilon \right ) := \counterterm \left ( \combgreen^r \right ) \, .
\end{equation}
More details in this direction can be found in \cite{Prinz_3} and in \cite{Panzer,vSuijlekom_Multiplicative} (using a different notation).
\end{rem}

\vspace{\baselineskip}

\begin{rem}[Hopf subalgebras and different gradings]
Furthermore, we remark that the existence of the Hopf subalgebras from \defnref{defn:hopf_subalgebras_renormalization_hopf_algebra} depends crucially on the grading \(\mathbf{g}\). In particular, for the grading induced by the first Betti number these Hopf subalgebras exist if and only if the local QFT has only one vertex, for the coupling-constant grading if and only if the local QFT has for each vertex a different coupling constant and always for the residue grading \cite{Prinz_3}.
\end{rem}

\subsection{Associating the renormalization Hopf algebra to a local QFT} \label{ssec:associating_the_renormalization_hopf_algebra_to_a_given_local_qft}

In this subsection we describe a problem which may occur in the construction of the renormalization Hopf algebra \(\mathscr{H}_{\mathcal{Q}}\) to a given local QFT \(\mathcal{Q}\) using \defnref{defn:renormalization_hopf_algebra}. Then, we present four different solutions to still obtain a renormalization Hopf algebra (which are not isomorphic if the problem occurs) and discuss their physical interpretation.

\vspace{\baselineskip}

\begin{pro} \label{pro:problem}
Given a general local QFT \(\Q\), \defnref{defn:renormalization_hopf_algebra} may not yield a well-defined Hopf algebra due to the following reason: Let \(\Q\) be such that there exist divergent Feynman graphs \(\gamma \in \mathscr{G}_\Q\) whose residue is not in the residue set, i.e.\ we have \(\sdd{\gamma} \geq 0\) and \(\res{\gamma} \notin \RQ\). Then given any Feynman graph \(\Gamma \in \mathscr{G}_\Q\) with \(\gamma \in \mathscr{D} \left ( \Gamma \right )\), the quotients of the form \(\Gamma / \gamma\) for \(\gamma \subsetneq \Gamma\) are ill-defined, as they generate a new vertex \(\res{\gamma} \notin \RQ^{[0]}\). As a consequence, the definitions of the coproduct and the antipode are ill-defined as well.
\end{pro}

\vspace{\baselineskip}

\begin{rem}
In order to remedy \proref{pro:problem} we need to change some of the definitions. This is explained in the following \solsaref{sol:solution_1}{sol:solution_2}{sol:solution_3}{sol:solution_4}. In order to distinguish the different objects, we use script letters for the objects as defined in \defnref{defn:renormalization_hopf_algebra} and calligraphic letters for the modified definitions.
\end{rem}

\vspace{\baselineskip}

\begin{defn}[Feynman graphs generated by residue sets 2] \label{defn:feynman_graphs_generated_by_residue_sets_2}
Let \(\mathcal{Q}\) be a local QFT with residue set \(\mathcal{R}_{\mathcal{Q}}\). Recall from \defnref{defn:feynman_graphs_generated_by_residue_sets_1} that we denote by \(\mathscr{G}_{\mathcal{Q}}\) the set of all one-particle irreducible (1PI) Feynman graphs\footnote{Again, we remark that the use of 1PI Feynman graphs, rather than connected Feynman graphs, is justified by \thmref{thm:generalized_furry_theorem}, as is discussed in \remref{rem:1pi_generalized_furrys_theorem}.} that can be generated by the residue set \(\mathcal{R}_{\mathcal{Q}}\) of \(\mathcal{Q}\). Moreover, we define the set \(\mathcal{G}_{\mathcal{Q}}\) of all 1PI Feynman graphs of \(\mathcal{Q}\) which does not contain superficially divergent subgraphs whose residue is not in the residue set \(\mathcal{R}_{\mathcal{Q}}\), i.e.
\begin{equation}
	\mathcal{G}_{\mathcal{Q}} := \left \{ \Gamma \in \mathscr{G}_{\mathcal{Q}} \, \left \vert \; \Gamma = \coprod_{m = 1}^M \Gamma_m \, , \; M \in \mathbb{N} \, : \; \res{\Gamma_m} \in \mathcal{R}_{\mathcal{Q}} \, , \; \forall m : \sdd{\Gamma_m} \geq 0 \right . \right \} \, .
\end{equation}
This set will be used in \solref{sol:solution_1}.
\end{defn}

\vspace{\baselineskip}

\begin{defn}[Set of superficially divergent subgraphs of a Feynman graph 2] \label{defn:subset_of_divergent_subgraphs_2}
Let \(\mathcal{Q}\) be a local QFT and \(\Gamma \in \mathscr{G}_{\mathcal{Q}}\) a Feynman graph of \(\mathcal{Q}\). Recall from \defnref{defn:subset_of_divergent_subgraphs_1} that we denote by \(\mathscr{D} \left ( \Gamma \right )\) the set of superficially divergent subgraphs of \(\Gamma\) and by \(\mathscr{D}^\prime \left ( \Gamma \right )\) the set of superficially divergent proper subgraphs of \(\Gamma\). Moreover, we define the two additional sets \(\mathcal{D} \left ( \Gamma \right )\) and \(\mathcal{D}^\prime \left ( \Gamma \right )\), corresponding to \(\mathscr{D} \left ( \Gamma \right )\) and \(\mathscr{D}^\prime \left ( \Gamma \right )\), respectively, which do not contain Feynman graphs with superficially divergent subgraphs whose residue is not in the residue set \(\mathcal{R}_{\mathcal{Q}}\), i.e.\
\begin{subequations}
\begin{equation}
	\mathcal{D} \left ( \Gamma \right ) := \left \{ \gamma \in \mathscr{D} \left ( \Gamma \right ) \, \left \vert \; \gamma = \coprod_{m = 1}^M \gamma_m \, , \; M \in \mathbb{N} \, : \; \res{\gamma_m} \in \mathcal{R}_{\mathcal{Q}} \, , \; \forall m \right . \right \}
\end{equation}
and
\begin{equation}
	\mathcal{D}^\prime \left ( \Gamma \right ) := \left \{ \gamma \in \mathcal{D} \left ( \Gamma \right ) \, \left \vert \; \gamma \subsetneq \Gamma \right . \right \} \, ,
\end{equation}
\end{subequations}
These sets will be used in \solref{sol:solution_2}.
\end{defn}

\vspace{\baselineskip}

\begin{sol} \label{sol:solution_1}
The first solution to \proref{pro:problem} to replace the Feynman graph set \(\mathscr{G}_\mathcal{Q}\) from \defnref{defn:feynman_graphs_generated_by_residue_sets_1} by \(\mathcal{G}_\mathcal{Q}\) from \defnref{defn:feynman_graphs_generated_by_residue_sets_2}. Then, we can construct the renormalization Hopf algebra as in \defnref{defn:renormalization_hopf_algebra}, which we now denote by \(\HQ\).
\end{sol}

\vspace{\baselineskip}

\begin{sol} \label{sol:solution_2}
The second solution to \proref{pro:problem} is to simply remove all divergent Feynman graphs whose residue is not in the residue set from the sets of divergent subgraphs, i.e.\ replace the sets \(\mathscr{D} \left ( \cdot \right )\) and \(\mathscr{D}^\prime \left ( \cdot \right )\) from \defnref{defn:subset_of_divergent_subgraphs_1} by \(\mathcal{D} \left ( \cdot \right )\) and \(\mathcal{D}^\prime \left ( \cdot \right )\) from \defnref{defn:subset_of_divergent_subgraphs_2}, respectively. Then, we can construct the renormalization Hopf algebra as in \defnref{defn:renormalization_hopf_algebra}, which we now again denote by \(\HQ\).
\end{sol}

\vspace{\baselineskip}

\begin{sol} \label{sol:solution_3}
The third solution to \proref{pro:problem} is to add all missing residues to the residue set and set its weights to the value of a divergent Feynman graph with this particular residue (if there exist two or more such graphs with different superficial degree of divergence we take the highest for uniqueness, although it suffices to be divergent). This enlarges also the set of Feynman graphs. Then, we can define the Hopf algebra using this enlarged set of Feynman graphs as in \defnref{defn:renormalization_hopf_algebra}, which we now again denote by \(\HQ\).
\end{sol}

\vspace{\baselineskip}

\begin{sol} \label{sol:solution_4}
The fourth solution to \proref{pro:problem} works only in special cases: Given that there exist tree diagrams with the residue of a divergent Feynman graph whose residue is not in the residue set. Then, we can construct the renormalization Hopf algebra as in \defnref{defn:renormalization_hopf_algebra}, with the only difference that we define the shrinking process of the aforementioned graphs by replacing them with the sum over the corresponding trees, which we now again denote by \(\HQ\). We remark that this procedure is a priori non-local, but could be interpreted in the sense of a generalized Slavnov-Taylor-type identity.
\end{sol}

\vspace{\baselineskip}

\begin{rem}
Equivalently, the Hopf algebra from \solref{sol:solution_1} could be constructed in two different ways: The first possibility is to define \(\mathscr{H}_\Q\) as the \(\mathbb{Q}\)-algebra generated by the set \(\mathscr{G}_\Q\) of Feynman graphs with the multiplication and unit as in \defnref{defn:renormalization_hopf_algebra}. Then, we define the ideal
\begin{equation}
	\mathfrak{i}_\mathcal{Q} := \left ( \Gamma \in \left ( \mathscr{G}_\Q \setminus \GQ \right ) \right ) \, ,
\end{equation}
generated by all divergent Feynman graphs whose residue is not in the residue set, and consider the the quotient
\begin{equation}
	\HQ := \mathscr{H}_\Q / \mathfrak{i}_\mathcal{Q} \, , \label{eqn:ideal_Hopf_algebra}
\end{equation}
on which we can define the additional Hopf algebra structures as in \defnref{defn:renormalization_hopf_algebra}. The set of Feynman graphs from the quotient Hopf algebra \(\HQ\) is then precisely the set as defined in \defnref{defn:feynman_graphs_generated_by_residue_sets_2}. The second possibility is to use the Hopf algebra from \solref{sol:solution_3} and consider the quotient by the ideal given as the sum of \(\mathfrak{i}_\mathcal{Q}\) and the ideal generated by all Feynman graphs with vertices which are not in the residue set \(\RQ\), which is a Hopf ideal inside the Hopf algebra from \solref{sol:solution_3} via \colref{cor:residue_hopf_ideals} together with \propref{prop:sums_hopf_ideals_are_hopf_ideals}.
\end{rem}

\vspace{\baselineskip}

\begin{rem}[Physical interpretation] \label{rem:physical_interpretation}
Physically, \proref{pro:problem} states that divergent Feynman graphs whose residue is not in the residue set contribute in principle to a divergent Green's function which cannot be renormalized if the corresponding vertex is missing in the local QFT. However, there could be still three possibilities that the unrenormalized Feynman rules remedy the problem themselves: The first one is that the problematic Feynman graphs itself or the corresponding restricted combinatorial Green's functions turn out to be in the kernel of the unrenormalized Feynman rules which corresponds to \solref{sol:solution_1}. The second one is that the problematic Feynman graphs itself or the corresponding restricted combinatorial Green's functions turn out to be already finite when applying the unrenormalized Feynman rules which corresponds to \solref{sol:solution_2}. However, if this is not the case, we need to add the corresponding vertices with suitable Feynman rules in order to absorb the divergences of the corresponding restricted combinatorial Green's functions via multiplicative renormalization corresponding to \solref{sol:solution_3}. Luckily, for all established physical local QFTs this situation did not appear so far. Finally, the last scenario could be still circumvented, if corresponding tree graphs exist together with a generalized Slavnov-Taylor-type identity rendering this a priori non-local process local, which corresponds to \solref{sol:solution_4}.
\end{rem}

\vspace{\baselineskip}

\begin{exmp}[QED]
An illuminating example to \proref{pro:problem} is QED since we need to apply both, \solref{sol:solution_1} and \solref{sol:solution_3}: Consider QED with its combinatorics as a renormalizable local QFT (i.e.\ the superficial degree of divergence of a Feynman graph depends only on its external leg structure). Then, the Feynman graphs contributing to the three- and four-point function are divergent --- however in contrast to non-abelian quantum gauge theories, there is no three- and four-photon vertex present to absorb the corresponding divergences. Luckily, when summing all Feynman graphs of a given loop order we have the following cancellations after applying the unrenormalized Feynman rules: The Feynman graphs contributing to the three-point function cancel pairwise due to Furry's Theorem, cf.\ \thmref{thm:generalized_furry_theorem} for a generalization thereof and the divergences of the Feynman graphs contributing to the four-point function cancel pairwise due to gauge invariance \cite{Aldins_Brodsky_Dufner_Kinoshita}. Thus, QED is a renormalizable local QFT after all without the need to add a three- and four-photon vertex to the theory.
\end{exmp}

\vspace{\baselineskip}

\begin{rem}[The situation of QGR-QED]
The situation of QGR-QED is worse than the one for QED, since QGR is non-renormalizable as a local QFT (in particular, the superficial degree of divergence of a pure gravity Feynman graph depends only on its loop number). However, for the two-loop propagator Feynman graphs considered in this article the generalization of Furry's theorem given in \thmref{thm:generalized_furry_theorem} suffices, since Feynman graphs with self-loops (in the mathematical literature also known as ``roses'') vanish in the renormalization process. In particular, this ensures that the graviton-photon 2-point function vanishes, which would be in principle possible via quantum correction (whose corresponding Feynman graphs are divergent due to the superficial degree of divergence). Furthermore we remark that there exist also graphs whose external leg structure consists of any combination of even numbers of matter particles, such as four photon or four fermion graphs, as they could be glued together via gravitons. For this scenario we suggest \solref{sol:solution_4} to avoid introducing the corresponding vertices, as the corresponding trees are already part of the theory. The corresponding generalized Slavnov-Taylor-type identities will be studied in future work.
\end{rem}

\vspace{\baselineskip}

\begin{defn}[Renormalization Hopf algebra associated to a local QFT] \label{defn:renormalization_hopf_algebra_associated_to_a_local_qft}
Let \(\Q\) be a local QFT. Then we denote by \(\HQ\) one of the following Hopf algebras: If \defnref{defn:renormalization_hopf_algebra} is well-defined, we denote \(\HQ\) the renormalization Hopf algebra of \defnref{defn:renormalization_hopf_algebra}. Otherwise, we denote by \(\HQ\) the Hopf algebra obtained after applying \solsoref{sol:solution_1}{sol:solution_2}{sol:solution_3}{sol:solution_4} to \defnref{defn:renormalization_hopf_algebra}. We call \(\HQ\) ``the renormalization Hopf algebra associated to \(\Q\)''.
\end{defn}

\vspace{\baselineskip}

\begin{rem}
The motivation for \defnref{defn:renormalization_hopf_algebra_associated_to_a_local_qft} is to simplify notation, as for the realm of this work it is not necessary to distinguish between \solsaref{sol:solution_1}{sol:solution_2}{sol:solution_3}{sol:solution_4}.
\end{rem}

\subsection{Hopf ideals and the renormalization Hopf algebra}

Recall the definition of a Hopf ideal from \defnref{Hopf_ideals_and_subalgebras}. Now, we study general properties of Hopf ideals and then specialize to the renormalization Hopf algebra associated to a local QFT. To this end, we prove general results for Hopf ideals and a condition for Hopf ideals in the renormalization Hopf algebra which yields some particular Hopf ideals as corollaries. This is of physical interest since symmetries generating Hopf ideals have a similar subdivergence structure and are thus compatible with their renormalization treatment. In particular we show that the ideal generated by all Feynman graphs having at least one self-loop (``rose'') is a Hopf ideal. This is useful, as these Feynman integrals vanish for kinematic renormalization schemes and can thus already be set to zero in the renormalization Hopf algebra.

\vspace{\baselineskip}

\begin{prop}[Sums of Hopf ideals are Hopf ideals] \label{prop:sums_hopf_ideals_are_hopf_ideals}
Let \(H\) be a Hopf algebra over a field with characteristic zero and \(\left \{ \mathfrak{i}_n \right \}_{n=1}^N\) be a set of \(N\) non-empty Hopf ideals, where \(N \in \mathbb{N}_{\geq 1} \cup \infty\). Then the sum
\begin{equation}
	\mathfrak{i}_{\Sigma} := \sum_{n = 1}^N \mathfrak{i}_n \, ,
\end{equation}
i.e.\ the ideal \(\mathfrak{i}_{\Sigma}\) generated by sums of the generators of all Hopf ideals in the set \(\left \{ \mathfrak{i}_n \right \}_{n=1}^N\), is also a Hopf ideal in \(H\), i.e.\ \(\mathfrak{i}_{\Sigma}\) satisfies:
\begin{enumerate}
\item \(\Delta \left ( \mathfrak{i}_{\Sigma} \right ) \subseteq H \otimes \mathfrak{i}_{\Sigma} + \mathfrak{i}_{\Sigma} \otimes H\)
\item \(\coone \left ( \mathfrak{i}_{\Sigma} \right ) = 0\)
\item \(S \left ( \mathfrak{i}_{\Sigma} \right ) \subseteq \mathfrak{i}_{\Sigma}\)
\end{enumerate}
\end{prop}

\begin{proof}
This follows directly by the linearity of the involved maps.
\end{proof}

\vspace{\baselineskip}

\begin{prop}[Special products of Hopf ideals are Hopf ideals] \label{prop:sprods_hopf_ideals_are_hopf_ideals}
Let \(H\) be a Hopf algebra over a field with characteristic zero and \(\left \{ \mathfrak{i}_n \right \}_{n=1}^N\) be a set of \(N\) non-empty Hopf ideals, where \(N \in \mathbb{N}_{\geq 1} \cup \infty\) and \(k \in \left \{ 1, \ldots, N \right \}\). Then the special ``product''
\begin{equation}
	\mathfrak{i}_{\Pi k} := \prod_{n = 1}^N \mathfrak{i}_n + \mathfrak{i}_k \, ,
\end{equation}
i.e.\ the ideal \(\mathfrak{i}_{\Pi k}\) generated by products of the generators of all Hopf ideals and the sum of a particular Hopf ideal in the set \(\left \{ \mathfrak{i}_n \right \}_{n=1}^N\), is also a Hopf ideal in \(H\), i.e.\ \(\mathfrak{i}_{\Pi k}\) satisfies:
\begin{enumerate}
\item \(\Delta \left ( \mathfrak{i}_{\Pi k} \right ) \subseteq H \otimes \mathfrak{i}_{\Pi k} + \mathfrak{i}_{\Pi k} \otimes H\)
\item \(\coone \left ( \mathfrak{i}_{\Pi k} \right ) = 0\)
\item \(S \left ( \mathfrak{i}_{\Pi k} \right ) \subseteq \mathfrak{i}_{\Pi k}\)
\end{enumerate}
\end{prop}

\begin{proof}
This follows directly by the linearity and multiplicativity of the involved maps.
\end{proof}

\vspace{\baselineskip}

\begin{prop}[Condition for Hopf ideals\footnote{A similar result was found independently in \cite{Borinsky_PhD}.}] \label{prop:condition_hopf_ideals}
Let \(\Q\) be a local QFT with residue set \(\mathcal{R}_{\Q}\) and \(\HQ\) its renormalization Hopf algebra.\footnote{Defined either via \defnref{defn:renormalization_hopf_algebra} or, if this definition fails, as described in \proref{pro:problem} via one of the solutions described in \solsaref{sol:solution_1}{sol:solution_2}{sol:solution_3}{sol:solution_4}, cf.\ \defnref{defn:renormalization_hopf_algebra_associated_to_a_local_qft}.} Let furthermore \(\emptyset \subsetneq \mathscr{S} \subseteq \mathscr{G}_{\Q}\) be a non-empty set of Feynman graphs and denote via
\begin{equation}
	\mathfrak{i}_{\mathscr{S}} := \left ( \Gamma \in \mathscr{S} \right )_{\HQ}
\end{equation}
the ideal generated by the set \(\mathscr{S}\).\footnote{In general \(\mathfrak{i}_{\mathscr{S}}\) will not be finitely generated.} Then, \(\mathfrak{i}_{\mathscr{S}}\) is a Hopf ideal, i.e.\ \(\mathfrak{i}_{\mathscr{S}}\) satisfies:
\begin{enumerate}
\item \(\Delta \left ( \mathfrak{i}_{\mathscr{S}} \right ) \subseteq \HQ \otimes \mathfrak{i}_{\mathscr{S}} + \mathfrak{i}_{\mathscr{S}} \otimes \HQ\)
\item \(\coone \left ( \mathfrak{i}_{\mathscr{S}} \right ) = 0\)
\item \(S \left ( \mathfrak{i}_R \right ) \subseteq \mathfrak{i}_{\mathscr{S}}\)
\end{enumerate}
if and only if the set \(\mathscr{S}\) is such that for all graphs \(\Gamma \in \mathscr{S}\) and for all corresponding graphs \(\gamma \in \mathscr{D} \left ( \Gamma \right )\) we have that either \(\gamma \in \mathscr{S}\) or \(\Gamma / \gamma \in \mathscr{S}\)
\end{prop}

\begin{proof}
By the multiplicativity of the coproduct and the antipode, i.e.\
\begin{align}
	\Delta \left ( \Gamma_1 \Gamma_2 \right ) & = \Delta \left ( \Gamma_1 \right ) \Delta \left ( \Gamma_2 \right )\\
\intertext{and}
	S \left ( \Gamma_1 \Gamma_2 \right ) & = S \left ( \Gamma_1 \right ) S \left ( \Gamma_2 \right )
\end{align}
for \(\Gamma_1, \Gamma_2 \in \HQ\), it suffices to check all three conditions on the level of generators, i.e.\ let \(\Gamma \in \mathscr{S}\): Then, the first property follows from the definition of the coproduct
\begin{equation}
\begin{split}
	\Delta \left ( \Gamma \right ) & = \sum_{\gamma \in \mathscr{D} \left ( \Gamma \right )} \gamma \otimes \Gamma / \gamma\\
	& \subseteq \HQ \otimes \mathfrak{i}_{\mathscr{S}} + \mathfrak{i}_{\mathscr{S}} \otimes \HQ \, .
\end{split}
\end{equation}
The second property follows directly from the fact that \(\mathscr{S} \neq \emptyset\) and thus \(\mathfrak{i}_\mathscr{S} \neq \emptyset\). Finally, the third property follows from the normalization \(S \left ( \one \right ) = \one\) and the recursive definition of the antipode
\begin{equation}
\begin{split}
	S \left ( \Gamma \right ) & = - \Gamma \sum_{\gamma \in \mathscr{D}^\prime \left ( \Gamma \right )} S \left ( \gamma \right ) \Gamma / \gamma\\
	& \subseteq \mathfrak{i}_{\mathscr{S}} \, .
\end{split}
\end{equation}
\end{proof}

\vspace{\baselineskip}

\begin{cor}[Ideals generated by self-loops are Hopf ideals] \label{cor:tadpole_hopf_ideal}
Given the situation of \propref{prop:condition_hopf_ideals}, we denote by \(\emptyset \subsetneq \mathscr{S}_S \subsetneq \mathscr{G}_\Q\) the set of all Feynman graphs being or containing at least one self-loop (``rose''), i.e.\
\begin{equation}
	\mathscr{S}_S := \left \{ \Gamma \in \mathscr{G}_{\Q} \, \left \vert \; \exists \gamma \subseteq \Gamma \, : \; \# \gamma^{[0]} = 1 \, , \; \# \gamma^{[1]} \geq 1 \right . \right \} \, .
\end{equation}
Furthermore, we denote by \(\mathfrak{i}_S\) the ideal generated by the set \(\mathscr{S}_S\) in the renormalization Hopf algebra \(\HQ\). Then, \(\mathfrak{i}_S\) is a Hopf ideal in \(\HQ\).
\end{cor}

\begin{proof}
Let \(\Gamma \in \mathscr{S}_S\) be a generator of \(\mathfrak{i}_S\), i.e.\ there exists at least one subgraph \(\gamma_S \subseteq \Gamma\) which is a single self-loop, and let \(\gamma_\mathscr{D} \in \mathscr{D} \left ( \Gamma \right )\). Then we have either \(\gamma_S \subseteq \gamma_\mathscr{D}\) or \(\gamma_S \subseteq \Gamma / \gamma_\mathscr{D}\) and thus either \(\gamma_\mathscr{D} \in \mathscr{S}_S\) or \(\Gamma / \gamma_\mathscr{D} \in \mathscr{S}_S\). Finally, applying \propref{prop:condition_hopf_ideals} finishes the proof.
\end{proof}

\vspace{\baselineskip}

\begin{rem}
We remark that self-loop graphs are in the kernel of the renormalized Feynman rules for kinematic renormalization schemes, i.e.\ it makes physically sense to set them already in the renormalization Hopf algebra to zero.
\end{rem}

\vspace{\baselineskip}

\begin{cor}[Ideals generated by residues are Hopf ideals] \label{cor:residue_hopf_ideals}
Given the situation of \propref{prop:condition_hopf_ideals} and let \(r \in \RQ\) be a residue, then we denote by \(\emptyset \subsetneq \mathscr{S}_r \subsetneq \mathscr{G}_\Q\) the set of all Feynman graphs in \(\HQ\) having residue \(r\) or having a subgraph with residue \(r\), i.e.\
\begin{equation}
	\mathscr{S}_r := \left \{ \Gamma \in \mathscr{G}_{\Q} \, \left \vert \; \exists \gamma \subseteq \Gamma \, : \; \res{\gamma} = r \right . \right \} \, .
\end{equation}
Furthermore, we denote by \(\mathfrak{i}_r\) the ideal generated by the set \(\mathscr{S}_r\) in the renormalization Hopf algebra \(\HQ\). Then, \(\mathfrak{i}_r\) is a Hopf ideal in \(\HQ\).
\end{cor}

\begin{proof}
Let \(\Gamma \in \mathscr{G}_{\mathcal{Q}}\) be a generator of \(\mathfrak{i}_r\), i.e.\ there exists at least one subgraph \(\gamma_r \subseteq \Gamma\) with \(\res{\gamma_r} = r\), and let \(\gamma_\mathscr{D} \in \mathscr{D} \left ( \Gamma \right )\). Then the following three situations can occur:
\begin{enumerate}
\item \(\gamma_r \cap \gamma_\mathscr{D} = \emptyset\)
\item \(\gamma_r \cap \gamma_\mathscr{D} \neq \emptyset\) and \(\gamma_r \subsetneq \gamma_\mathscr{D}\)
\item \(\gamma_r \subseteq \gamma_\mathscr{D}\)
\end{enumerate}
Observe that for each respective situation we have:
\begin{enumerate}
\item \(\gamma_r \subseteq \Gamma / \gamma_\mathscr{D}\) and thus \(\Gamma / \gamma_\mathscr{D} \in \mathscr{S}_r\)
\item \(\gamma_r / \left ( \gamma_\mathscr{D} \cap \gamma_r \right ) \subseteq \Gamma / \gamma_\mathscr{D}\) and \(\res{\gamma_r / \left ( \gamma_\mathscr{D} \cap \gamma_r \right )} = r\) and thus \(\Gamma / \gamma_\mathscr{D} \in \mathscr{S}_r\)
\item \(\gamma_r \subseteq \gamma_\mathscr{D}\) and thus \(\gamma_\mathscr{D} \in \mathscr{S}_r\)
\end{enumerate}
Finally, applying \propref{prop:condition_hopf_ideals} finishes the proof.
\end{proof}

\vspace{\baselineskip}

\begin{rem}
\colref{cor:residue_hopf_ideals} states that it is compatible with renormalization to set all Feynman graphs with a given residue to zero in the corresponding renormalization Hopf algebra.
\end{rem}

\vspace{\baselineskip}

\begin{rem}
\propref{prop:sums_hopf_ideals_are_hopf_ideals} and \propref{prop:sprods_hopf_ideals_are_hopf_ideals} state that, in the case of the renormalization Hopf algebra associated to a local QFT, special linear combinations of Feynman graph sets \(\left \{ \mathscr{S}_n \right \}_{n=1}^N\), which each individually generate Hopf ideals \(\mathfrak{i}_n := \left ( \mathscr{S}_n \right ) _H\), again generate Hopf ideals. These results will be used e.g.\ in \cite{Prinz_3} to study the renormalization of general local Quantum Gauge Theories, similar to \cite{vSuijlekom_QED,vSuijlekom_QCD}.
\end{rem}

\subsection{The renormalization Hopf algebra of QGR-QED} \label{ssec:the_renormalization_hopf_algebra_of_qgr_qed}

Now, we examine renormalization Hopf algebra associated to QGR-QED, cf.\ \defnref{defn:renormalization_hopf_algebra_associated_to_a_local_qft}. Furthermore, we proove a generalization of Furry's Theorem in \thmref{thm:generalized_furry_theorem} which also includes external gravitons and graviton ghosts. This is in particular useful, since the calculation shows that for the two-loop propagator combinatorial Green's functions, given explicitly in \ssecref{ssec:combinatorial_green_functions}, all divergent subgraphs whose residue is not in the set \(\mathcal{R}_{\text{QGR-QED}}\) are either of this type or belong to the ideal generated by self-loop graphs (``roses''). Thus, when constructing the renormalization Hopf algebra of QGR-QED for the realm of two-loop propagator graphs, we first consider the quotient via the ideal generated by all amplitudes captured via Generalized Furry's Theorem as in \solref{sol:solution_1} and then consider the quotient via the Hopf ideal generated by all self-loop graphs. Finally, we conclude that all graphs that are set to zero in the construction of the renormalization Hopf algebra of QGR-QED are in the kernel of the renormalized Feynman rules. This is due to the fact that self-loops vanish in kinematic renormalization schemes and the Generalized Furry's Theorem. In the following, fermion edges are denoted by \(\residue{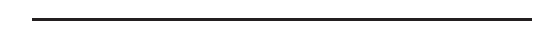}\), photon edges by \(\residue{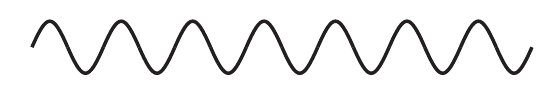}\), graviton edges by \(\residue{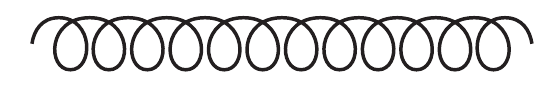}\), photon-ghost edges by \(\residue{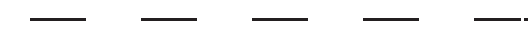}\) and graviton-ghost edges by \(\residue{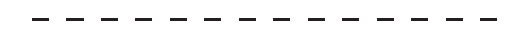}\).

\vspace{\baselineskip}

\begin{rem}[Residue set of QGR-QED] \label{rem:weighted_residue_set_qgr_qed}
Recall from \defnref{defn:residue_set_of_a_local_qft_and_their_weights}, how to obtain the corresponding residue set \(\mathcal{R}_{\text{QGR-QED}}\) from the QGR-QED Lagrange density
\begin{equation}
\begin{split}
\mathcal{L}_{\text{QGR-QED}} & = \left ( \leh + \lmd \right ) \dif V_g \\ & \phantom{=} + \mathcal{L}_{\text{GF}} + \mathcal{L}_{\text{Ghost}} \, ,
\end{split}
\end{equation}
which was introduced in \ssecref{ssec:Lagrange_Density_of_QGR-QED}.
As we restrict the residue set to \(\order{\gcoupling^2}\) in this article, we obtain the following finite residue set \(\mathcal{R}_{\text{QGR-QED}}\). It splits, as usual, into a disjoint union of vertex residues \(\mathcal{R}_{\text{QGR-QED}}^{[0]}\) and edge residues \(\mathcal{R}_{\text{QGR-QED}}^{[1]}\), i.e.\
\begin{equation}
	\mathcal{R}_{\text{QGR-QED}} = \mathcal{R}_{\text{QGR-QED}}^{[0]} \amalg \mathcal{R}_{\text{QGR-QED}}^{[1]} \, .
\end{equation}
Concretely, we have
\begin{subequations}
\begin{align}
\begin{split}
	\mathcal{R}_{\text{QGR-QED}}^{[0]} & = \left \{ \tresidue{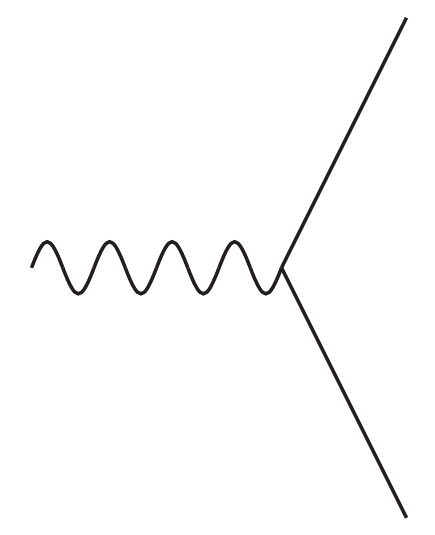} , \tresidue{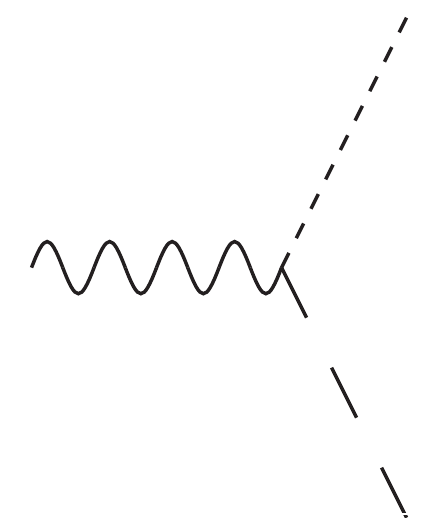} , \tresidue{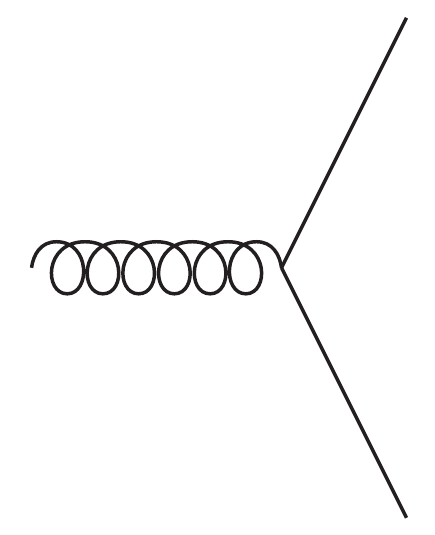} , \tresidue{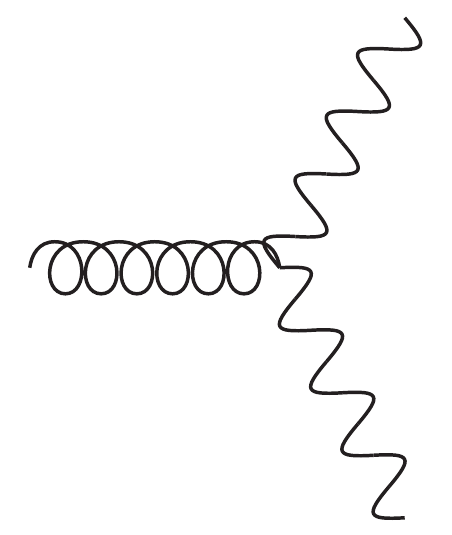} , \tresidue{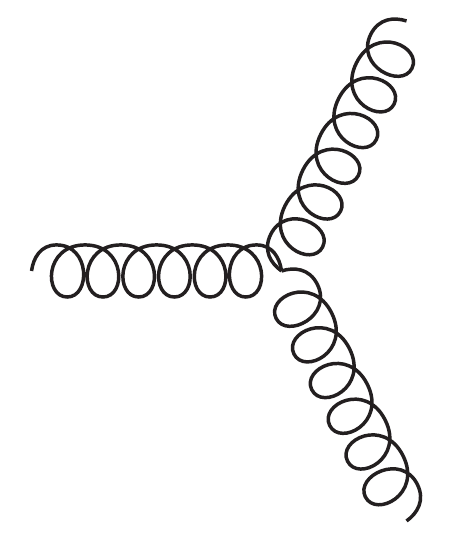} , \tresidue{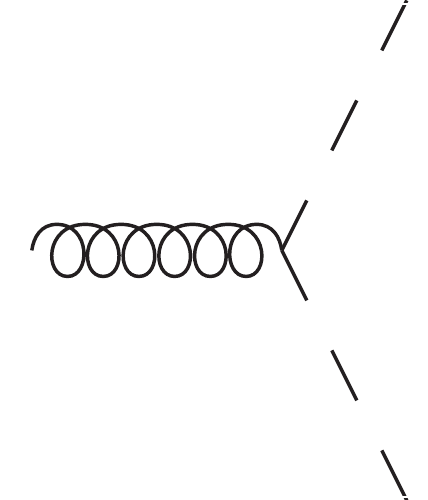} , \tresidue{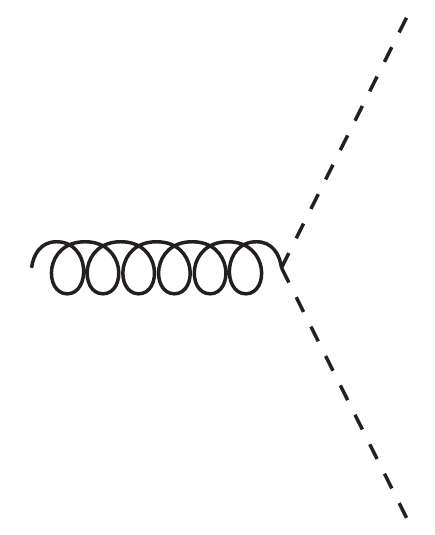} , \residue{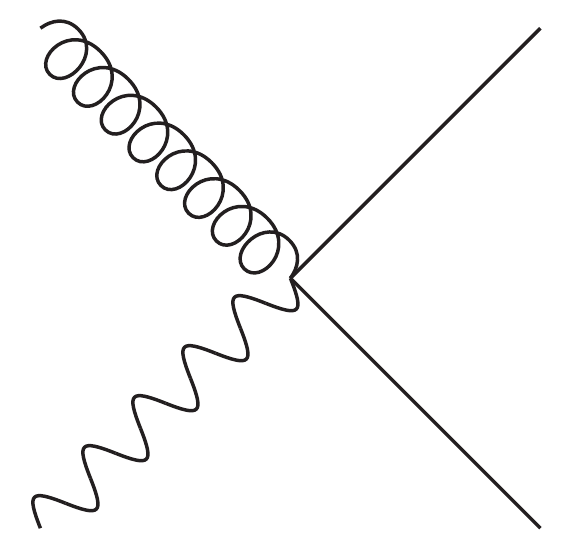} , \residue{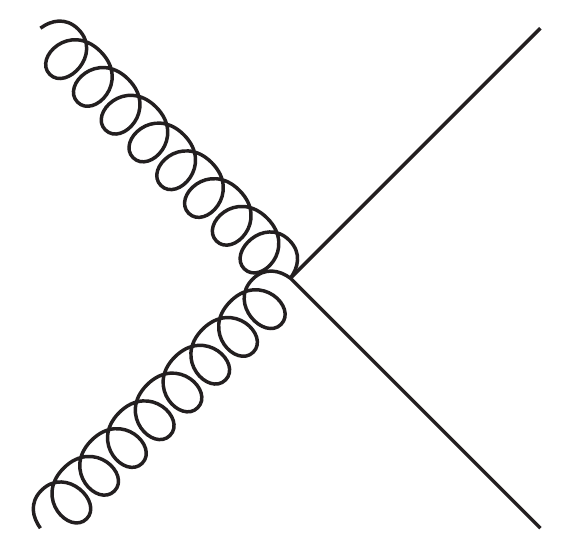} , \right . \\
& \phantom{ = \{} \left . \residue{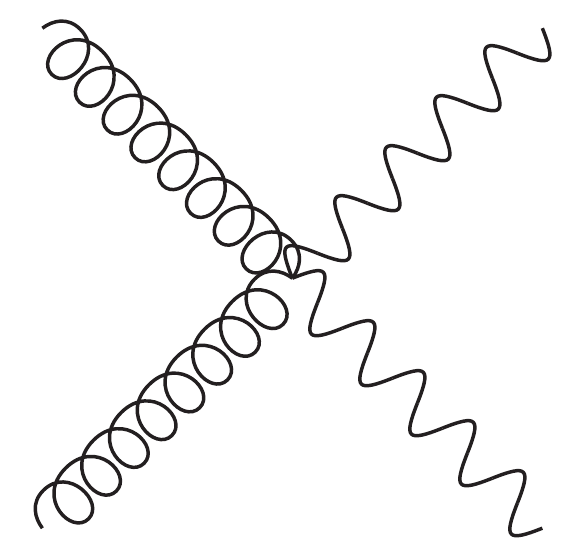} , \residue{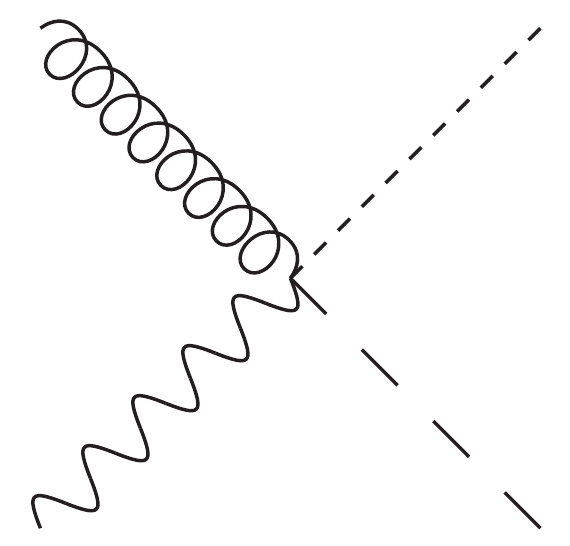} , \residue{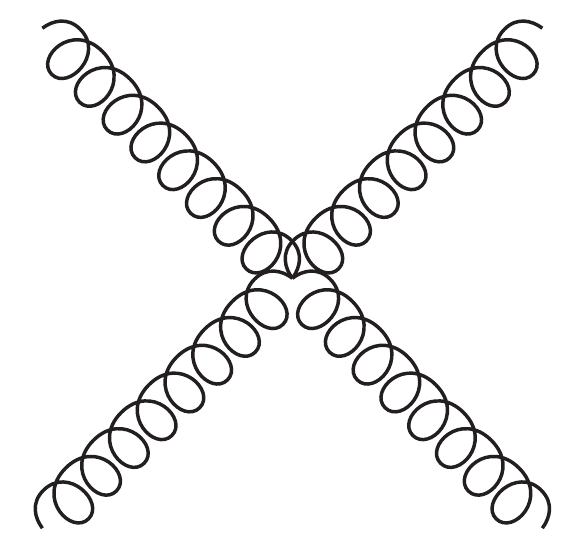} , \residue{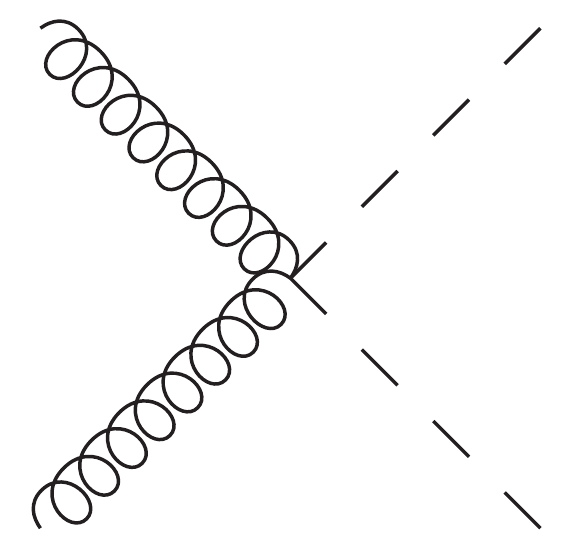}, \residue{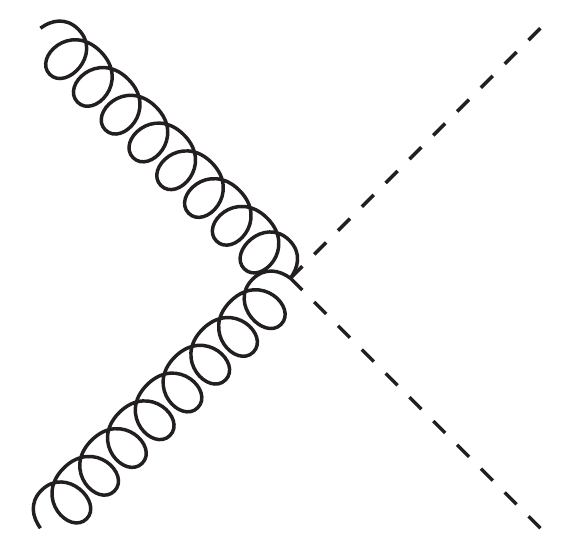} \right \}
\end{split}
\intertext{and}
\mathcal{R}_{\text{QGR-QED}}^{[1]} & = \left \{ \residue{ff} , \residue{pp} , \residue{gg} , \residue{aa} , \residue{hh} \right \} \, . \label{eqn:edge_residue_qgr-qed}
\end{align}
\end{subequations}
Furthermore, their corresponding weights and coupling constants read:
{\allowdisplaybreaks
\begin{subequations}
\begin{alignat}{2}
\sdd{\residue{ff}} & = 1 \qquad \qquad \qquad \cpl{\residue{ff}} && = 1\\
\sdd{\residue{pp}} & = 2 \qquad \qquad \qquad \cpl{\residue{pp}} && = 1\\
\sdd{\residue{gg}} & = 2 \qquad \qquad \qquad \cpl{\residue{gg}} && = 1\\
\sdd{\residue{aa}} & = 2 \qquad \qquad \qquad \cpl{\residue{aa}} && = 1\\
\sdd{\residue{hh}} & = 2 \qquad \qquad \qquad \cpl{\residue{hh}} && = 1\\
\sdd{\tresidue{pff}} & = 0 \qquad \qquad \qquad \cpl{\tresidue{pff}} && = \ecoupling\\
\sdd{\tresidue{pah}} & = 2 \qquad \qquad \qquad \cpl{\tresidue{pah}} && = \gcoupling\\
\sdd{\tresidue{gff}} & = 1 \qquad \qquad \qquad \cpl{\tresidue{gff}} && = \gcoupling\\
\sdd{\tresidue{gpp}} & = 2 \qquad \qquad \qquad \cpl{\tresidue{gpp}} && = \gcoupling\\
\sdd{\tresidue{ggg}} & = 2 \qquad \qquad \qquad \cpl{\tresidue{ggg}} && = \gcoupling\\
\sdd{\residue{gpff}} & = 0 \qquad \qquad \qquad \cpl{\residue{gpff}} && = \ecoupling \gcoupling\\
\sdd{\residue{ggff}} & = 1 \qquad \qquad \qquad \cpl{\residue{ggff}} && = \gcoupling^2\\
\sdd{\residue{ggpp}} & = 2 \qquad \qquad \qquad \cpl{\residue{ggpp}} && = \gcoupling^2\\
\sdd{\residue{gpah}} & = 2 \qquad \qquad \qquad \cpl{\residue{gpah}} && = \gcoupling^2\\
\sdd{\residue{gggg}} & = 2 \qquad \qquad \qquad \cpl{\residue{gggg}} && = \gcoupling^2\\
\sdd{\residue{ggaa}} & = 2 \qquad \qquad \qquad \cpl{\residue{ggaa}} && = \gcoupling^2\\
\sdd{\residue{gghh}} & = 2 \qquad \qquad \qquad \cpl{\residue{gghh}} && = \gcoupling^2
\end{alignat}
\end{subequations}
}
\end{rem}

\vspace{\baselineskip}

\begin{thm}[Generalized Furry's Theorem] \label{thm:generalized_furry_theorem}
Consider QGR-QED with the Lagrange density \(\mathcal{L}_\textup{QGR-QED}\) given in \eqnref{eqn:lagrange_qgr-qed_complete}. Let \(r\) be the residue of an amplitude, set \(p \left ( r \right )\) to be the sum of its external photon and photon ghost edges and \(f \left ( r \right )\) to be the number of external fermion edges. Then the corresponding restricted combinatorial Green's functions \(\rescombgreen^r_\mathbf{r}\) are in the kernel of the unrenormalized Feynman rules \(\Phi_\textup{QGR-QED} \left ( \cdot \right )\) for each residue grading \(\mathbf{r}\) individually if \(p \left ( r \right )\) is odd and \(f \left ( r \right )\) is zero.\footnote{Obviously, it is then also true for the other two gradings presented in \defnref{defn:connectedness_gradings_renormalization_hopf_algebra}, as the residue grading is the finest.}
\end{thm}

\begin{proof}
This proof is the only part of this article where we need to consider explicitly the orientation of fermion edges. We denote by \(\mathcal{A}_\text{QGR-QED}^\text{Furry}\) the set of all amplitudes with residues \(r_i\) such that \(p \left ( r_i \right )\) is odd and \(f \left ( r_i \right )\) is zero. Furthermore, let
\begin{equation}
	\mathscr{C} \, : \quad \mathscr{G}_{\text{QGR-QED}} \to \mathscr{G}_{\text{QGR-QED}} \, , \quad \Gamma \mapsto \Gamma^{\mathscr{C}} \, ,
\end{equation}
be the fermion charge-conjugation operator on Feynman graphs, i.e.\ \(\Gamma^{\mathscr{C}}\) is the Feynman graph \(\Gamma\) with the orientation of all fermion lines reversed. Notice that \(\mathscr{C}\) is an involution. Moreover, let \(\FR{\cdot}\) denote the corresponding Feynman rules. Then, we denote
\begin{equation}
	\FRC{\Gamma} := \FR{\Gamma^{\mathscr{C}}} \, .
\end{equation}
Observe that we can factor \(\mathscr{C}\) on the level of Feynman rules into two commuting operations, i.e.
\begin{equation}
	\FRC{\cdot} \equiv \FRCRO{\cdot} \equiv \FRCOR{\cdot} \, ,
\end{equation}
where \(\mathscr{C}_{\text{R}}\) alters the Feynman rules to fermion charge-conjugated residues and \(\mathscr{C}_{\text{O}}\) reverses the order of the Dirac-matrices. Now, let \(\Gamma \in \mathscr{G}_{\text{QGR-QED}}\) be a Feynman graph with \(\res{\Gamma} \in \mathcal{A}_\text{QGR-QED}^\text{Furry}\). Then we immediately have
\begin{equation}
	\FRC{\Gamma} = \FRCR{\Gamma} \, ,
\end{equation}
as traces of Dirac matrices are invariant under a reversion of their order. Now, we analyze the action of \(\mathscr{C}_\text{R}\) on the Feynman rules acting on individual residues. We claim that the residue set \(\mathcal{R}_{\text{QGR-QED}}\) splits into a disjoint union
\begin{equation}
	\mathcal{R}_{\text{QGR-QED}} = \mathcal{R}_{\text{QGR-QED}}^+ \amalg \mathcal{R}_{\text{QGR-QED}}^-
\end{equation}
such that we have for all residues \(r_+ \in \mathcal{R}_{\text{QGR-QED}}^+\)
\begin{subequations}
\begin{align}
	\FRCR{r_+} & = + \FR{r_+} \label{eqn:FRCa}
\intertext{and for all residues \(r_- \in \mathcal{R}_{\text{QGR-QED}}^-\)}
	\FRCR{r_-} & = - \FR{r_-} \, .
\end{align}
\end{subequations}
Indeed, reversing the fermion particle flow is equivalent to reversing the fermion momentum, i.e.\ replace it by its negative. Since the corresponding Feynman rules are either independent or linear in fermion momenta, the independent ones belong to the set \(\mathcal{R}_{\text{QGR-QED}}^+\) and the linear ones to the set \(\mathcal{R}_{\text{QGR-QED}}^-\). In particular, it follows from the structure of the Lagrange density given in \eqnref{eqn:lagrange_qgr-qed_complete} that this property is independent of the number of gravitons attached to a vertex residue. Thus, it suffices to check this property on the level of propagator and three-valent vertex residues. Obviously, residues without fermions are independent of fermion momenta and thus belong to the set \(\mathcal{R}_{\text{QGR-QED}}^+\) as is the photon-fermion-antifermion vertex residue. Contrary, the fermion propagator and the graviton-fermion-antifermion vertex residue are linear in the fermion momenta and thus belong to the set \(\mathcal{R}_{\text{QGR-QED}}^-\). In total, we obtain (where \(\mathcal{R}_{\text{QGR-QED}}^{[0], 3}\) stands for three-valent vertices):
\begin{subequations}
\begin{align}
	\mathcal{R}_{\text{QGR-QED}}^{[0], 3, +} & = \left \{ \tresidue{pff} , \tresidue{pah} , \tresidue{gpp} , \tresidue{ggg} , \tresidue{gaa} , \tresidue{ghh} \right \}\\
	\mathcal{R}_{\text{QGR-QED}}^{[0], 3, -} & = \left \{ \tresidue{gff} \right \}
\intertext{and}
	\mathcal{R}_{\text{QGR-QED}}^{[1], +} & = \left \{ \residue{pp} , \residue{gg} , \residue{aa} , \residue{hh} \right \}\\
	\mathcal{R}_{\text{QGR-QED}}^{[1], -} & = \left \{ \residue{ff} \right \}
\end{align}
\end{subequations}
Now, let \(r \in \mathcal{A}_{\text{QGR-QED}}^{\text{Furry}}\) and we consider corresponding combinatorial Green's functions \(\rescombgreen^r_\mathbf{r}\) for some residue-grading multi-index \(\mathbf{r}\). We claim that the involution \(\mathscr{C}\) is fixed-point free when restricted to \(\rescombgreen^r_\mathbf{r}\), by abuse of notation now considered as a set. To show this claim, we first observe that the photon-fermion-antifermion vertex with an arbitrary number of gravitons attached to it is the only residue which allows to change \(p \left ( r \right )\). Thus, the requirement on \(p \left ( r \right )\) being odd implies that every Feynman graph \(\Gamma \in \rescombgreen^r_\mathbf{r}\) has at least one fermion loop. Thus, \(\mathscr{C}\) applied to any Feynman graph \(\Gamma \in \rescombgreen^r_\mathbf{r}\) is not the identity and furthermore \(\Gamma^{\mathscr{C}} \in \rescombgreen^r_\mathbf{r}\), since residues and gradings are preserved by \(\mathscr{C}\). Finally, we claim that for \(\Gamma \in \rescombgreen^r_\mathbf{r}\) the sum \(\Gamma + \Gamma^{\mathscr{C}}\) is in the kernel of the Feynman rules, i.e.\
\begin{equation}
	\FR{\Gamma + \Gamma^{\mathscr{C}}} = 0 \, . \label{eqn:furry_thm_kernel_feynman_rules}
\end{equation}
We conclude this by showing that in this case \(\Gamma\) is build from an odd number of residues belonging to the set \(\mathcal{R}_{\text{QGR-QED}}^-\), and thus
\begin{equation}
	\FR{\Gamma^{\mathscr{C}}} = - \FR{\Gamma}
\end{equation}
which directly implies \eqnref{eqn:furry_thm_kernel_feynman_rules} by the linearity of the Feynman rules. Indeed, since \(\res{\Gamma} \in \mathcal{A}_{\text{QGR-QED}}^{\text{Furry}}\) we conclude that \(\Gamma\) has to consist of an odd number of photon-fermion-antifermion (modulo an arbitrary number of gravitons) vertices and thus in particular at least one closed fermion loop as noted before. Since the fermion loops are closed (as \(\Gamma\) contains no external fermion edges), the number of fermion propagators equals the number of fermion vertices. Moreover, we remind that photon-fermion-antifermion (modulo an arbitrary number of gravitons) vertices are in the set \(\mathcal{R}_{\text{QGR-QED}}^+\), whereas fermion propagators and gravitons-fermion-antifermion vertices are in the set \(\mathcal{R}_{\text{QGR-QED}}^-\). But as we need to have an odd number of photon-fermion-antifermion (modulo an arbitrary number of gravitons) vertices by the previous argument, \(\Gamma\) needs to have an odd number of residues from the set \(\mathcal{R}_{\text{QGR-QED}}^-\) which finishes the proof.
\end{proof}

\vspace{\baselineskip}

\begin{rem} \label{rem:1pi_generalized_furrys_theorem}
\thmref{thm:generalized_furry_theorem} also justifies the restriction of connected Feynman graphs to 1PI Feynman graphs in \defnref{defn:feynman_graphs_generated_by_residue_sets_1} and \defnref{defn:feynman_graphs_generated_by_residue_sets_2}. More precisely, since the amplitude of the two-point vertex function of a photon and a graviton vanishes, the set of connected Feynman graphs which are not 1PI is a trivial extension of the set of 1PI Feynman graphs.
\end{rem}

\section{Combinatorial Green's functions, the coproduct structure and obstructions to multiplicative renormalization in QGR-QED} \label{sec:coproduct_structure_of_the_greens_function}

Now, we consider the coproduct structure on the two-loop propagator graphs in QGR-QED. We associate to QGR-QED its renormalization Hopf algebra \(\mathcal{H}_{\text{QGR-QED}}\), as was described in \ssecref{ssec:the_renormalization_hopf_algebra_of_qgr_qed}. Let furthermore \(\rescombgreen^r_{\mathbf{c}}\) be the combinatorial Green's function with residue \(r \in \mathcal{R}_{\text{QGR-QED}}\) and multi-index \(\mathbf{c} \in \mathbb{Z}^2\), as was defined in \defnref{def:combinatorial_greens_functions}. The multi-index \(\mathbf{c}\), introduced in \defnref{defn:multi-index}, is defined such that \(\mathbf{c} = (m,n)\) corresponds to a graphs with coupling constants of order \(\order{\gcoupling^m \ecoupling^n}\), with \(m, n \in \mathbb{Z}\). In the following subsections, we present the corresponding Green's functions and their coproducts. We remind that fermion edges are denoted by \(\residue{ff}\), photon edges by \(\residue{pp}\), graviton edges by \(\residue{gg}\), photon-ghost edges by \(\residue{aa}\) and graviton-ghost edges by \(\residue{hh}\). In this work, we draw oriented Feynman graph edges, i.e.\ fermion, photon-ghost and graviton-ghost edges, without orientation. This is understood as the sum of all Feynman graphs having all possible orientations.

\subsection{Combinatorial Green's functions in QGR-QED} \label{ssec:combinatorial_green_functions}

Now, we present the restricted combinatorial Green's functions: Non-symmetric graphs are drawn only once, but with the corresponding multiplicity.\footnote{The minus sings for propagator graphs are due to \defnref{def:combinatorial_greens_functions}, cf.\ \remref{rem:different_conventions_restricted_greens_functions}.} We remark that in the realm of two loop propagator graphs only one graph in \(\rescombgreen^{\cgreen{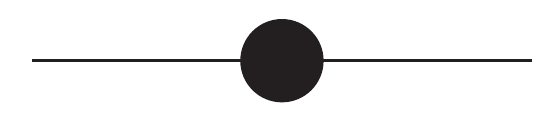}}_{(2,2)}\), \eqnref{eqn:c-ff_2-2}, contains a problematic subgraph which, however, is captured by Generalized Furry's Theorem, stated in \thmref{thm:generalized_furry_theorem}, and thus is presented in brackets. Below, we also present the one-loop two- and three-point amplitudes which are set to zero due to Generalized Furry's Theorem.

{\allowdisplaybreaks
\begin{align}
\rescombgreen^{\cgreen{c-ff}}_{(0,2)} & = - \ograph{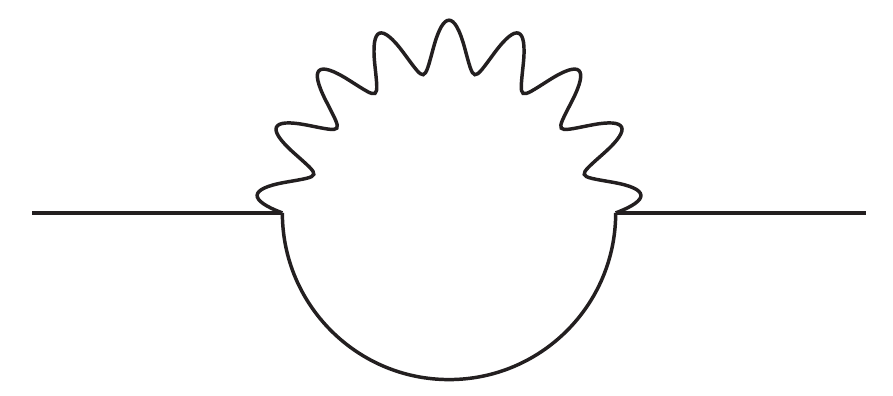}
\\
\rescombgreen^{\cgreen{c-ff}}_{(2,0)} & = - \ograph{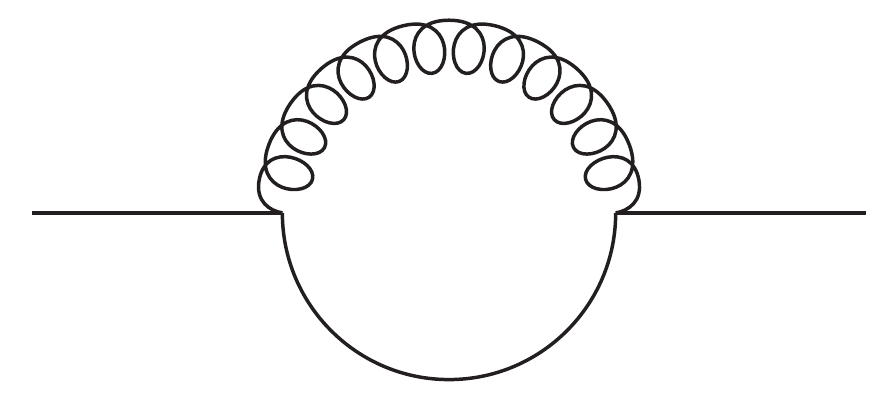}
\\
\rescombgreen^{\cgreen{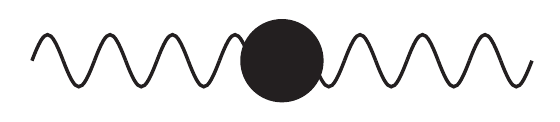}}_{(0,2)} & = - \frac{1}{2} \ograph{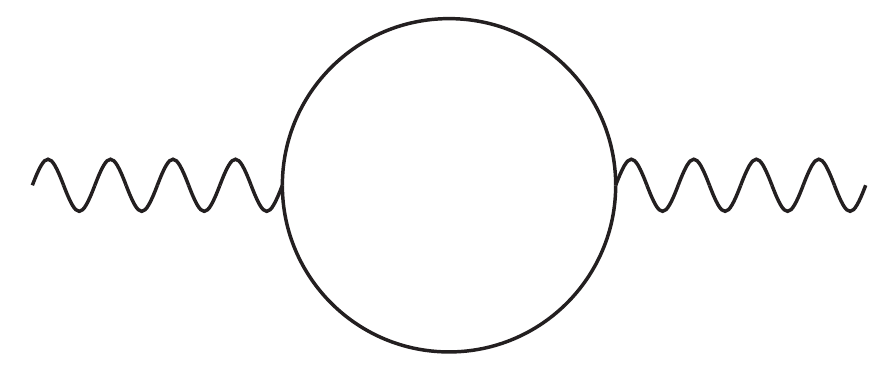}
\\
\rescombgreen^{\cgreen{c-pp}}_{(2,0)} & = - \ograph{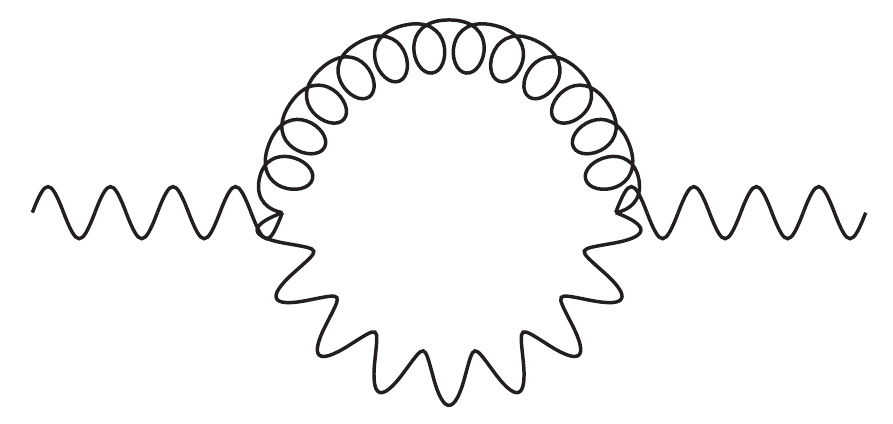} - \ograph{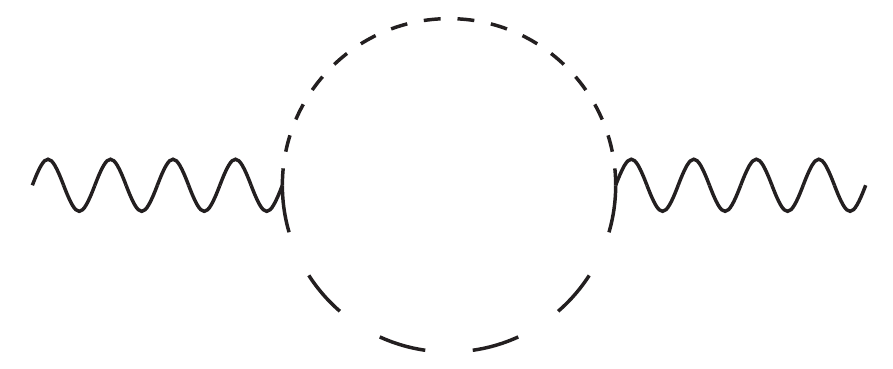}
\\
\rescombgreen^{\cgreen{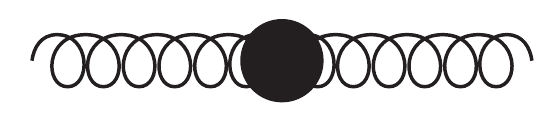}}_{(0,2)} & = 0
\\
\begin{split}
\rescombgreen^{\cgreen{c-gg}}_{(2,0)} & = - \frac{1}{2} \ograph{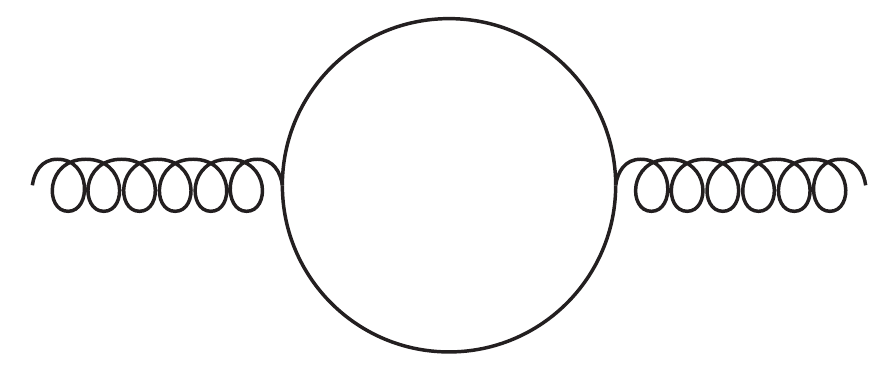} - \frac{1}{2} \ograph{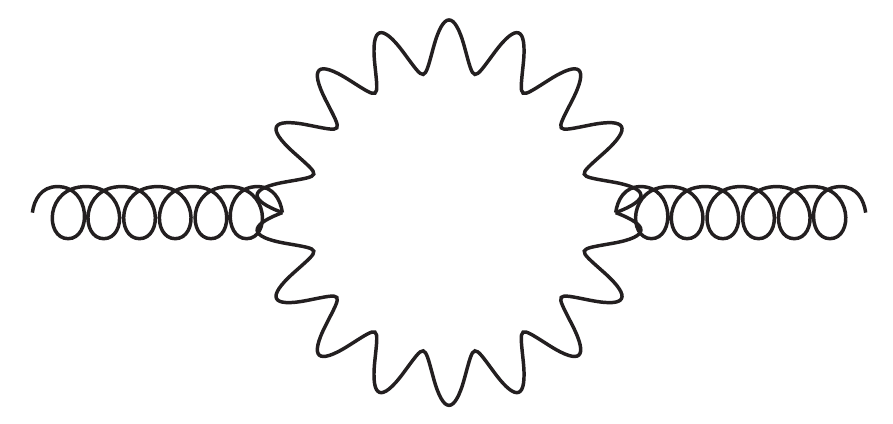} - \frac{1}{2} \ograph{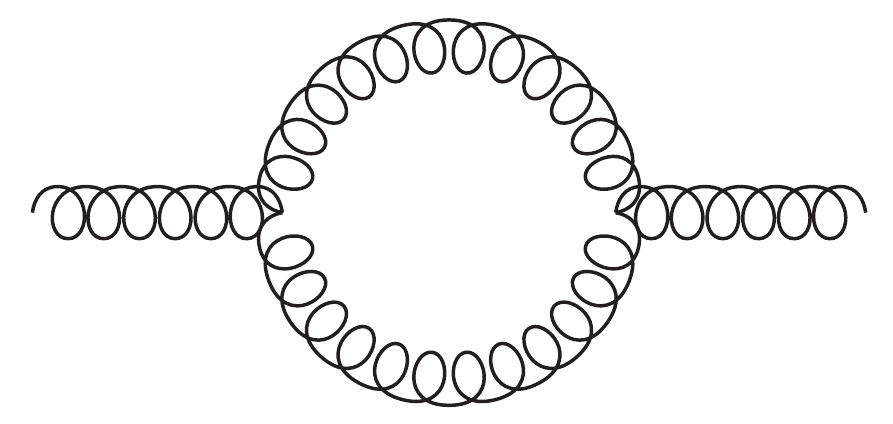} - \frac{1}{2} \ograph{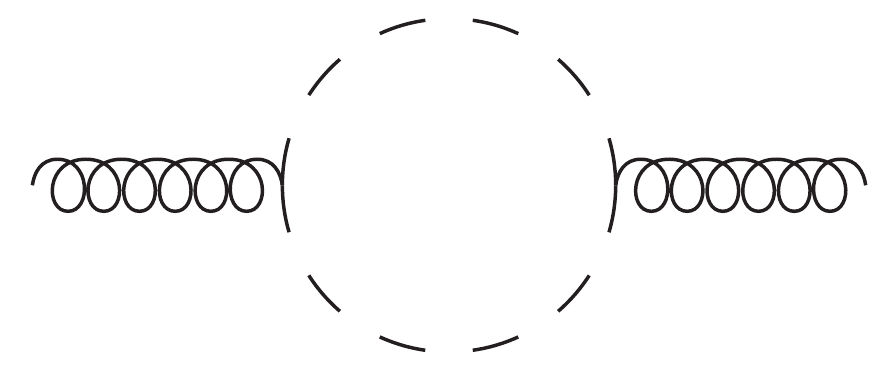} \\ & \phantom{ = } - \frac{1}{2} \ograph{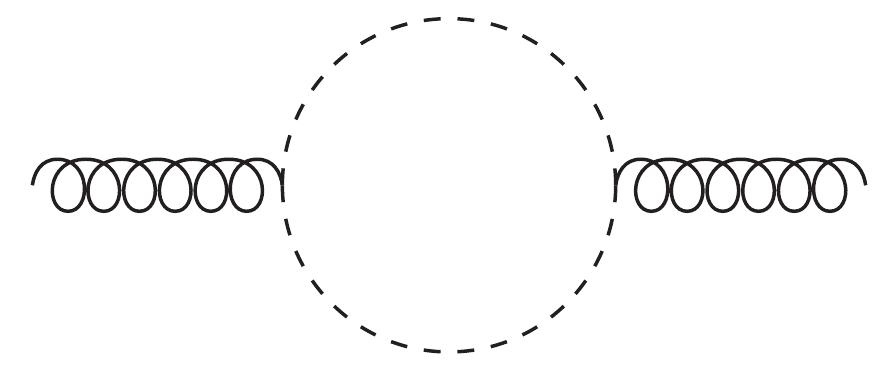}
\end{split}
\\
\rescombgreen^{\cgreen{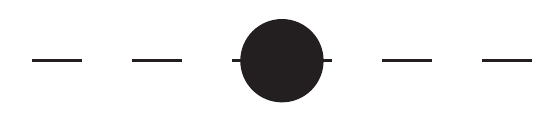}}_{(0,2)} & = 0
\\
\rescombgreen^{\cgreen{c-aa}}_{(2,0)} & = - \ograph{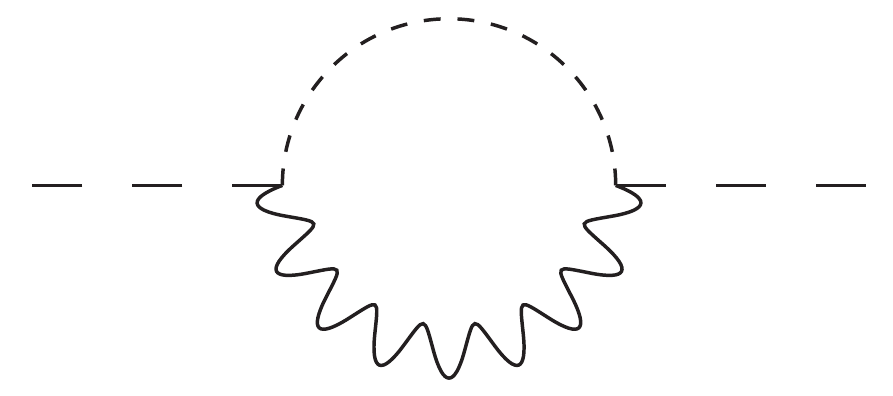} - \ograph{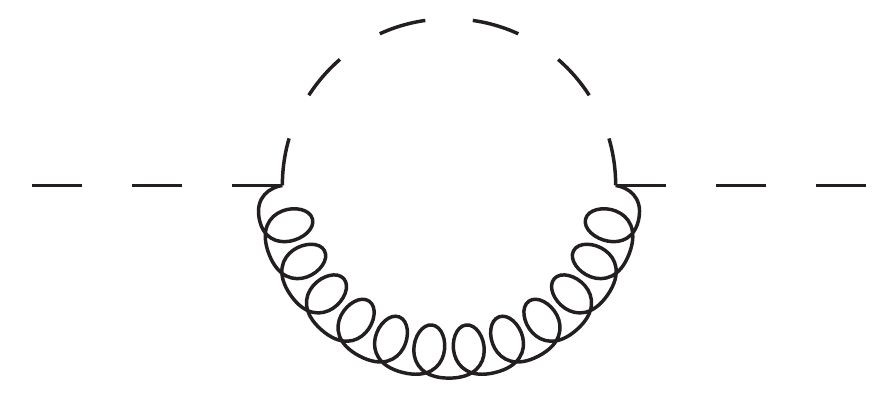}
\\
\rescombgreen^{\cgreen{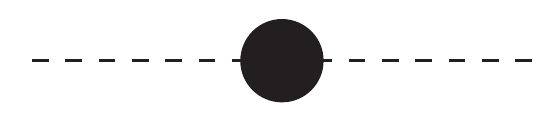}}_{(0,2)} & = 0
\\
\rescombgreen^{\cgreen{c-hh}}_{(2,0)} & = - \ograph{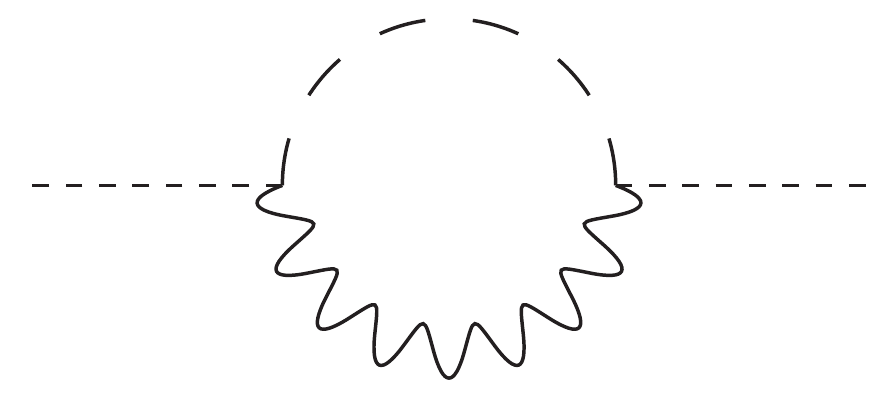} - \ograph{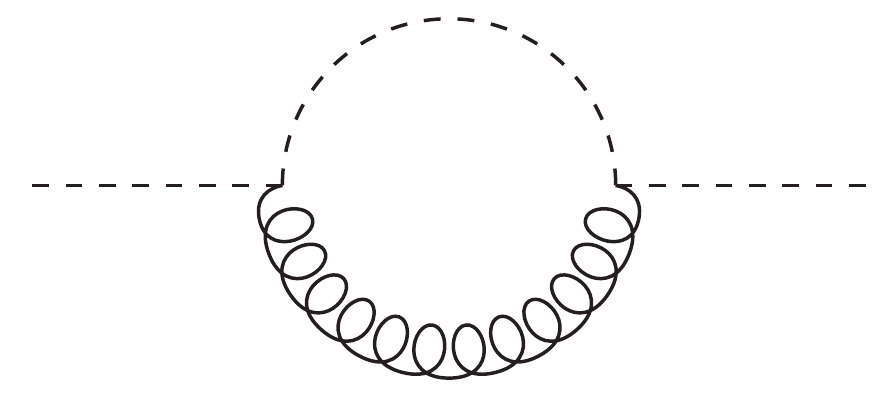}
\\
\rescombgreen^{\tcgreen{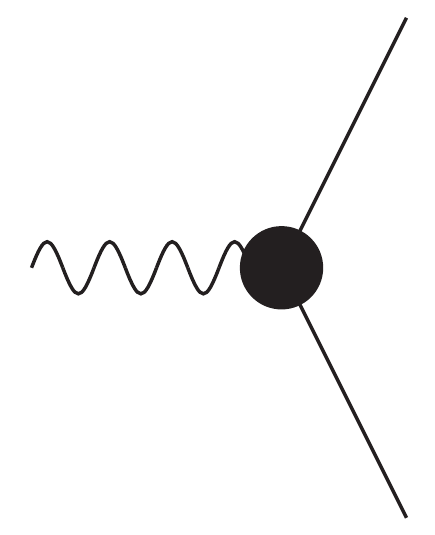}}_{(0,2)} & = \tgraph{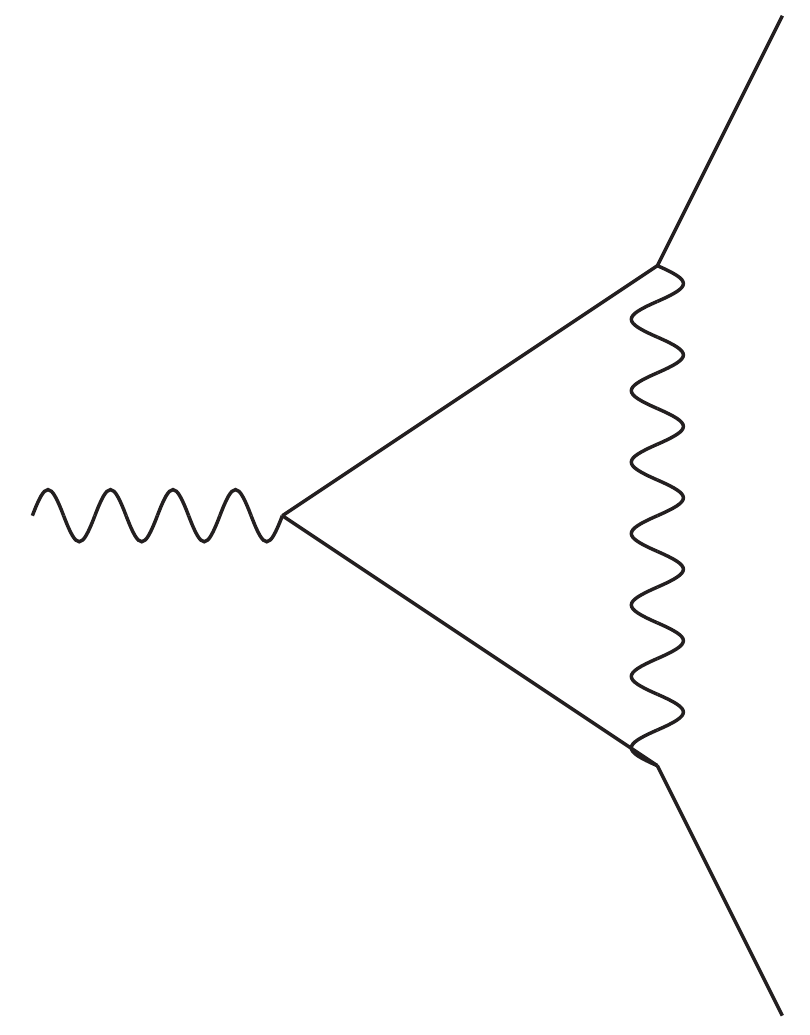}
\\
\rescombgreen^{\tcgreen{c-pff}}_{(2,0)} & = \tgraph{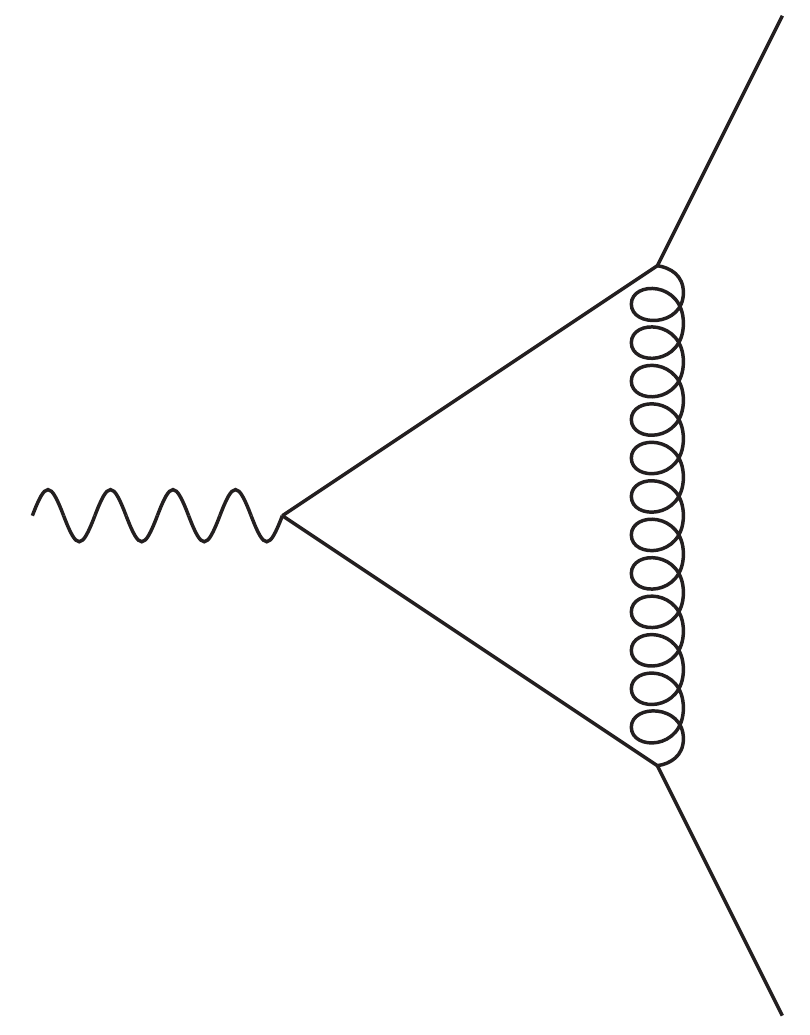} + 2 \tgraph{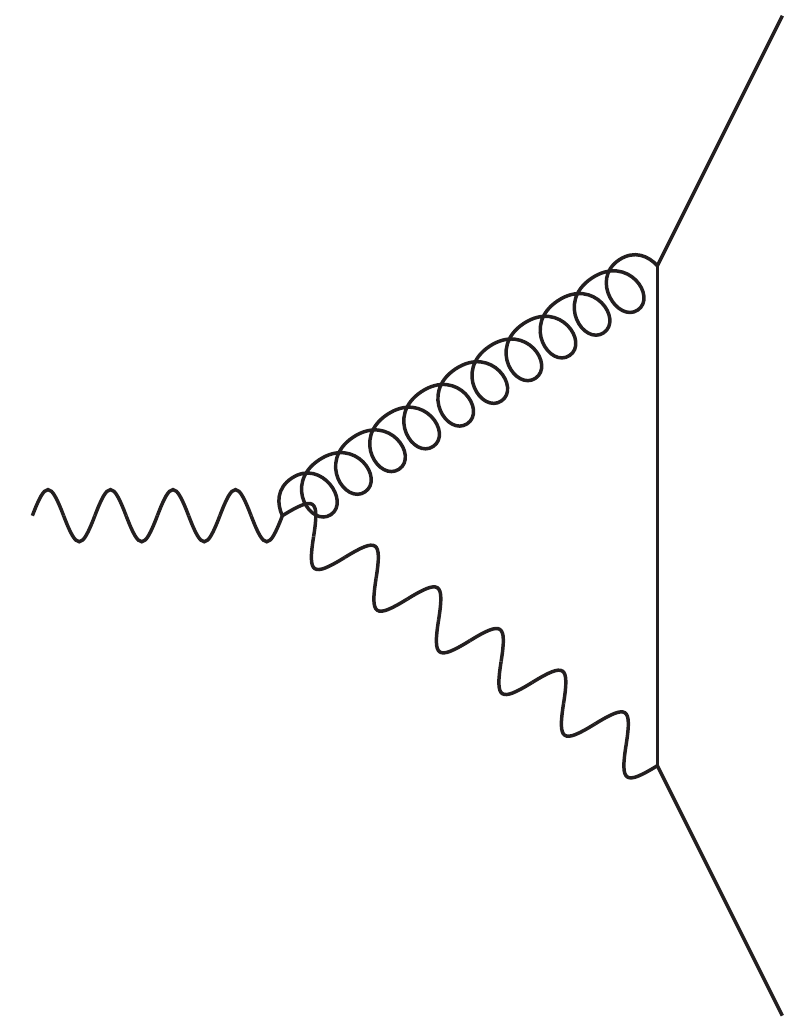} + 2 \ttgraph{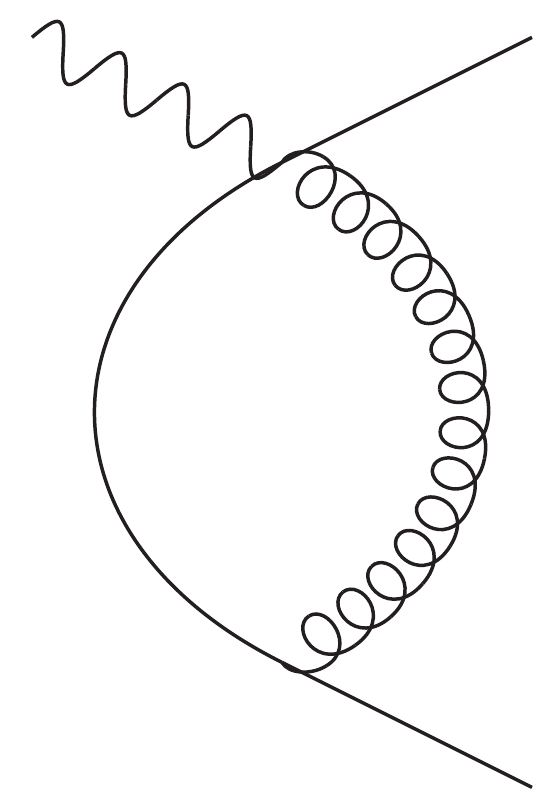} + \tgraph{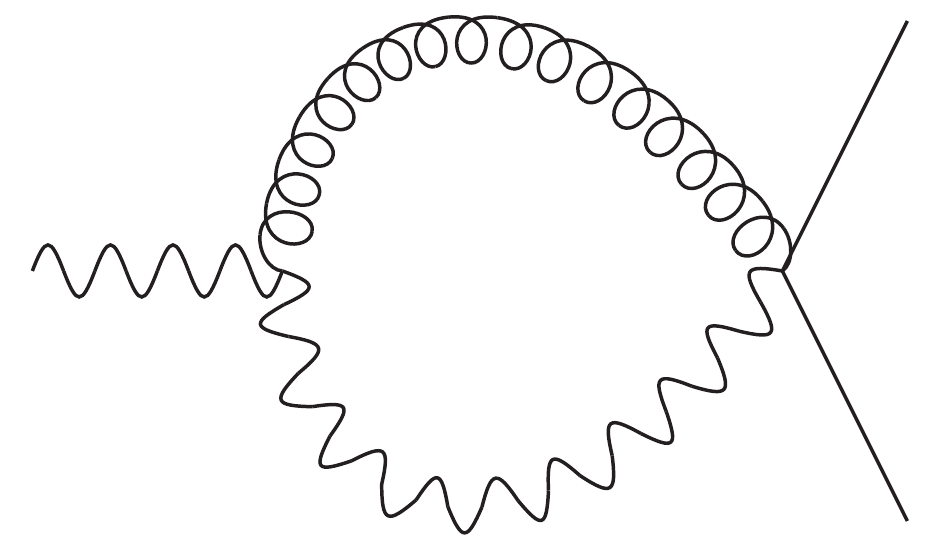}
\\
\rescombgreen^{\tcgreen{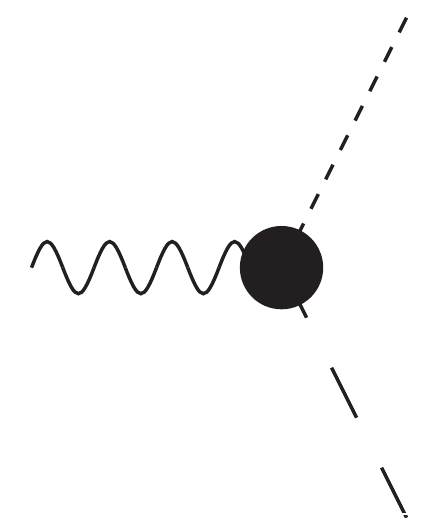}}_{(0,2)} & = 0
\\
\begin{split}
\rescombgreen^{\tcgreen{c-pah}}_{(2,0)} & = \tgraph{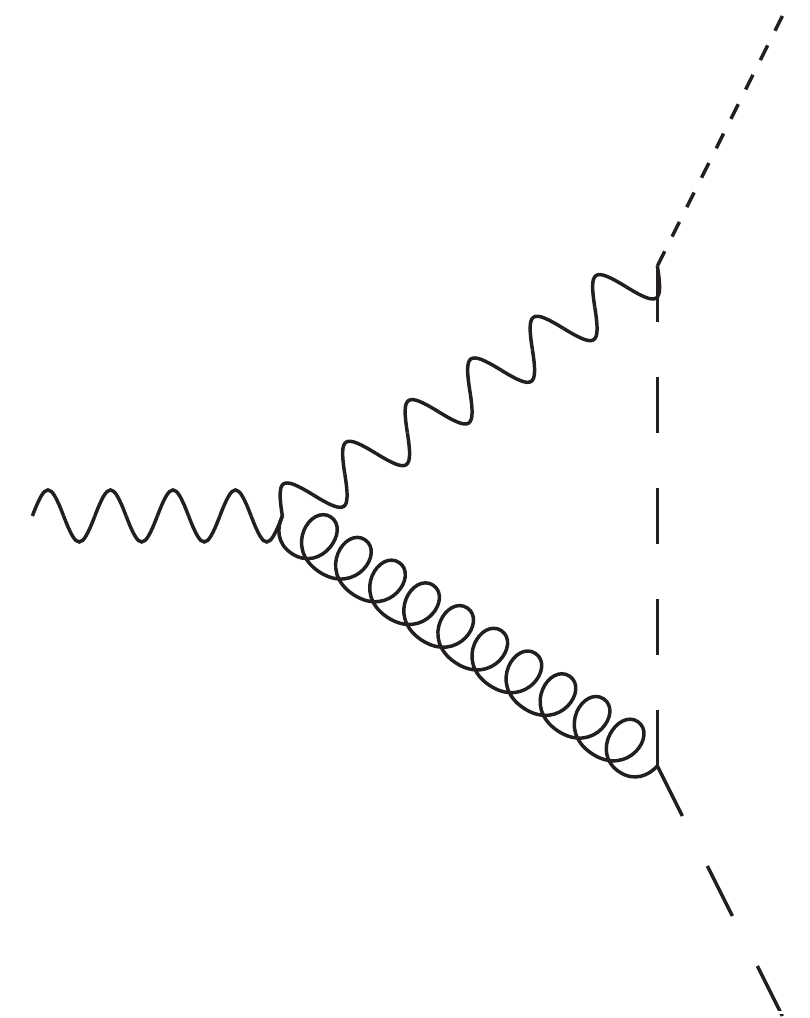} + \tgraph{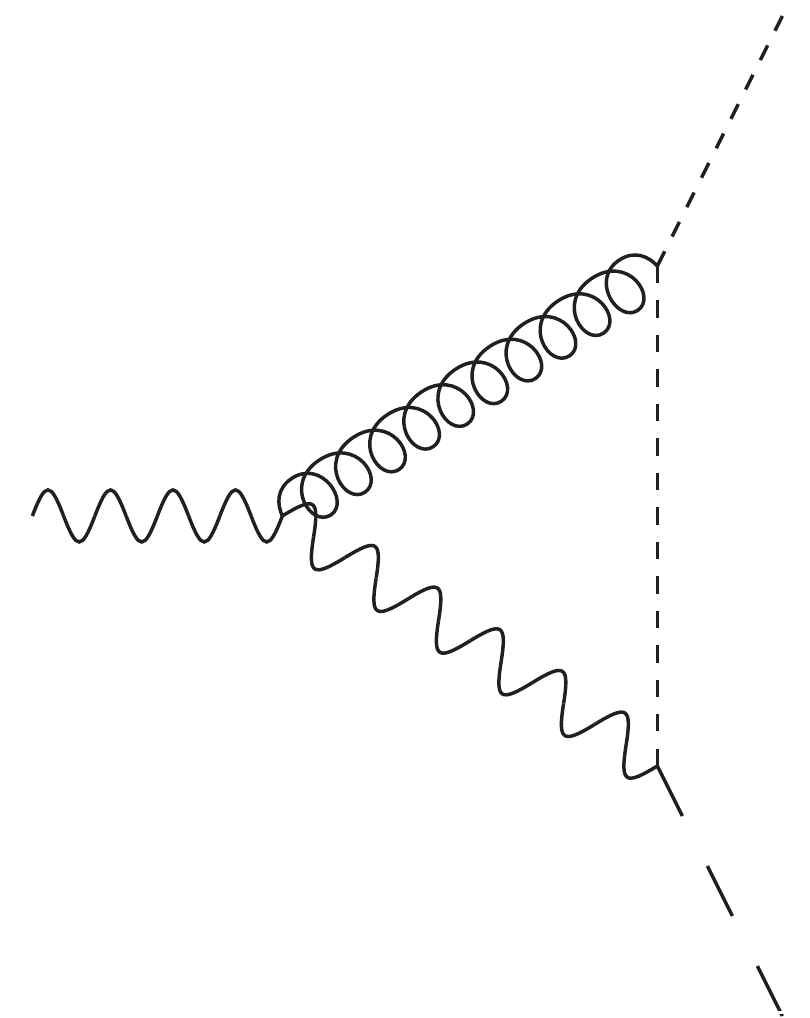} + \tgraph{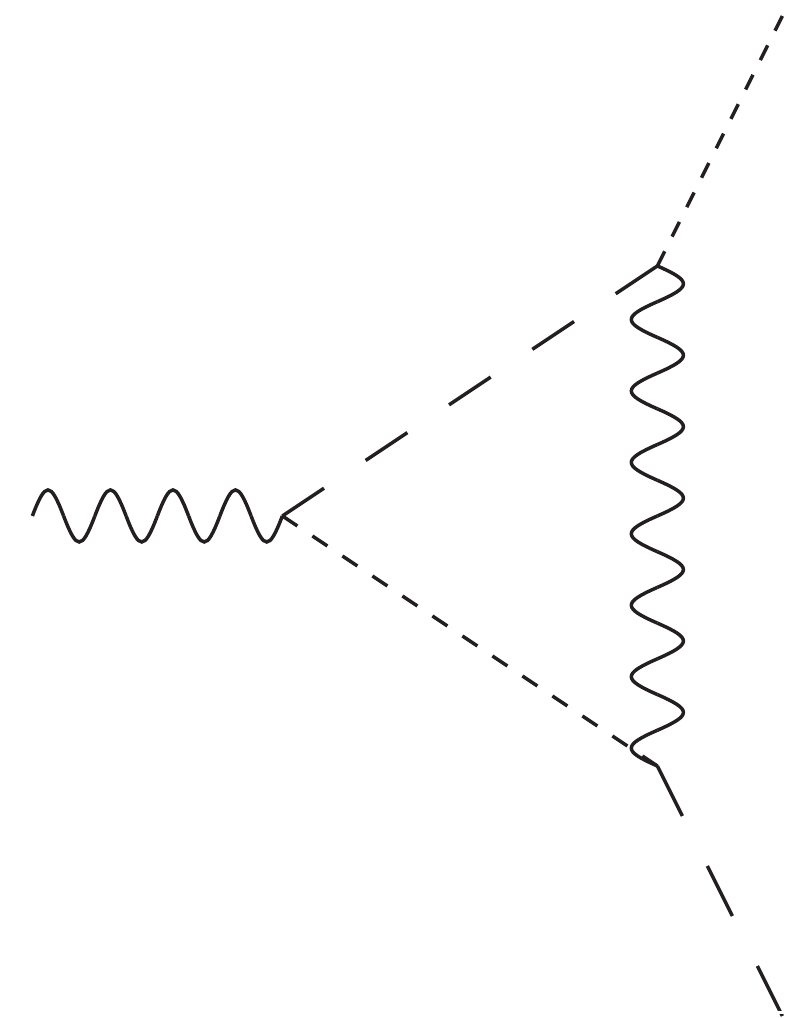} + \tgraph{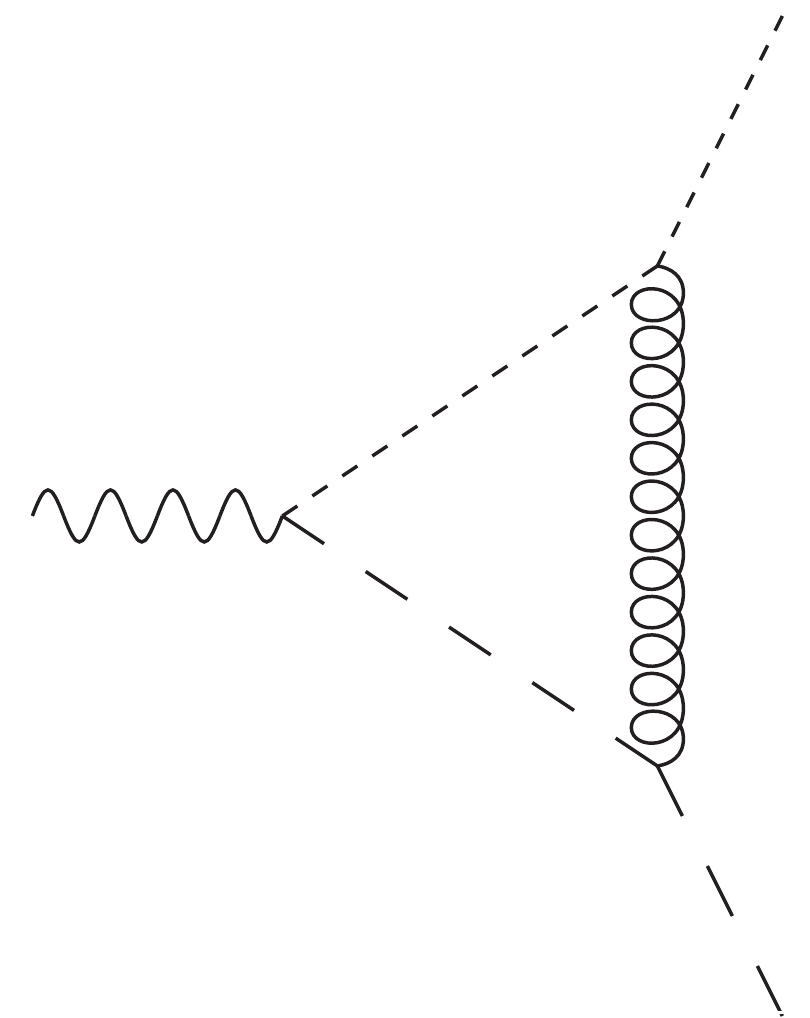} + \tgraph{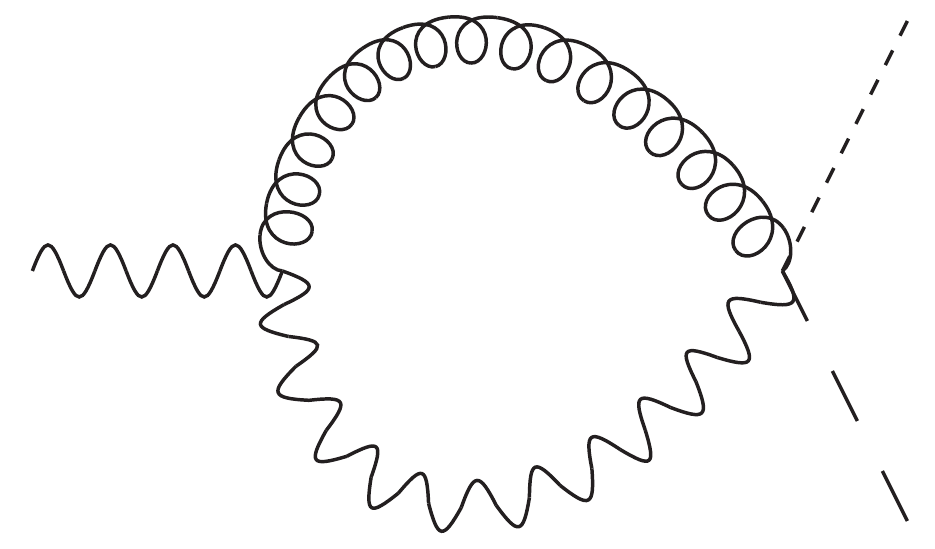} \\ & \phantom{ = } + \ttgraph{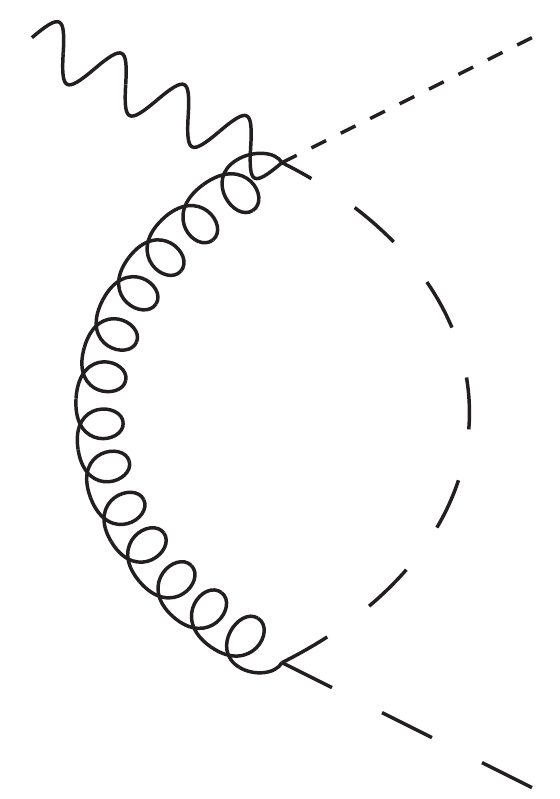} + \ttgraph{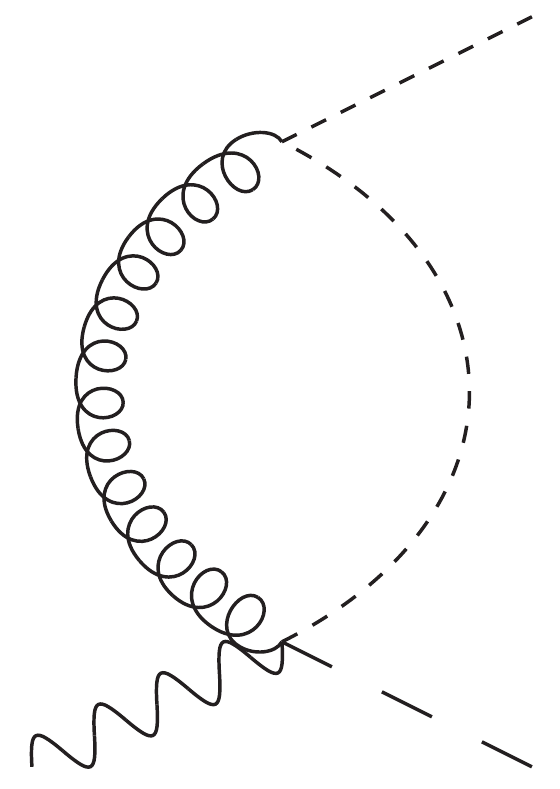}
\end{split}
\\
\rescombgreen^{\tcgreen{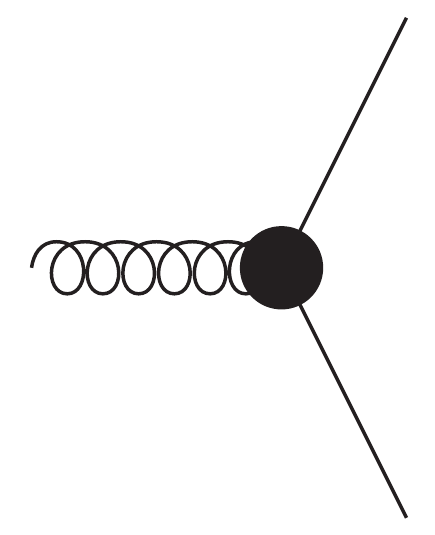}}_{(0,2)} & = \tgraph{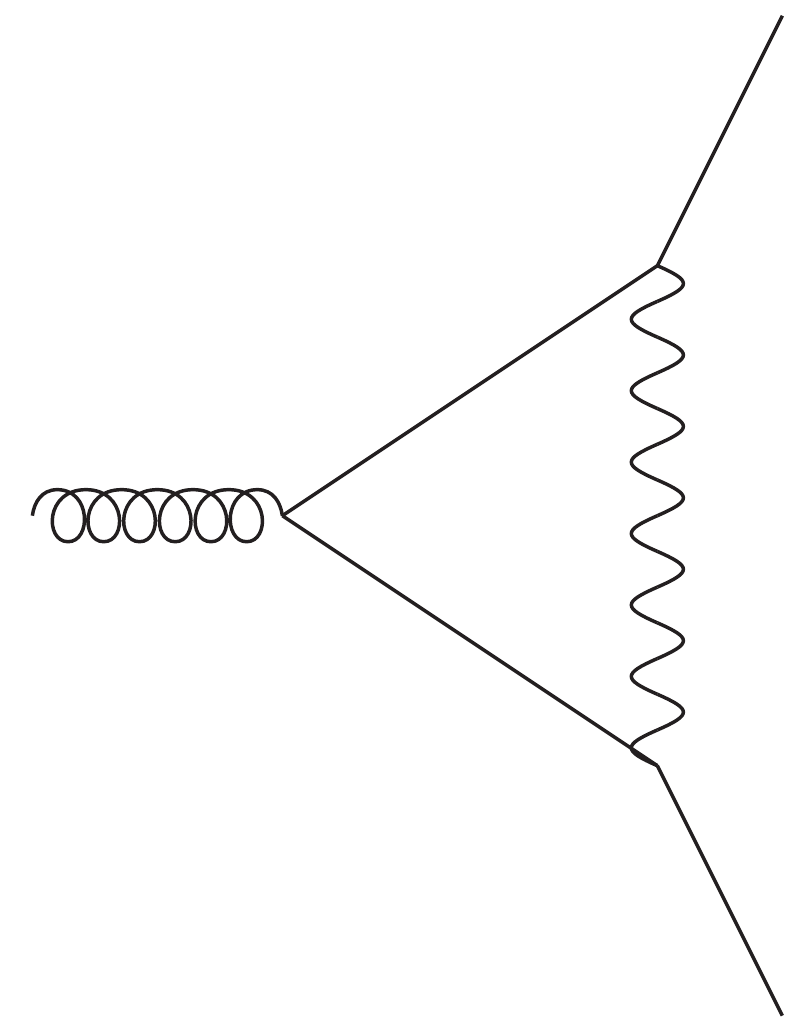} + \tgraph{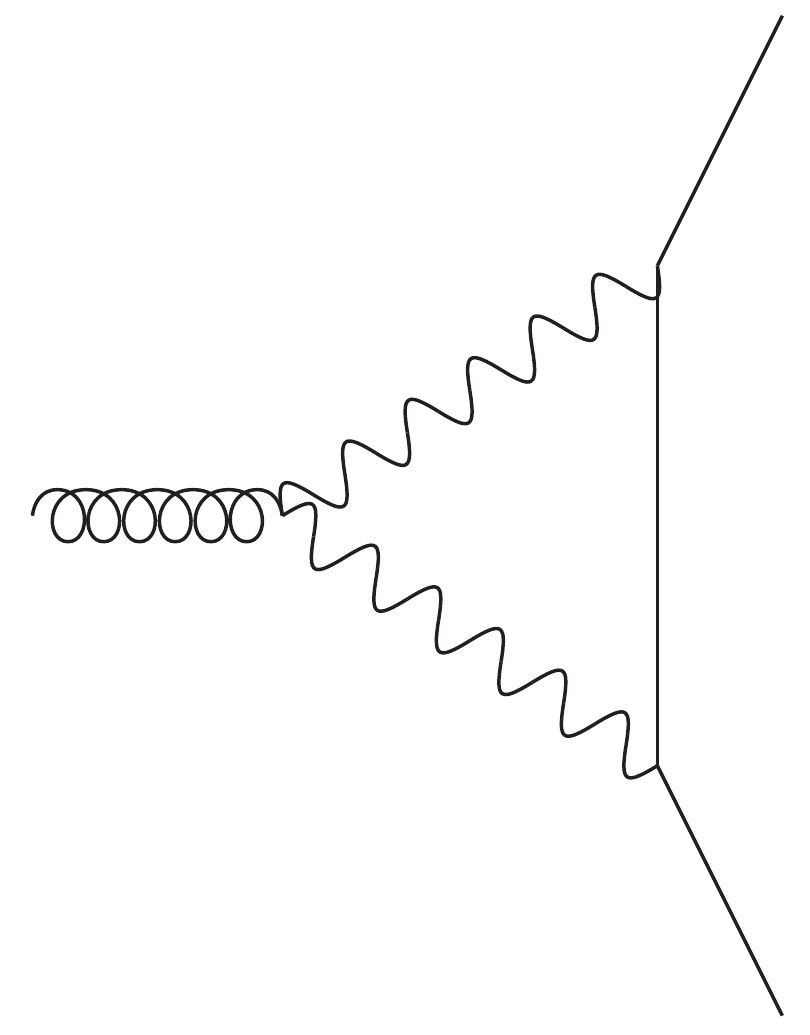} + 2 \ttgraph{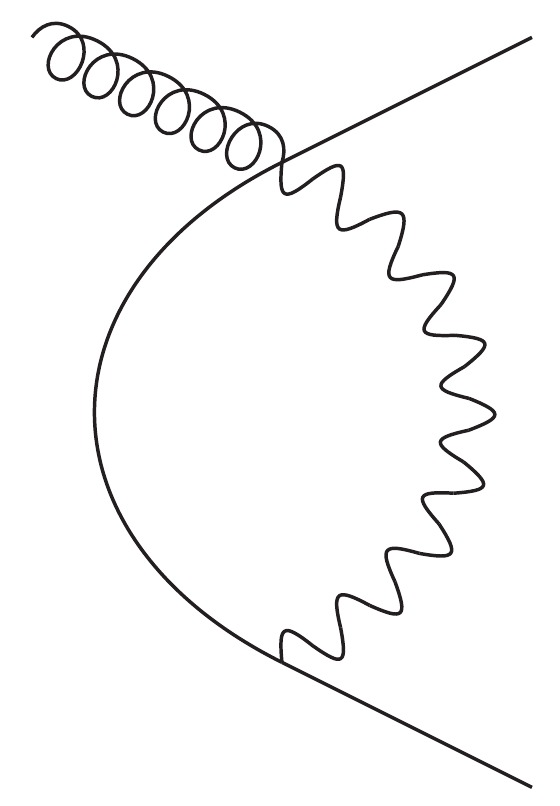}
\\
\rescombgreen^{\tcgreen{c-gff}}_{(2,0)} & = \tgraph{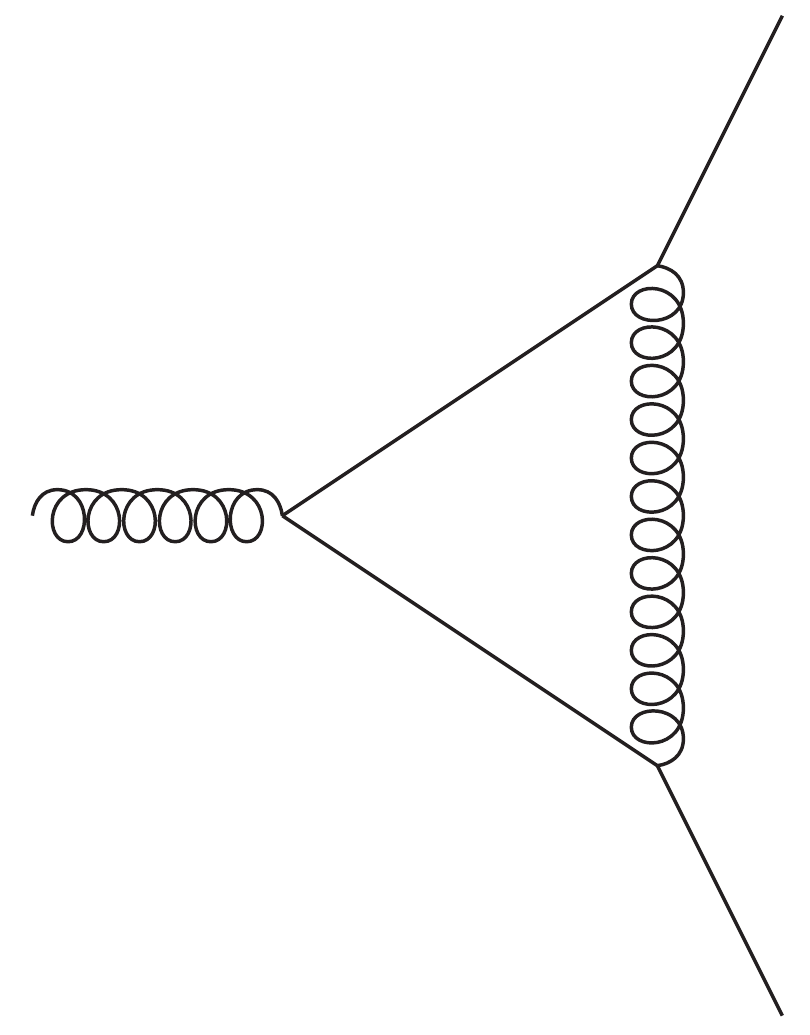} + \tgraph{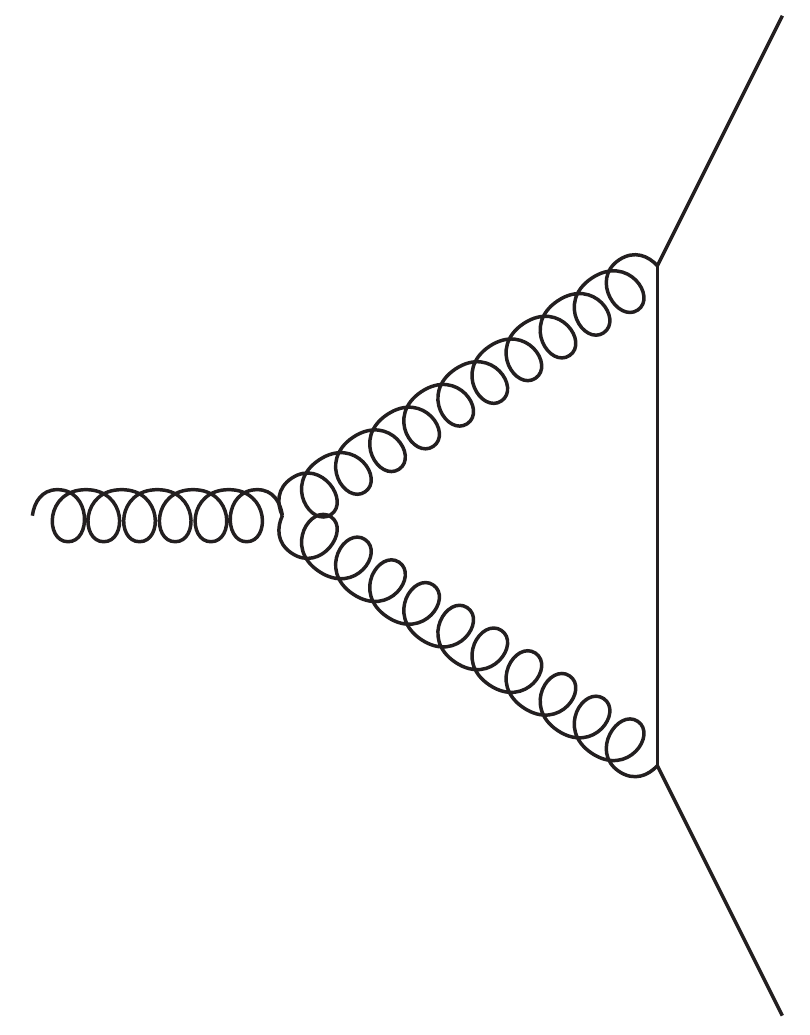} + 2 \ttgraph{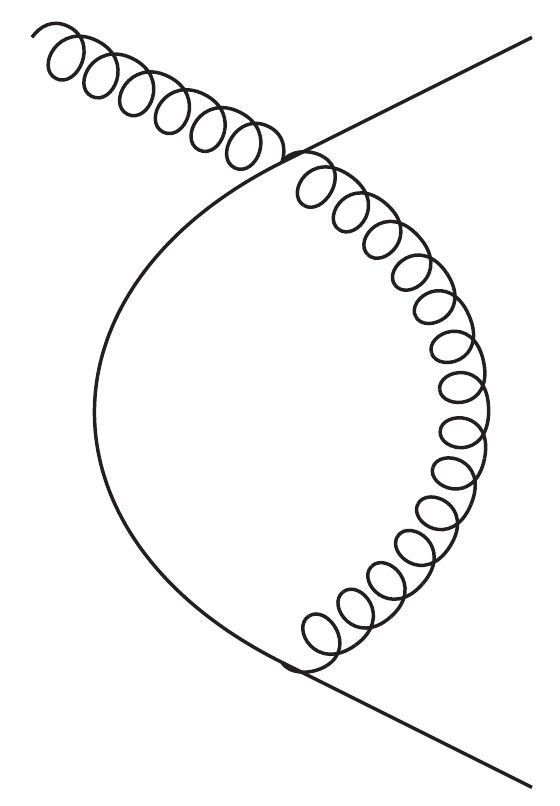} + \frac{1}{2} \tgraph{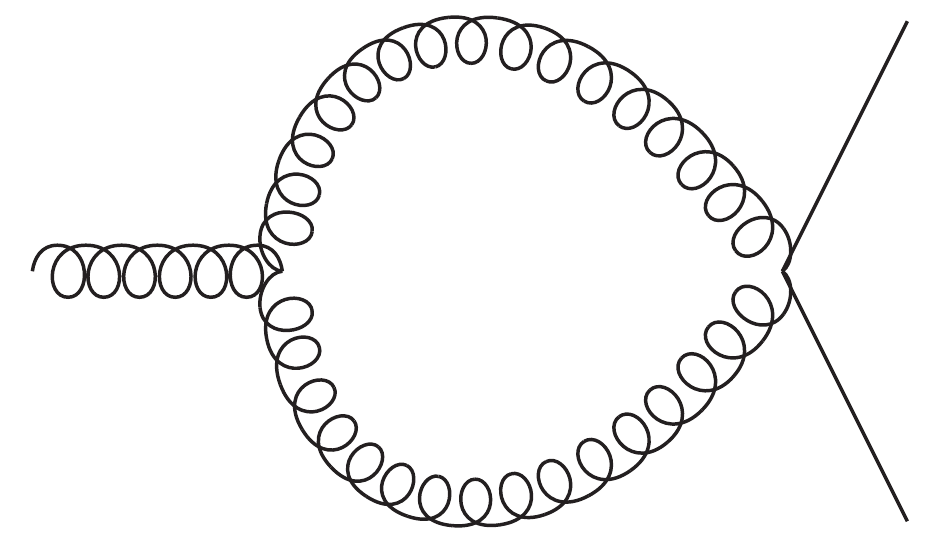}
\\
\rescombgreen^{\tcgreen{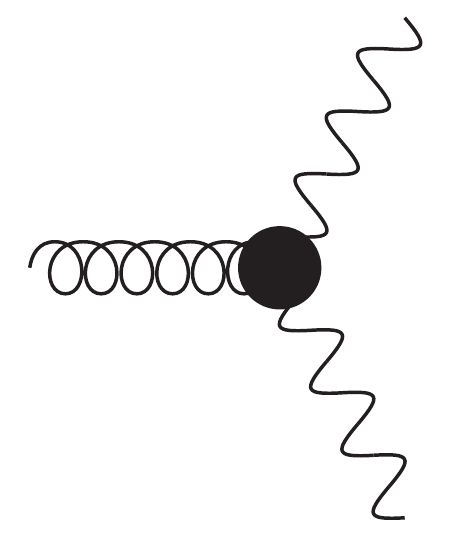}}_{(0,2)} & = \tgraph{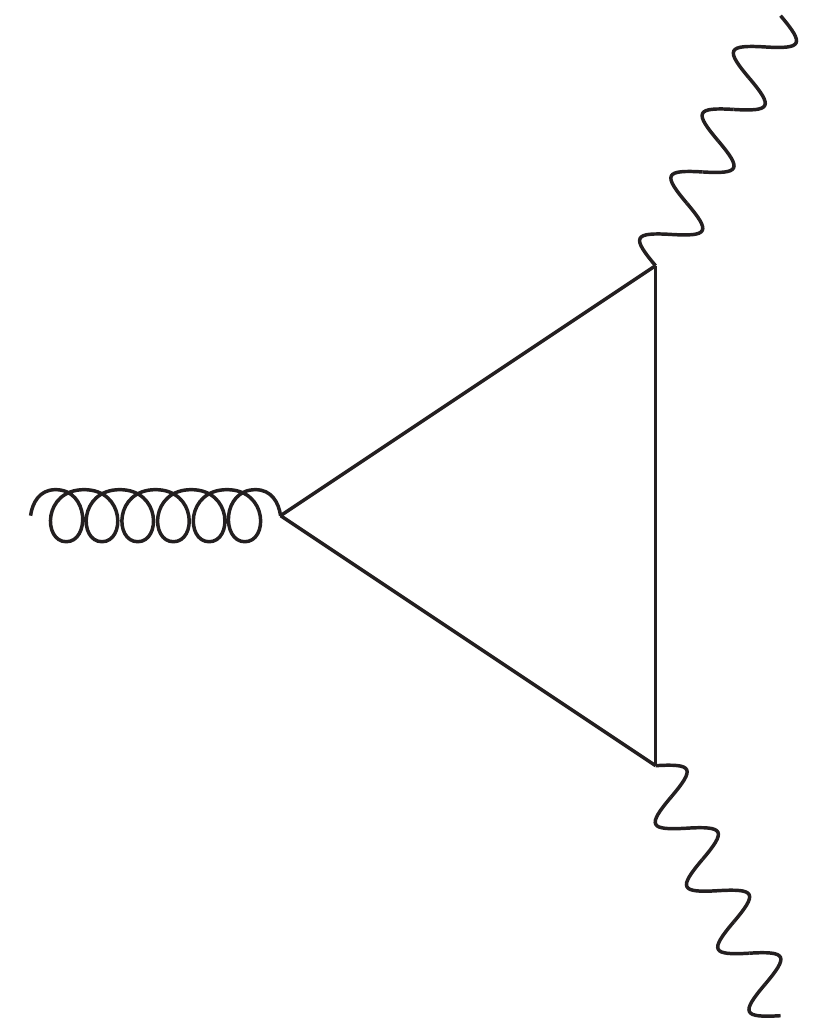} + \frac{2}{2} \ttgraph{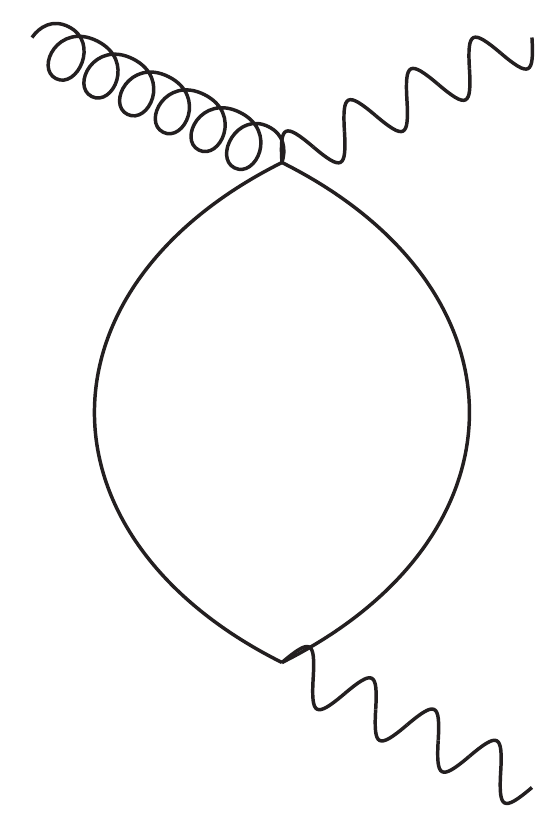}
\\
\begin{split}
\rescombgreen^{\tcgreen{c-gpp}}_{(2,0)} & = \tgraph{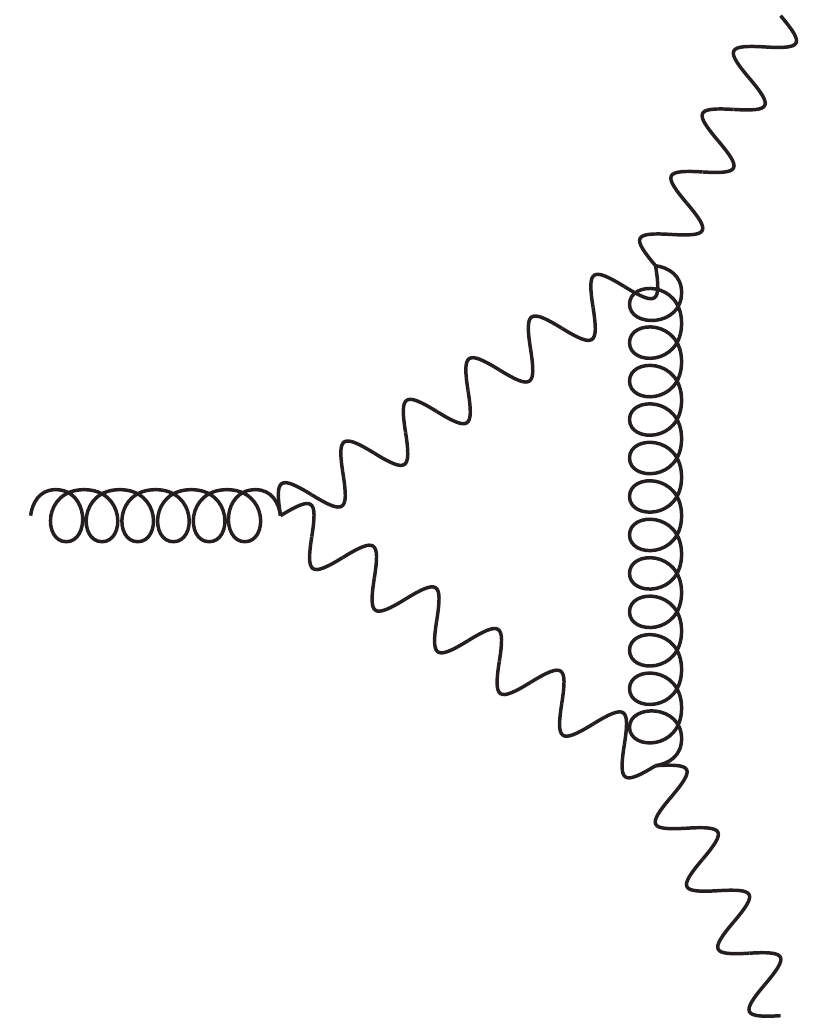} + \tgraph{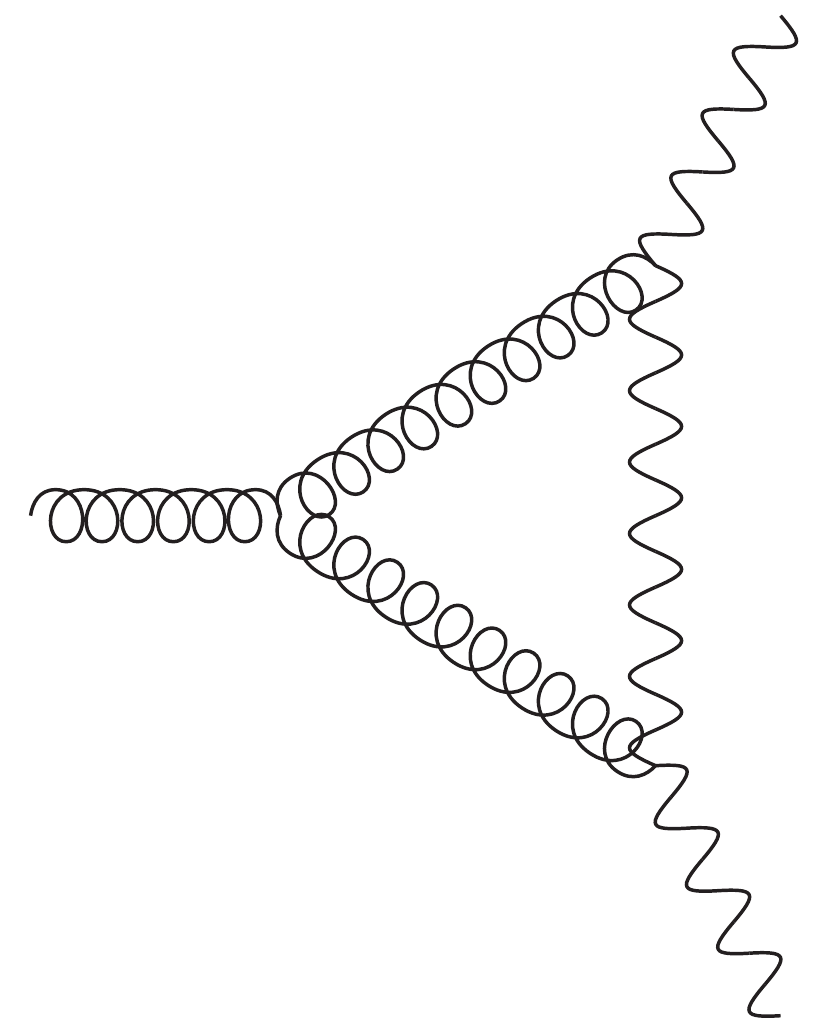} + \tgraph{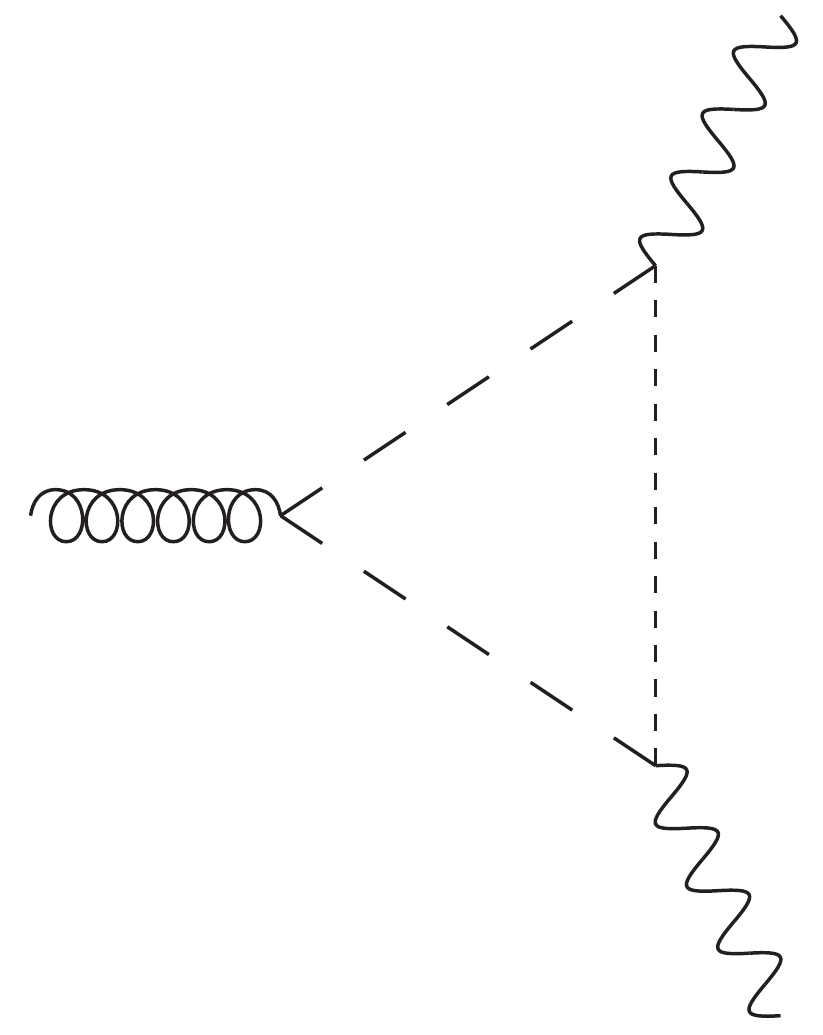} + \tgraph{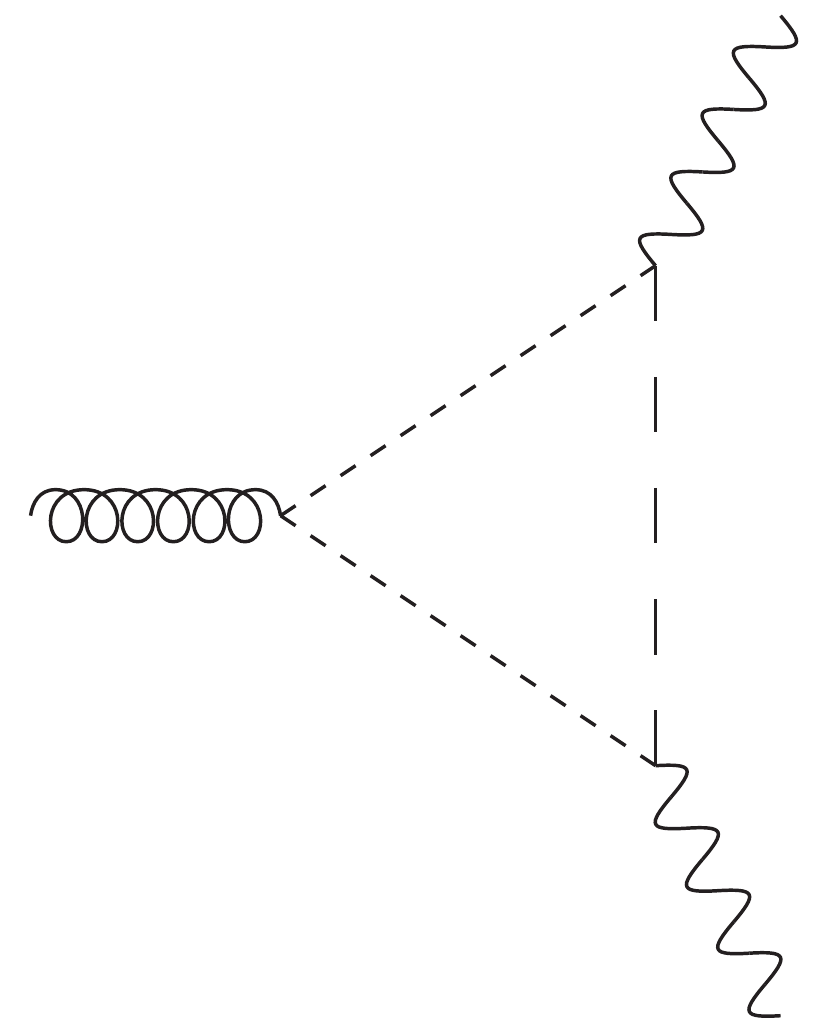} + 2 \ttgraph{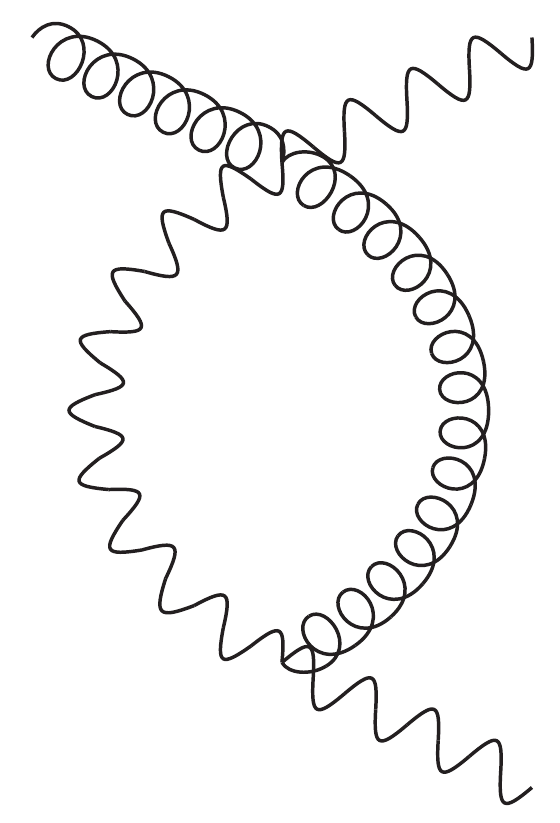} \\ & \phantom{ = } + \frac{1}{2} \tgraph{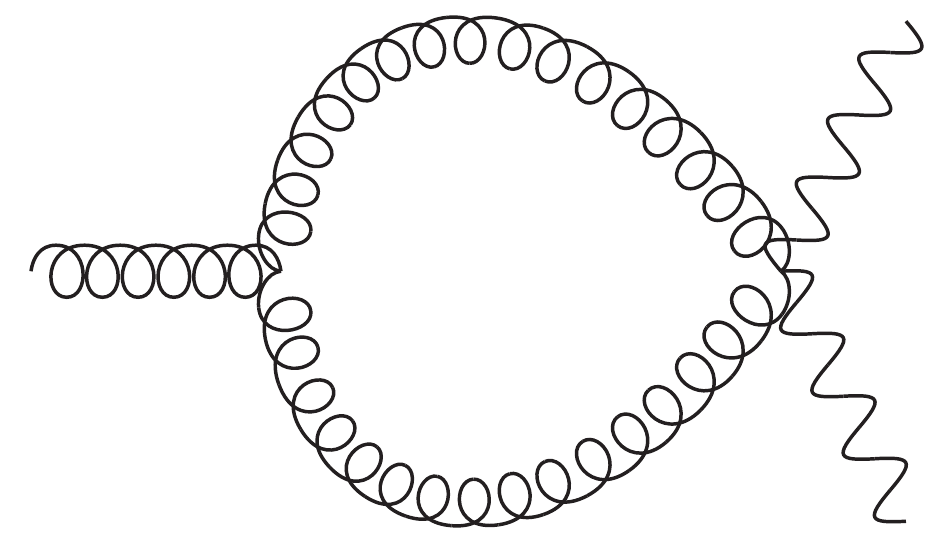} + 2 \ttgraph{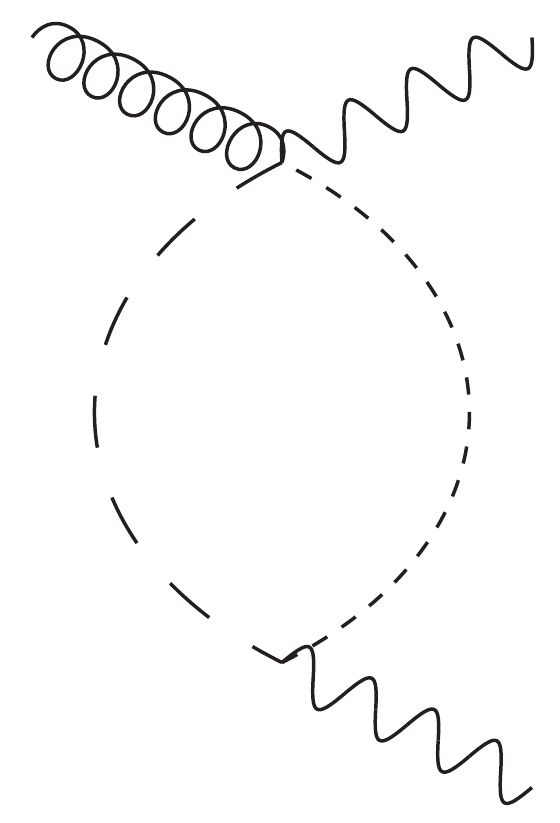}
\end{split}
\\
\rescombgreen^{\tcgreen{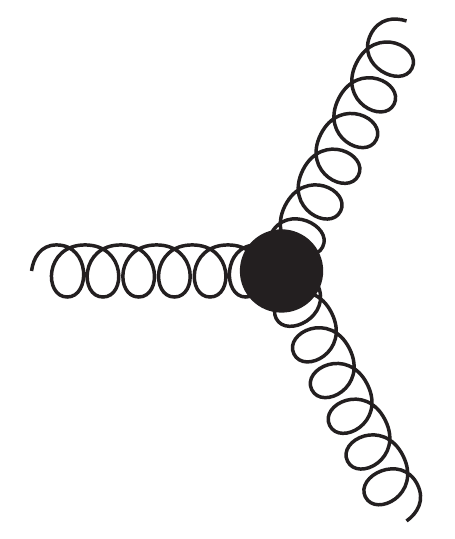}}_{(0,2)} & = 0
\\
\begin{split}
\rescombgreen^{\tcgreen{c-ggg}}_{(2,0)} & = \tgraph{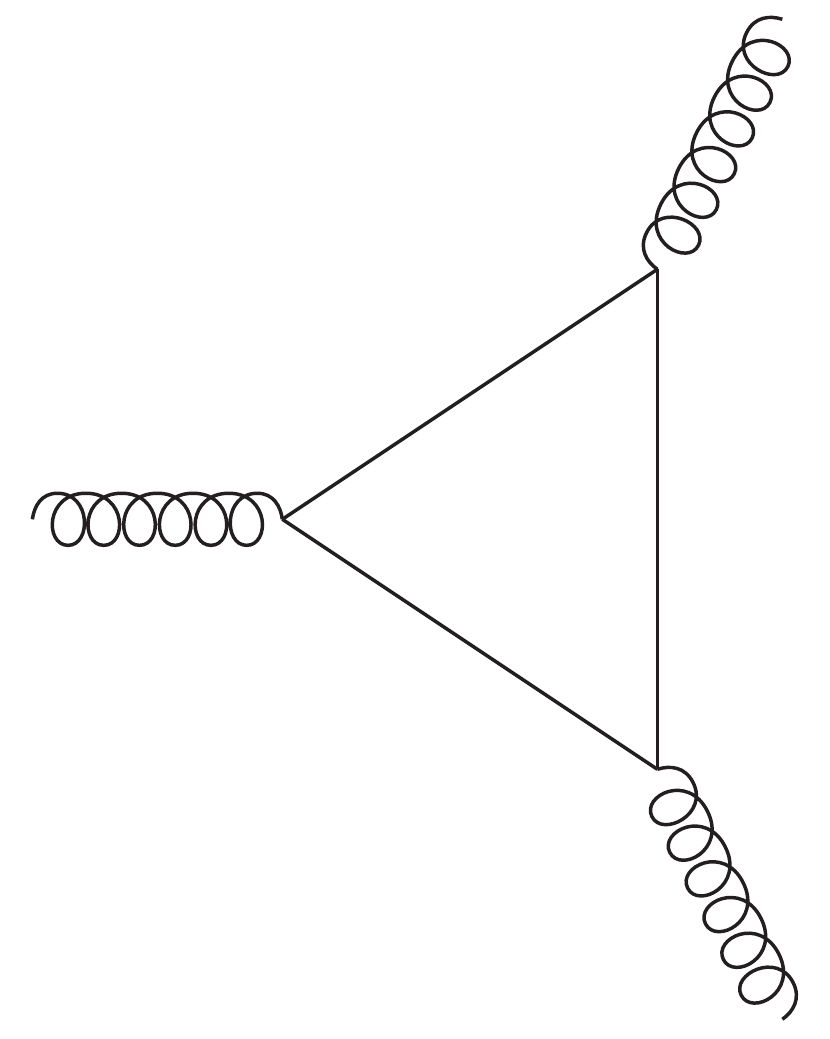} + \tgraph{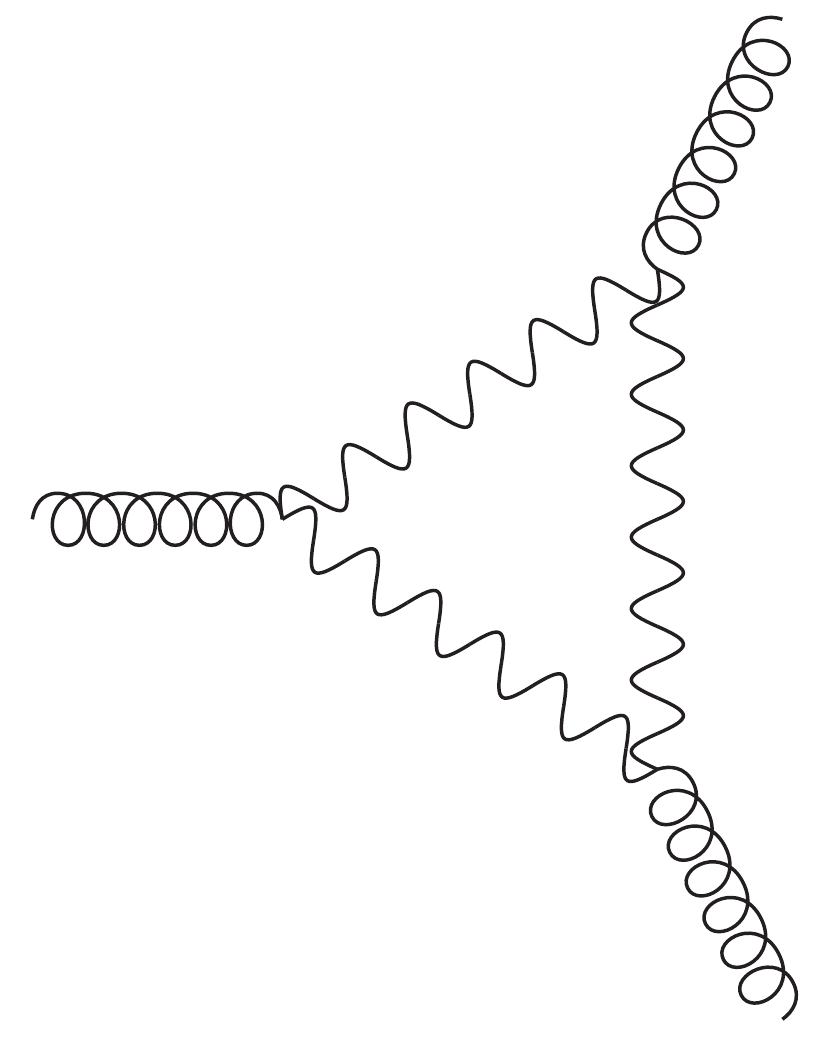} + \tgraph{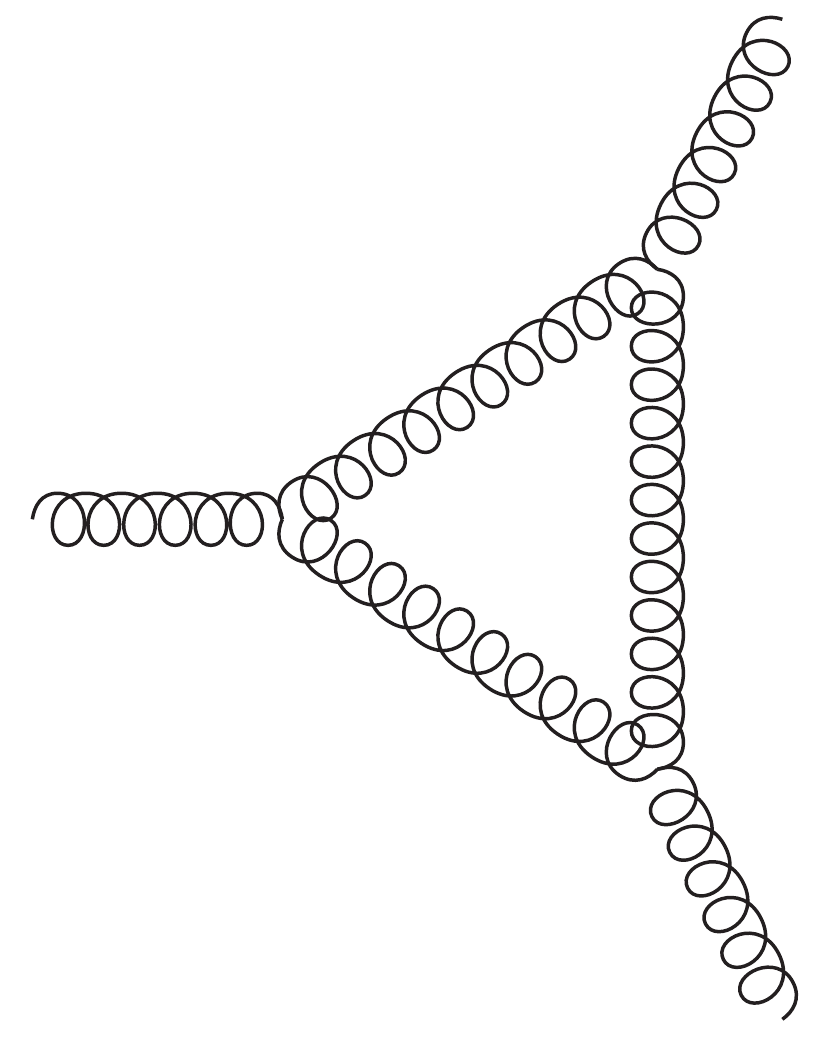} + \tgraph{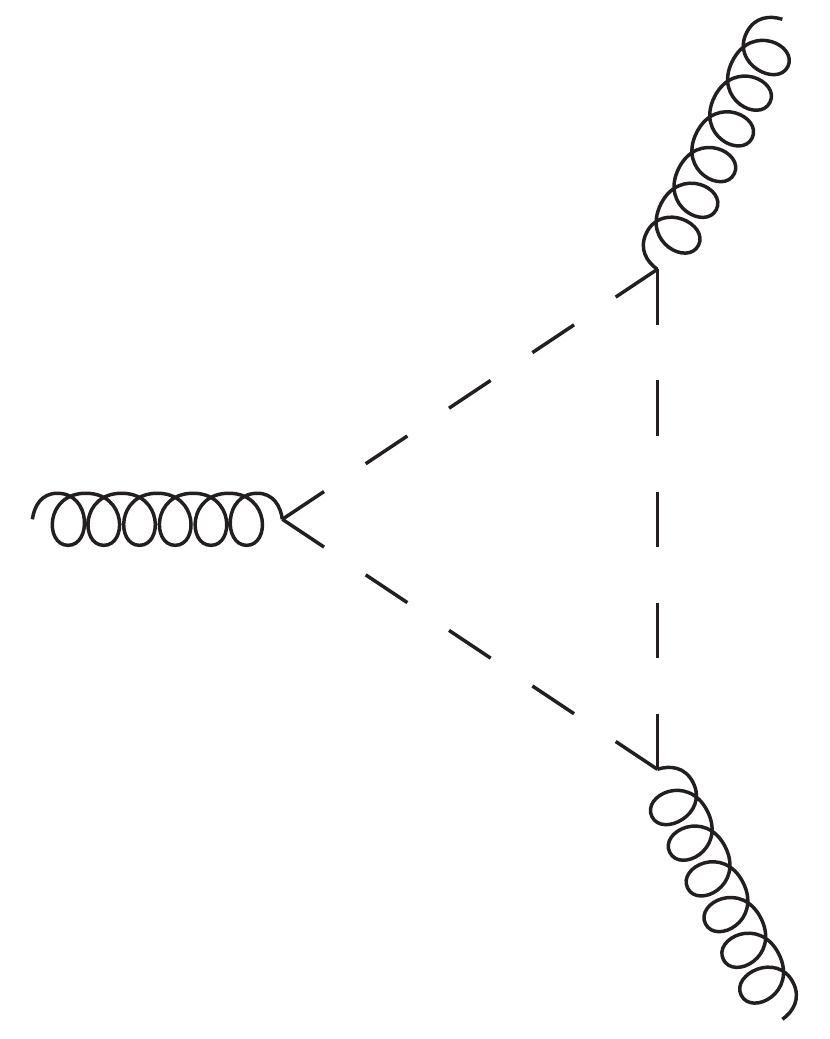} + \tgraph{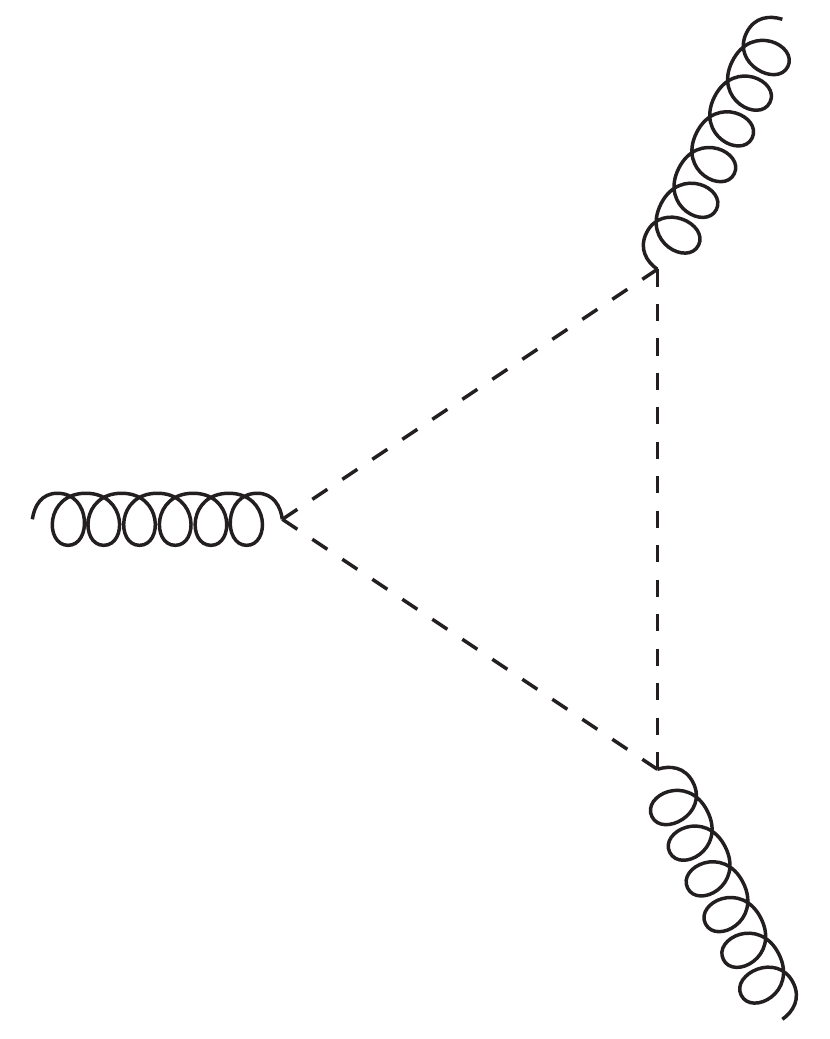} \\ & \phantom{ = } + \frac{3}{2} \tgraph{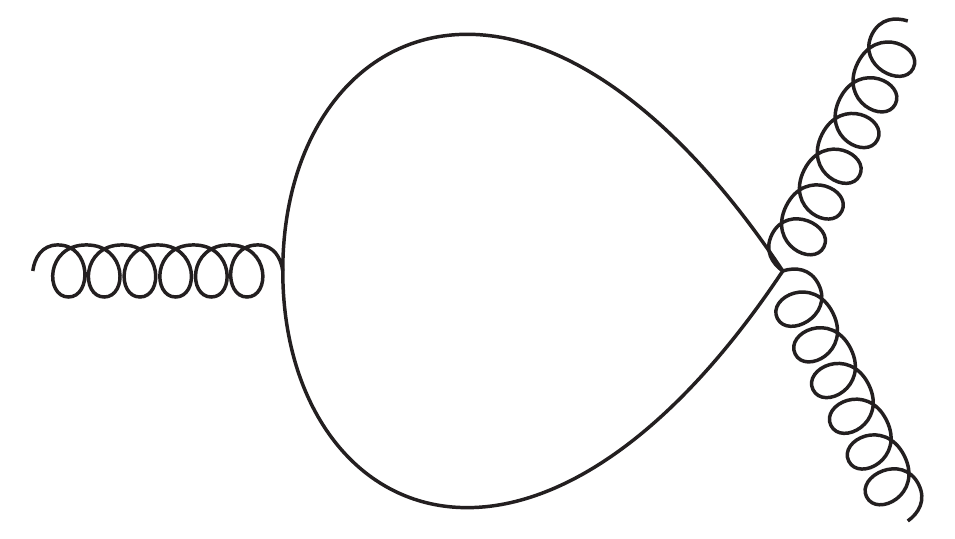} + \frac{3}{2} \tgraph{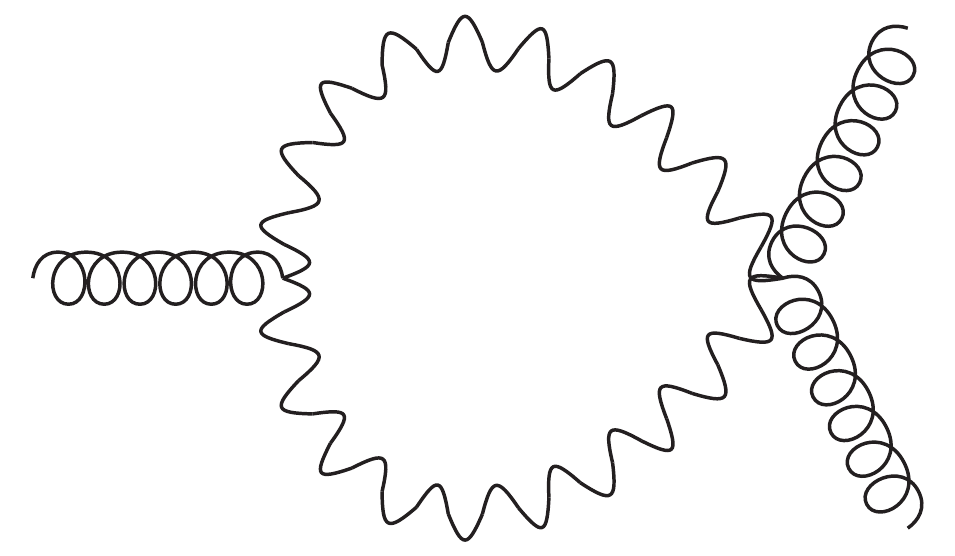} + \frac{3}{2} \tgraph{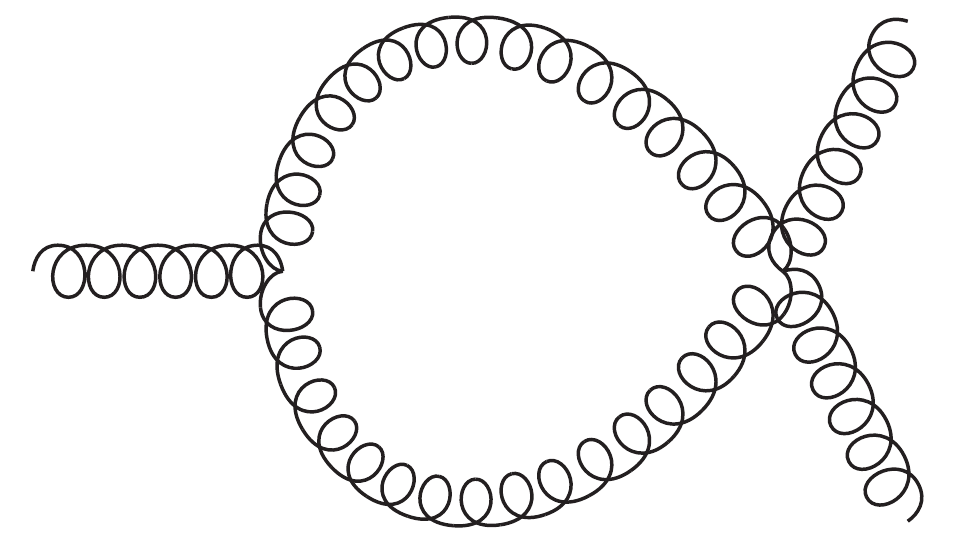} + \frac{3}{2} \tgraph{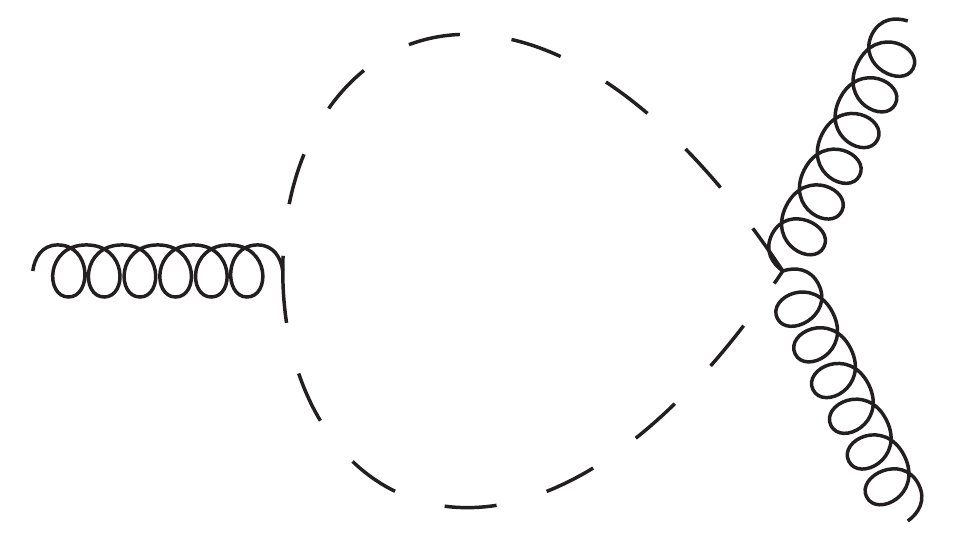} \\ & \phantom{ = } + \frac{3}{2} \tgraph{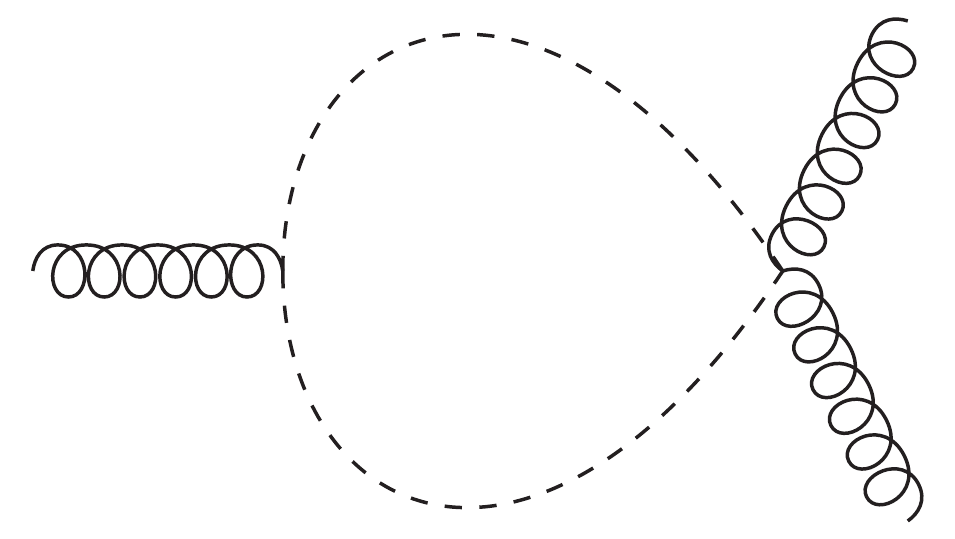}
\end{split}
\\
\rescombgreen^{\tcgreen{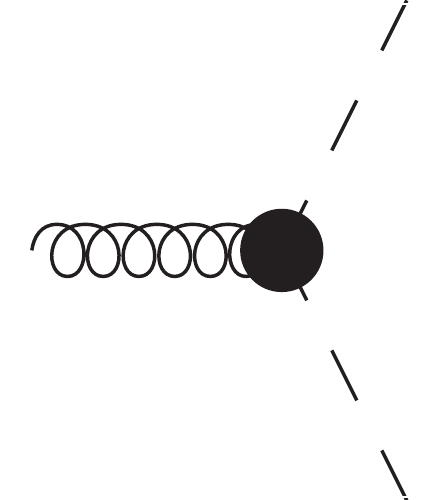}}_{(0,2)} & = 0
\\
\begin{split}
\rescombgreen^{\tcgreen{c-gaa}}_{(2,0)} & = \tgraph{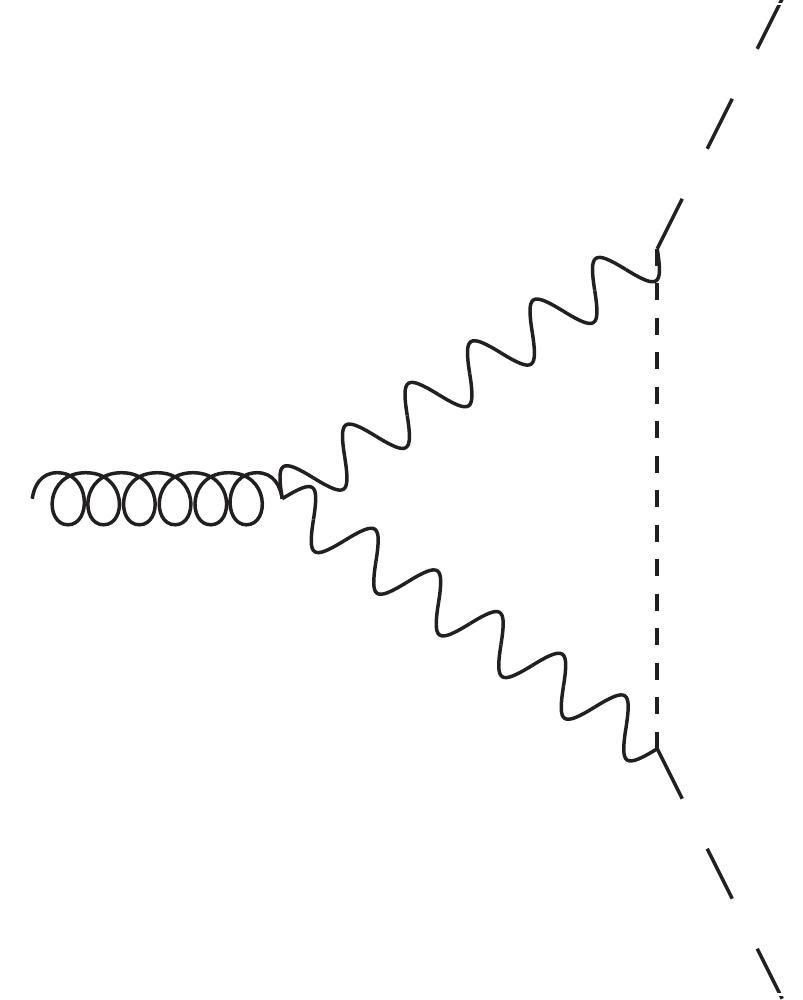} + \tgraph{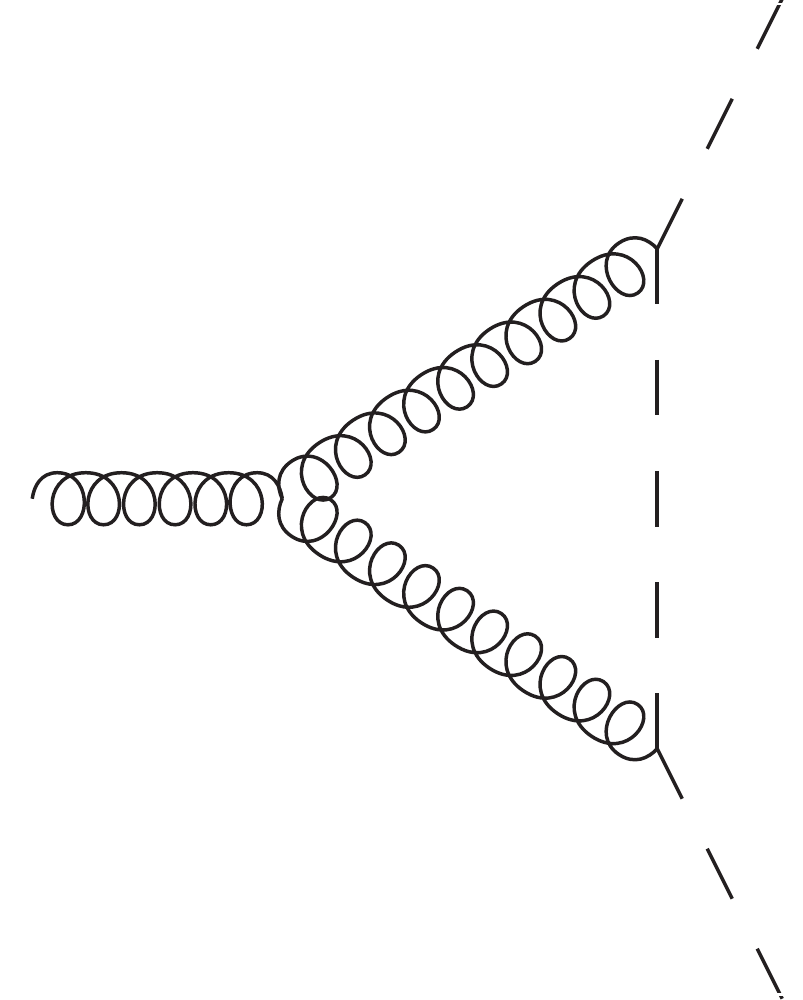} + \tgraph{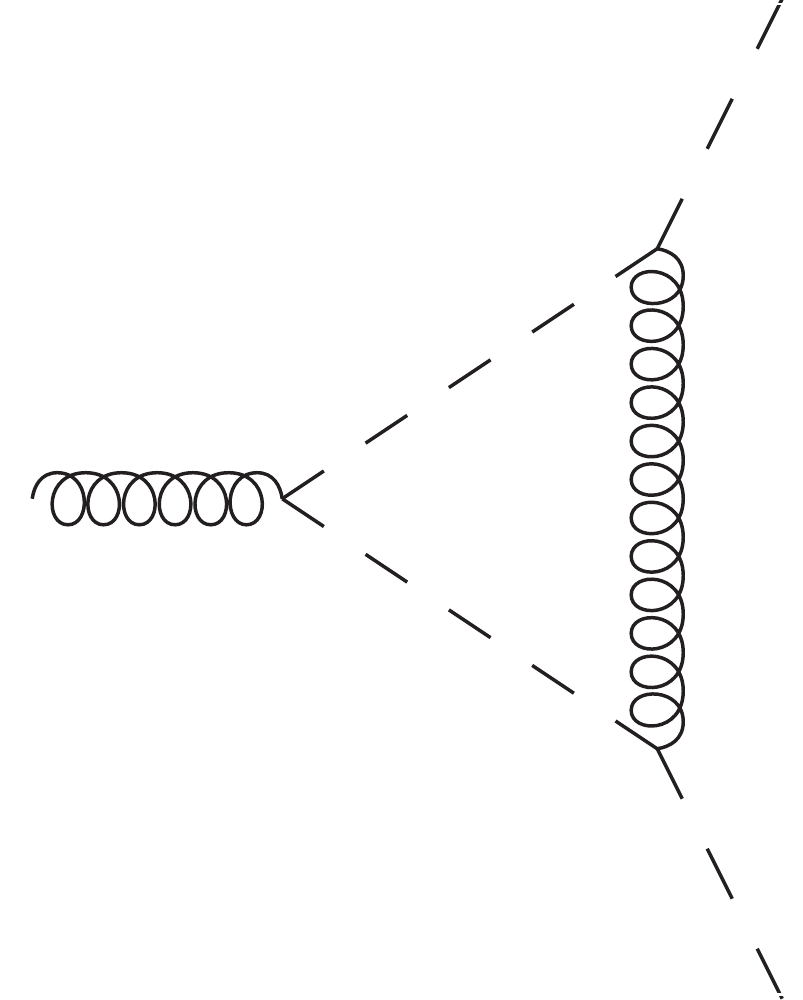} + \tgraph{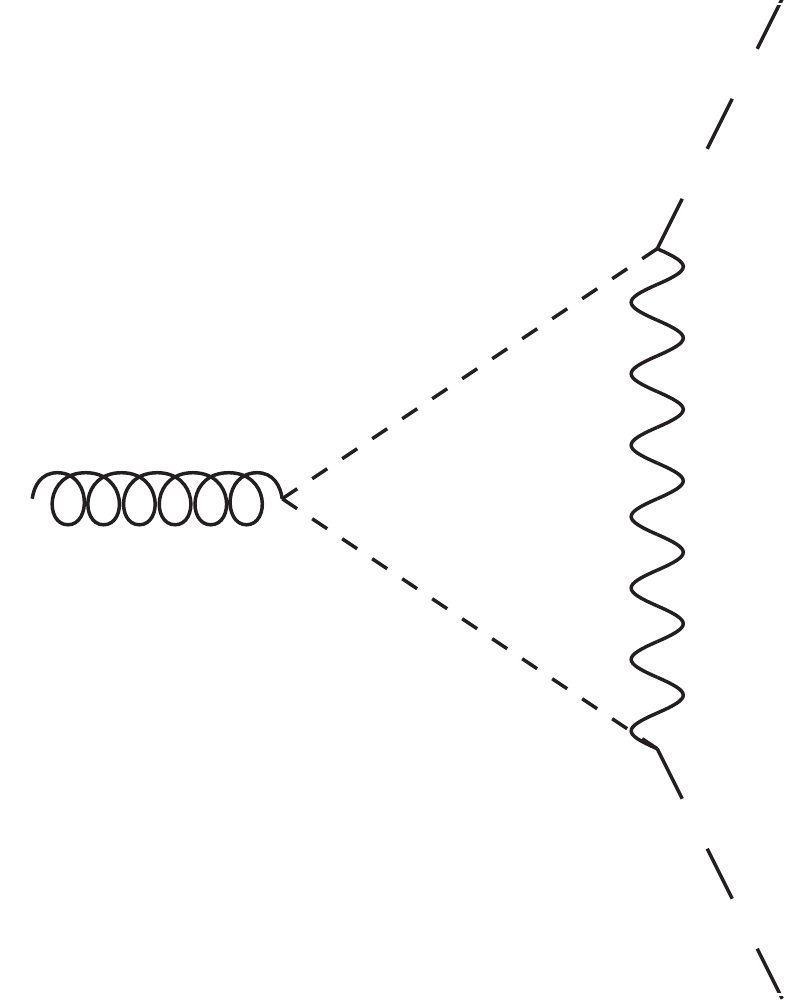} + \frac{1}{2} \tgraph{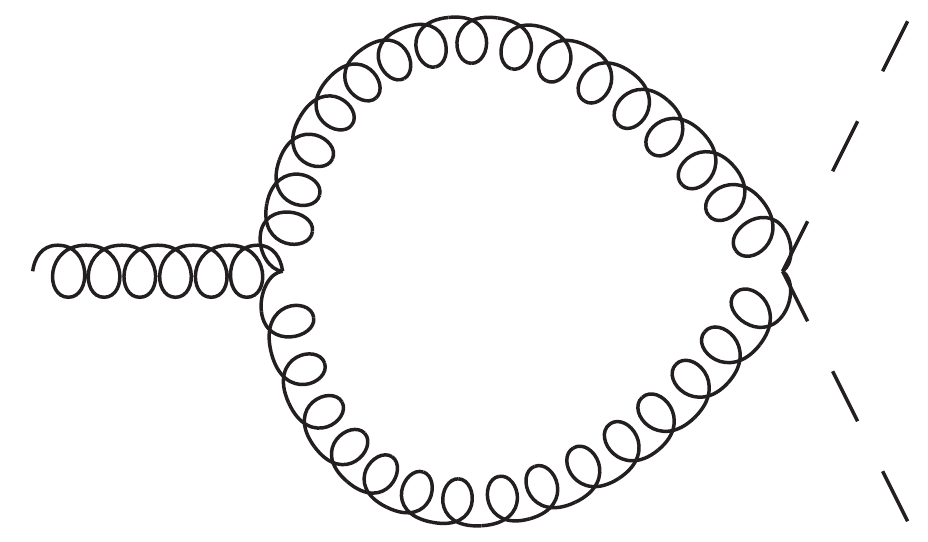} \\ & \phantom{ = } + 2 \ttgraph{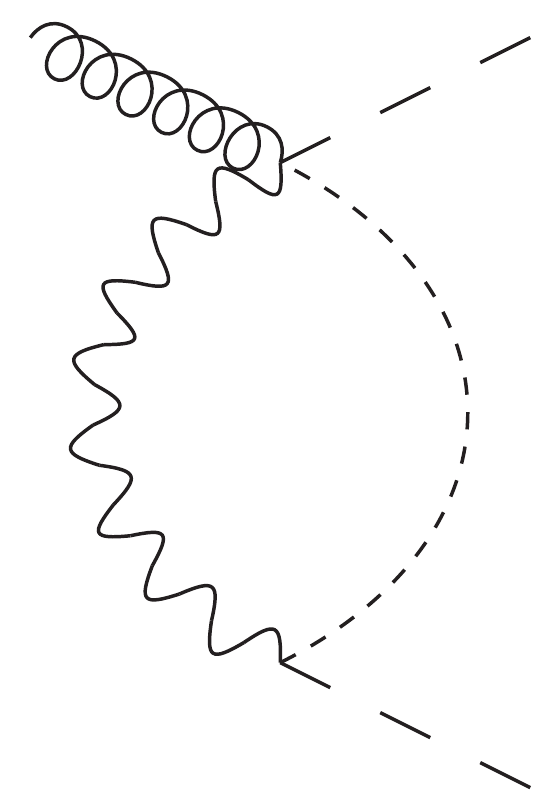} + 2 \ttgraph{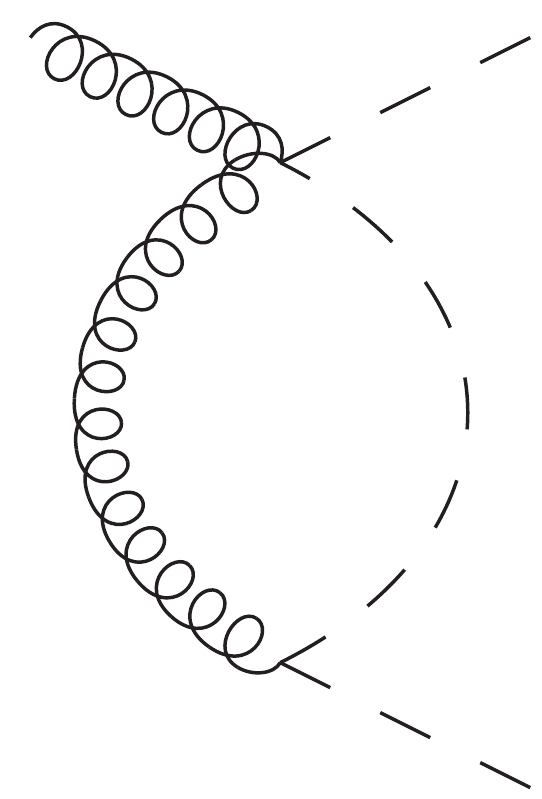}
\end{split}
\\
\rescombgreen^{\tcgreen{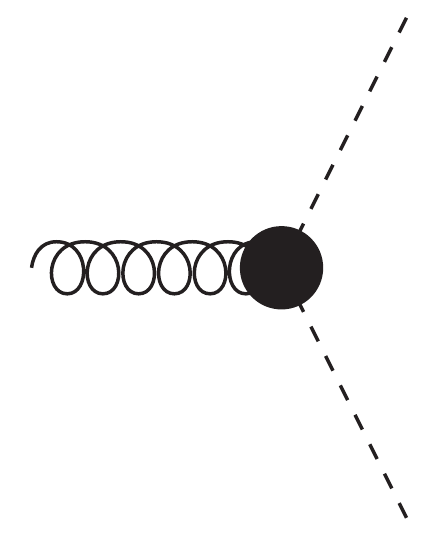}}_{(0,2)} & = 0
\\
\begin{split}
\rescombgreen^{\tcgreen{c-ghh}}_{(2,0)} & = \tgraph{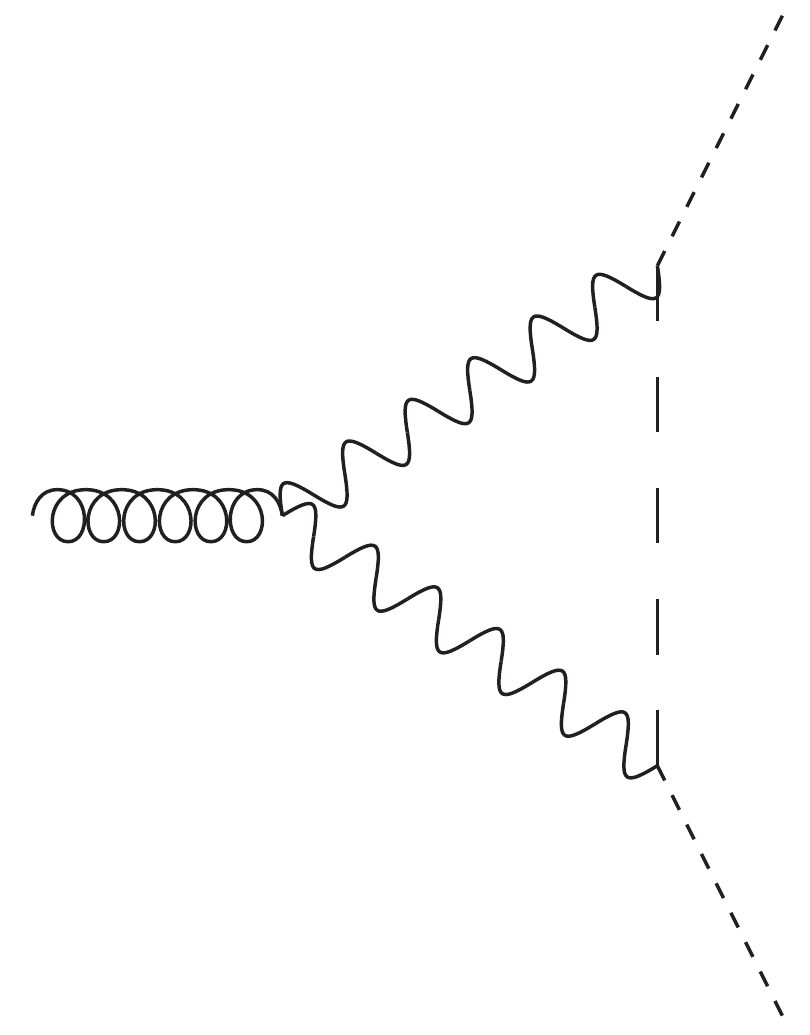} + \tgraph{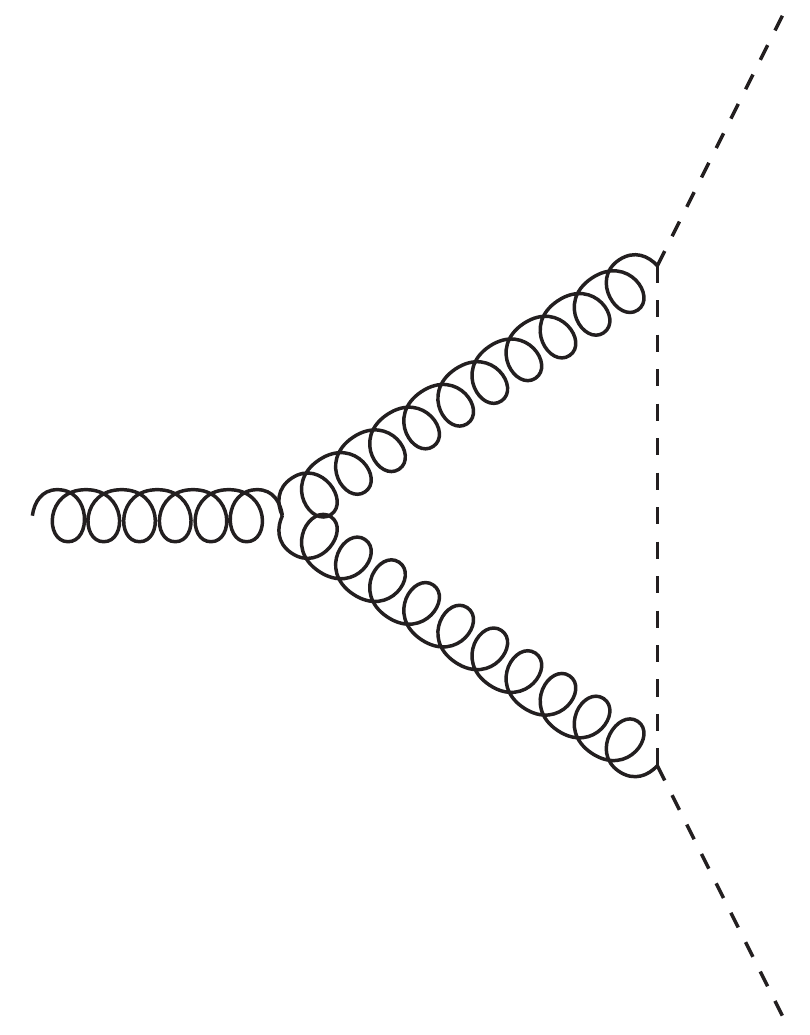} + \tgraph{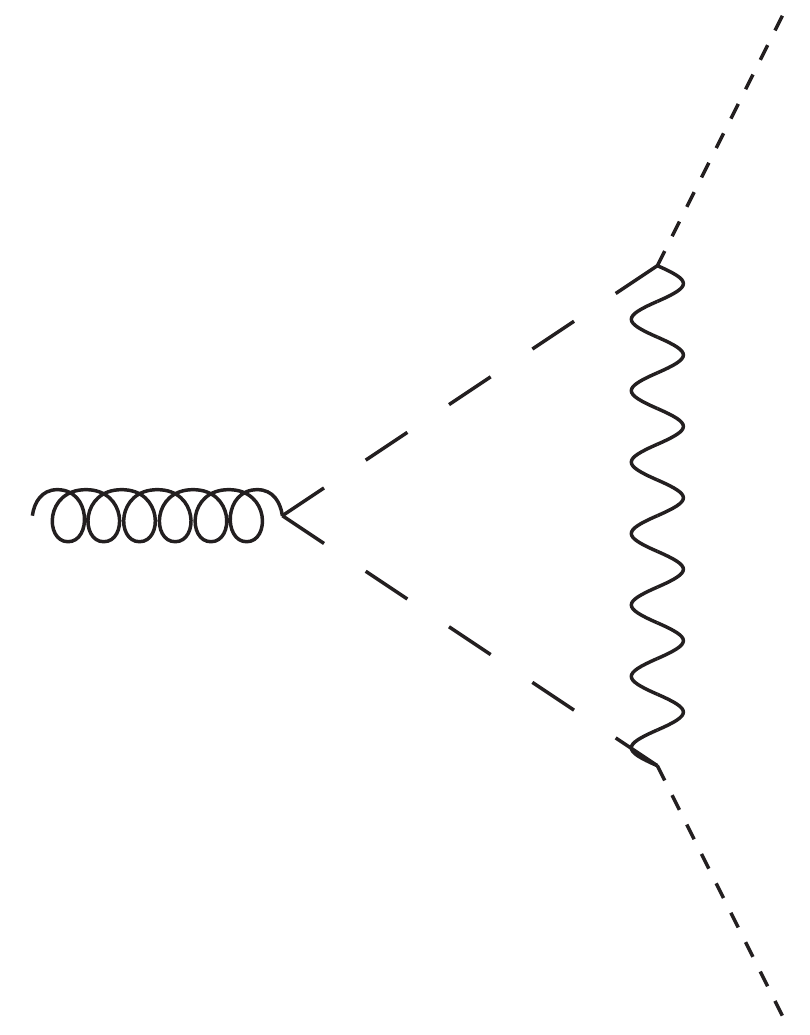} + \tgraph{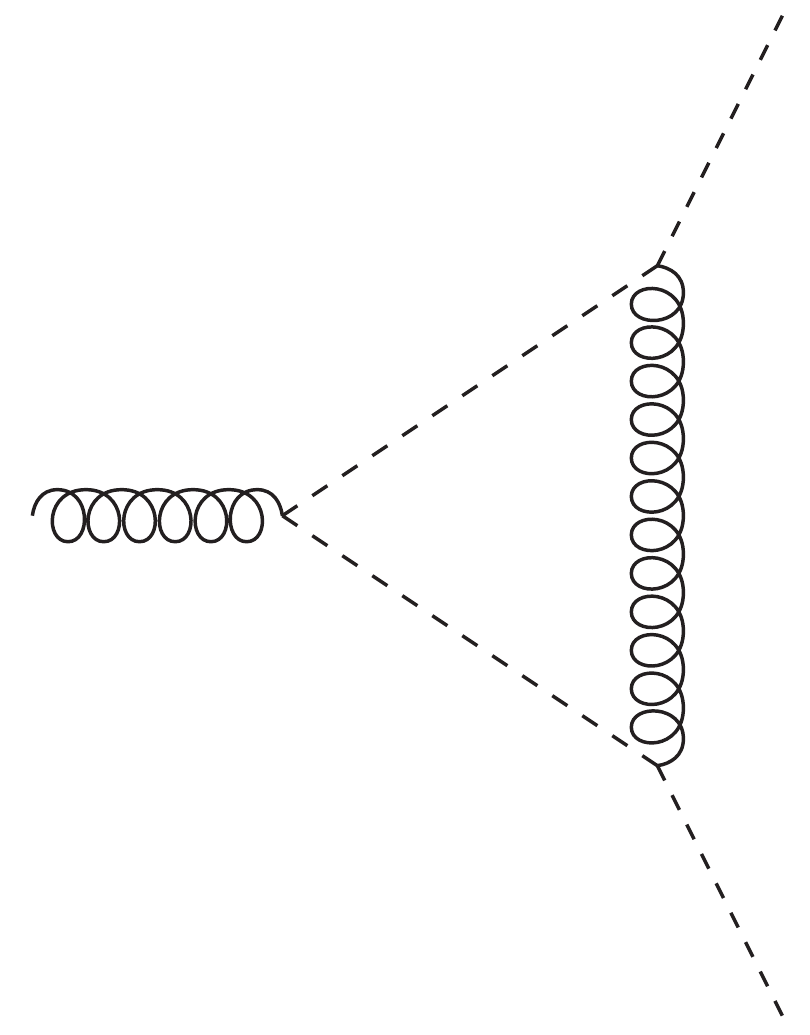} + \frac{1}{2} \tgraph{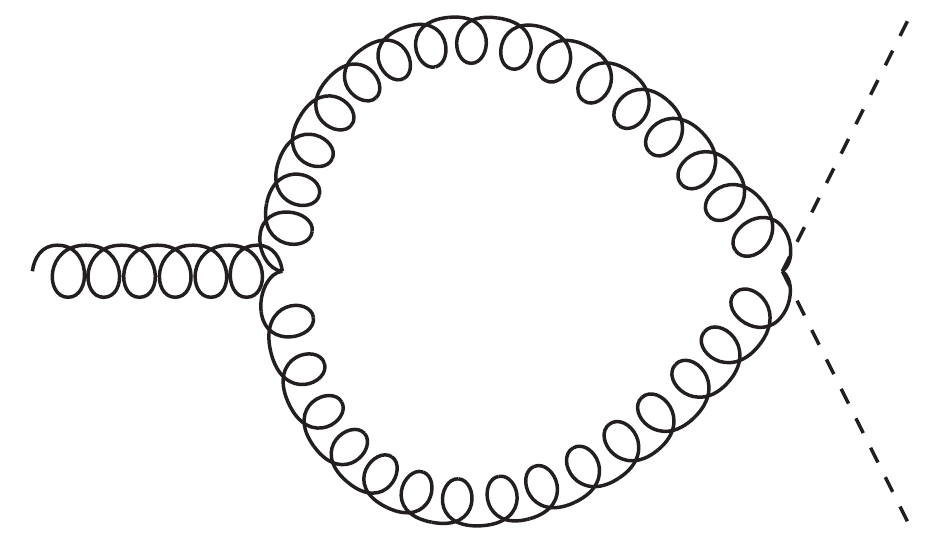} \\ & \phantom{ = } + 2 \ttgraph{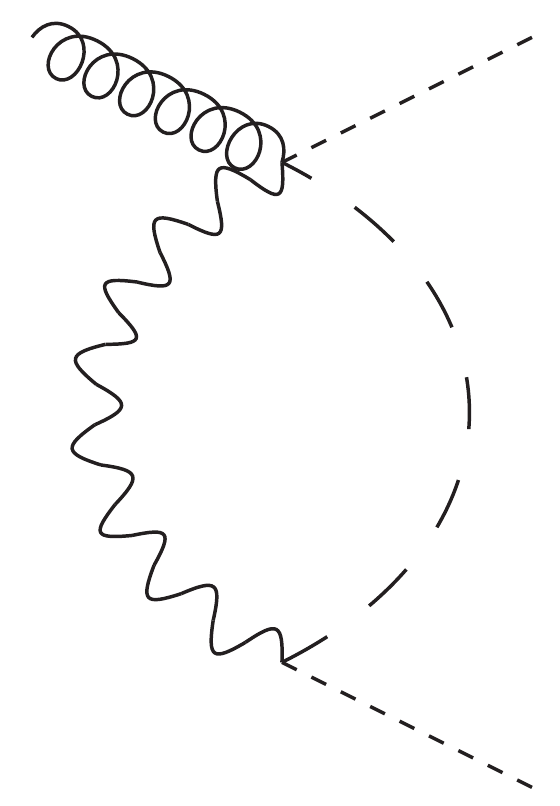} + 2 \ttgraph{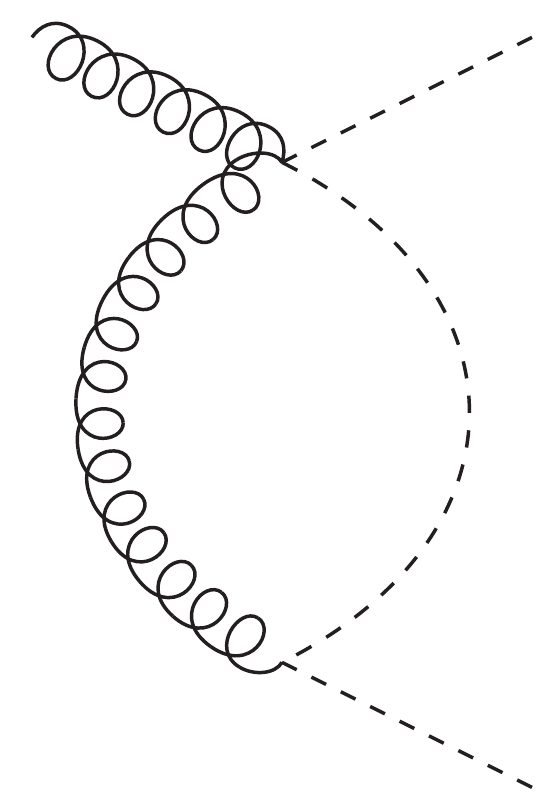}
\end{split}
\\
\rescombgreen^{\cgreen{c-ff}}_{(0,4)} & = - \graph{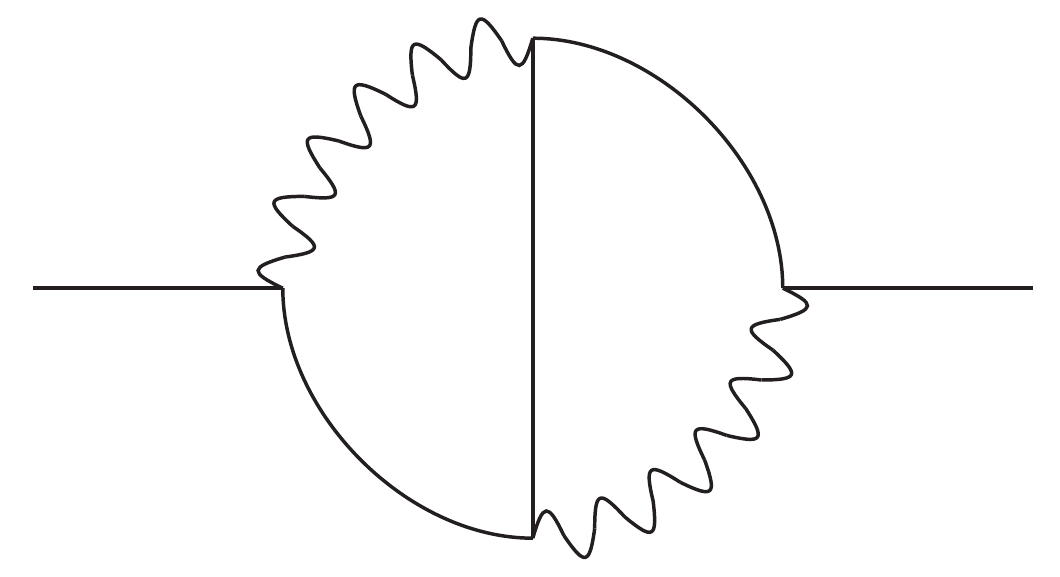} - \frac{1}{2} \graph{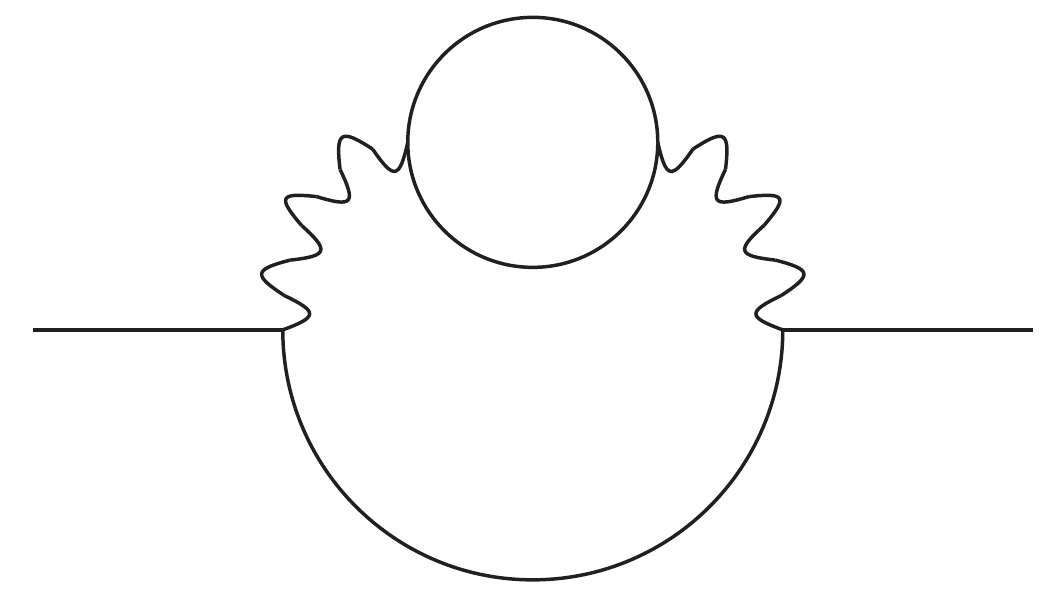} - \graph{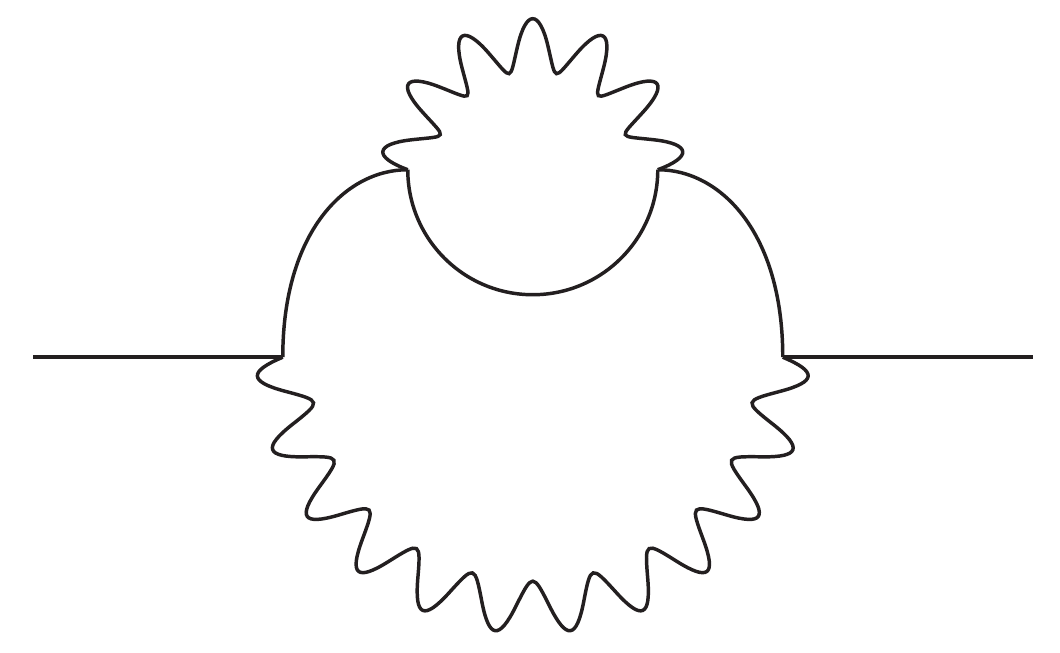}
\\
\begin{split} \label{eqn:c-ff_2-2}
\rescombgreen^{\cgreen{c-ff}}_{(2,2)} & = - 2 \graph{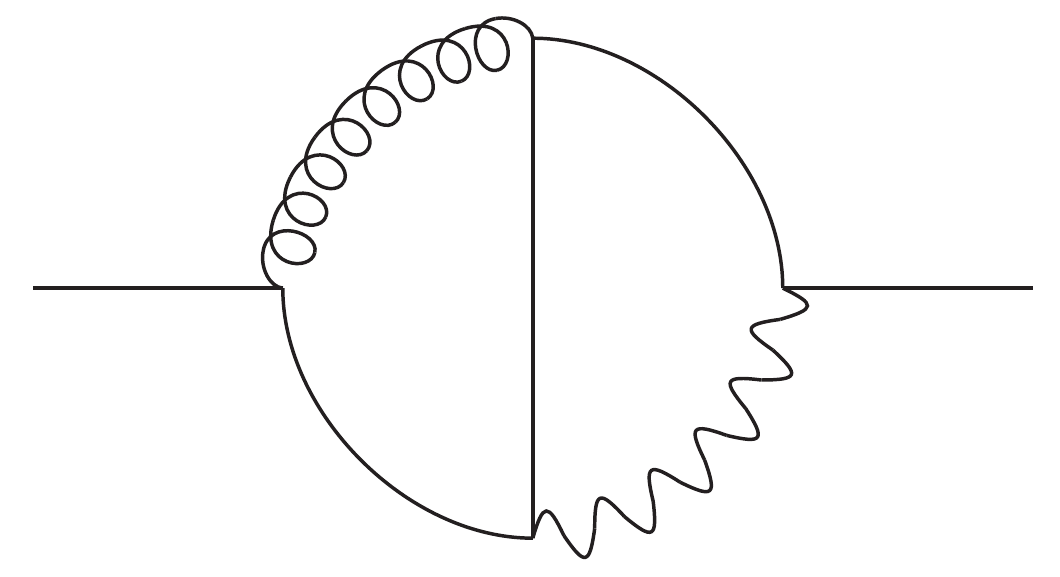} - \graph{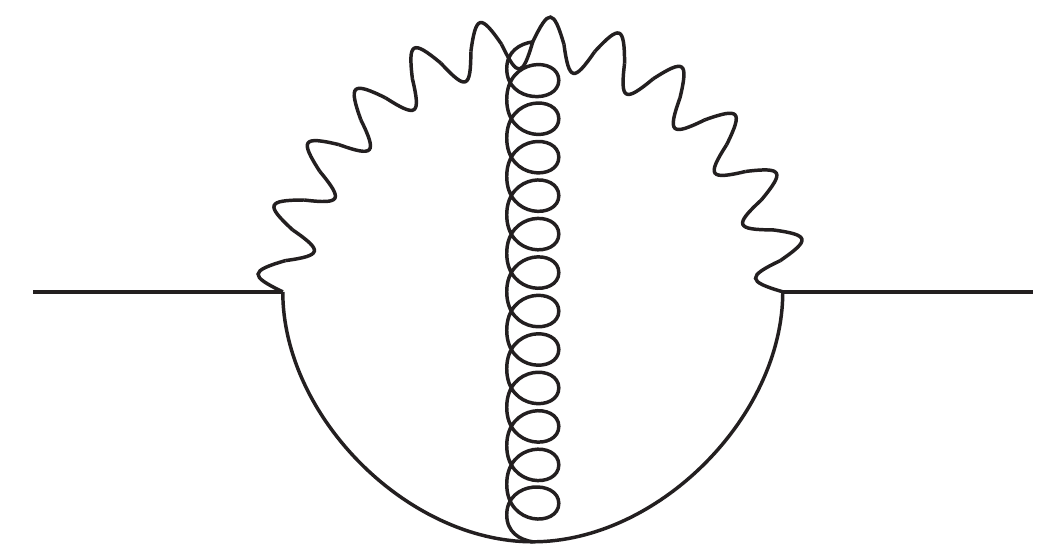} - 2 \graph{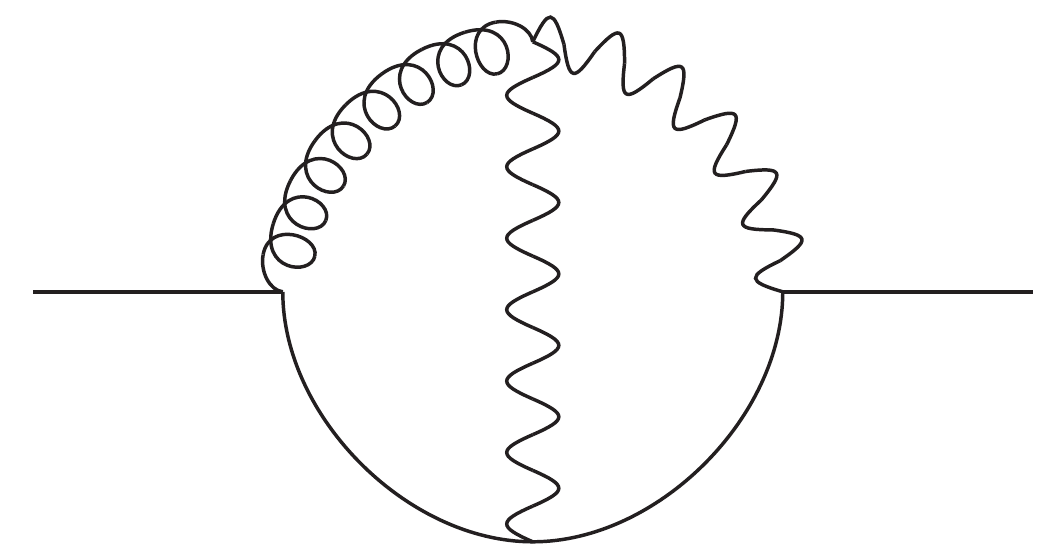} \\ & \phantom{ = } \left ( - \frac{2}{2} \graph{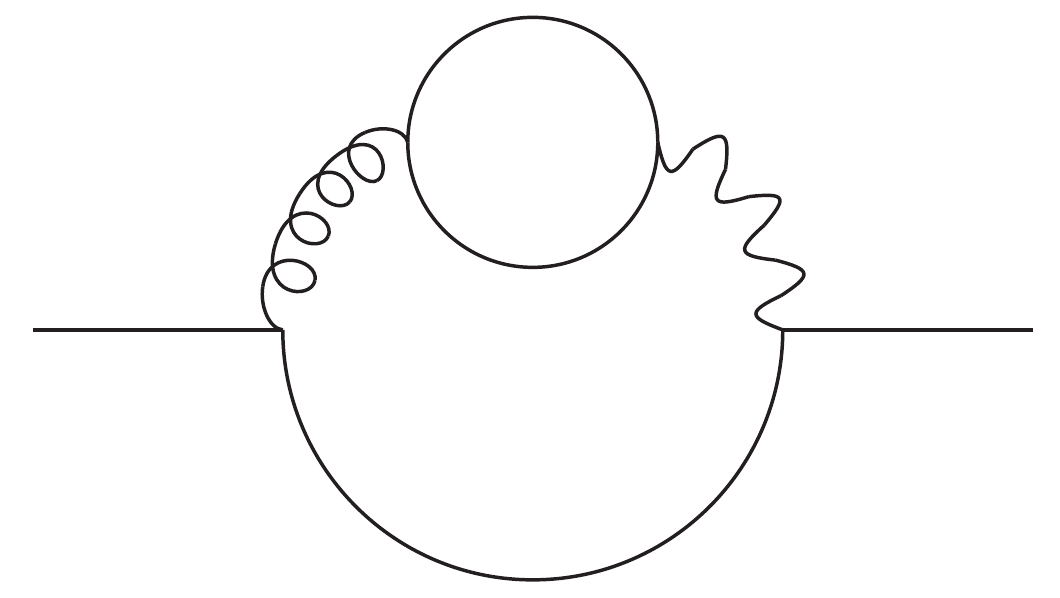} \right ) - \graph{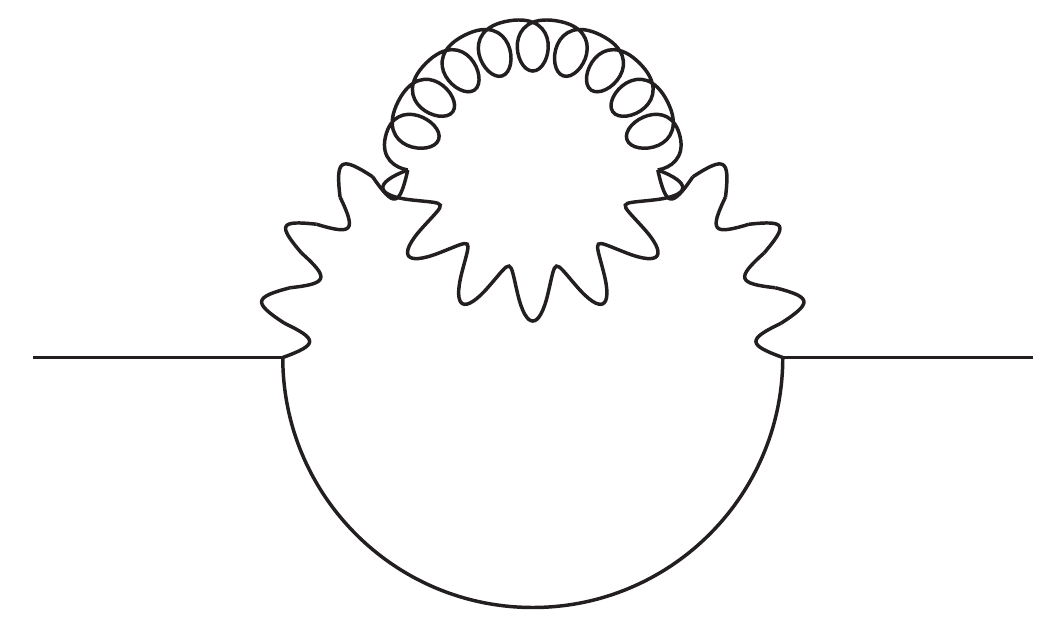} - \graph{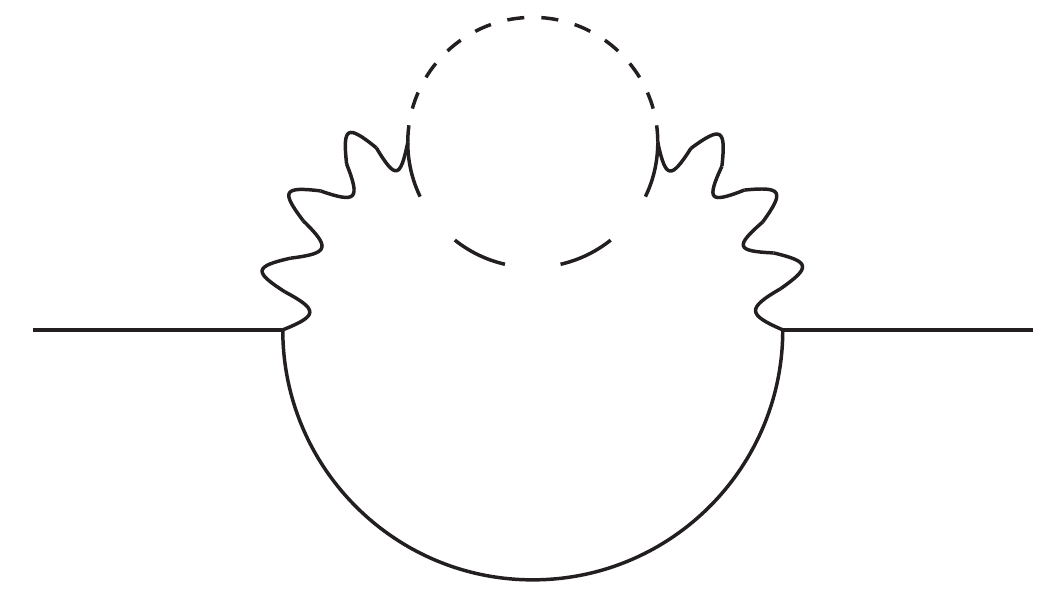} \\ & \phantom{ = } - \graph{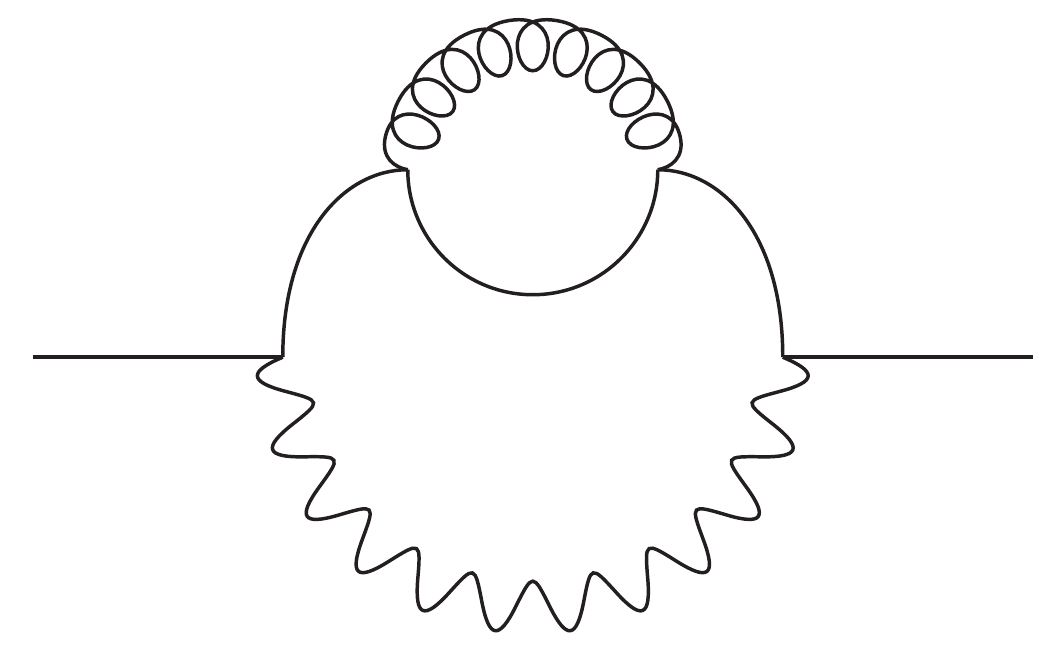} - \graph{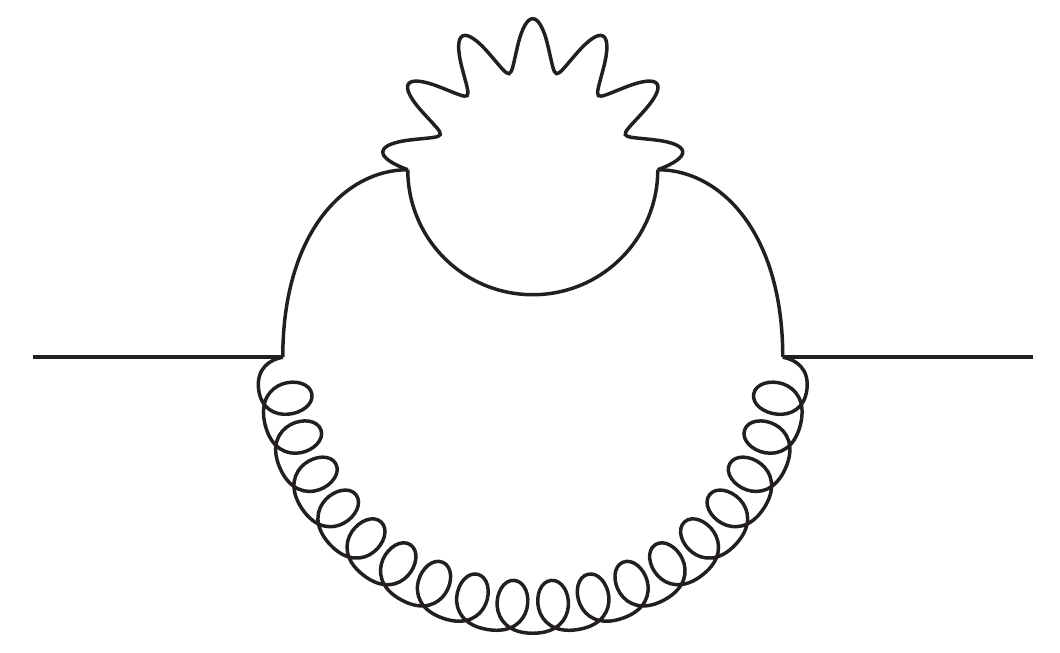} - 2 \ggraph{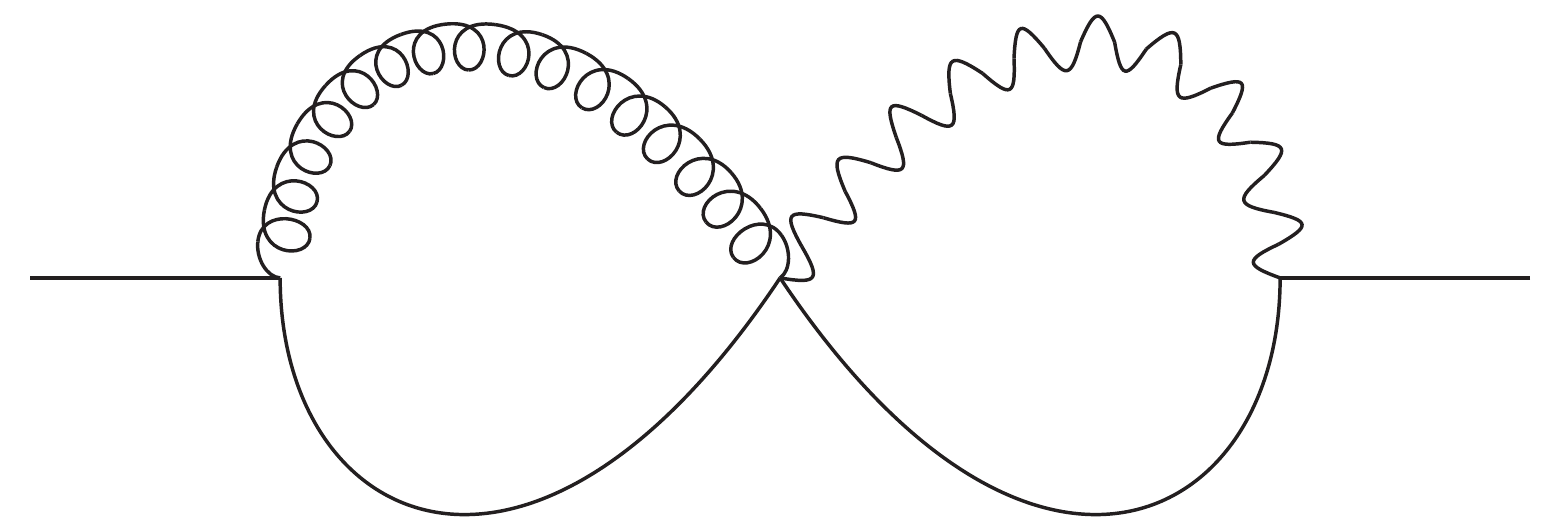} \\ & \phantom{ = } - 2 \graph{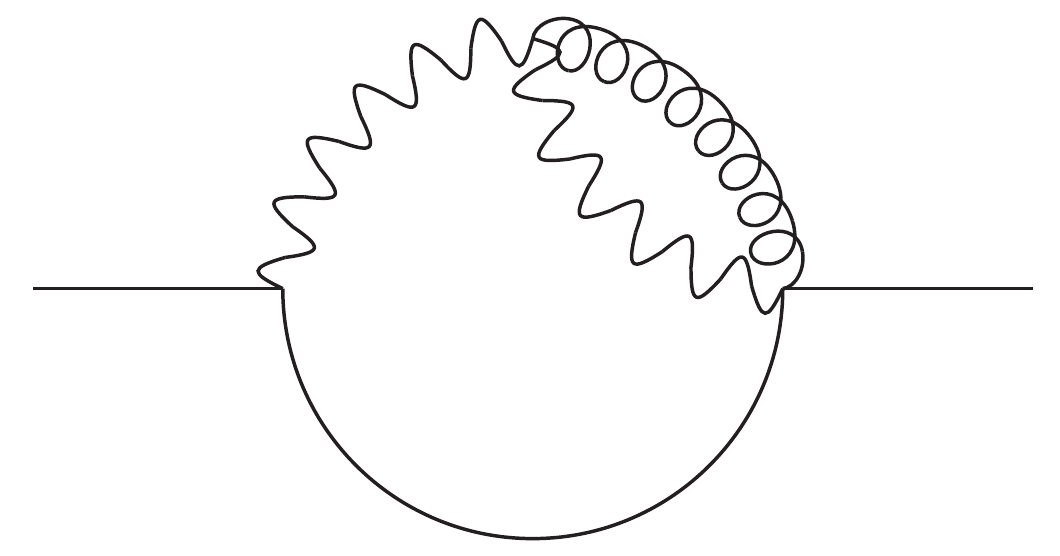} - 2 \graph{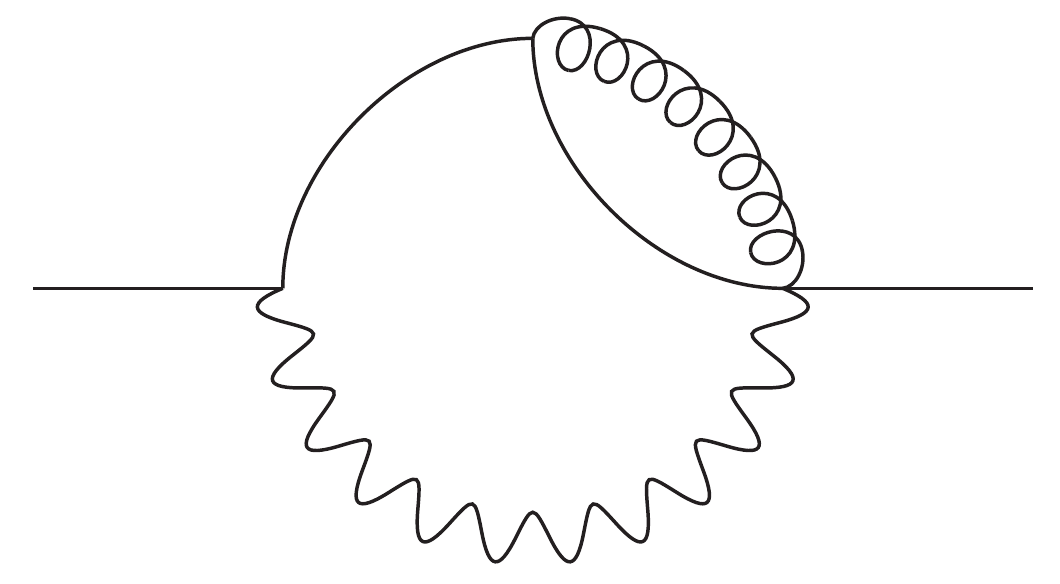} - 2 \graph{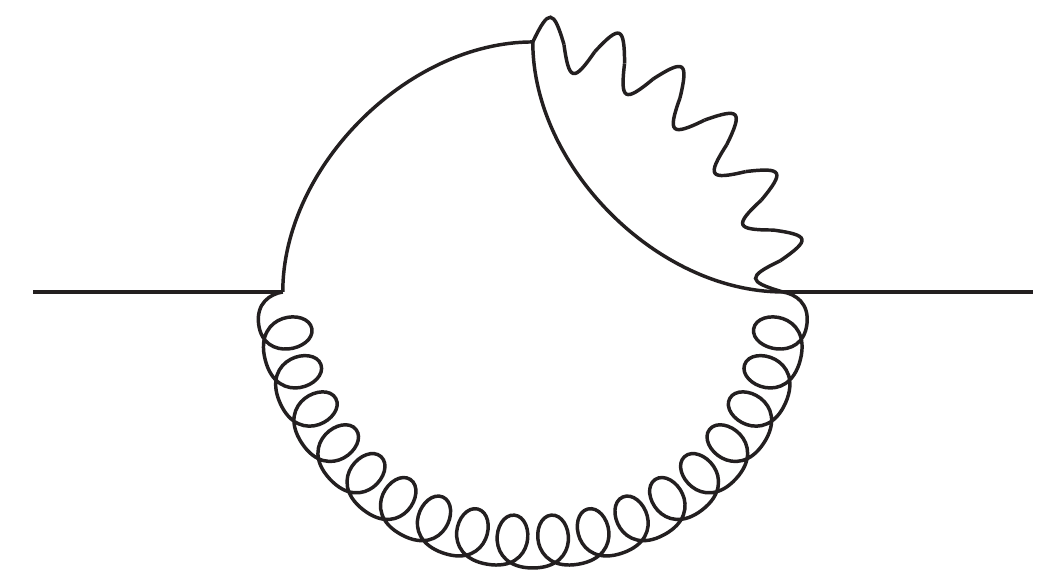}
\end{split}
\\
\begin{split}
\rescombgreen^{\cgreen{c-ff}}_{(4,0)} & = - \graph{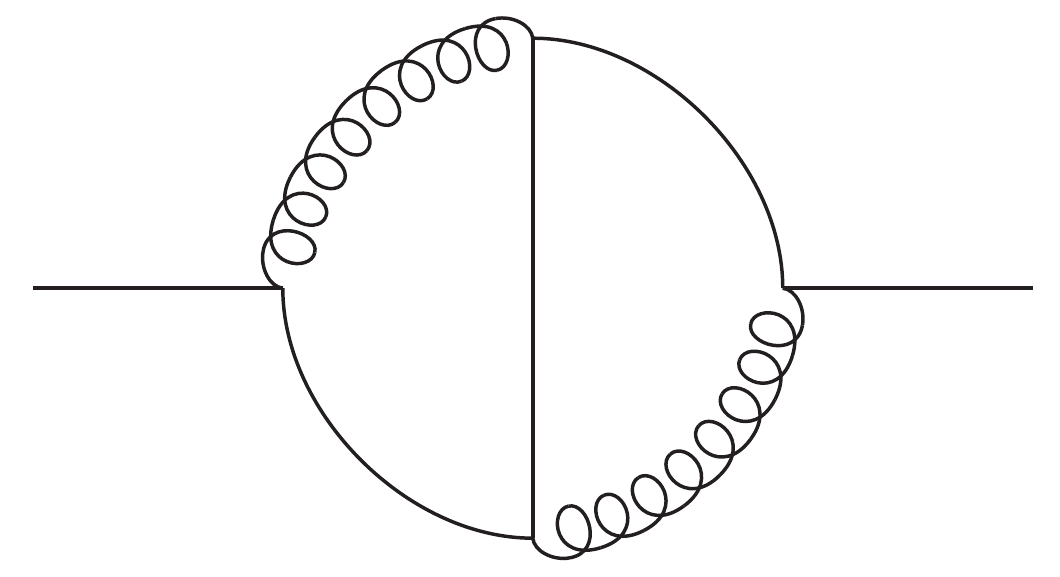} - \graph{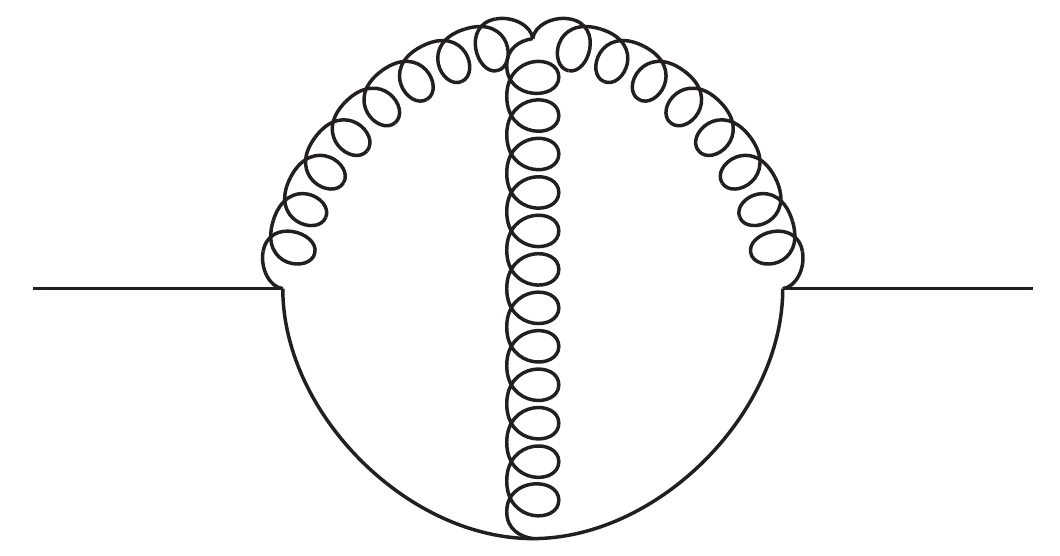} - \frac{1}{2} \graph{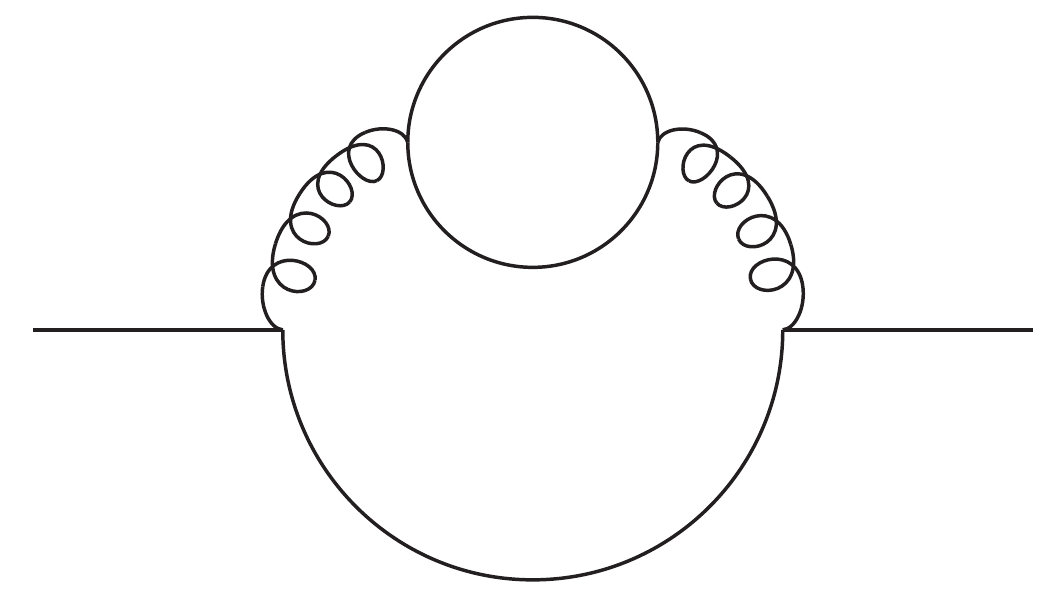} \\ & \phantom{ = } - \frac{1}{2} \graph{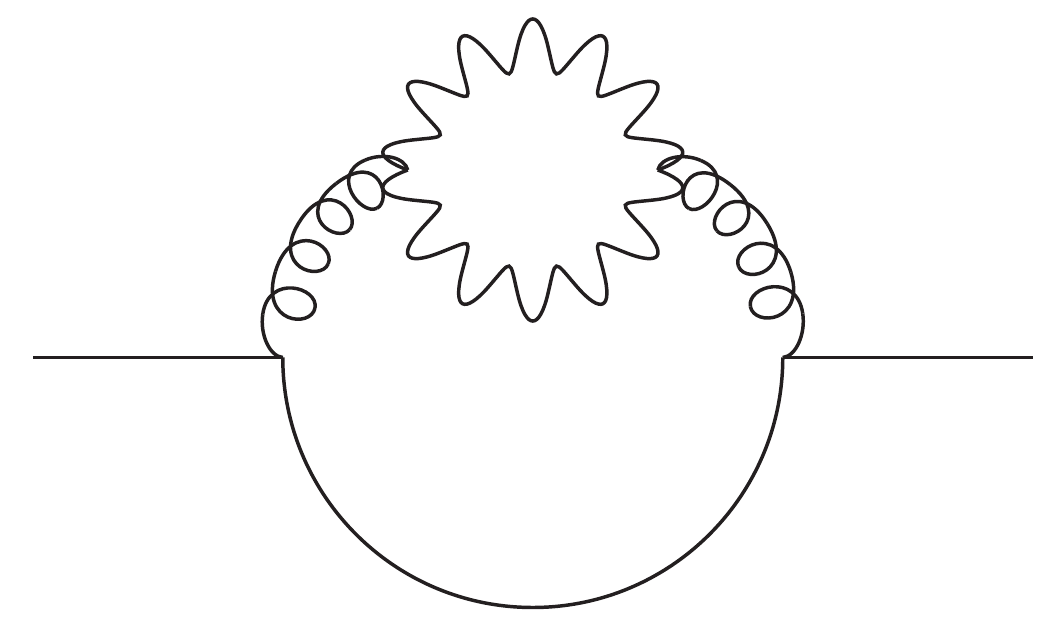} - \frac{1}{2} \graph{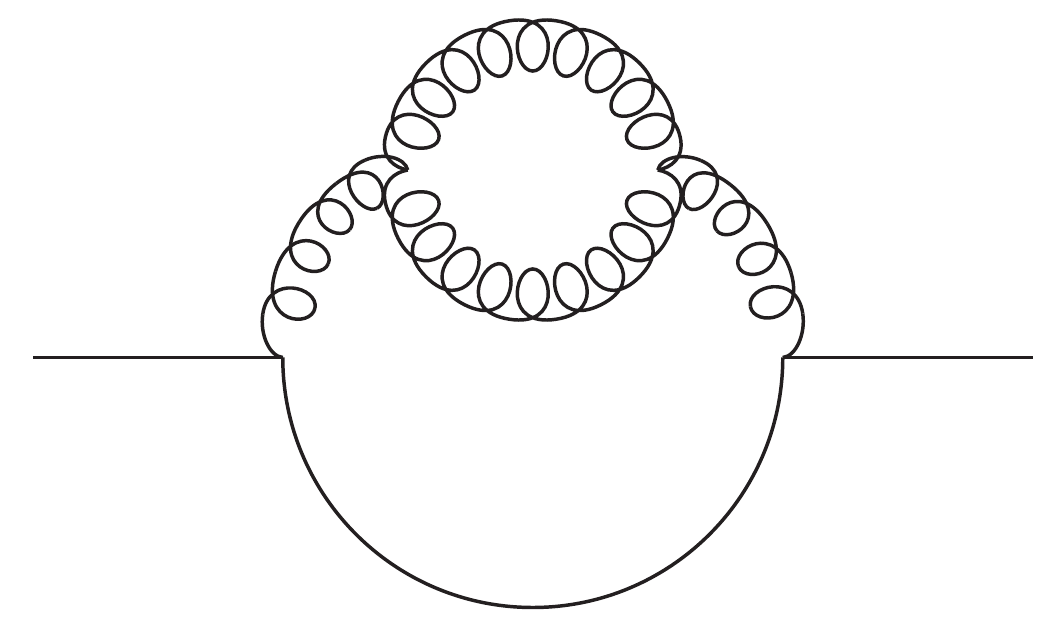} - \frac{1}{2} \graph{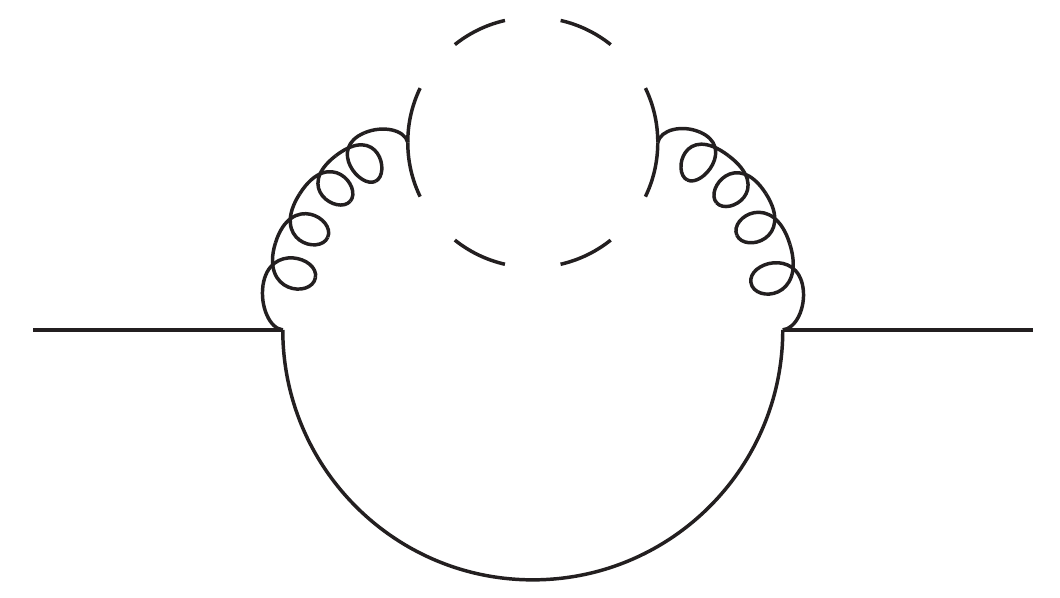} \\ & \phantom{ = } - \frac{1}{2} \graph{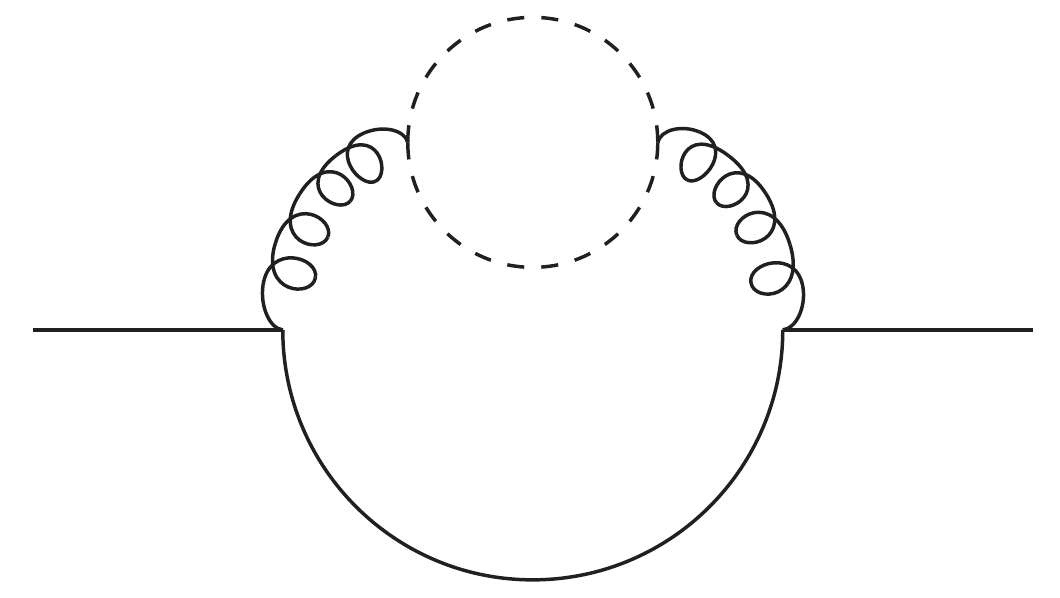} - \graph{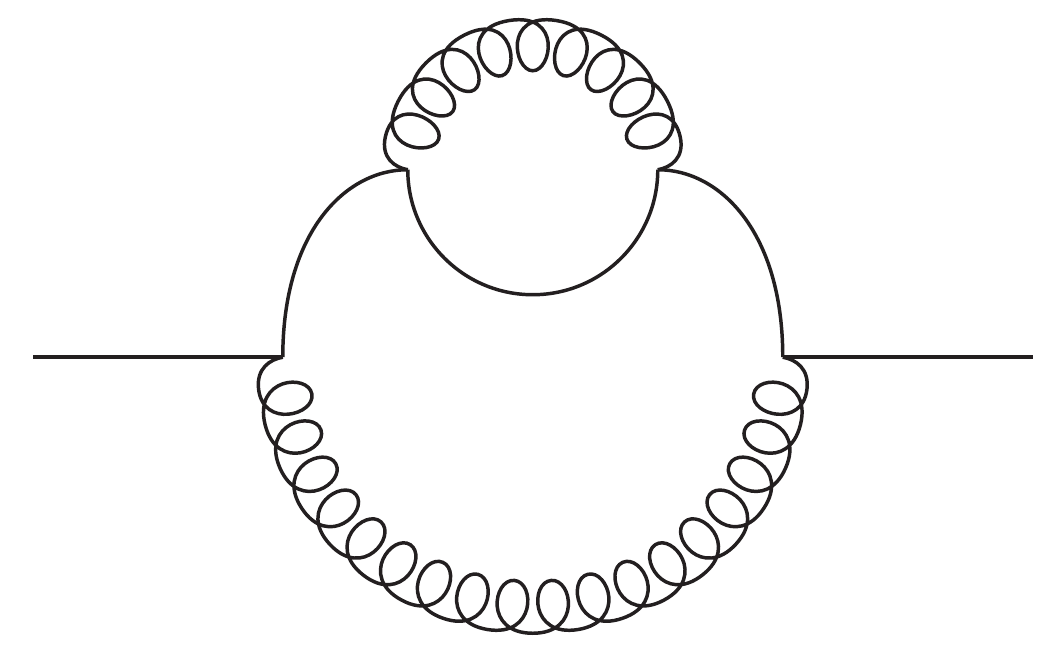} - \ggraph{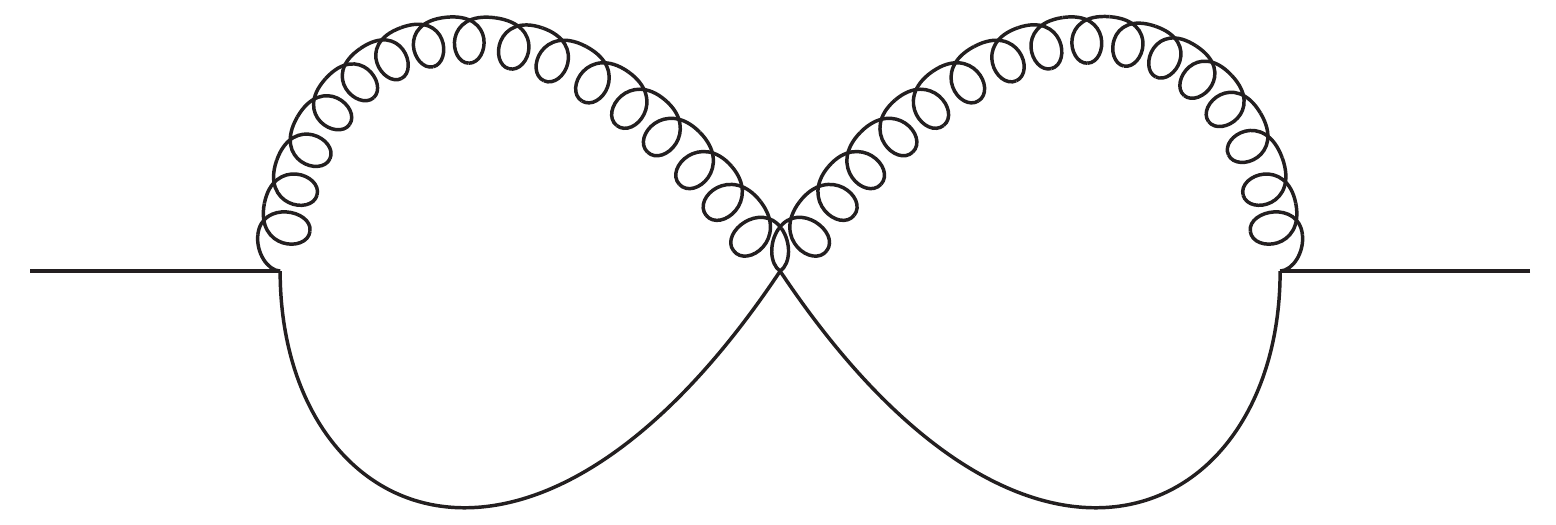} \\ & \phantom{ = } - \frac{2}{2} \graph{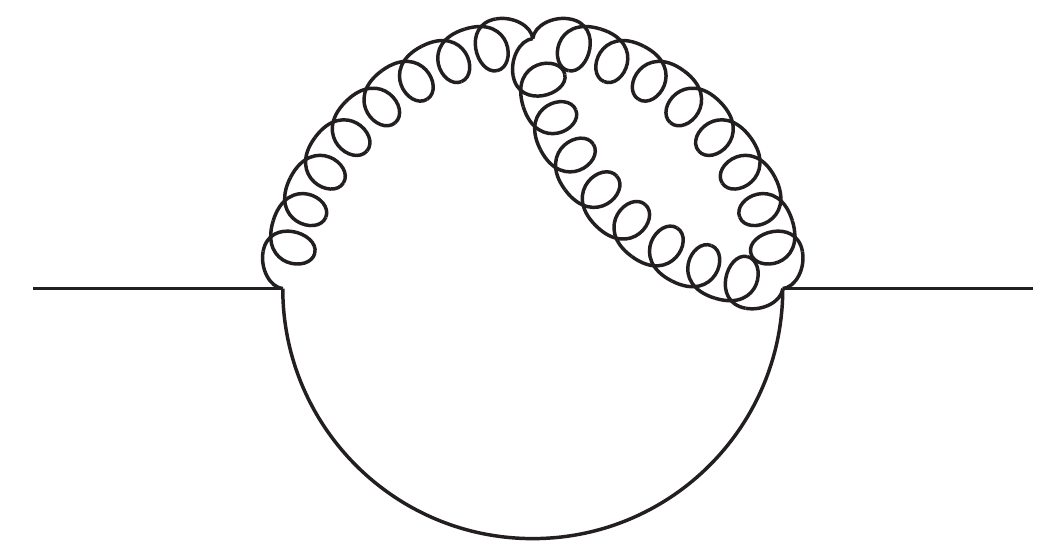} - 2 \graph{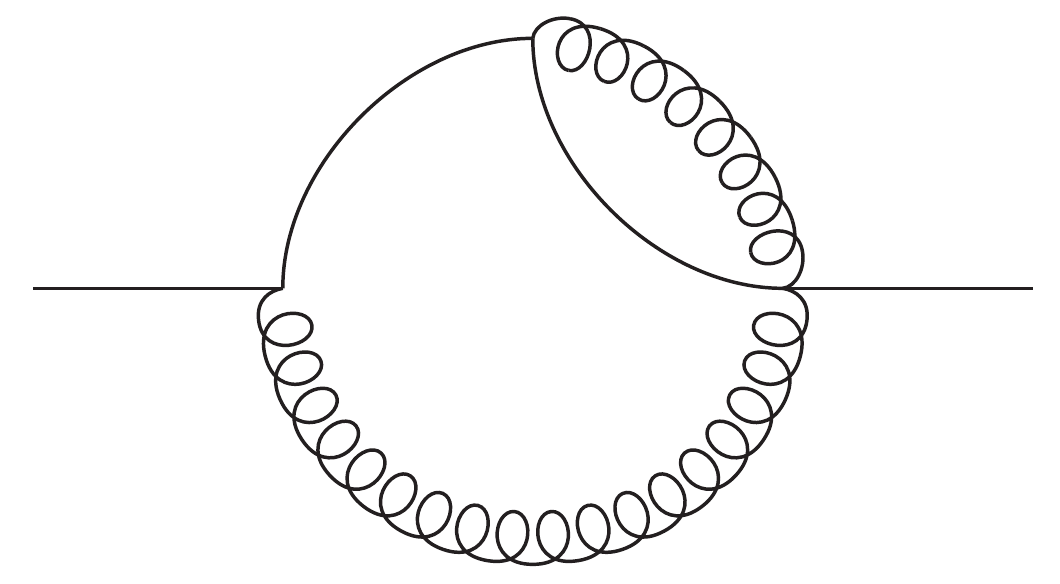}
\end{split}
\\
\rescombgreen^{\cgreen{c-pp}}_{(0,4)} & = - \frac{1}{2} \graph{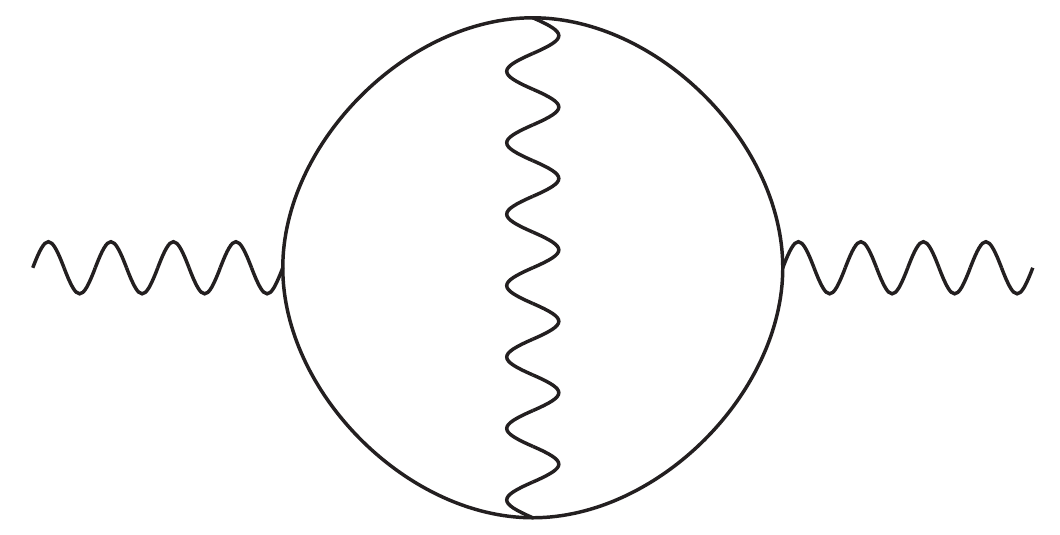} - \graph{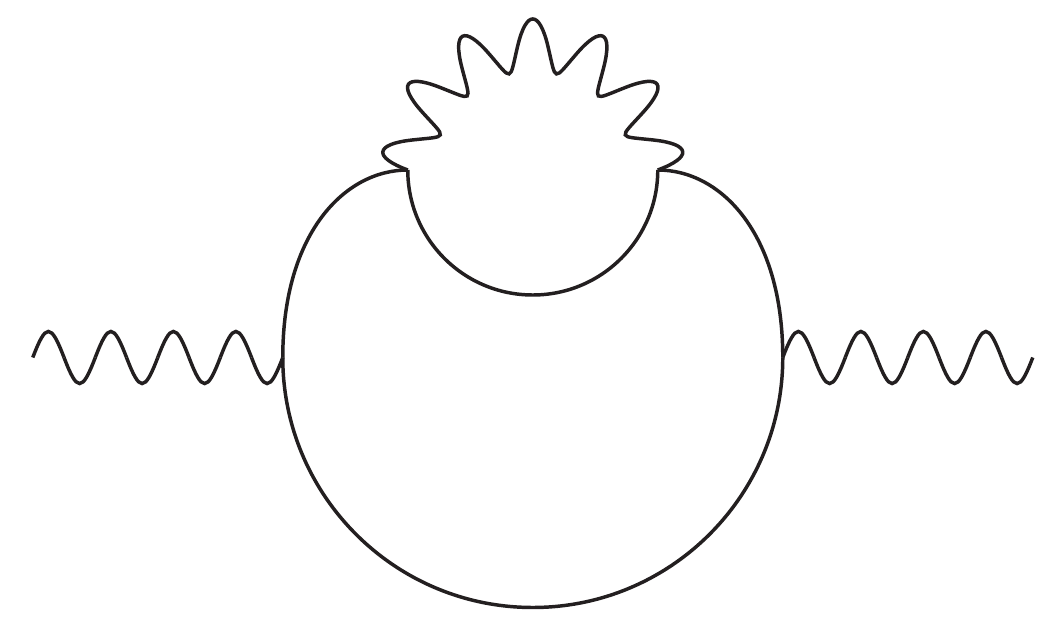}
\\
\begin{split}
\rescombgreen^{\cgreen{c-pp}}_{(2,2)} & = - \frac{1}{2} \graph{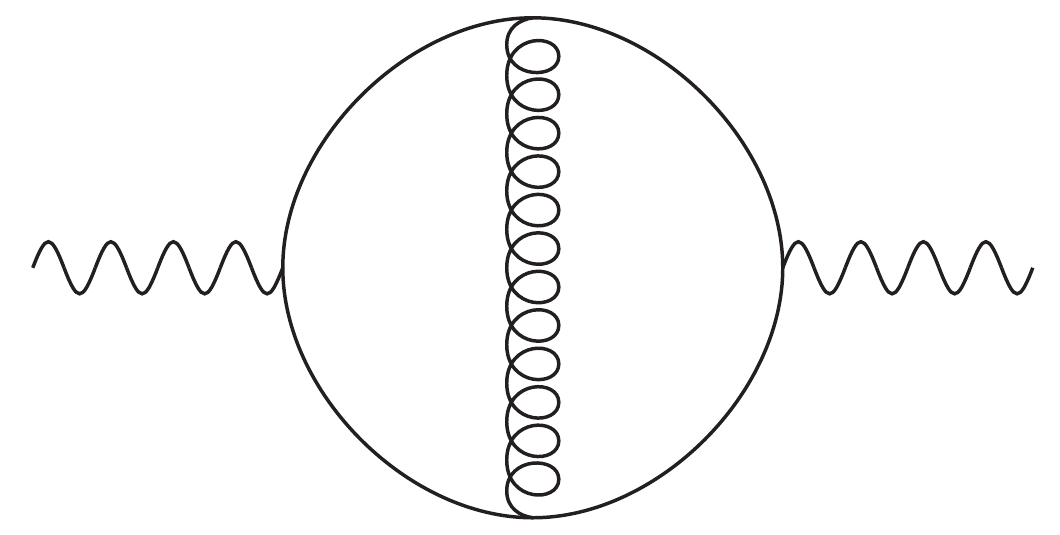} - 2 \graph{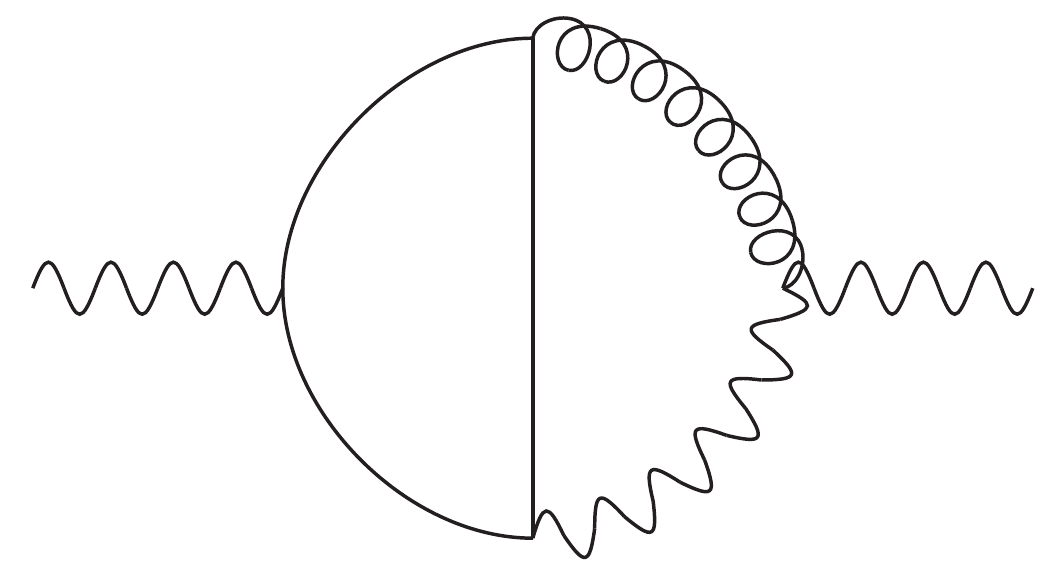} - \graph{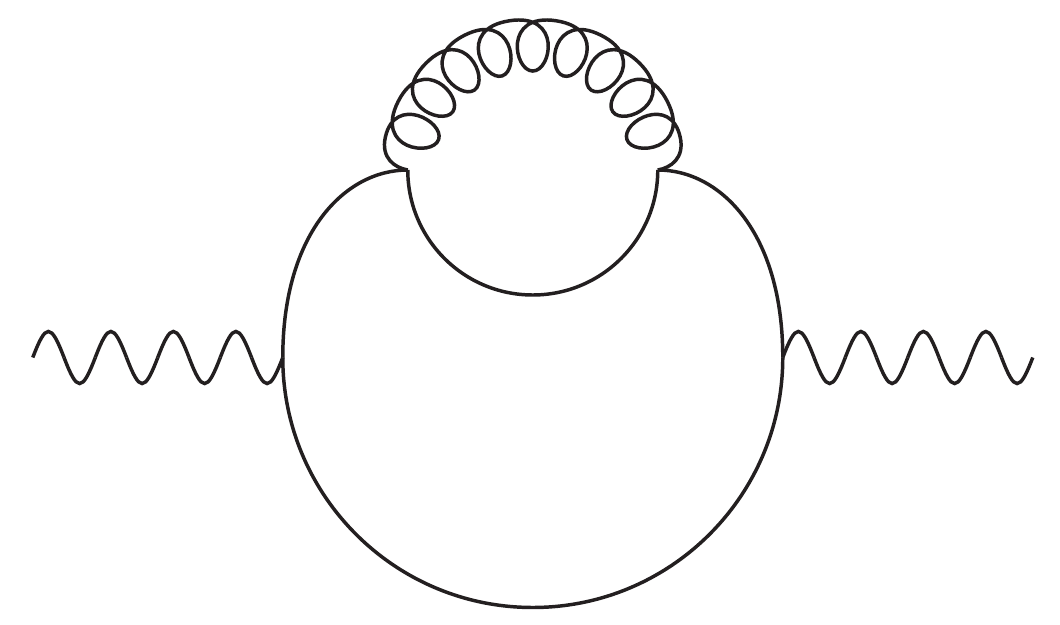} \\ & \phantom{ = } - \frac{1}{2} \graph{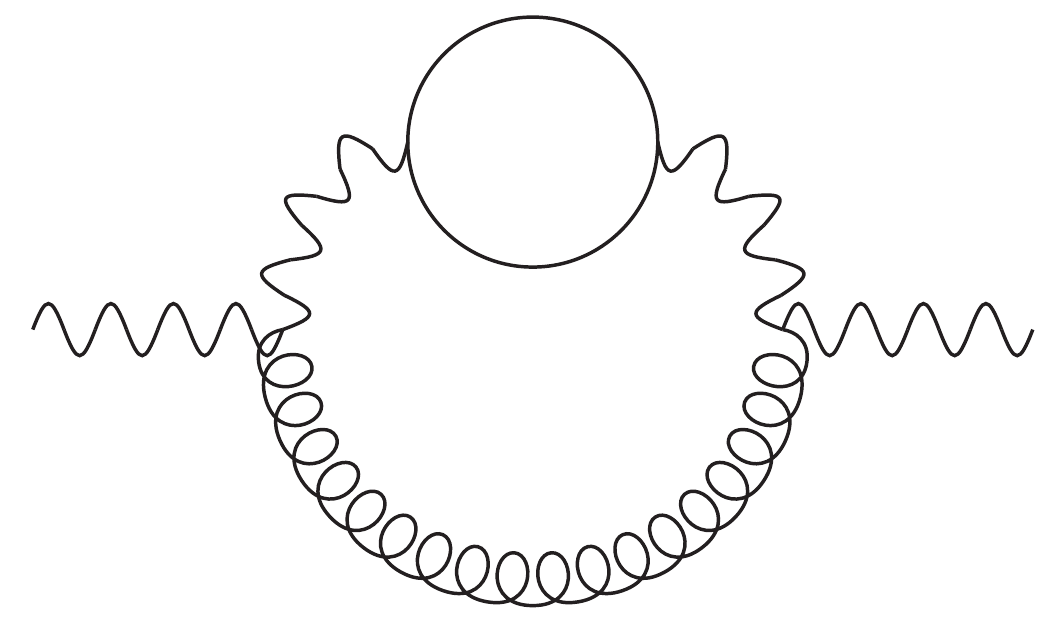} - \frac{2}{2} \ggraph{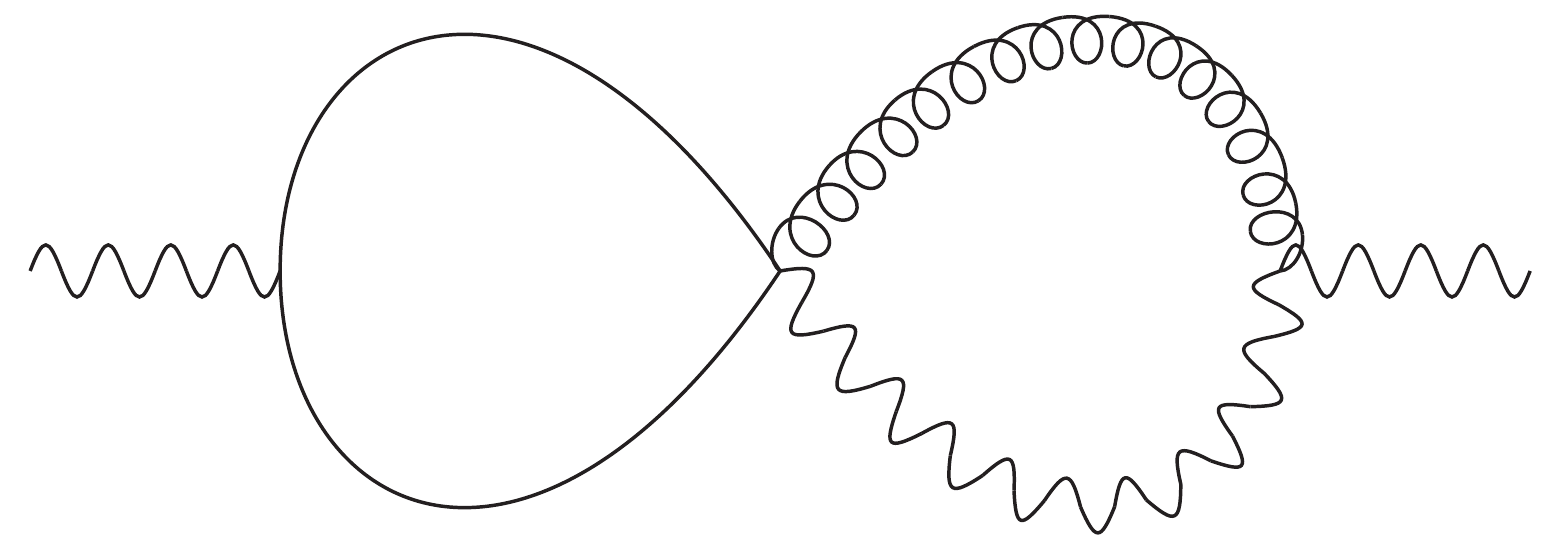} - 2 \graph{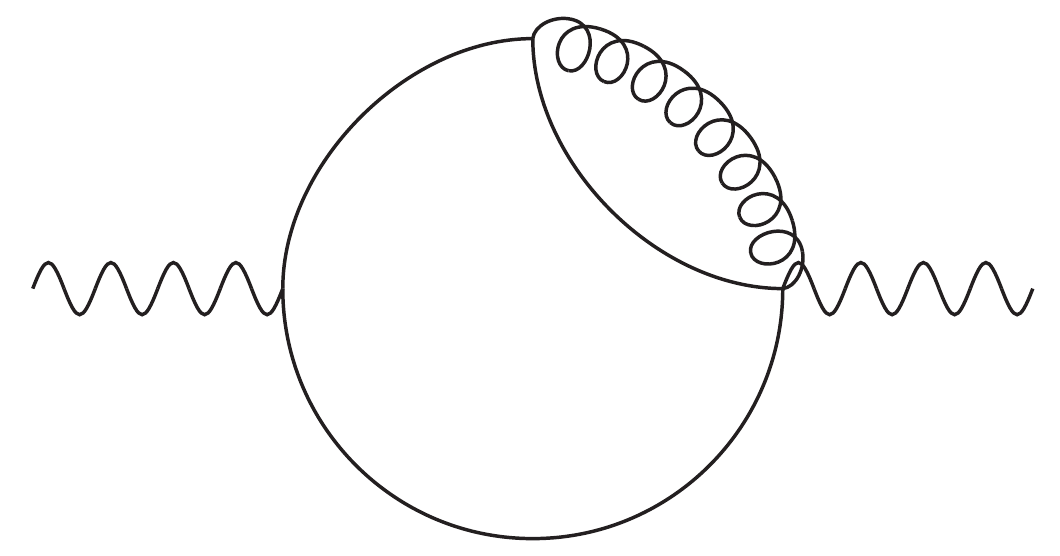} \\ & \phantom{ = } - \frac{2}{2} \graph{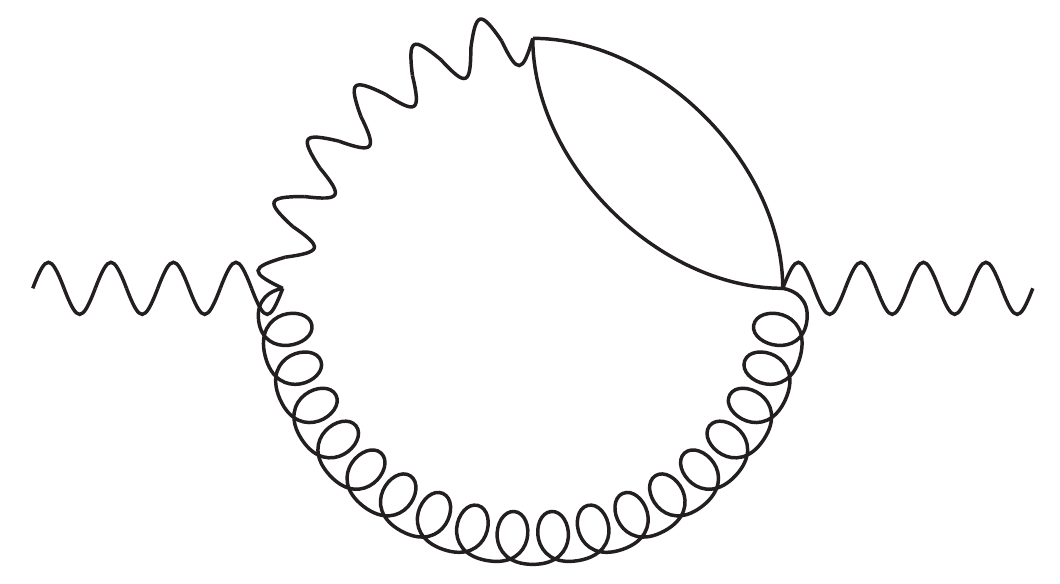}
\end{split}
\\
\begin{split}
\rescombgreen^{\cgreen{c-pp}}_{(4,0)} & = - \graph{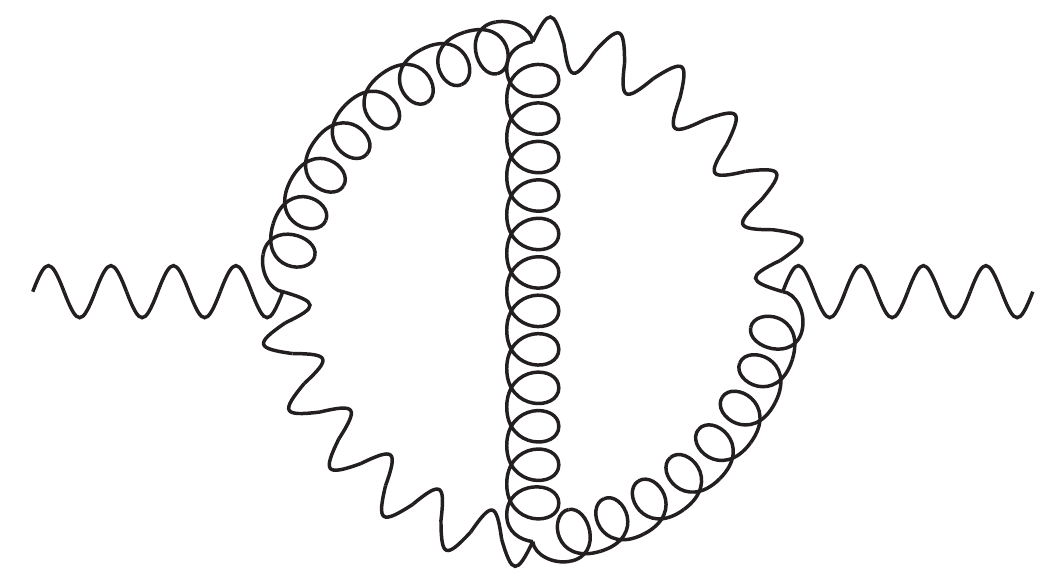} - \graph{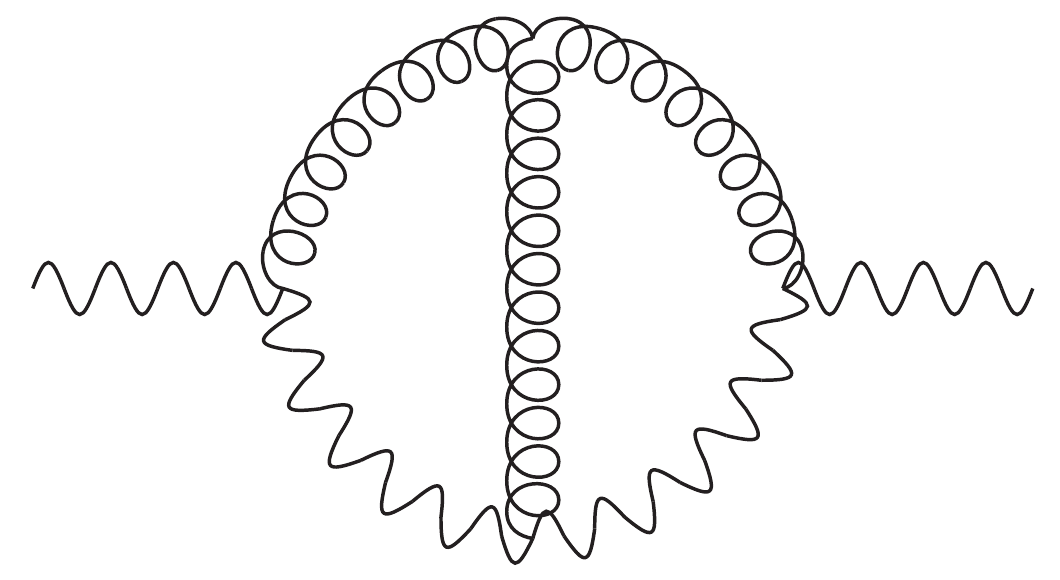} - 2 \graph{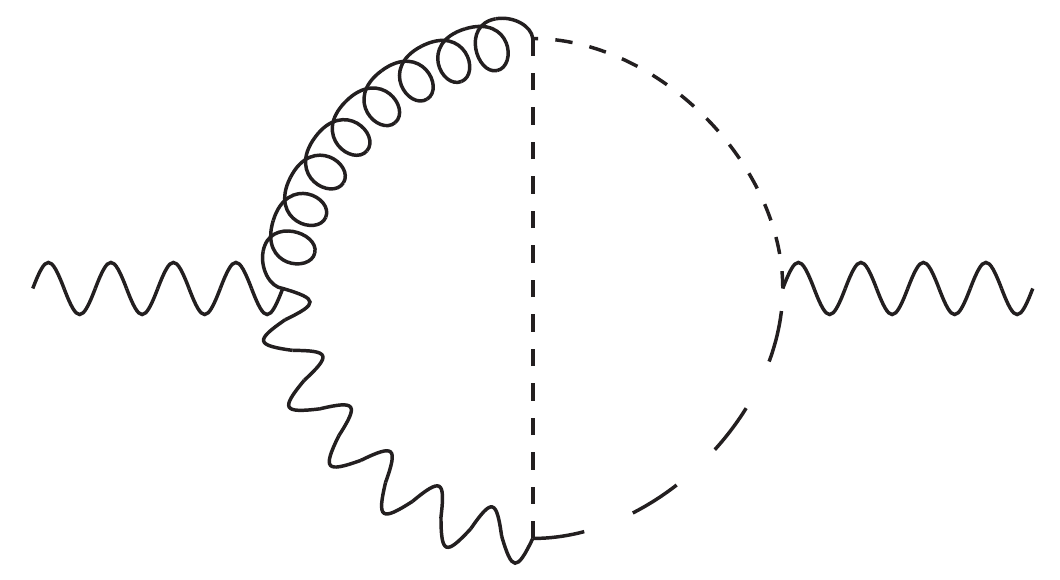} \\ & \phantom{ = } - 2 \graph{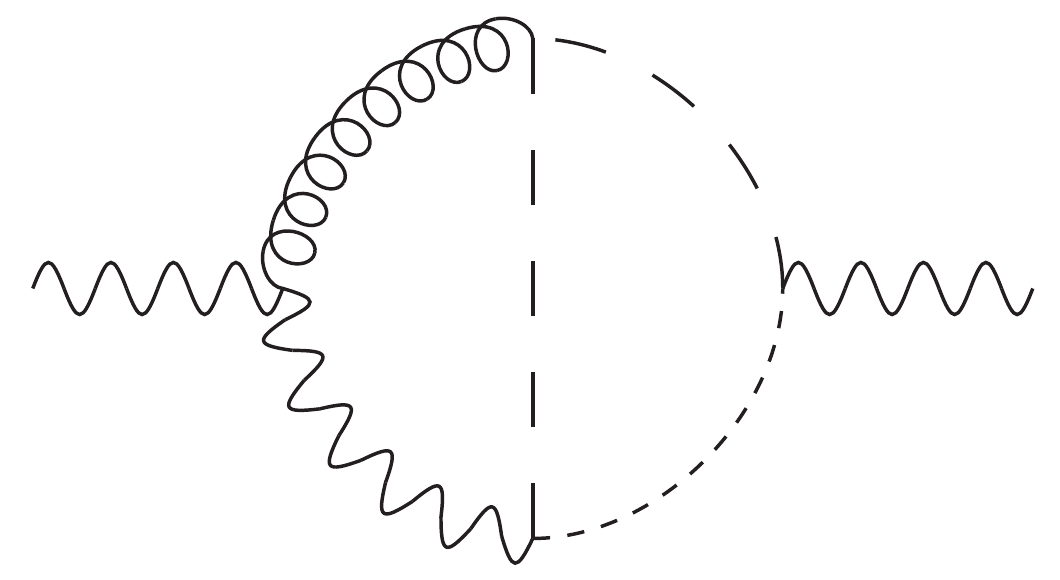} - \graph{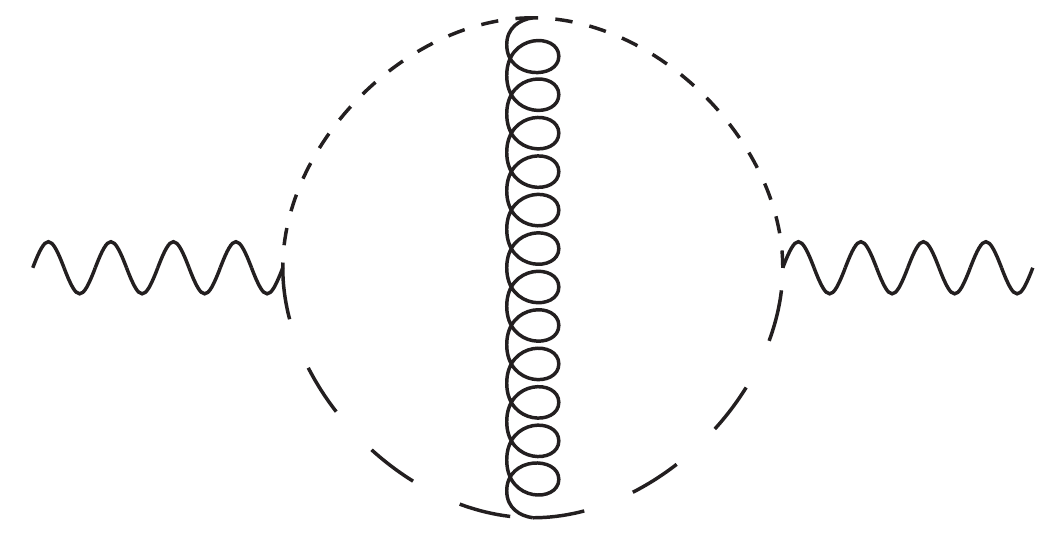} - \graph{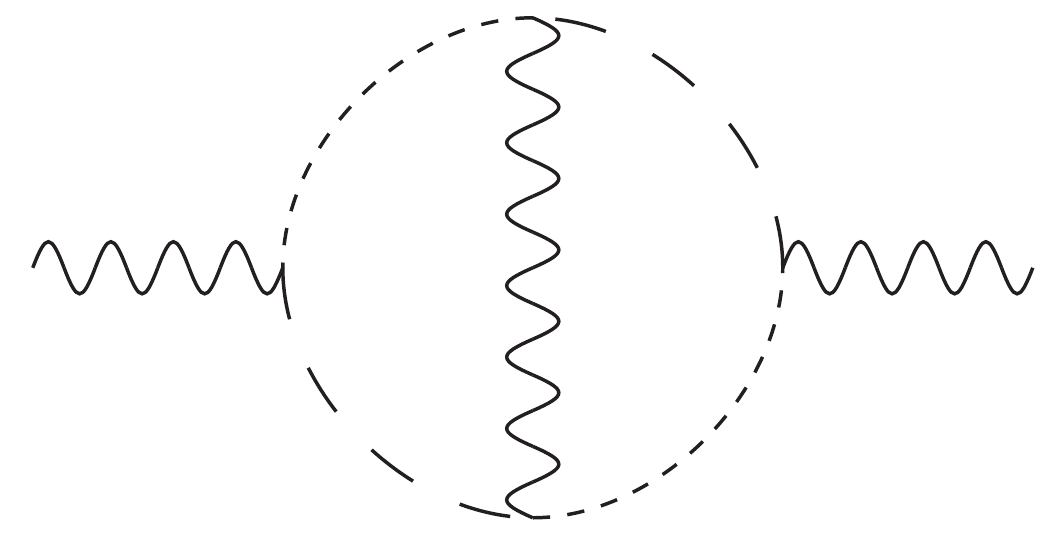} \\ & \phantom{ = } - \frac{1}{2} \graph{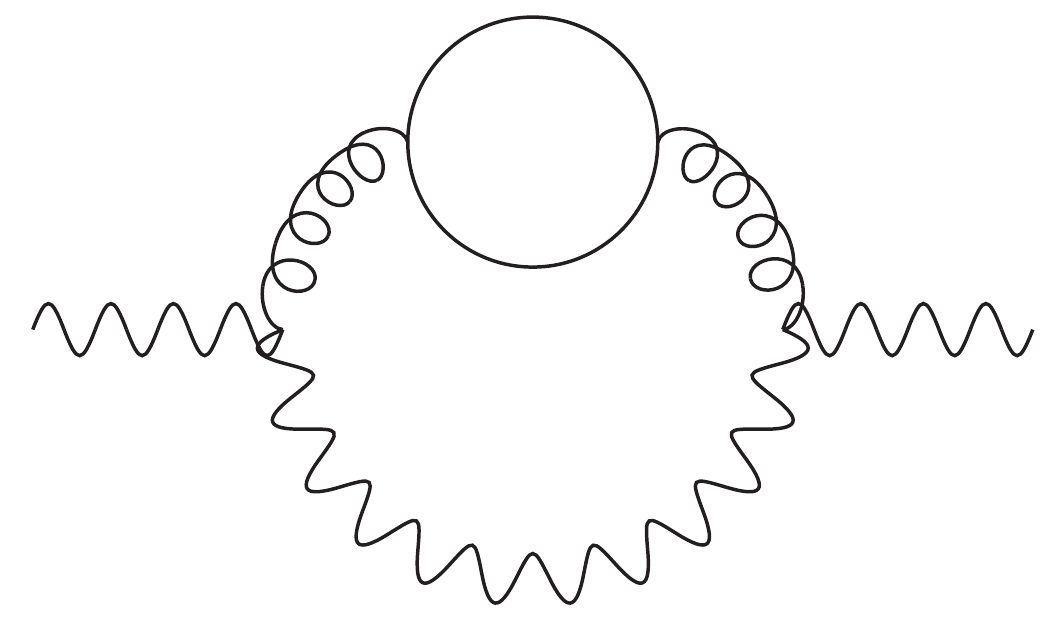} - \frac{1}{2} \graph{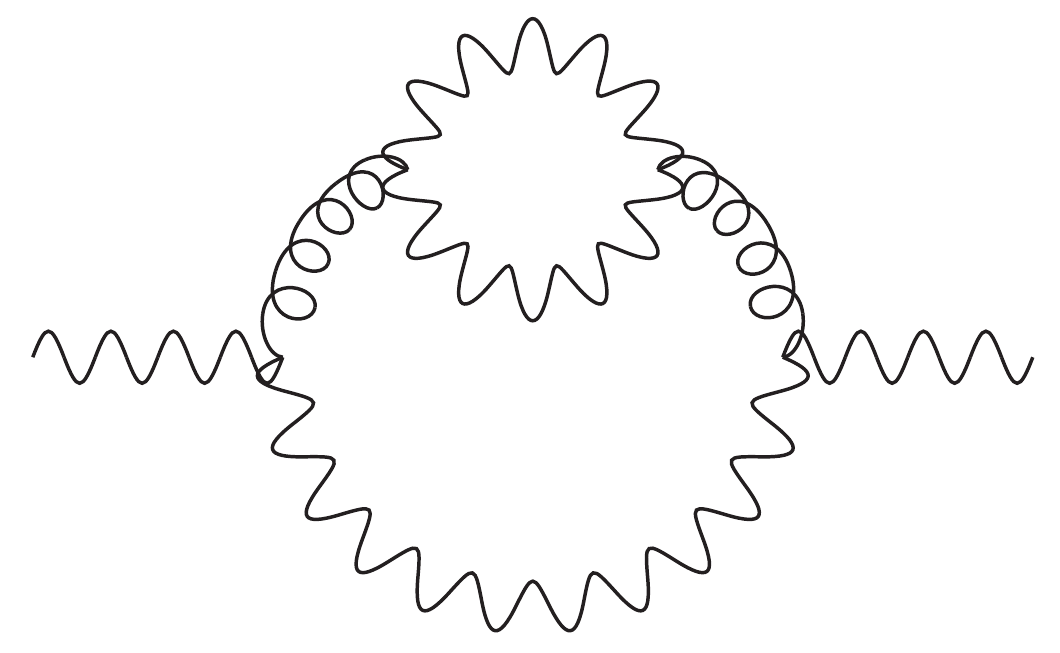} - \frac{1}{2} \graph{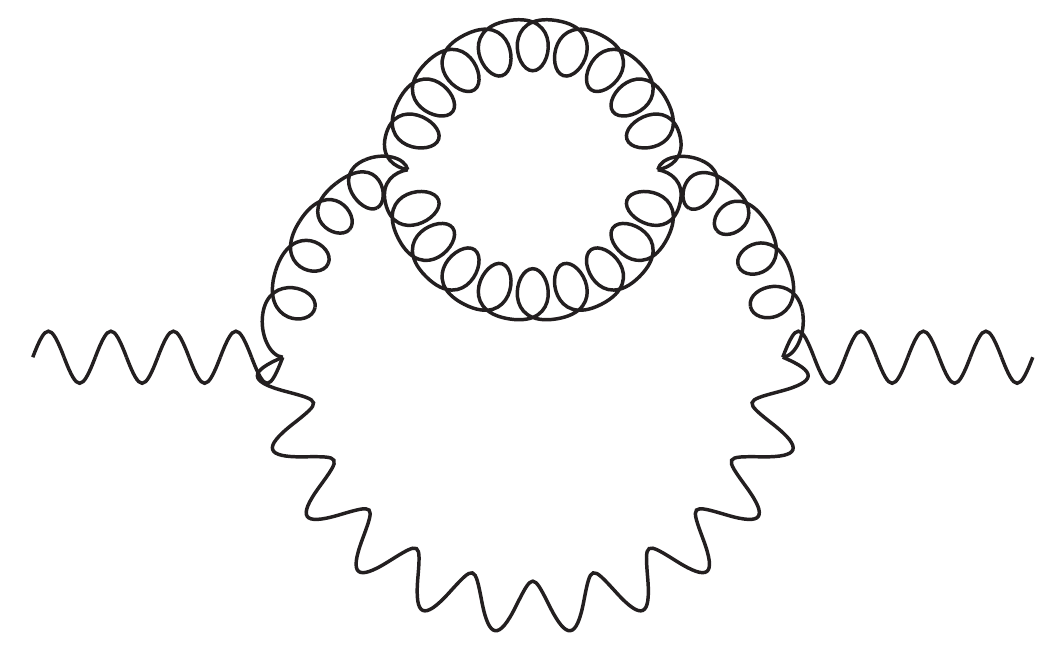} \\ & \phantom{ = } - \frac{1}{2} \graph{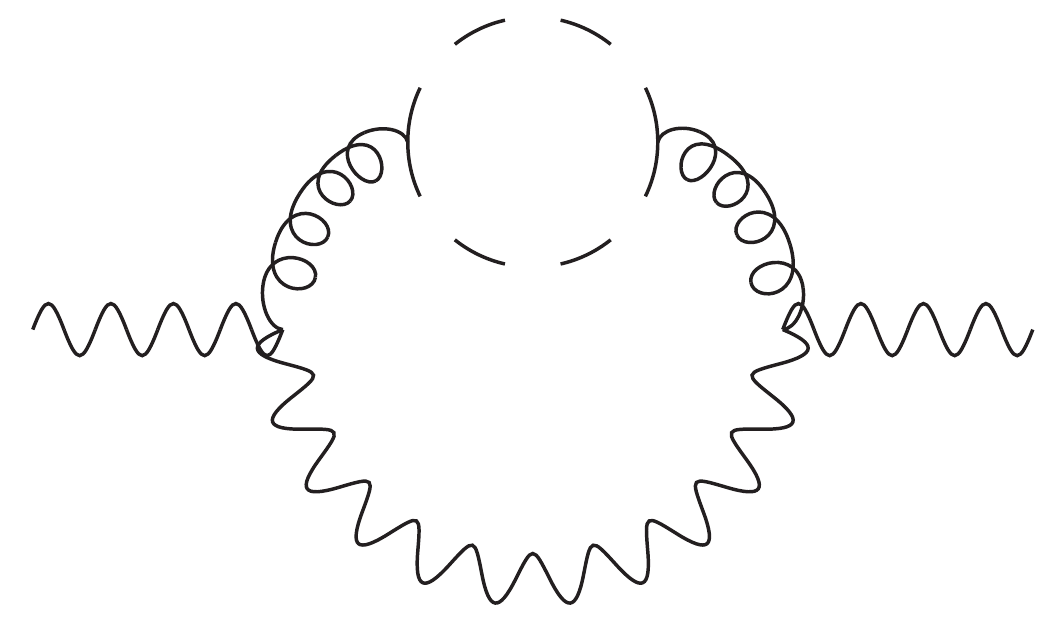} - \frac{1}{2} \graph{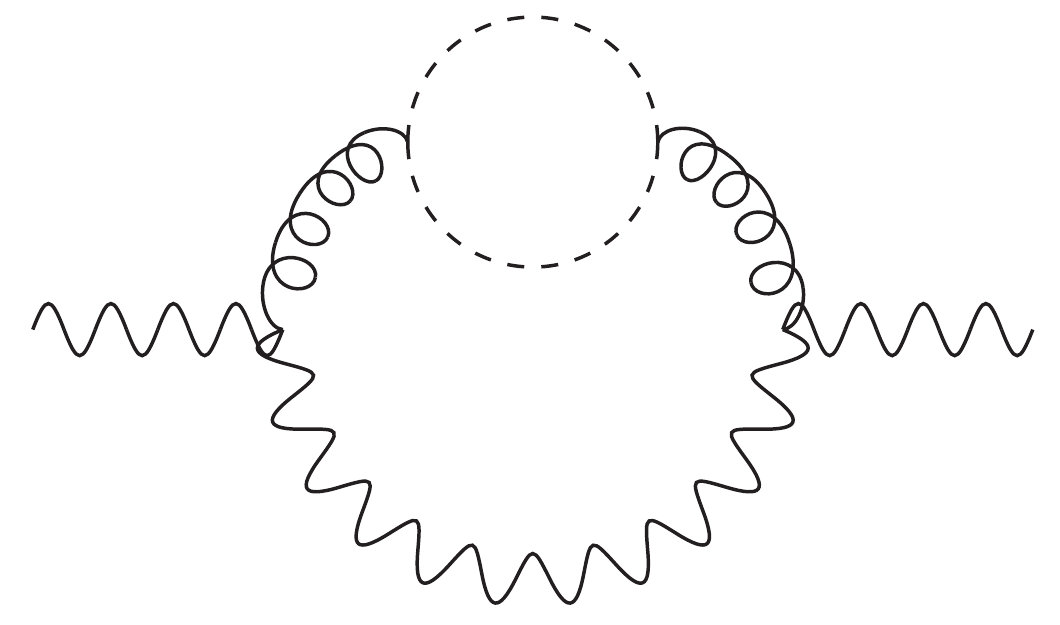} - \graph{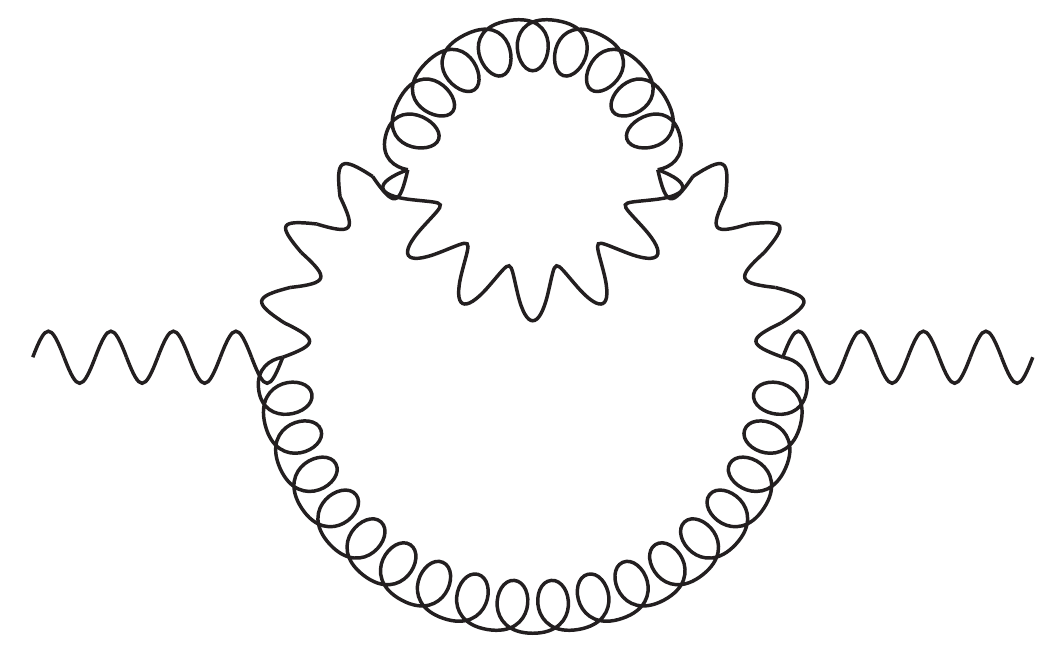} \\ & \phantom{ = } - \graph{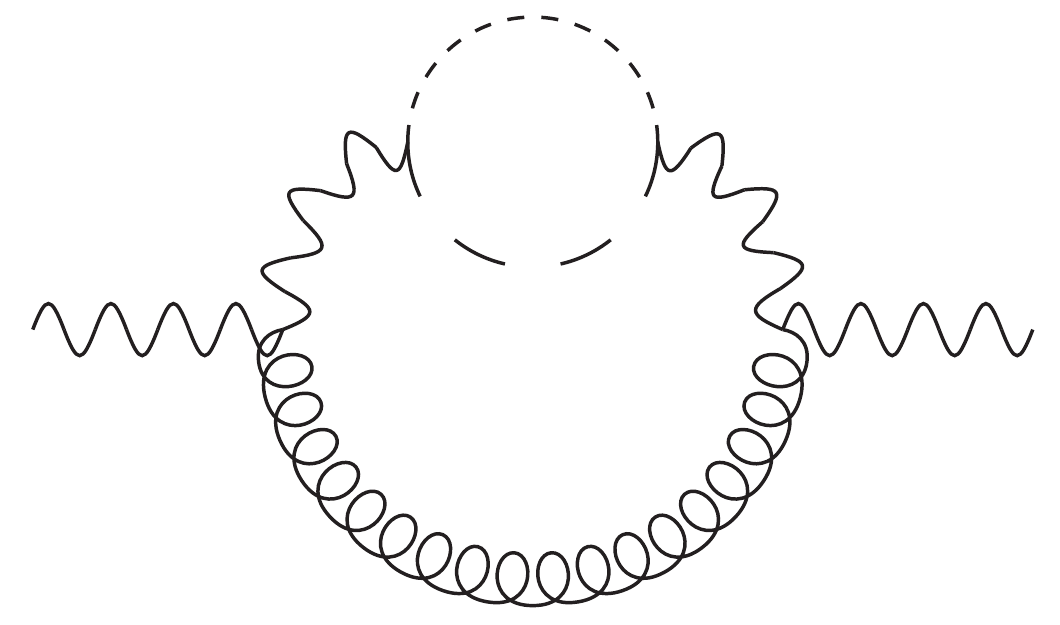} - \graph{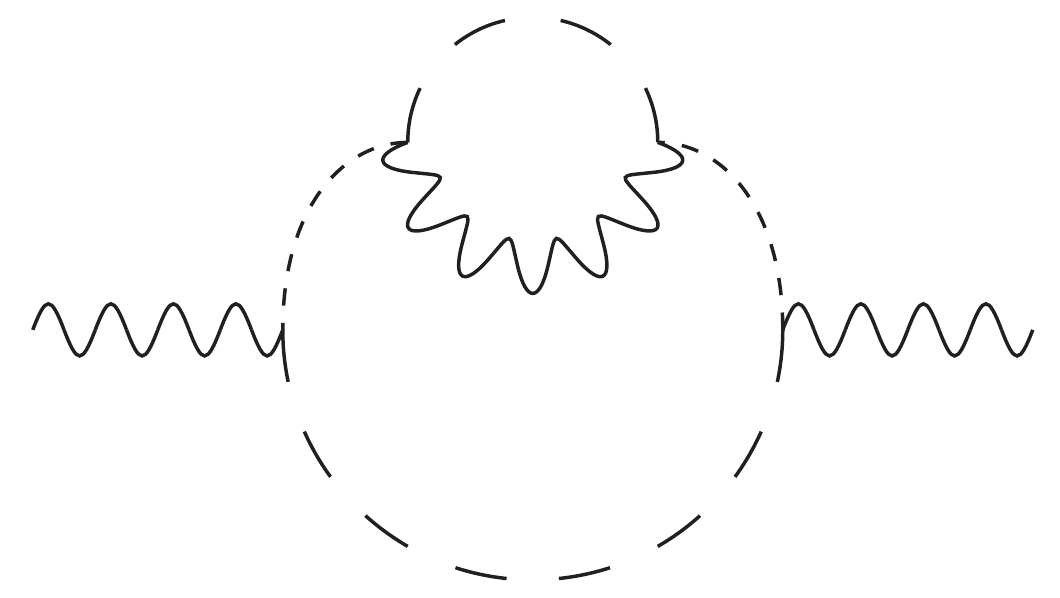} - \graph{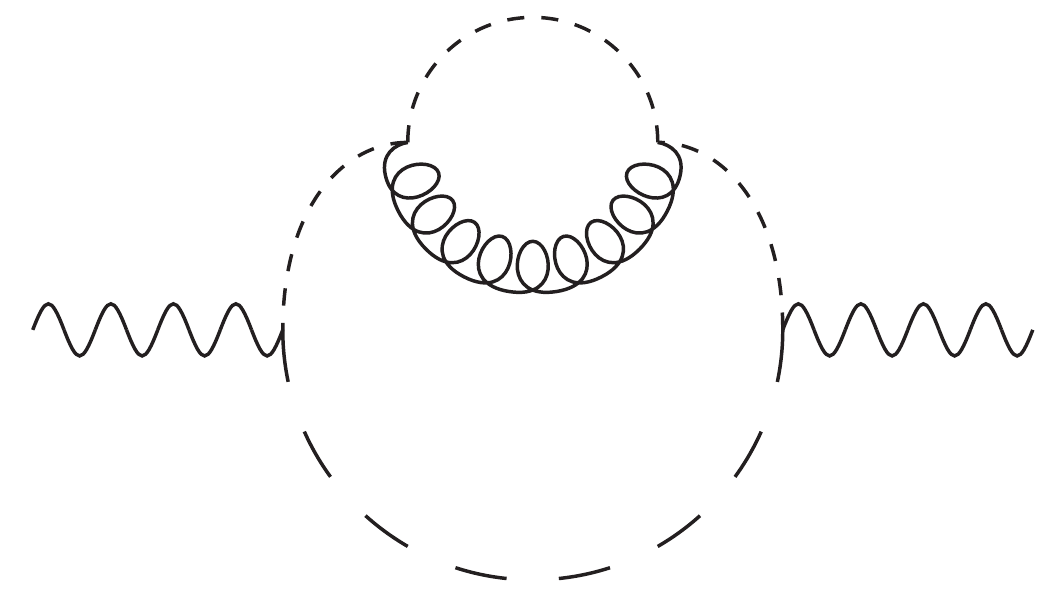} \\ & \phantom{ = } - \graph{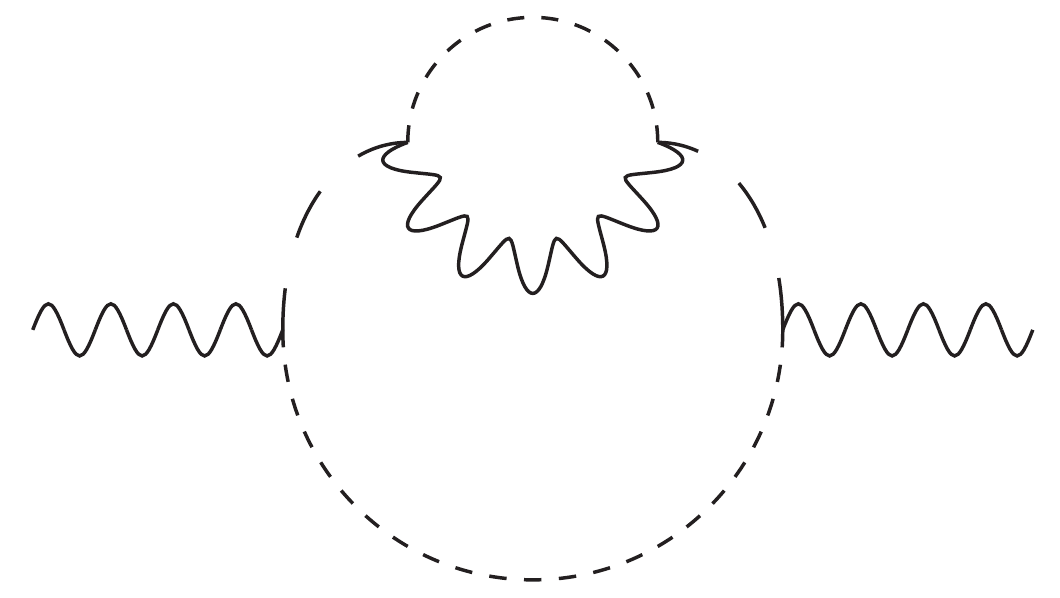} - \graph{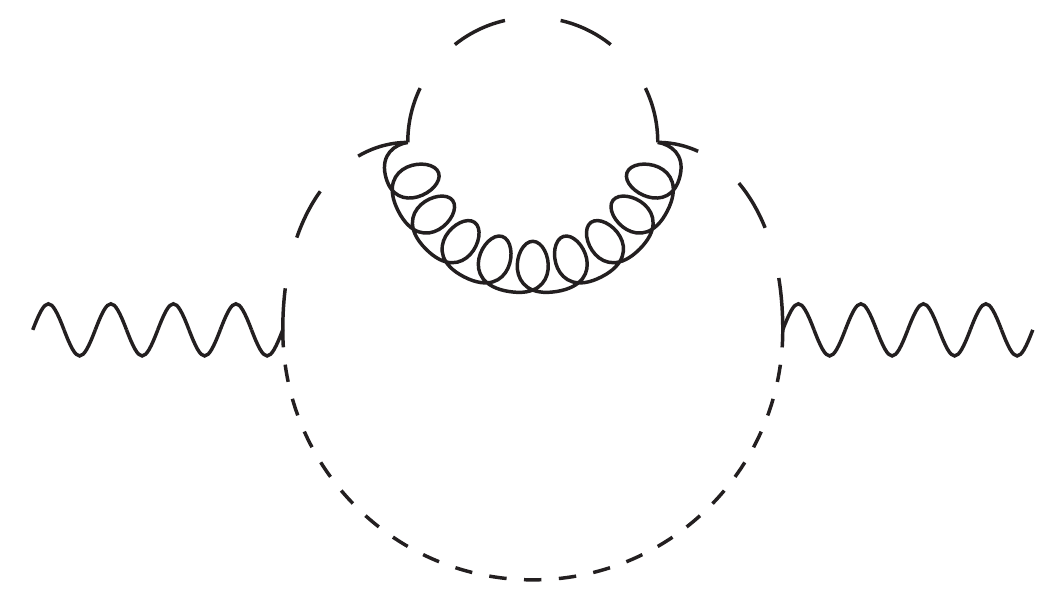} \ggraph{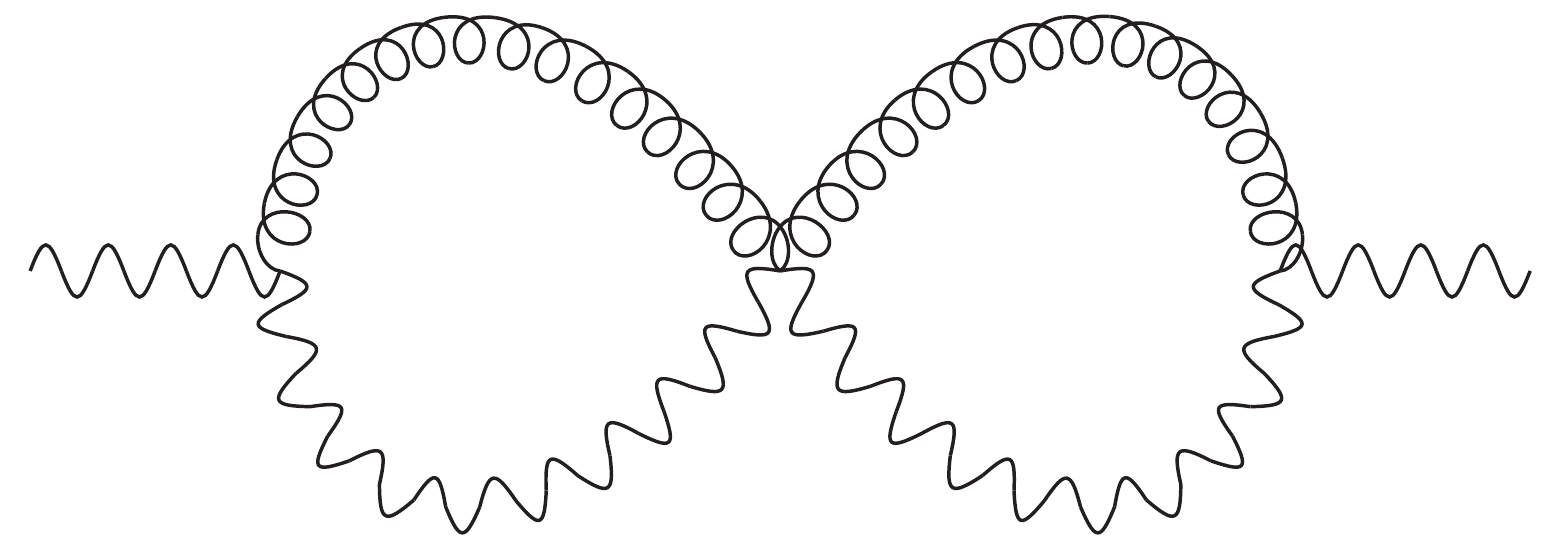} \\ & \phantom{ = } - 2 \ggraph{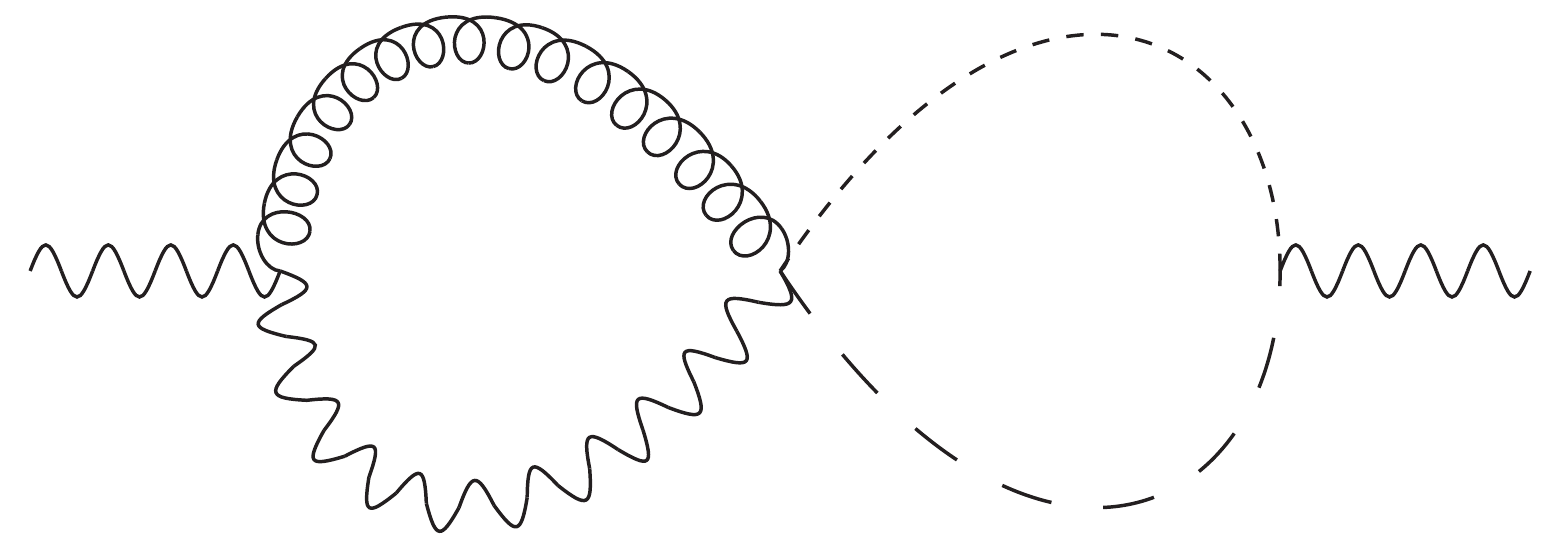} - \frac{2}{2} \graph{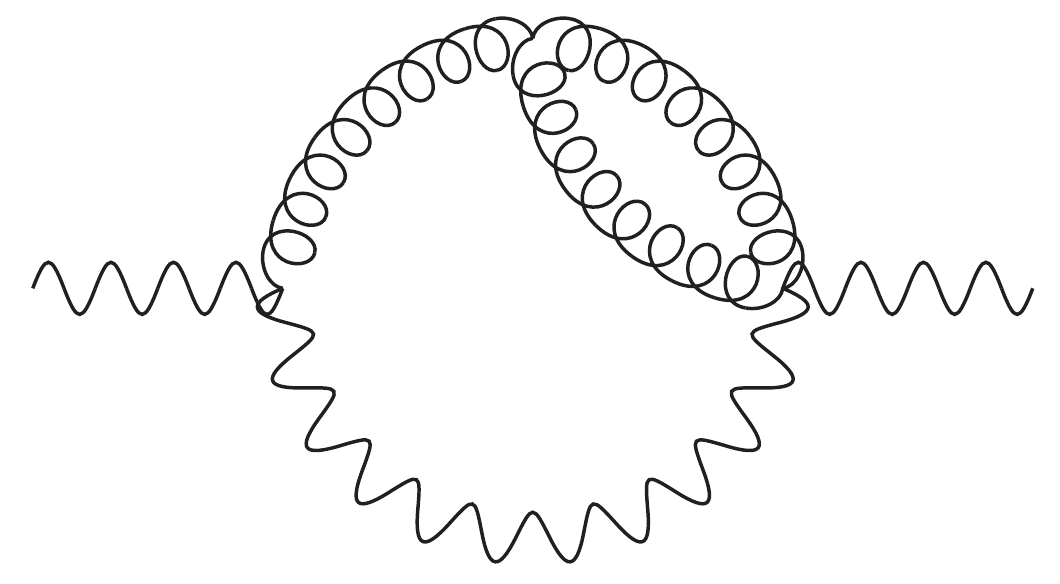} - 2 \graph{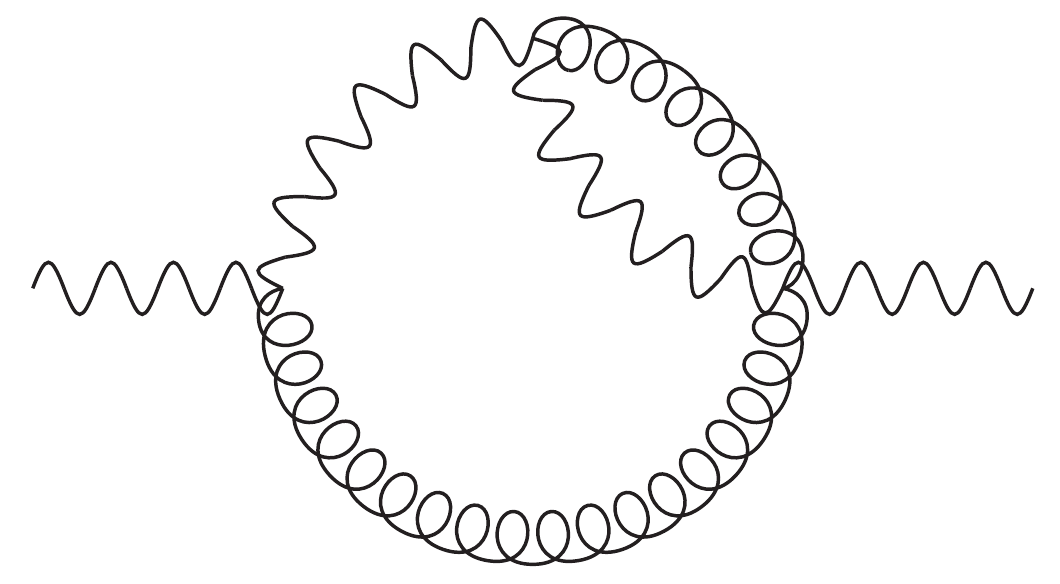} \\ & \phantom{ = } - 2 \graph{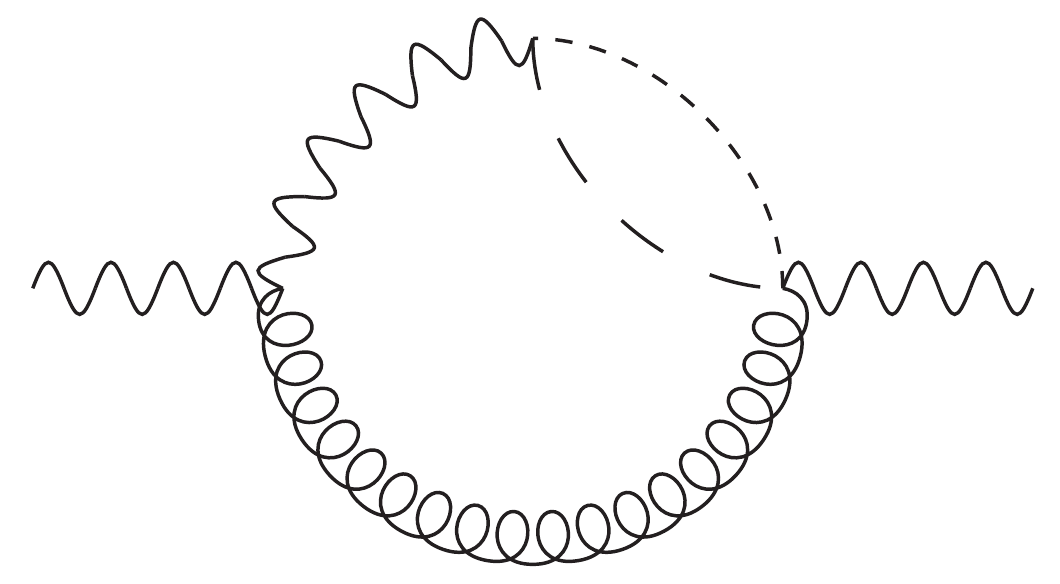} - 2 \graph{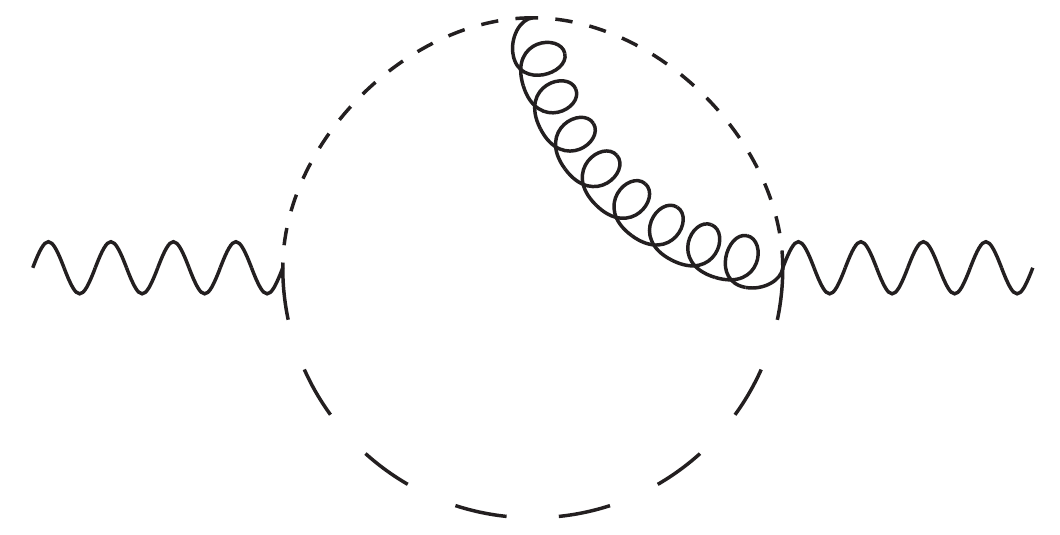} - 2 \graph{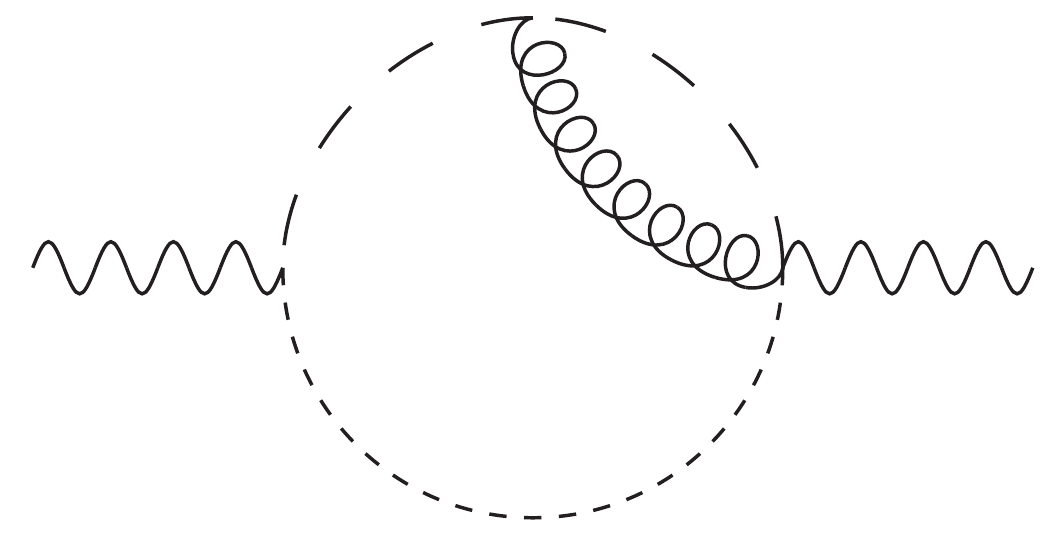}
\end{split}
\\
\rescombgreen^{\cgreen{c-gg}}_{(0,4)} & = 0
\\
\begin{split}
\rescombgreen^{\cgreen{c-gg}}_{(2,2)} & = - \frac{1}{2} \graph{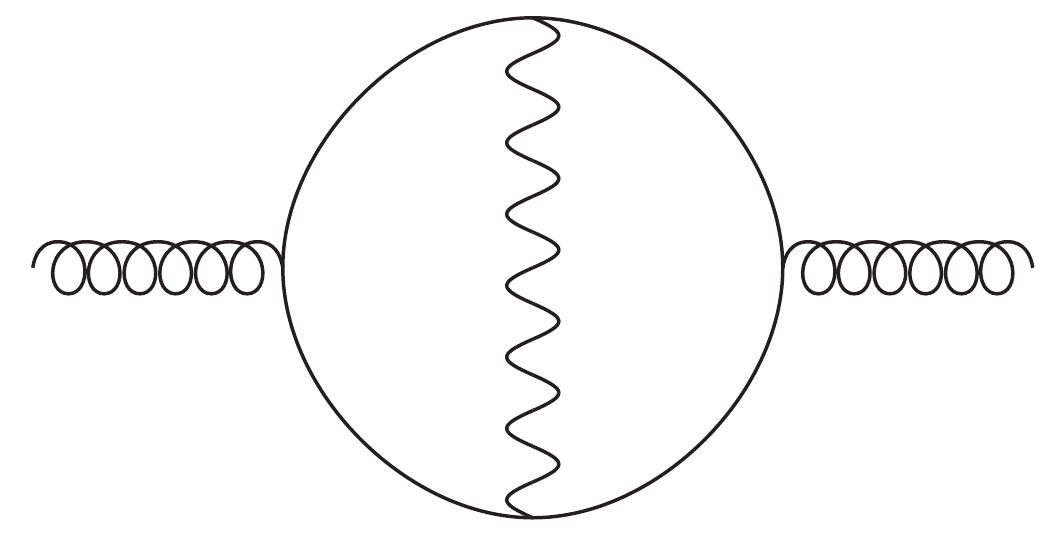} - \frac{2}{2} \graph{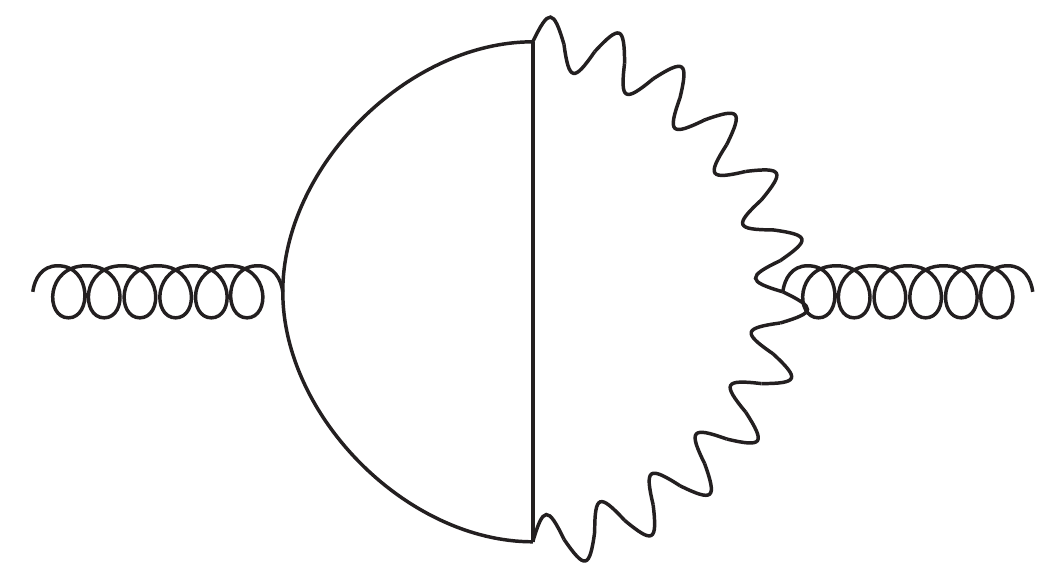} - \graph{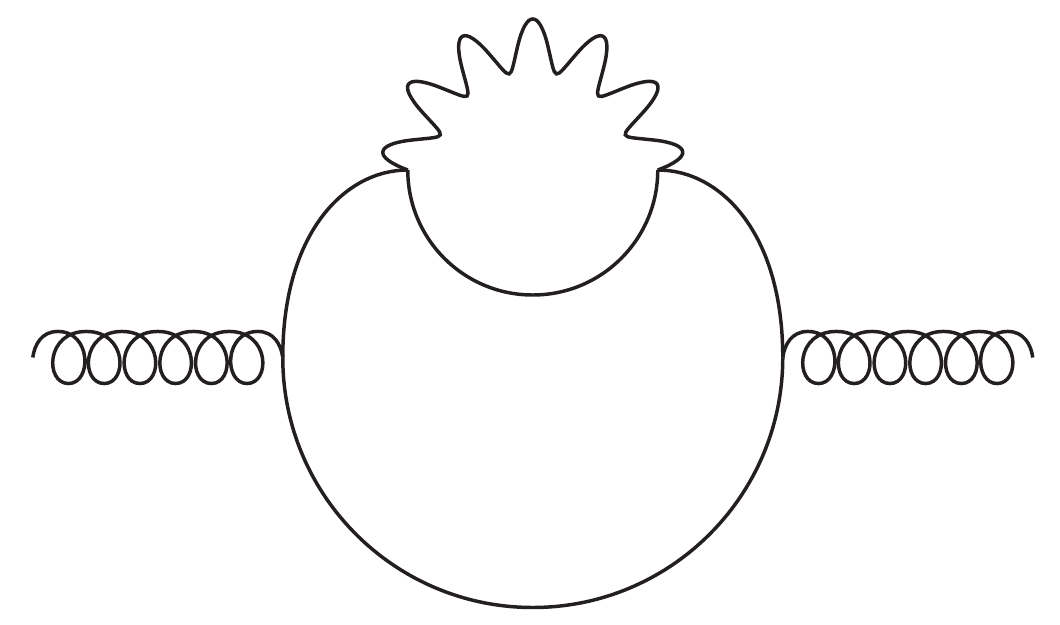} \\ & \phantom{ = } - \frac{1}{2} \graph{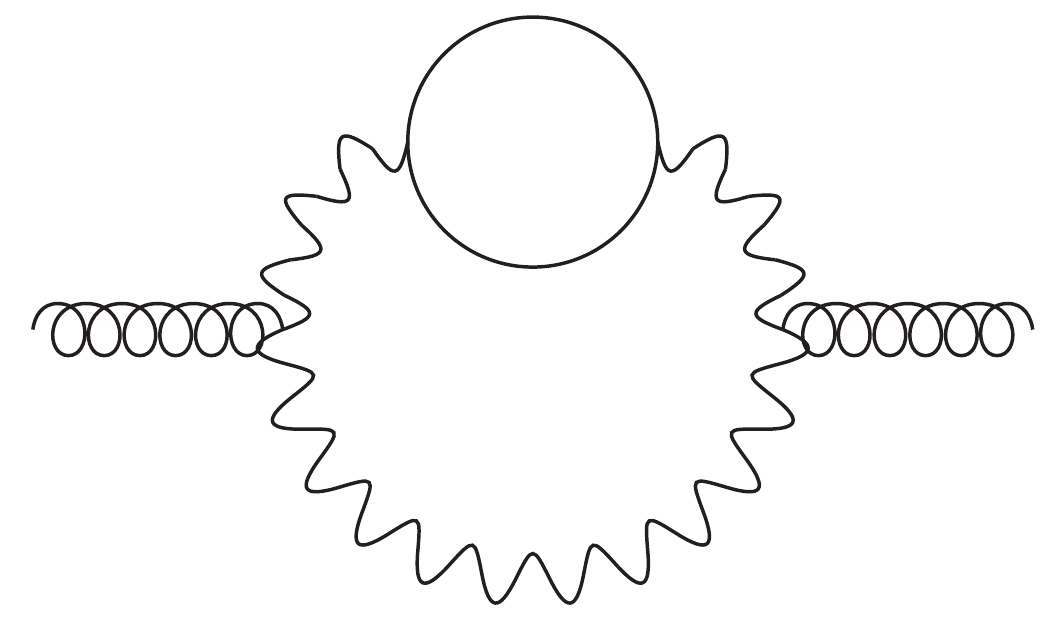} - 2 \graph{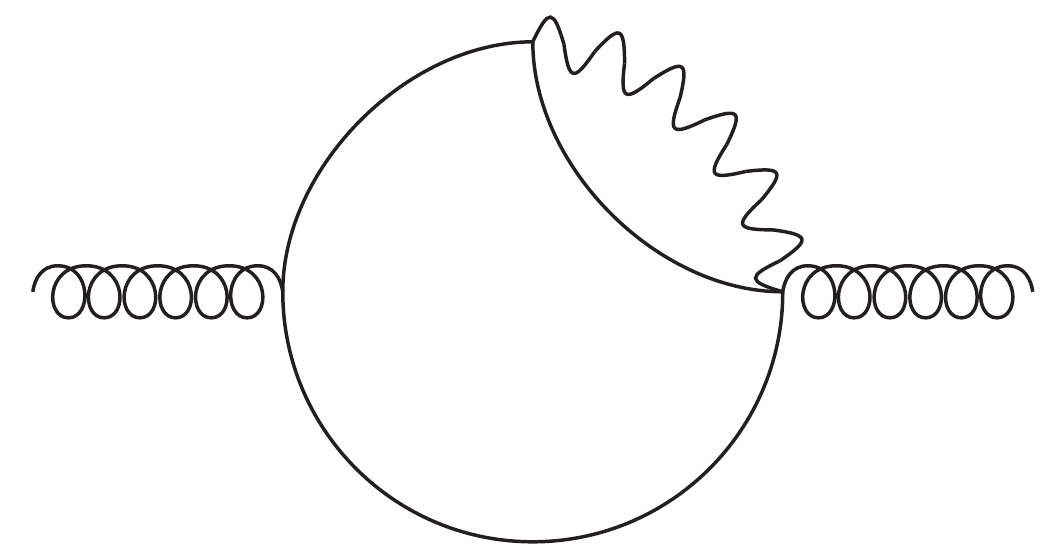} - \frac{2}{2} \graph{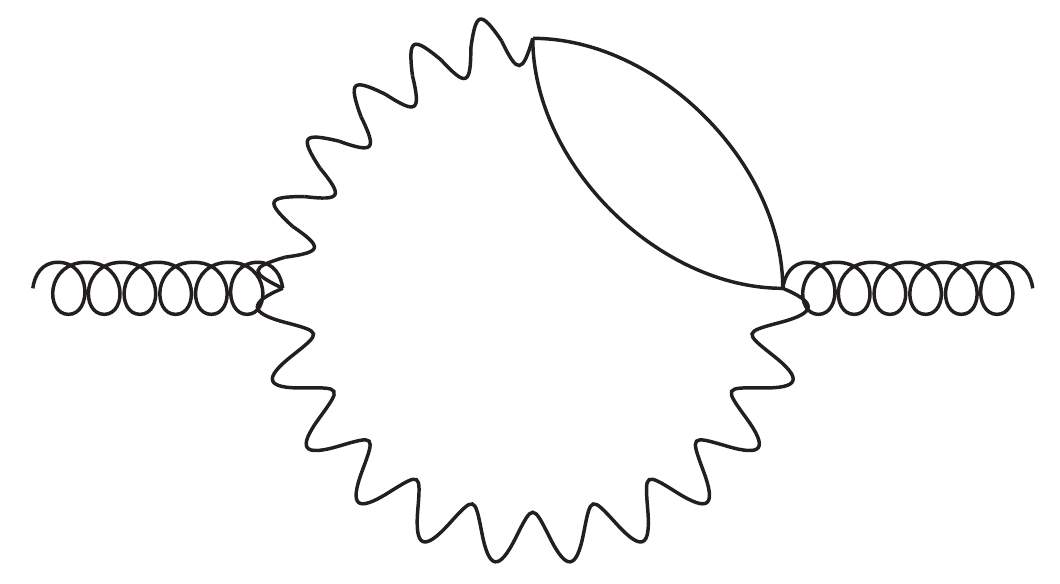}
\end{split}
\\
\begin{split}
\rescombgreen^{\cgreen{c-gg}}_{(4,0)} & = - \frac{1}{2} \graph{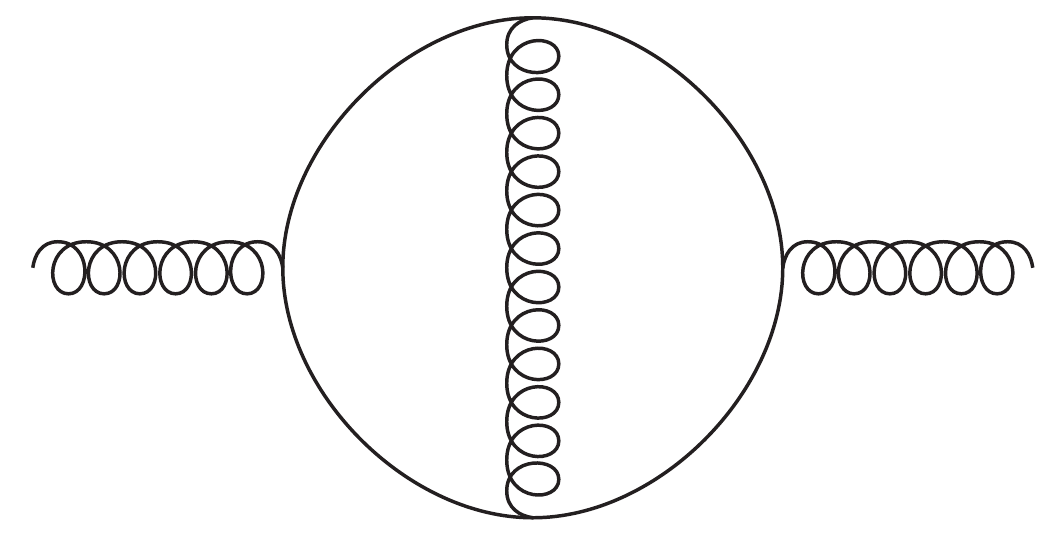} - \frac{2}{2} \graph{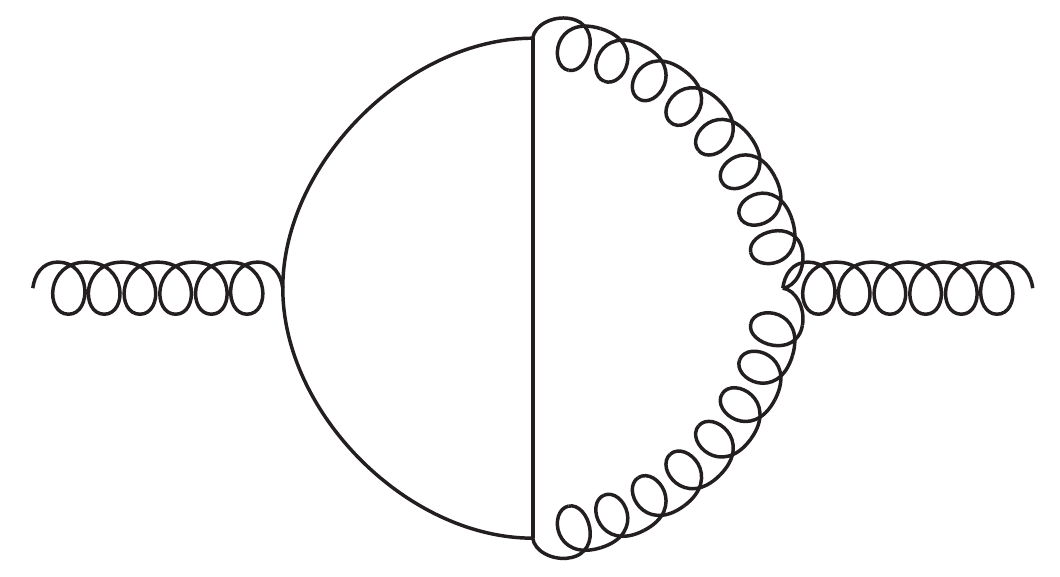} - \frac{1}{2} \graph{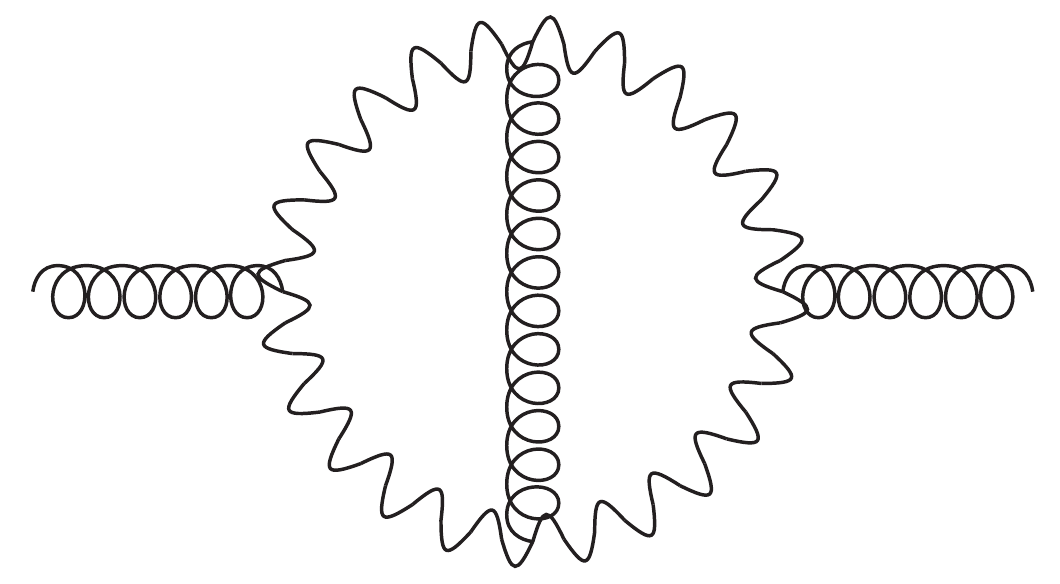} \\ & \phantom{ = } - \frac{2}{2} \graph{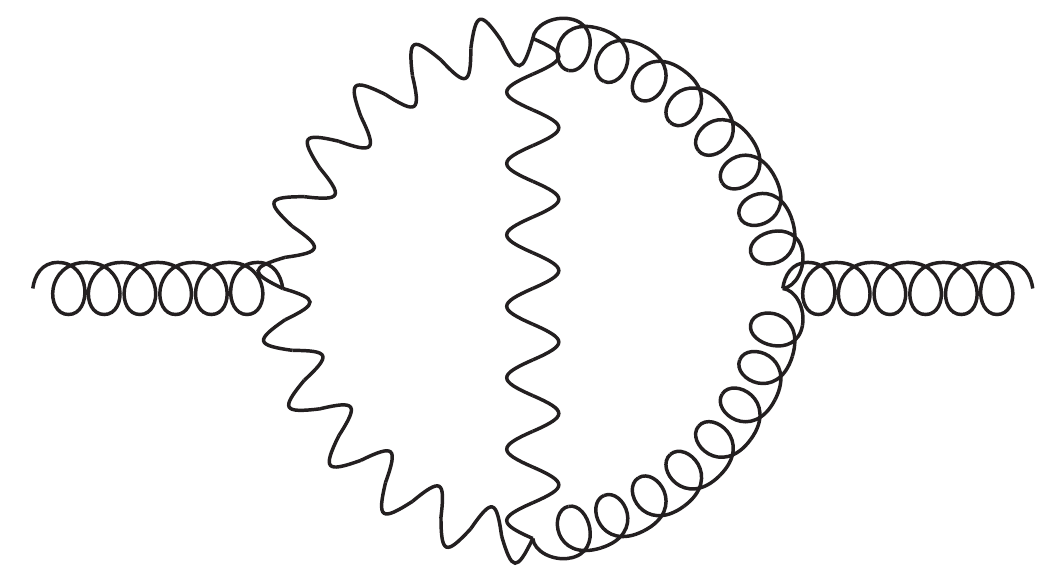} - \frac{2}{2} \graph{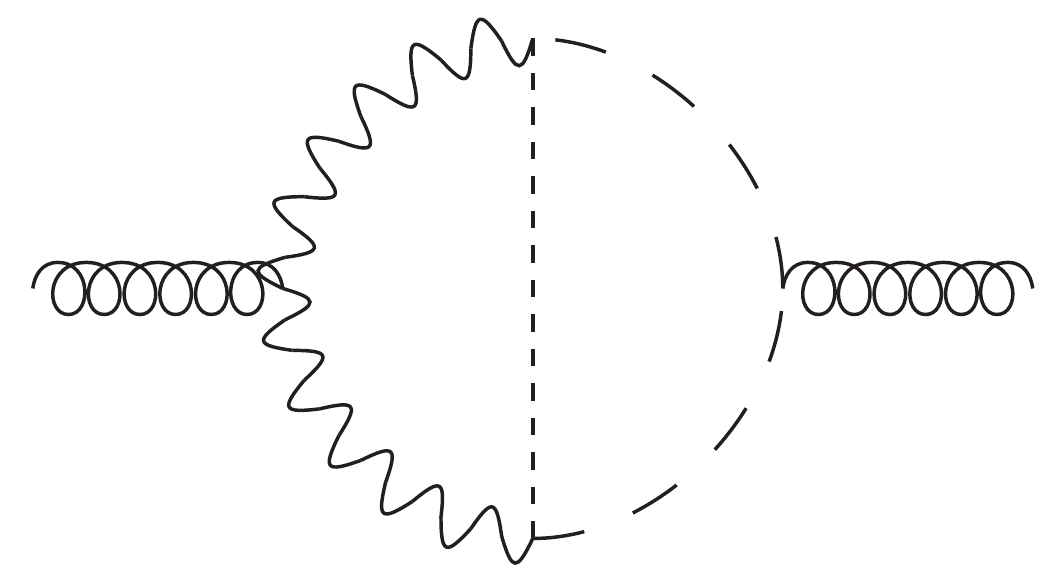} - \frac{2}{2} \graph{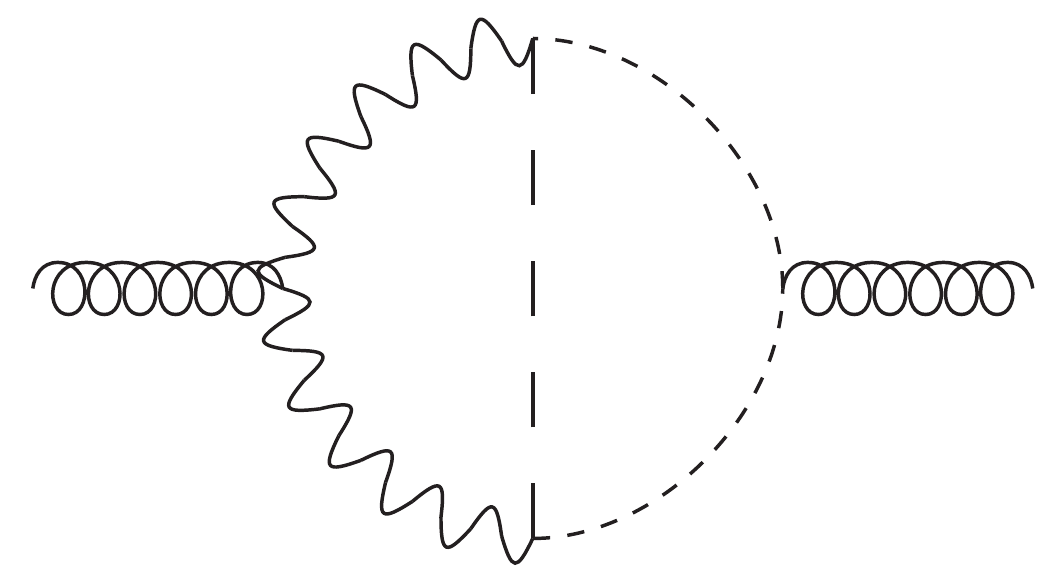} \\ & \phantom{ = } - \frac{1}{2} \graph{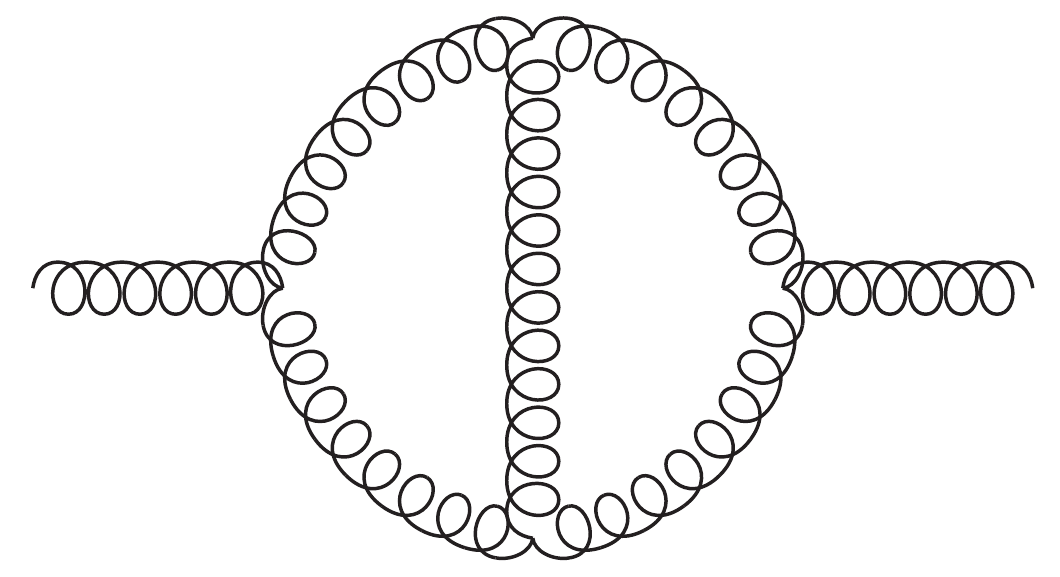} - \frac{2}{2} \graph{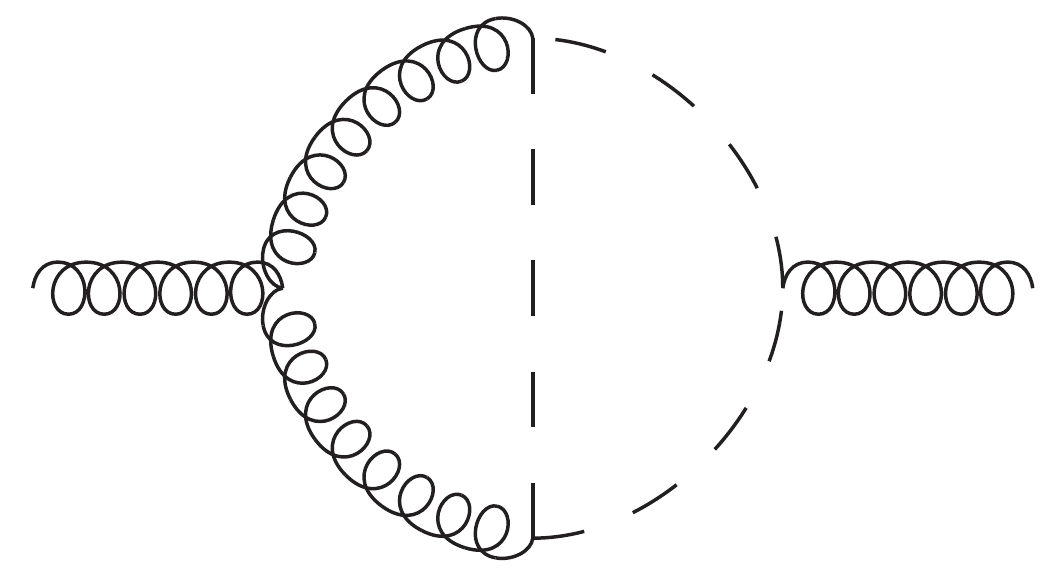} - \frac{2}{2} \graph{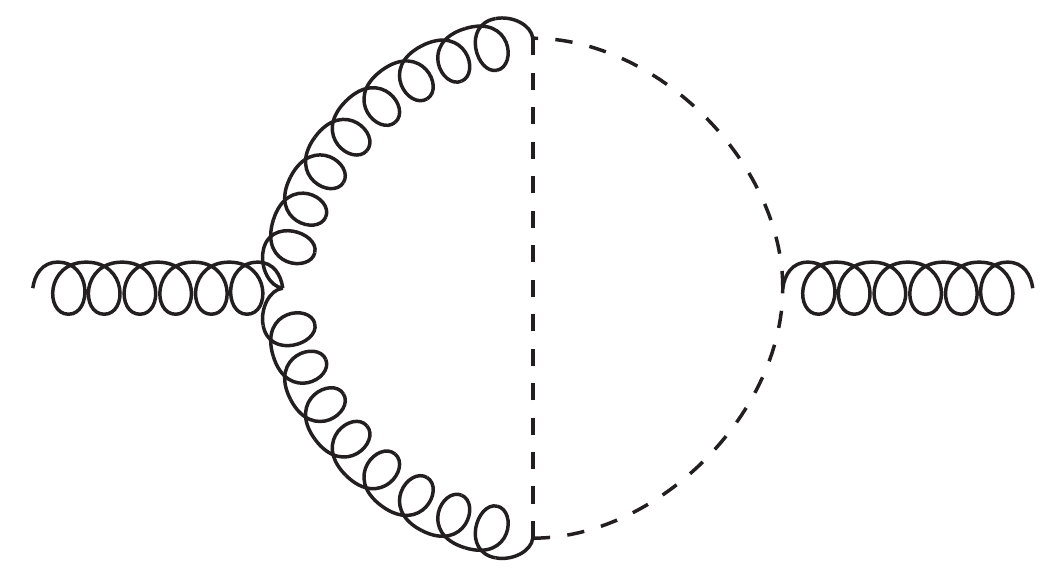} \\ & \phantom{ = } - \frac{1}{2} \graph{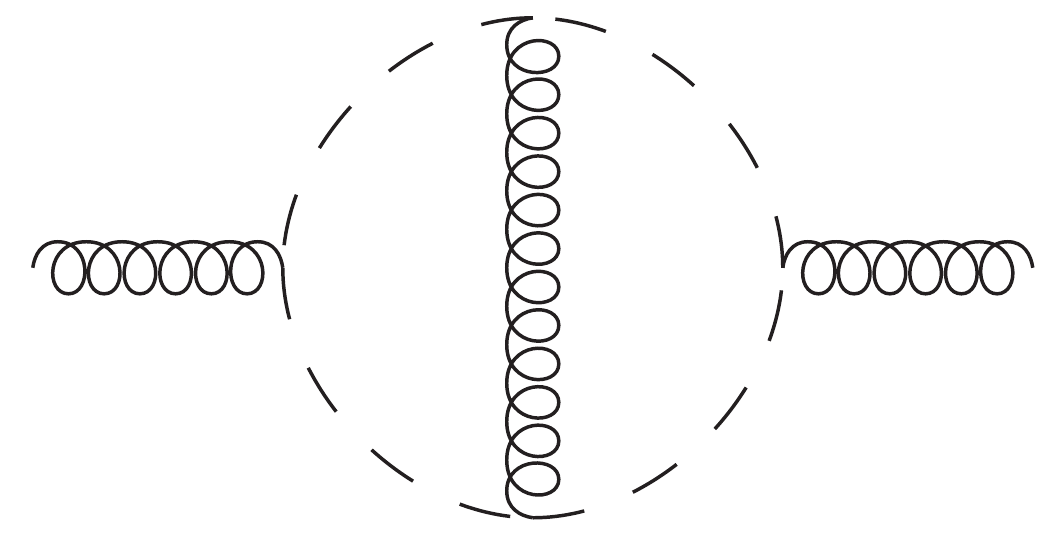} - \frac{2}{2} \graph{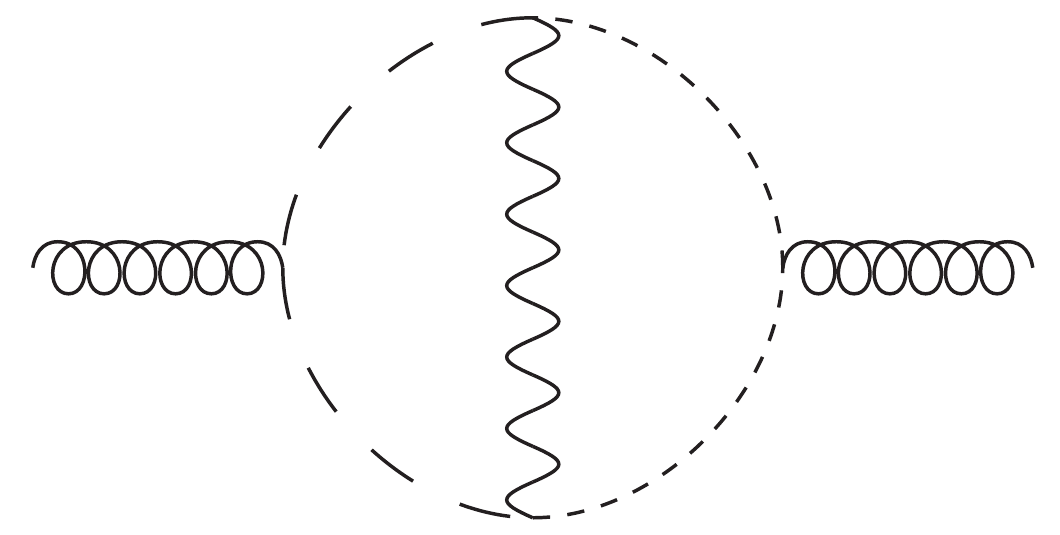} - \frac{1}{2} \graph{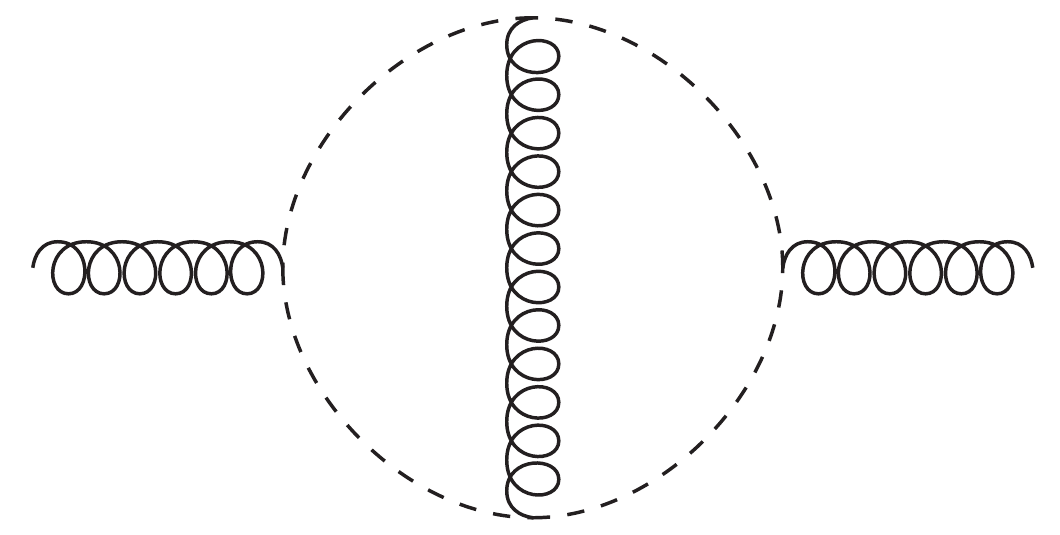} \\ & \phantom{ = } - \graph{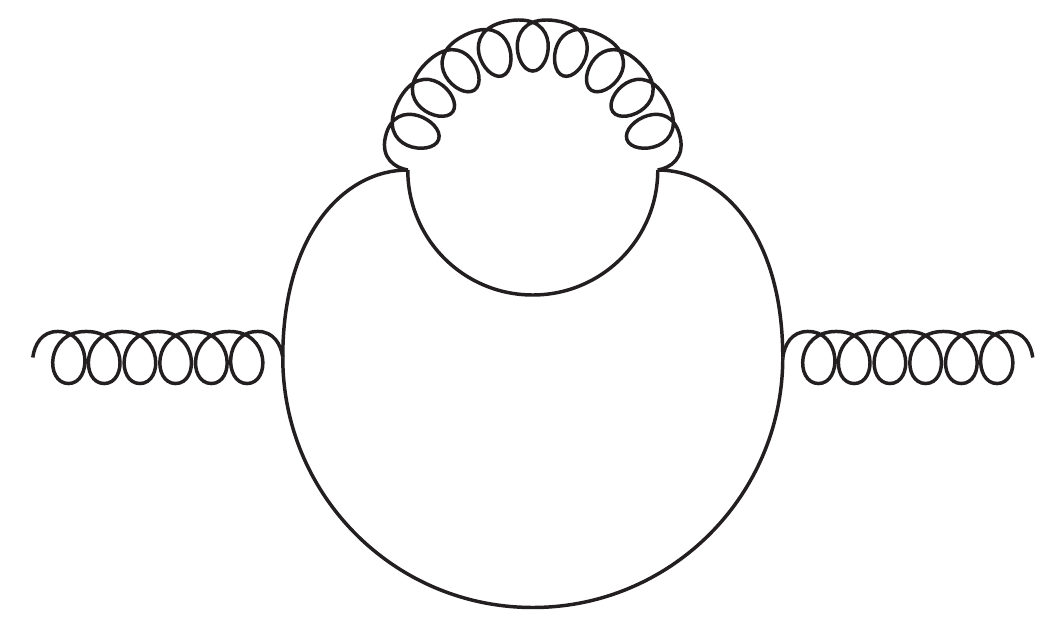} - \graph{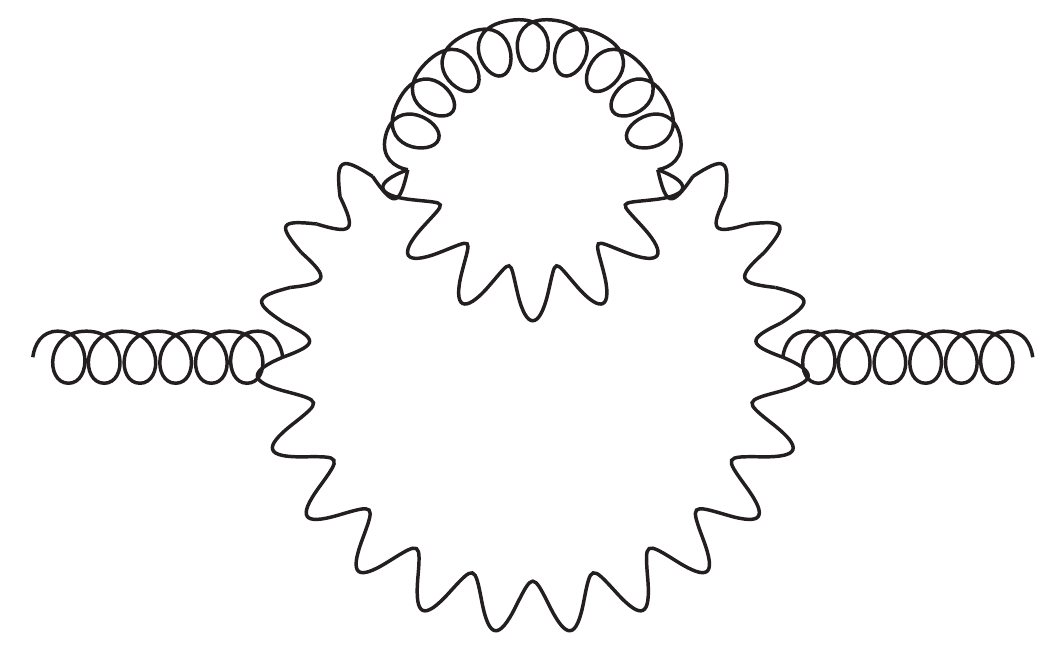} - \graph{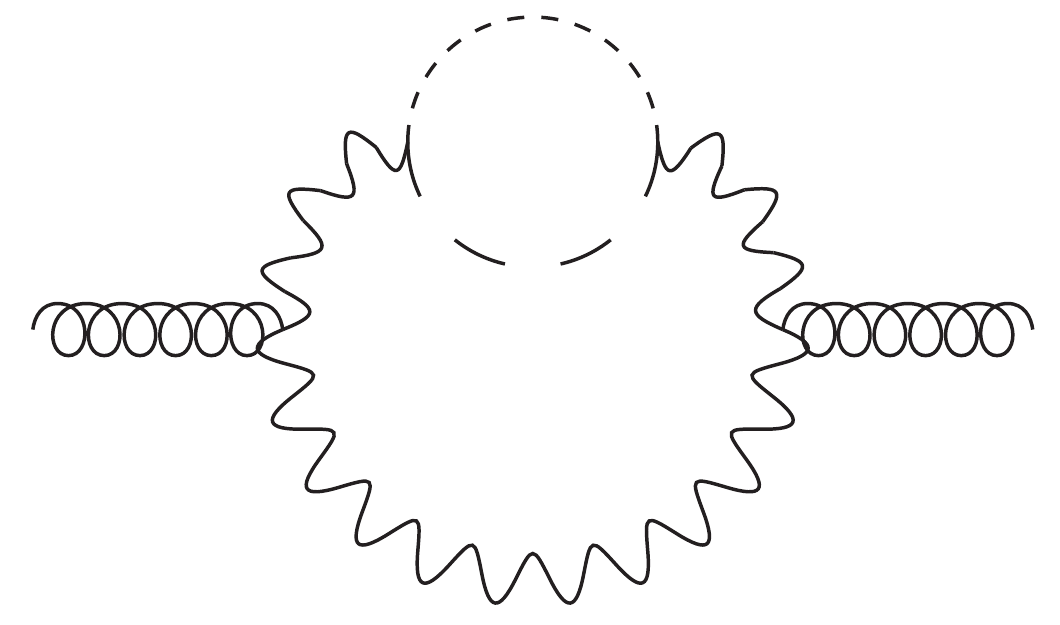} \\ & \phantom{ = } - \frac{1}{2} \graph{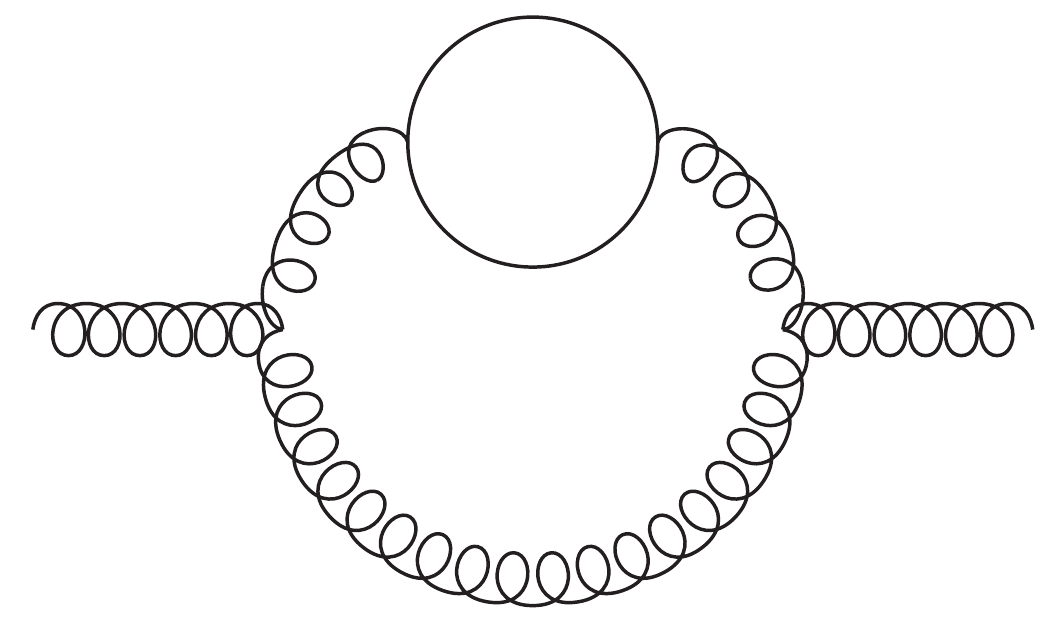} - \frac{1}{2} \graph{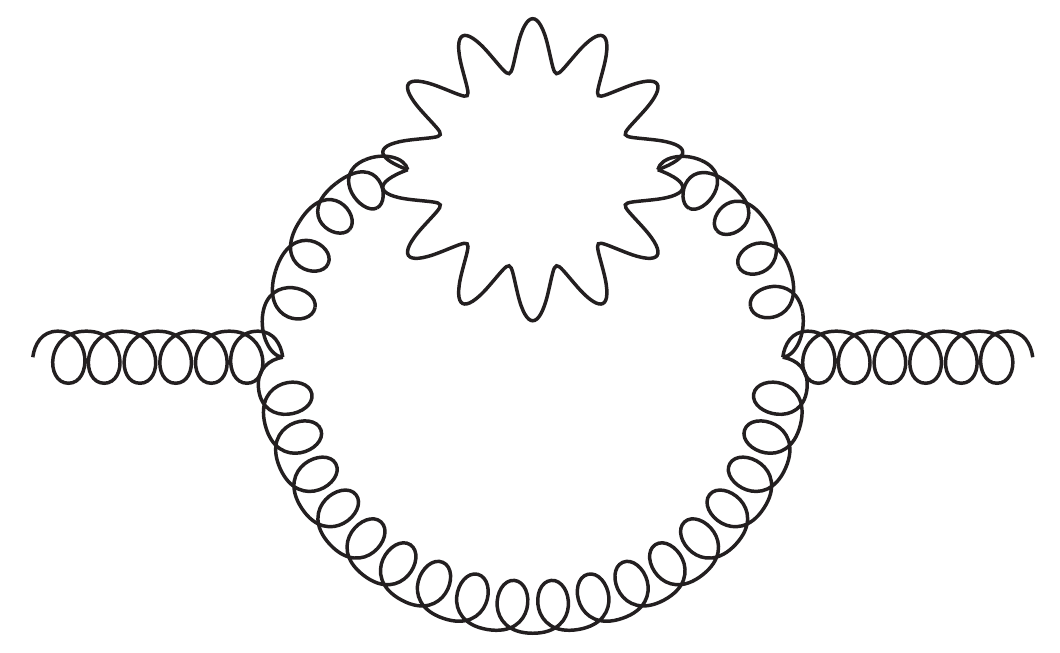} - \frac{1}{2} \graph{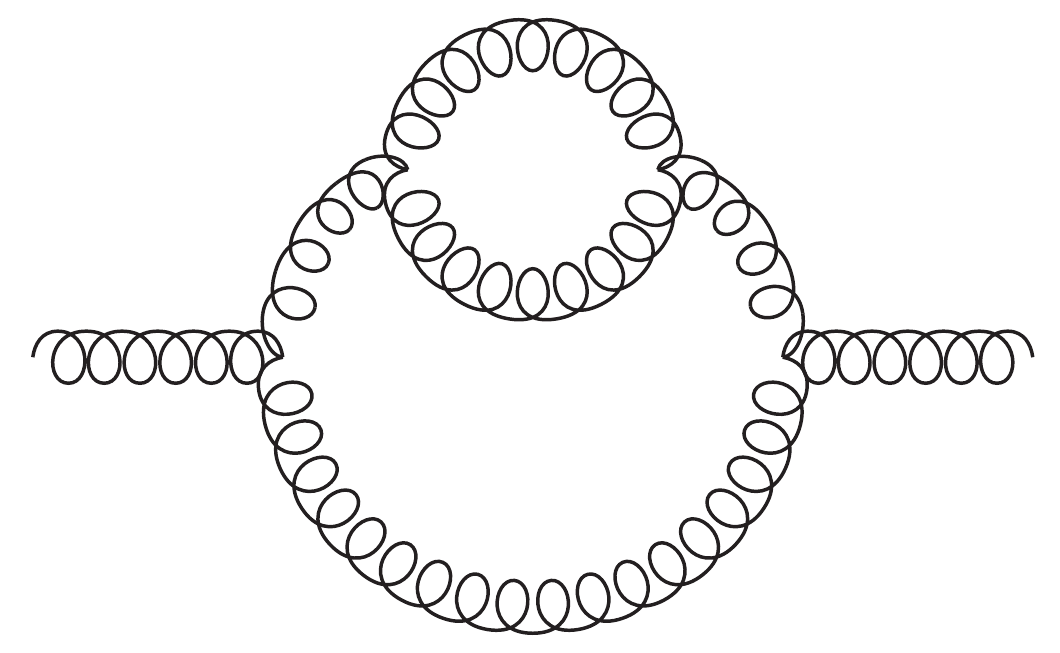} \\ & \phantom{ = } - \frac{1}{2} \graph{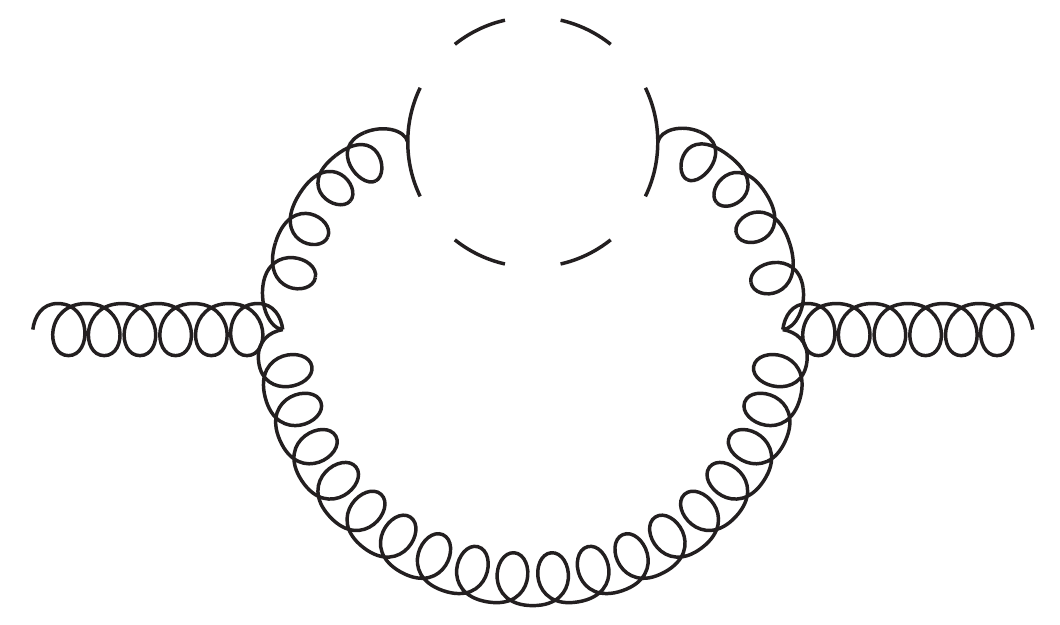} - \frac{1}{2} \graph{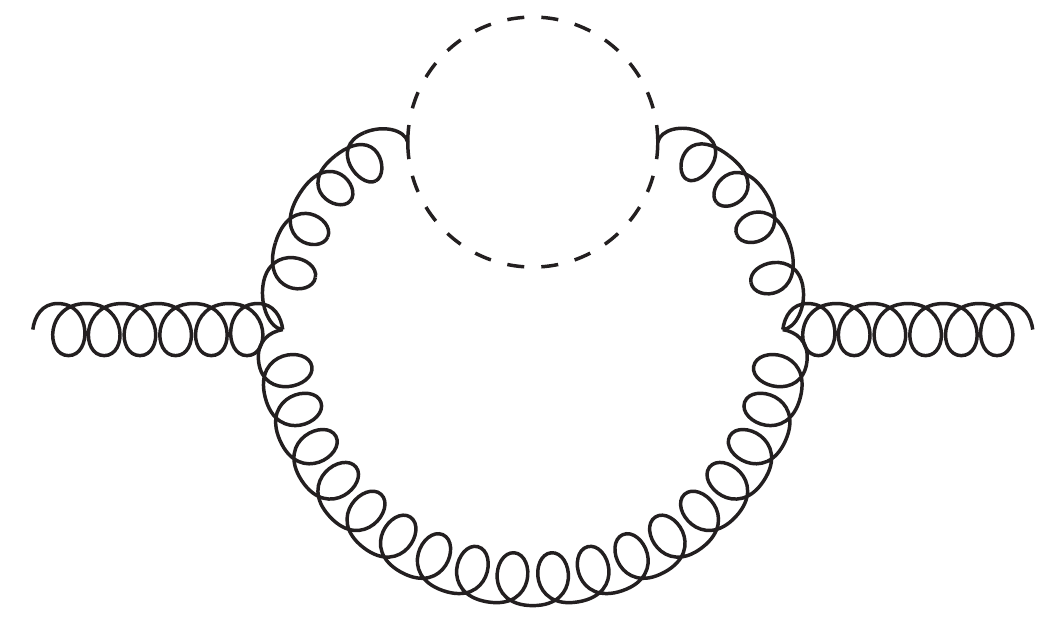} - \graph{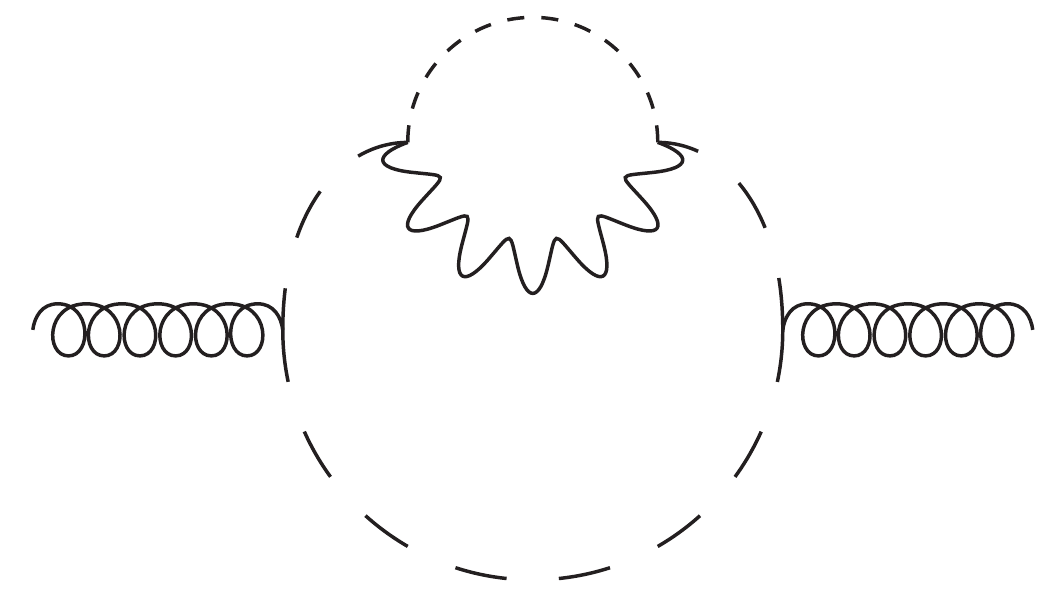} \\ & \phantom{ = } - \graph{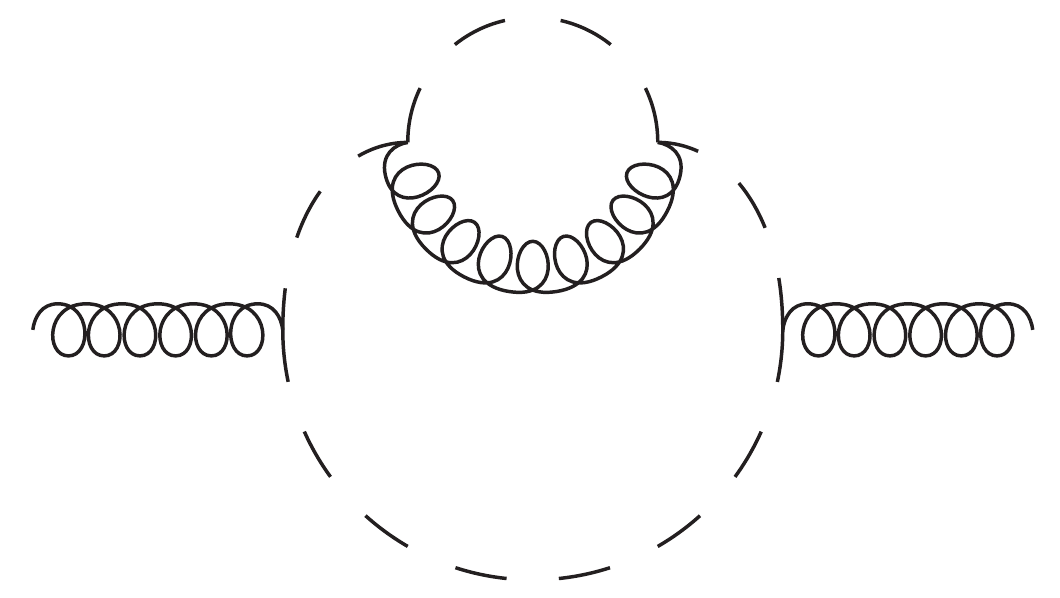} - \graph{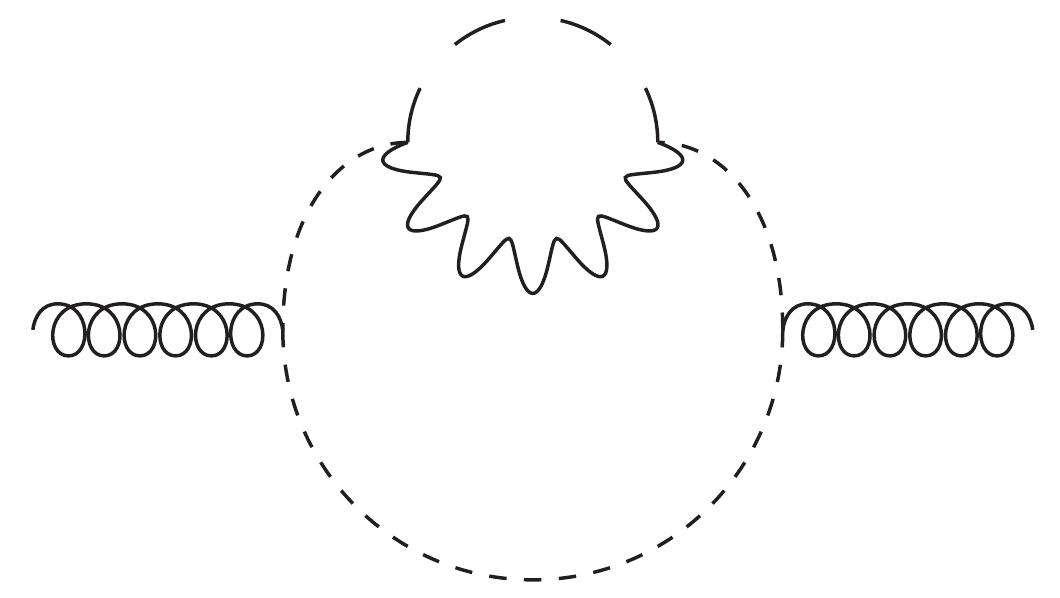} - \graph{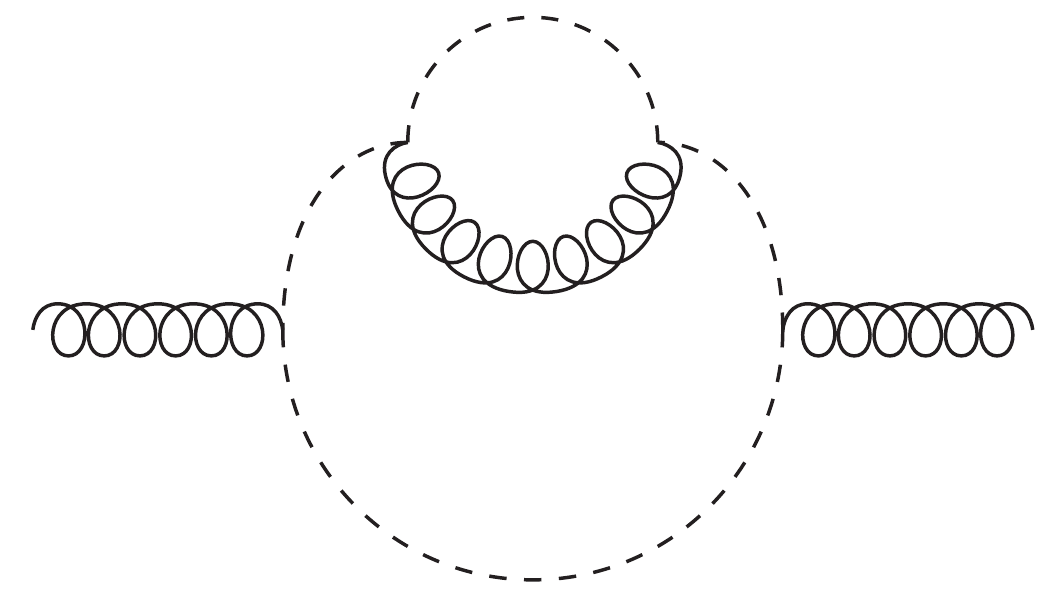} \\ & \phantom{ = } - \frac{2}{4} \ggraph{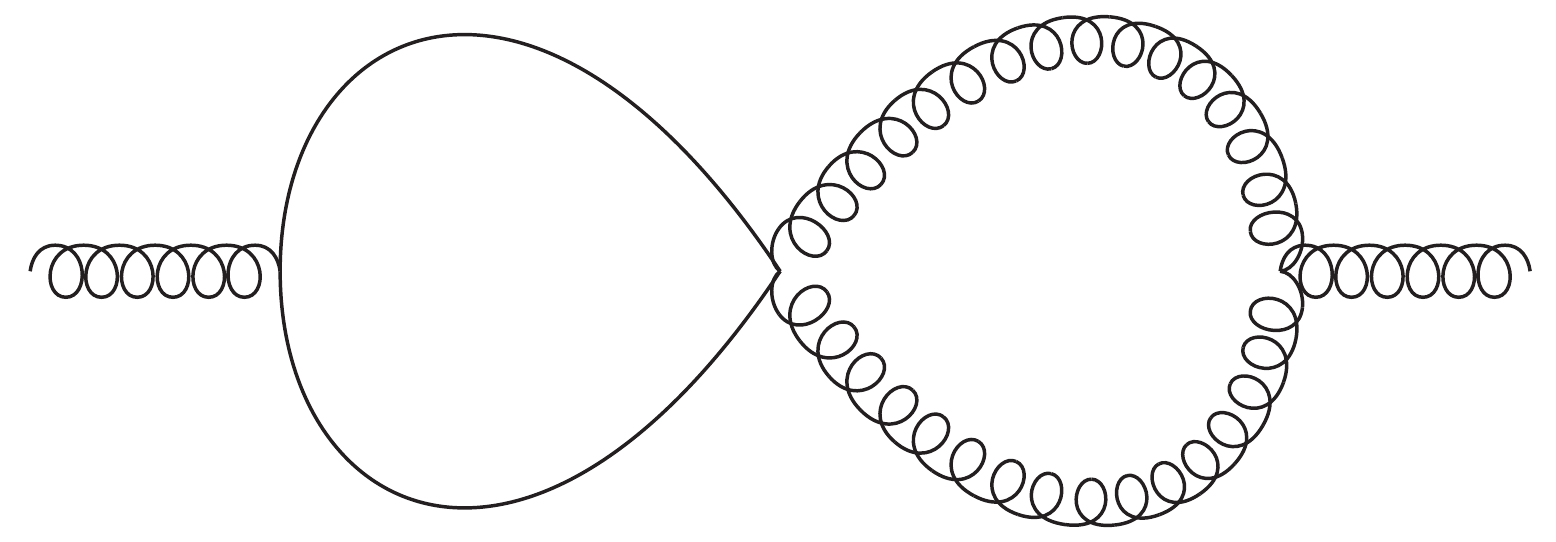} - \frac{2}{4} \ggraph{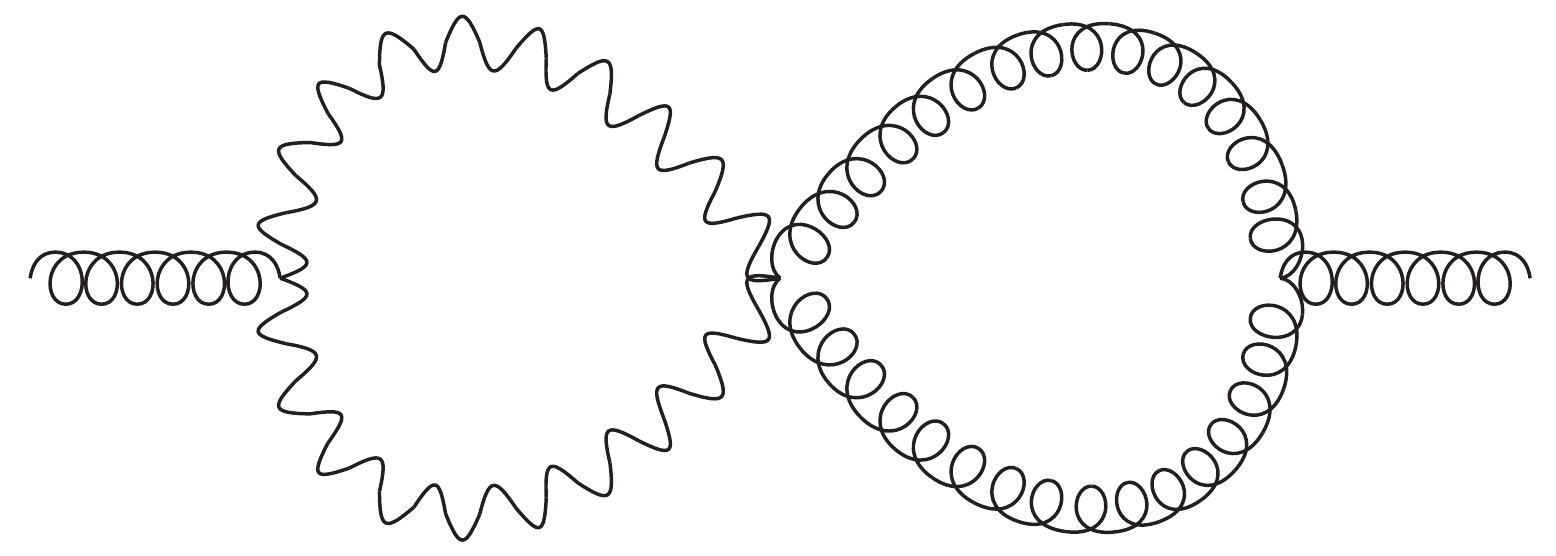} \\ & \phantom{ = } - \frac{1}{4} \ggraph{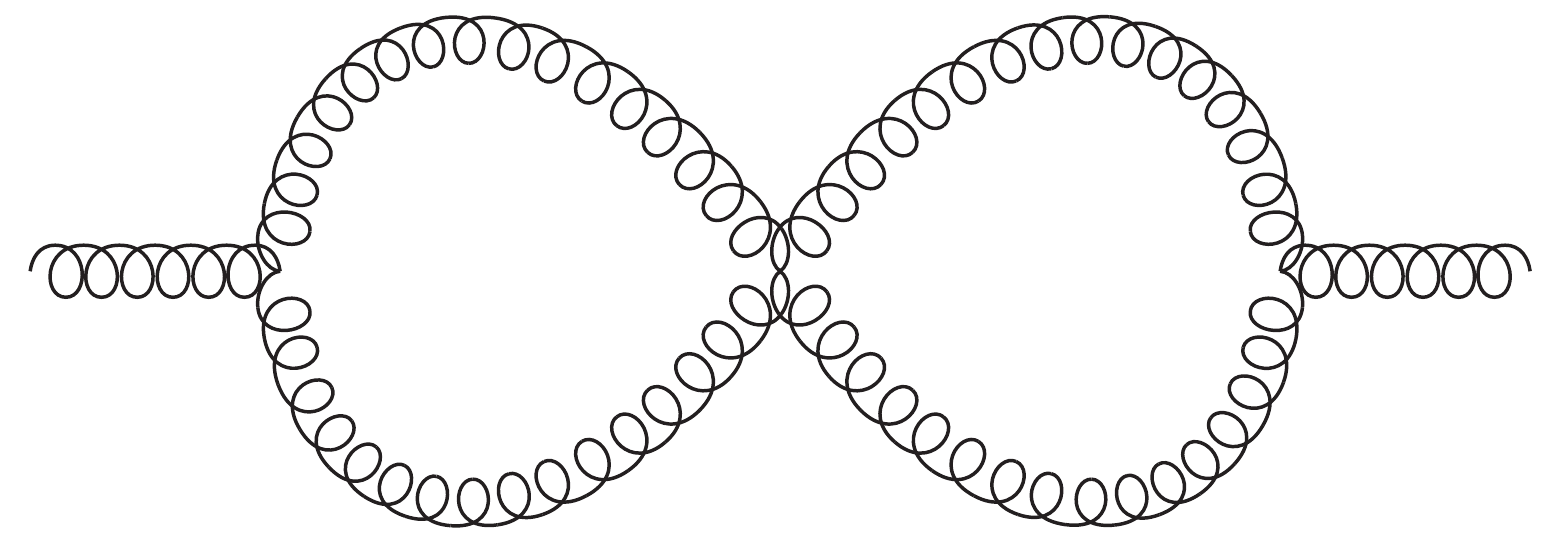} - \frac{2}{4} \ggraph{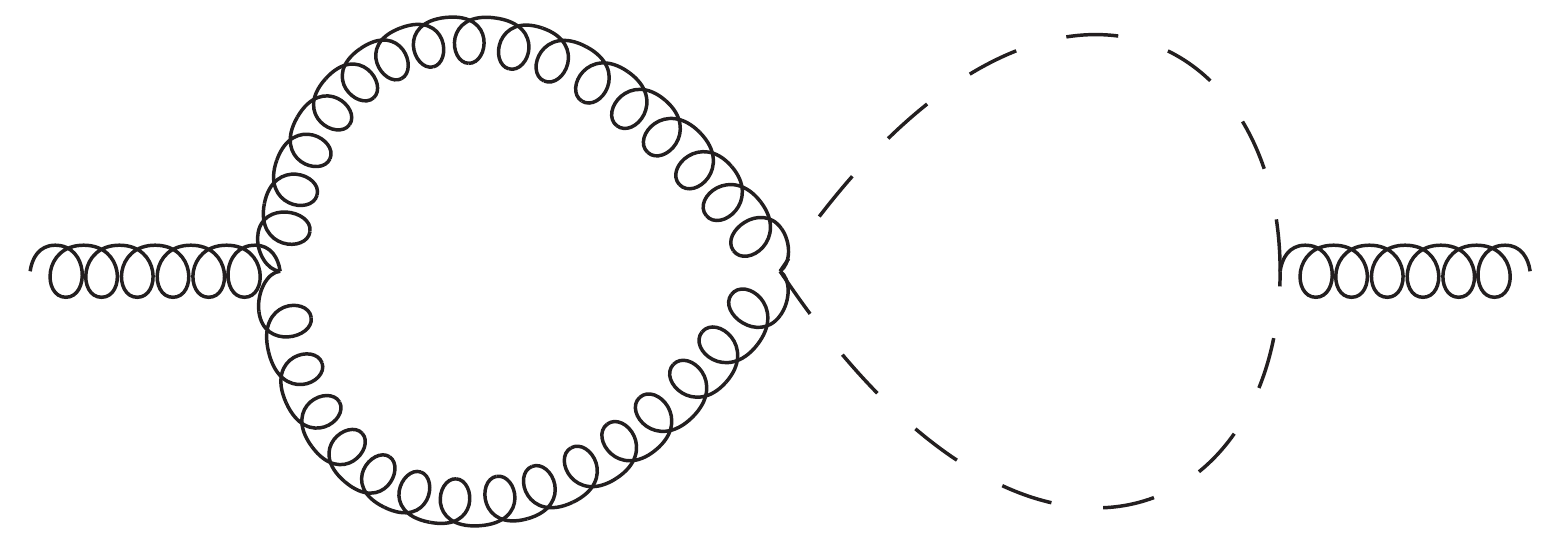} \\ & \phantom{ = } - \frac{2}{4} \ggraph{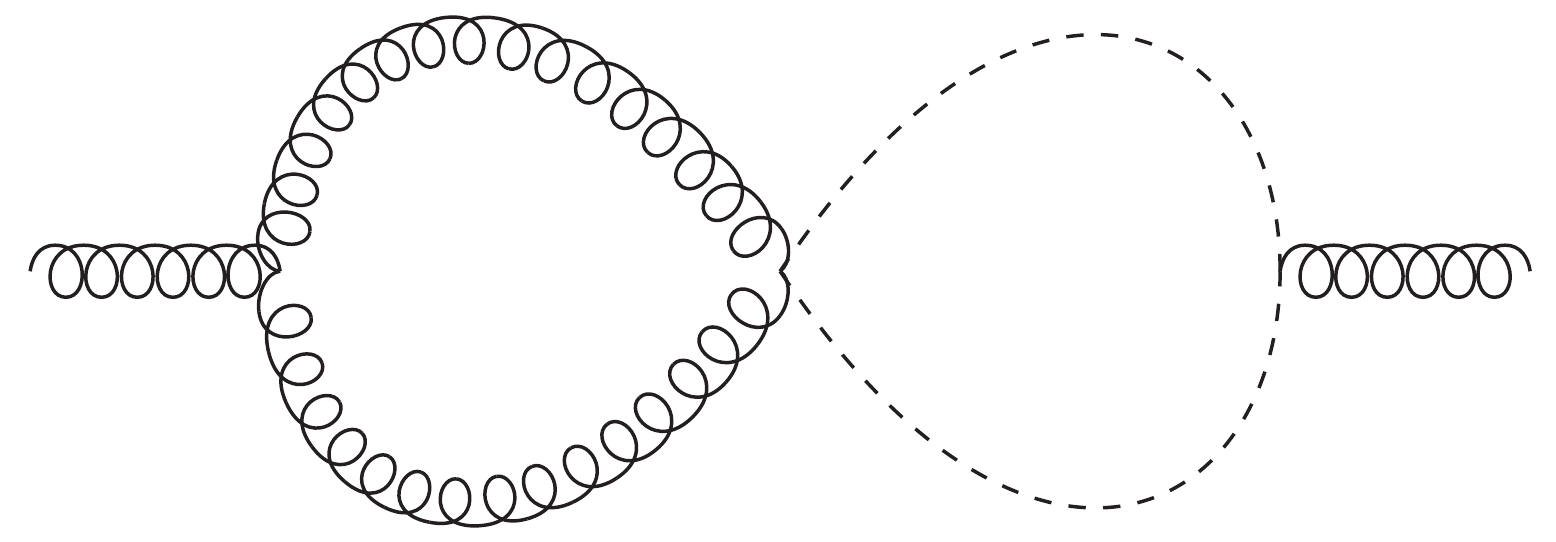} - 2 \graph{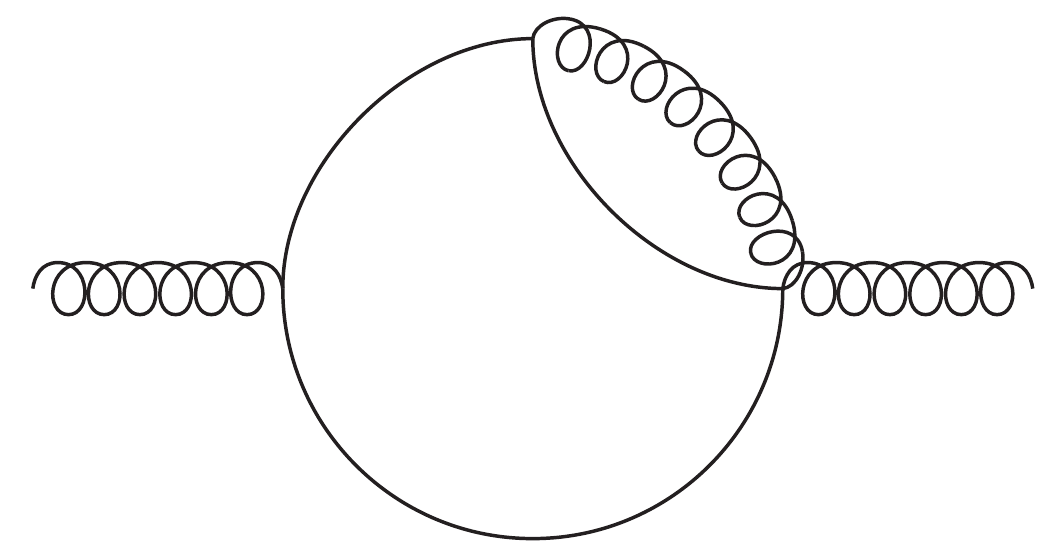} - 2 \graph{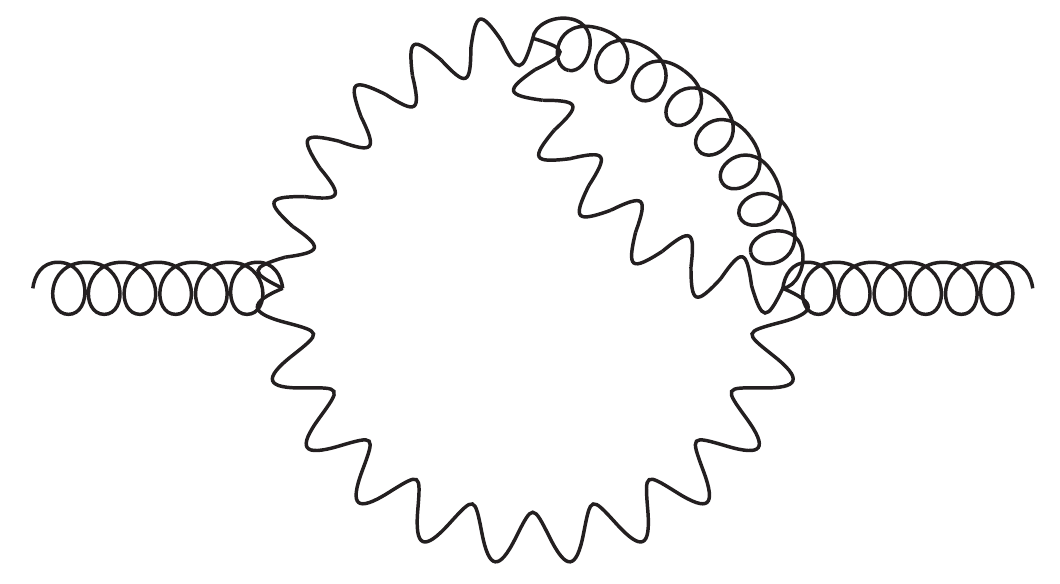} \\ & \phantom{ = } - 2 \graph{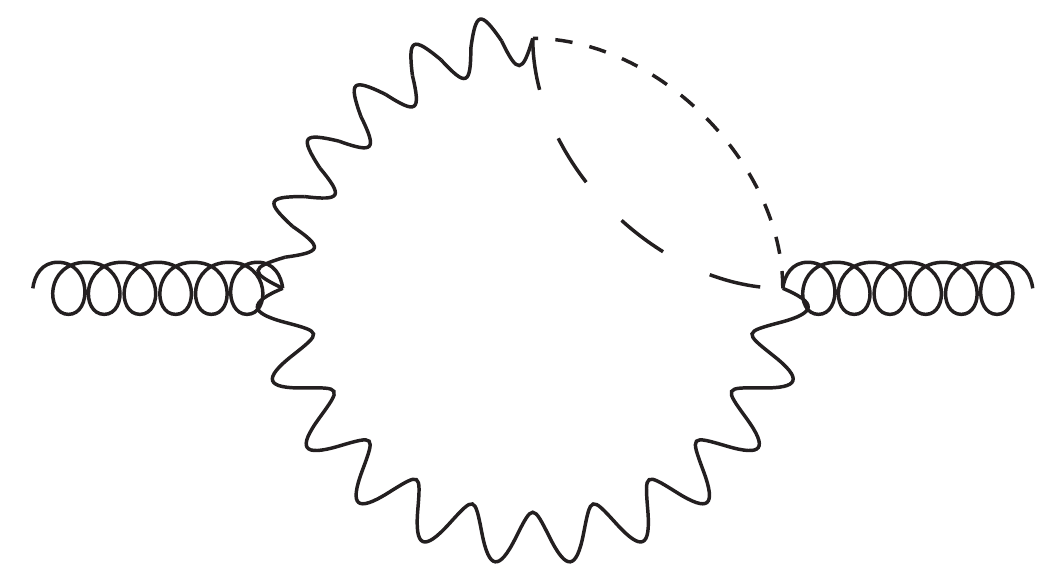} - \frac{2}{2} \graph{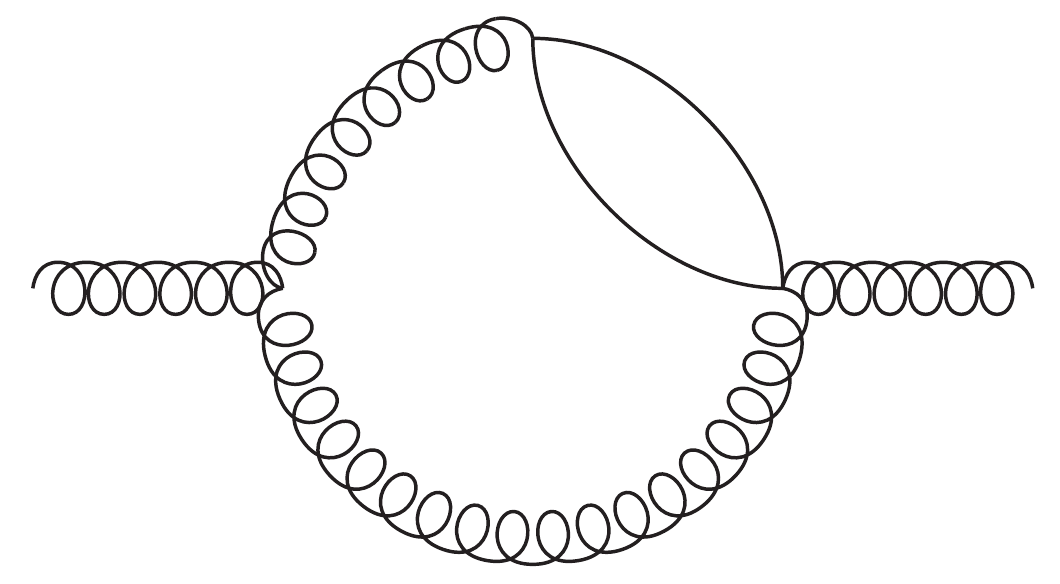} - \frac{2}{2} \graph{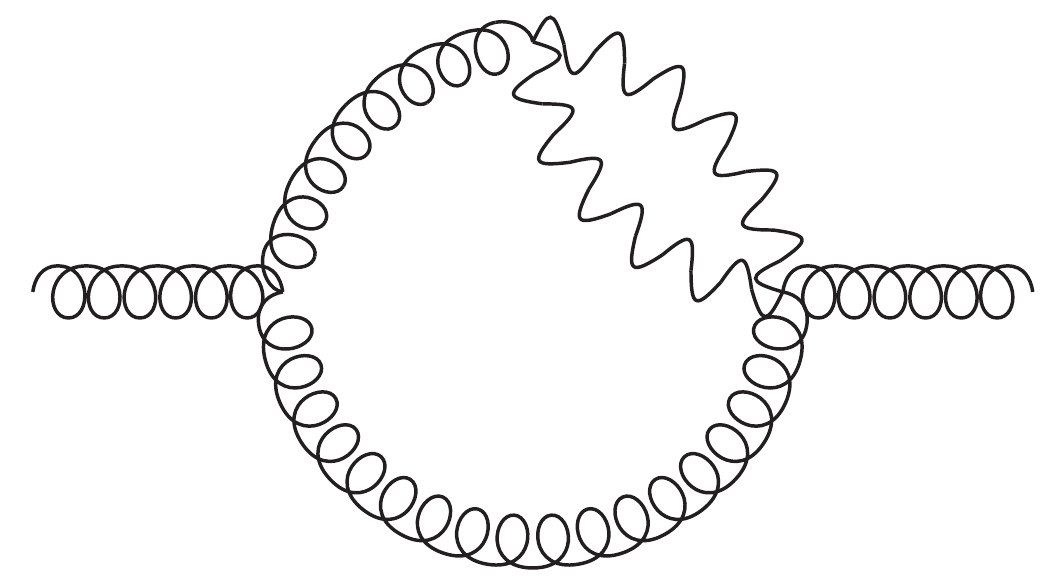} \\ & \phantom{ = } - \frac{2}{2} \graph{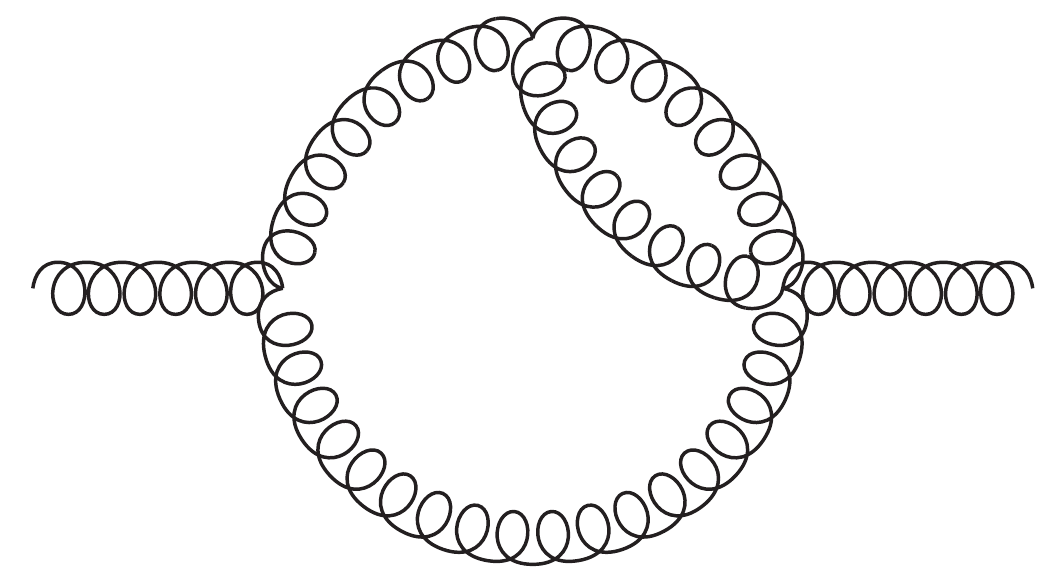} - \frac{2}{2} \graph{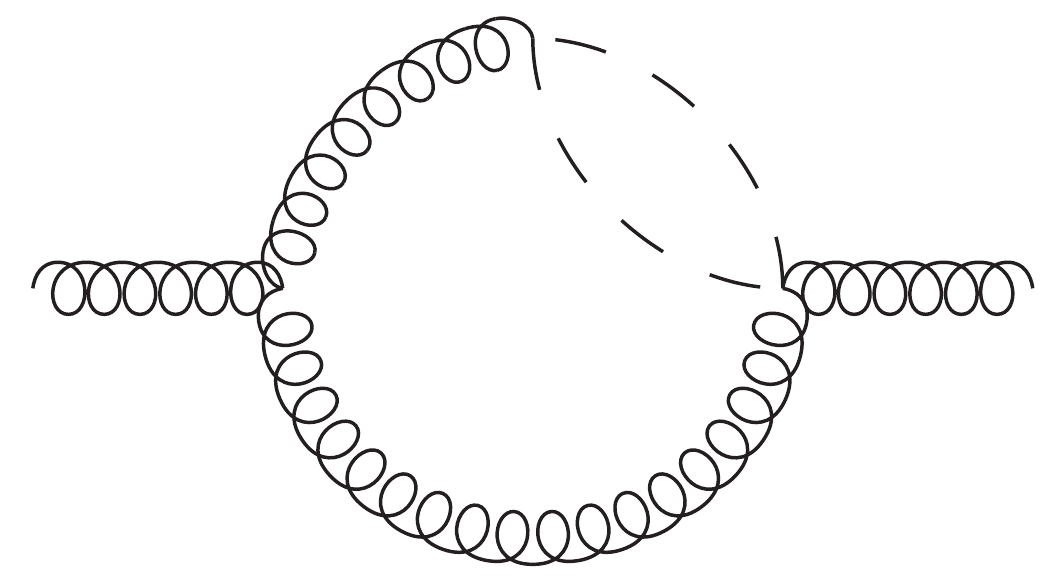} - \frac{2}{2} \graph{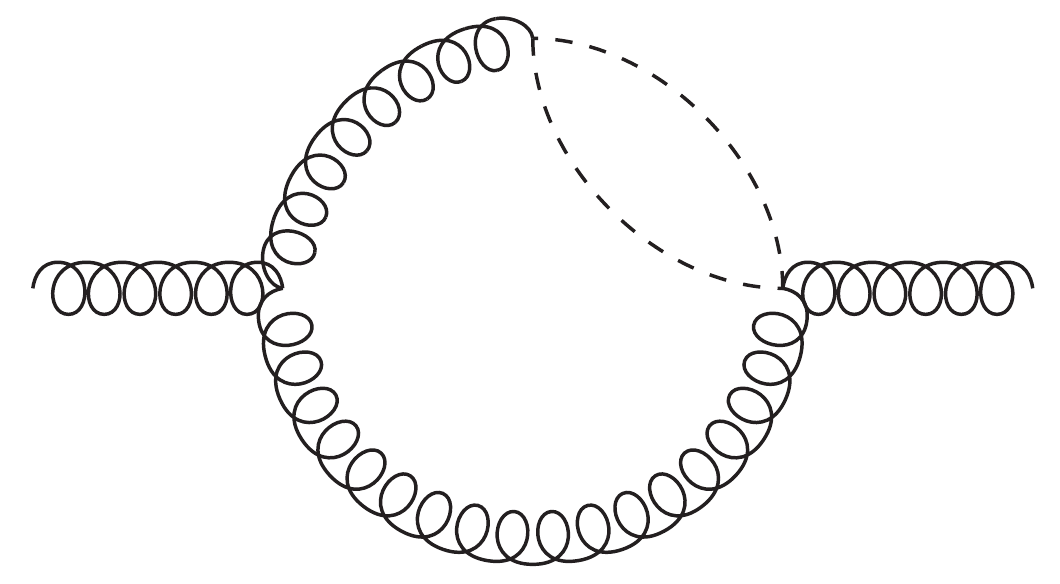} \\ & \phantom{ = } - 2 \graph{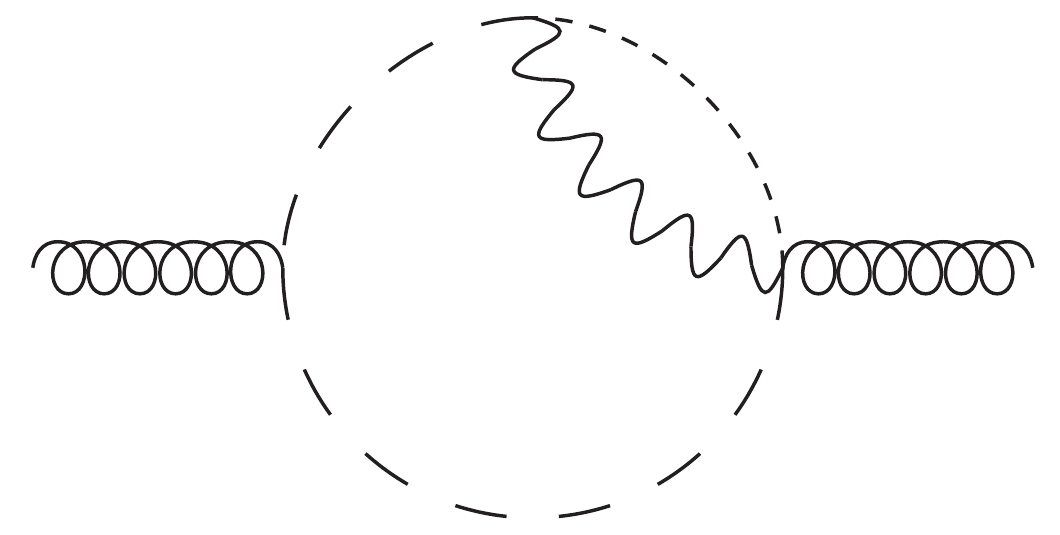} - 2 \graph{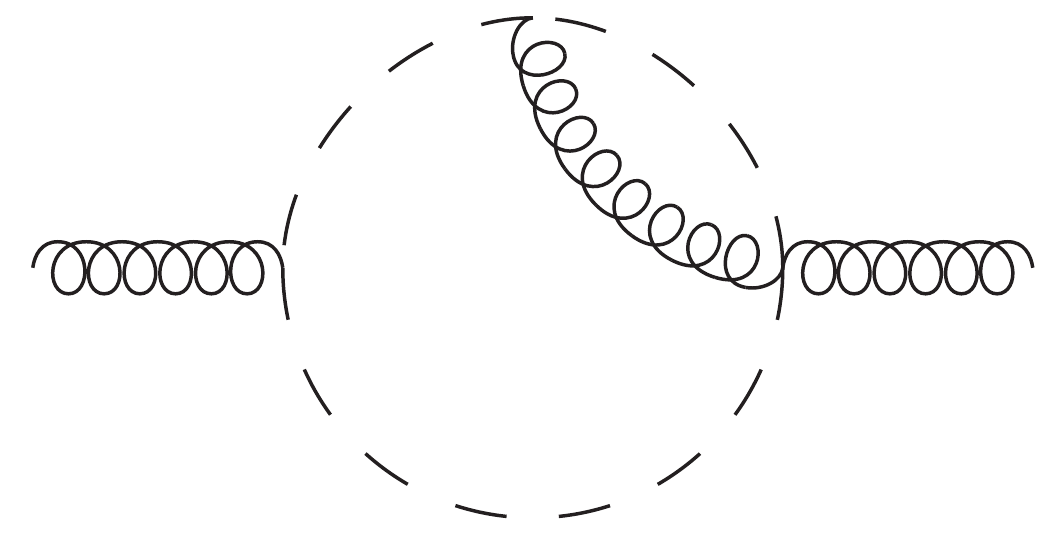} - 2 \graph{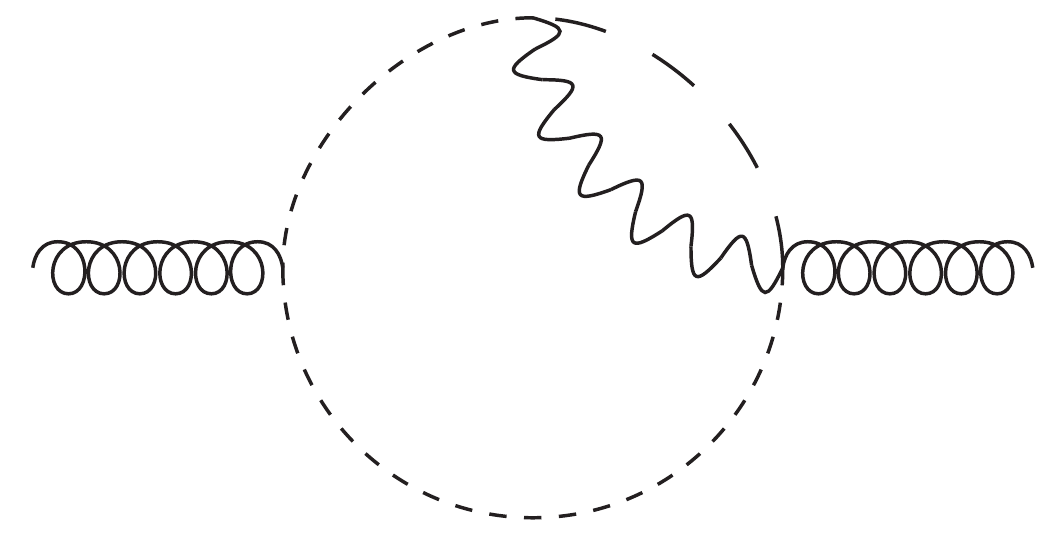} \\ & \phantom{ = } - 2 \graph{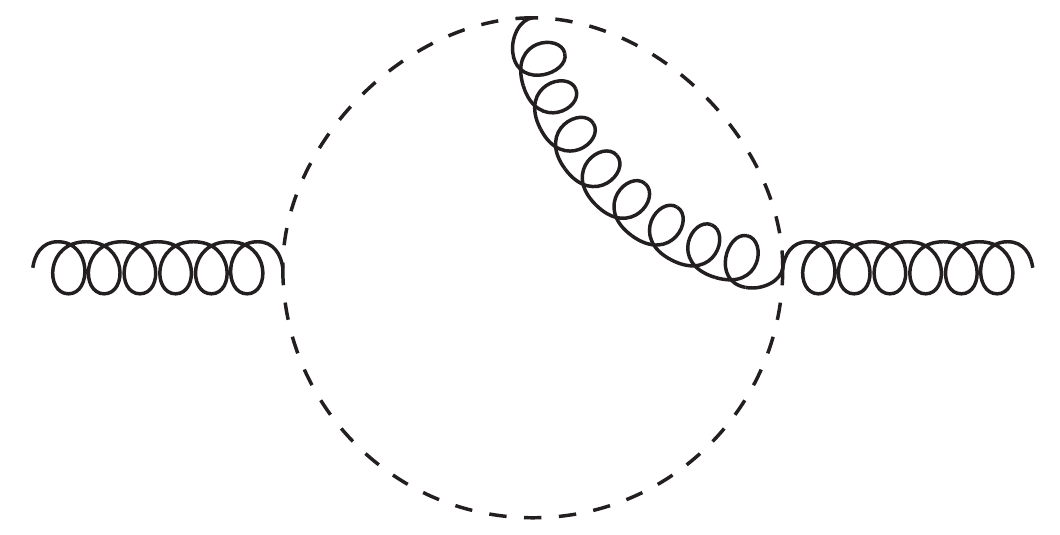}
\end{split}
\\
\rescombgreen^{\cgreen{c-aa}}_{(0,4)} & = 0
\\
\rescombgreen^{\cgreen{c-aa}}_{(2,2)} & = - \frac{1}{2} \graph{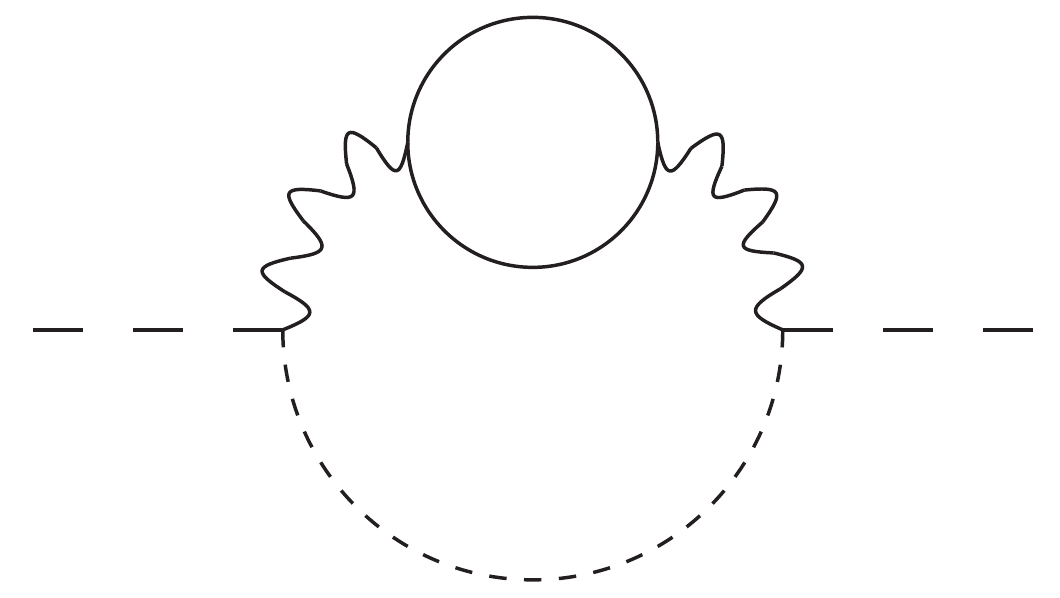}
\\
\begin{split}
\rescombgreen^{\cgreen{c-aa}}_{(4,0)} & = - \graph{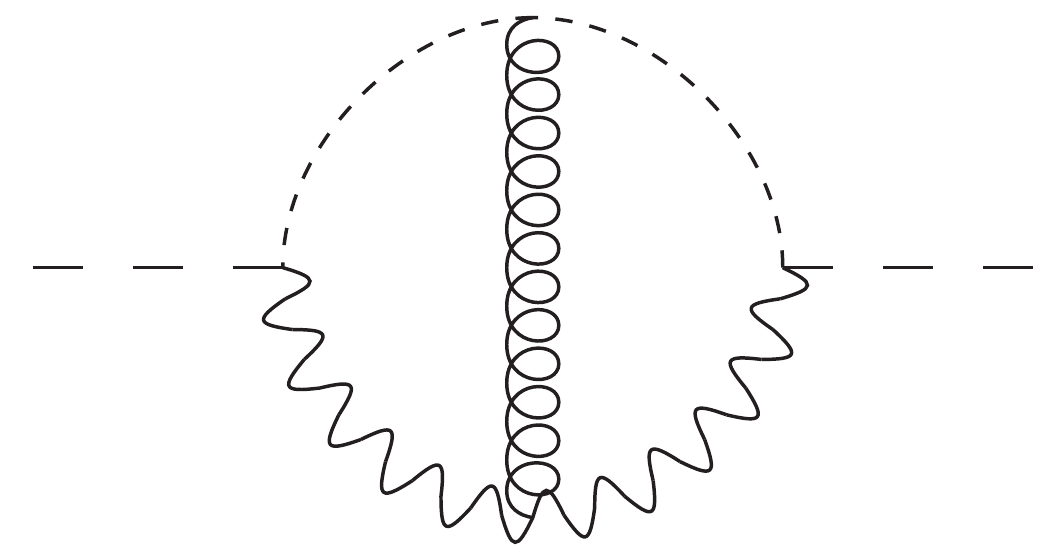} - 2 \graph{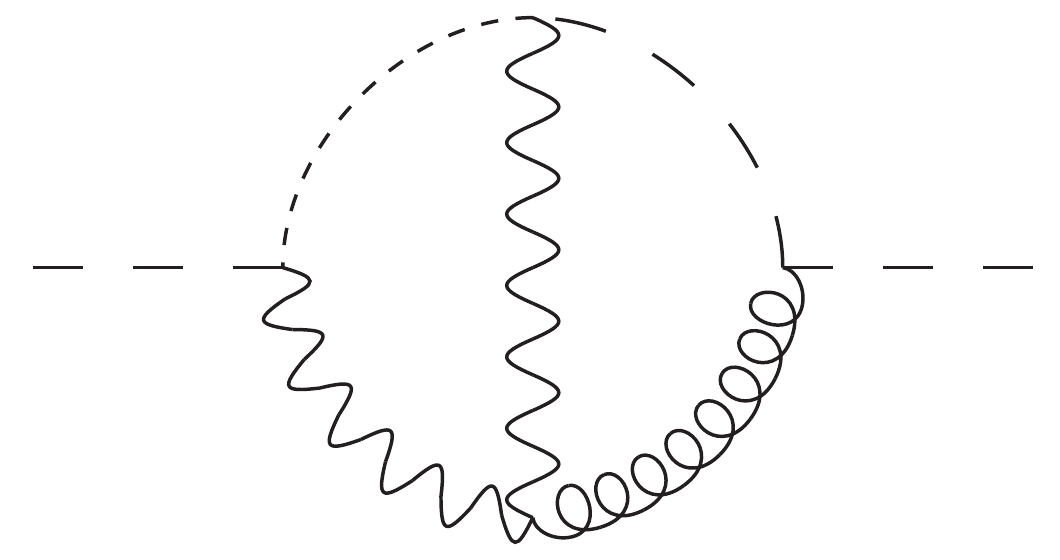} - 2 \graph{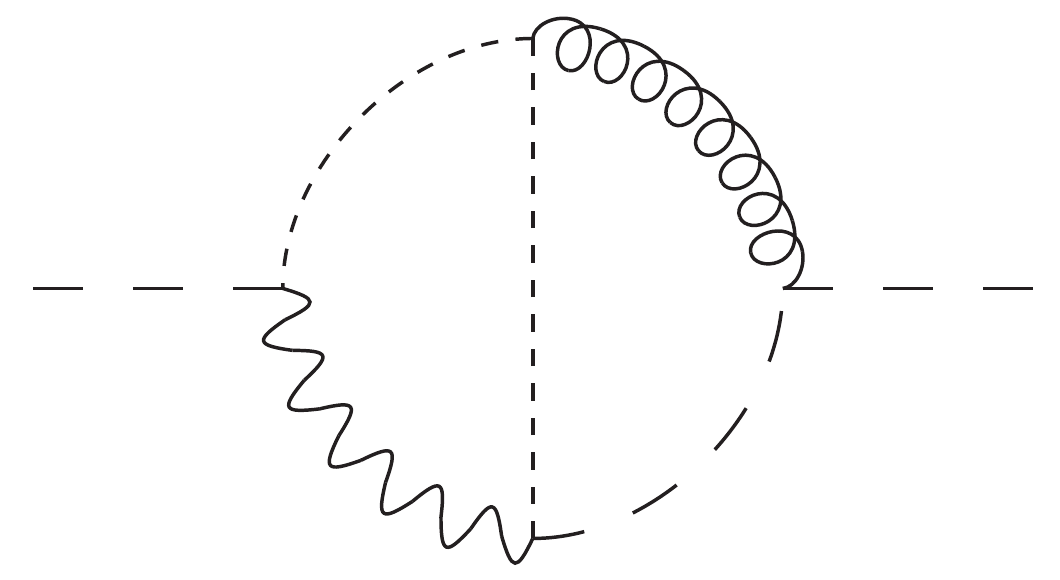} \\ & \phantom{ = } - \graph{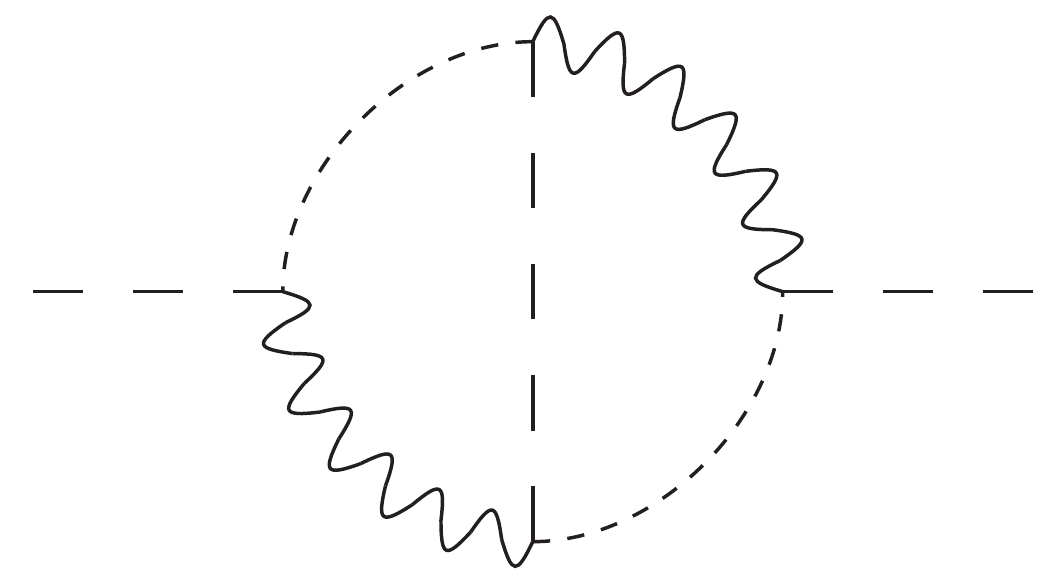} - \graph{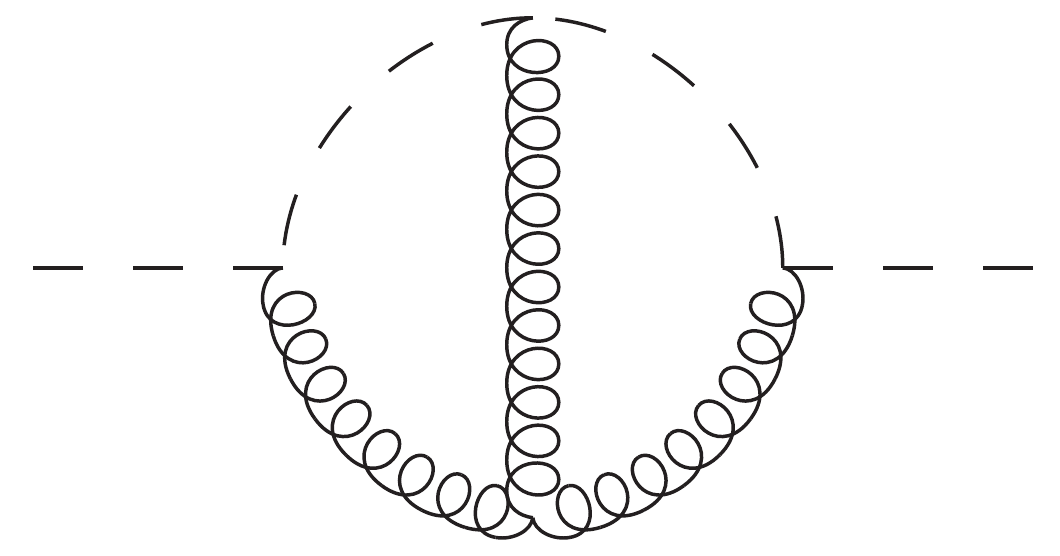} - \graph{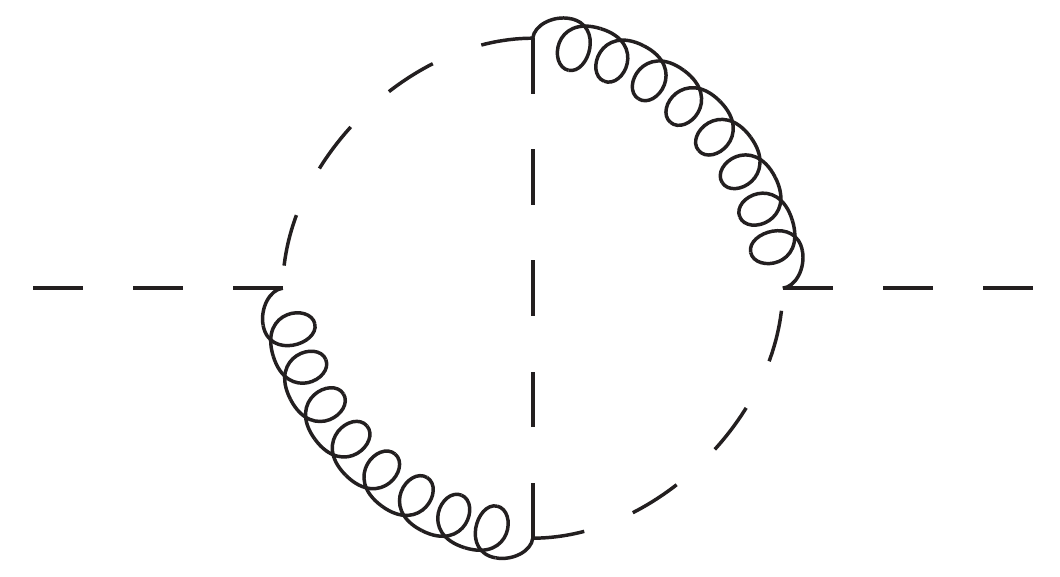} \\ & \phantom{ = } - \graph{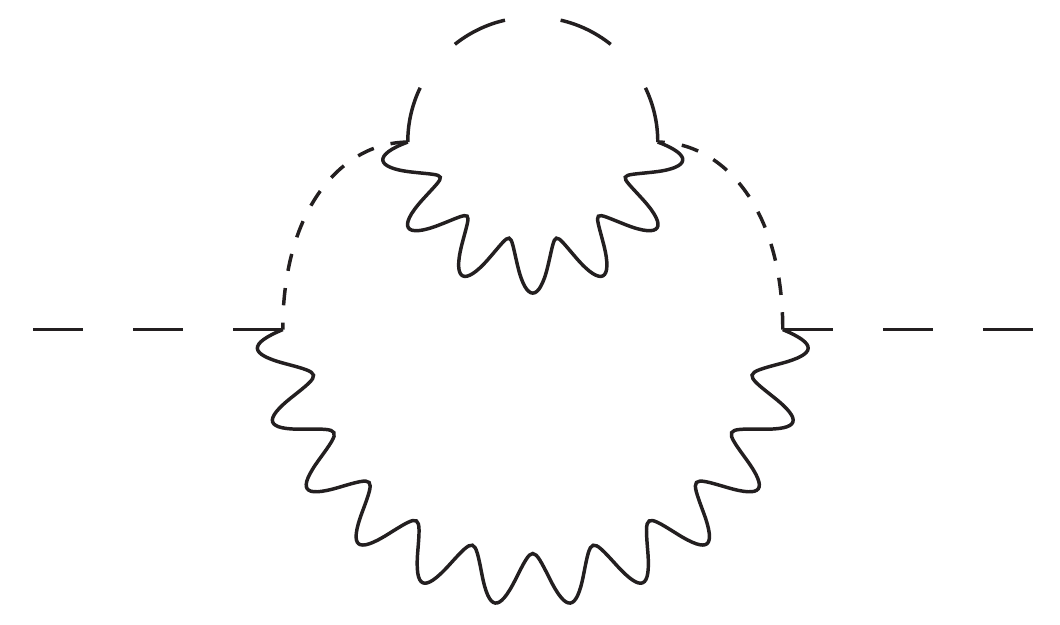} - \graph{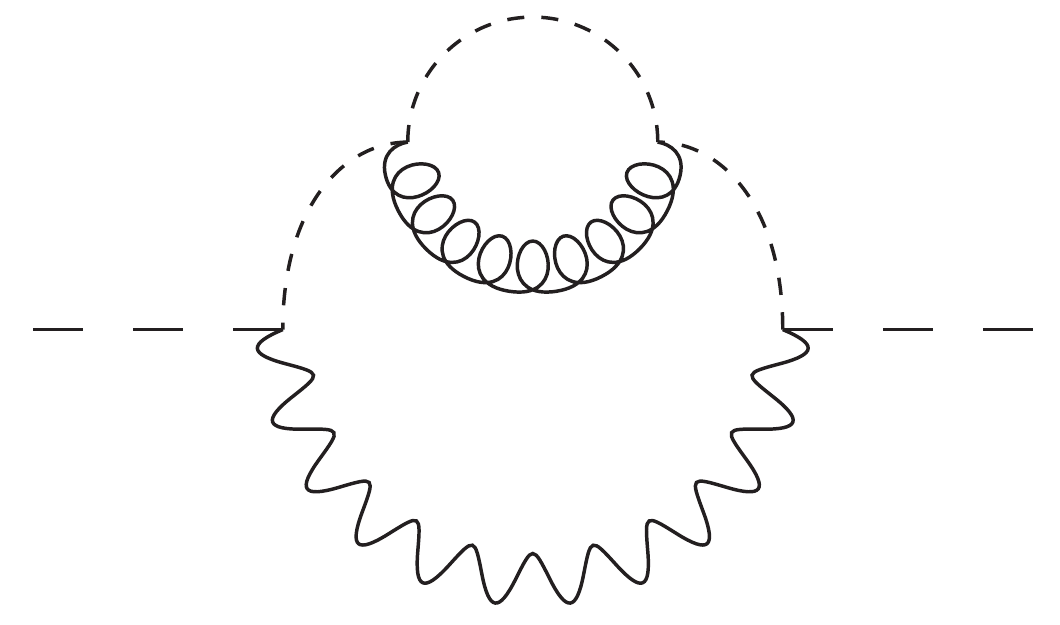} - \graph{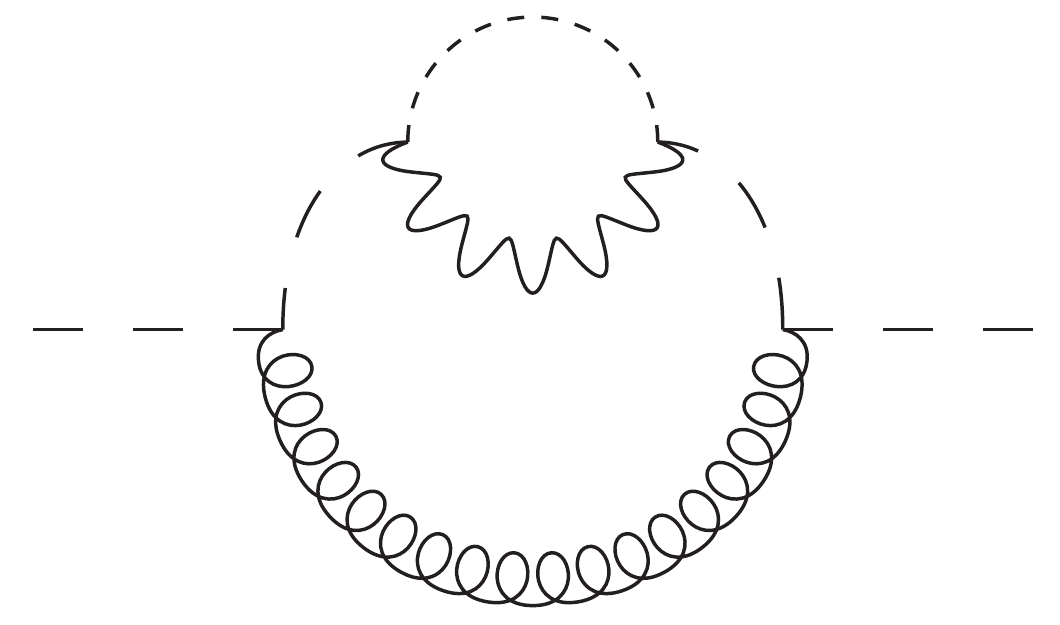} \\ & \phantom{ = } - \graph{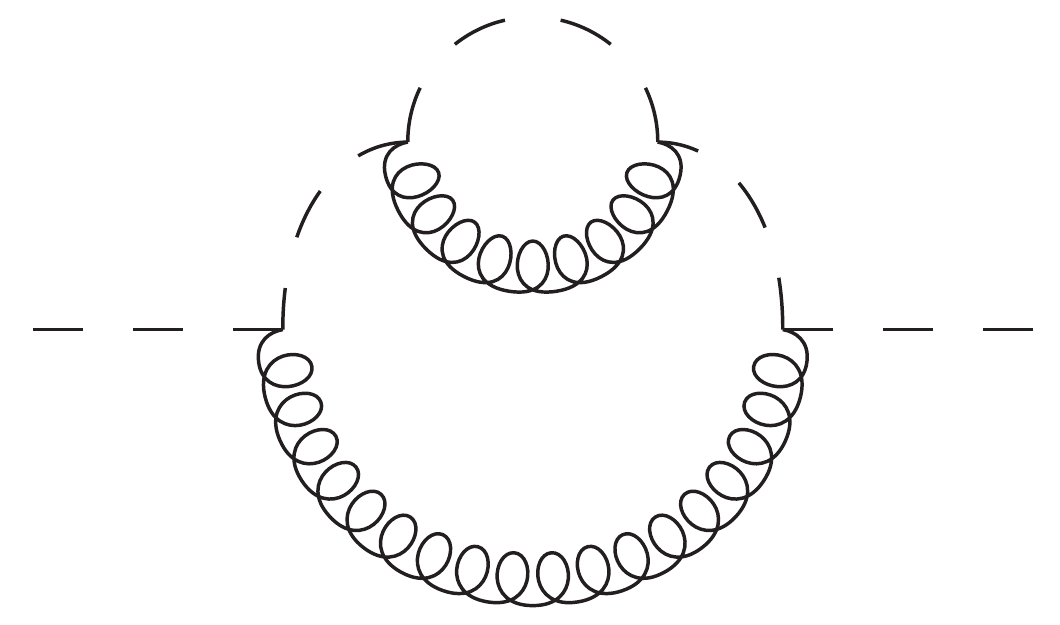} - \frac{1}{2} \graph{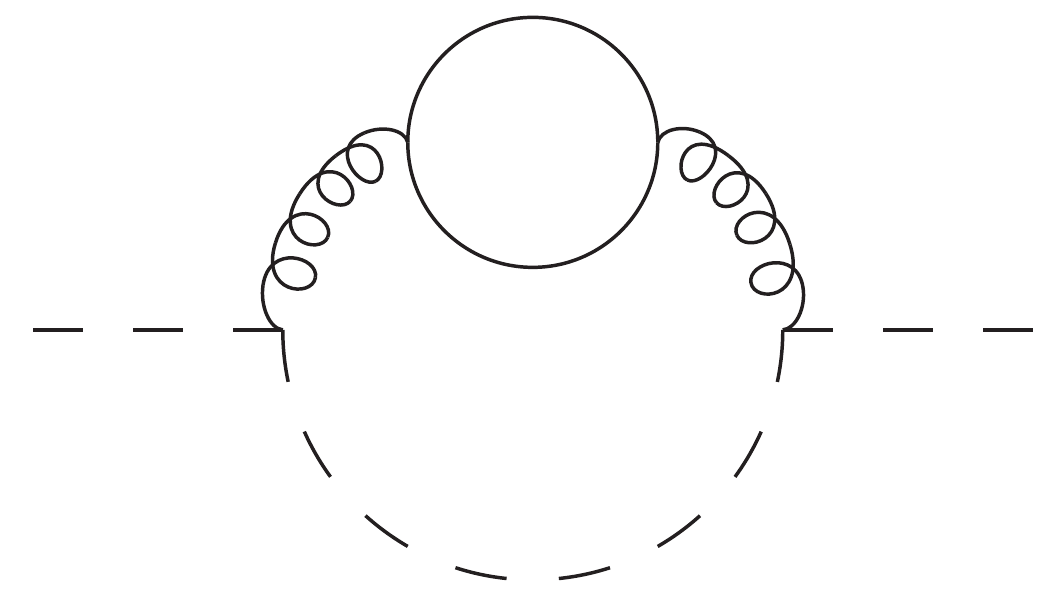} - \frac{1}{2} \graph{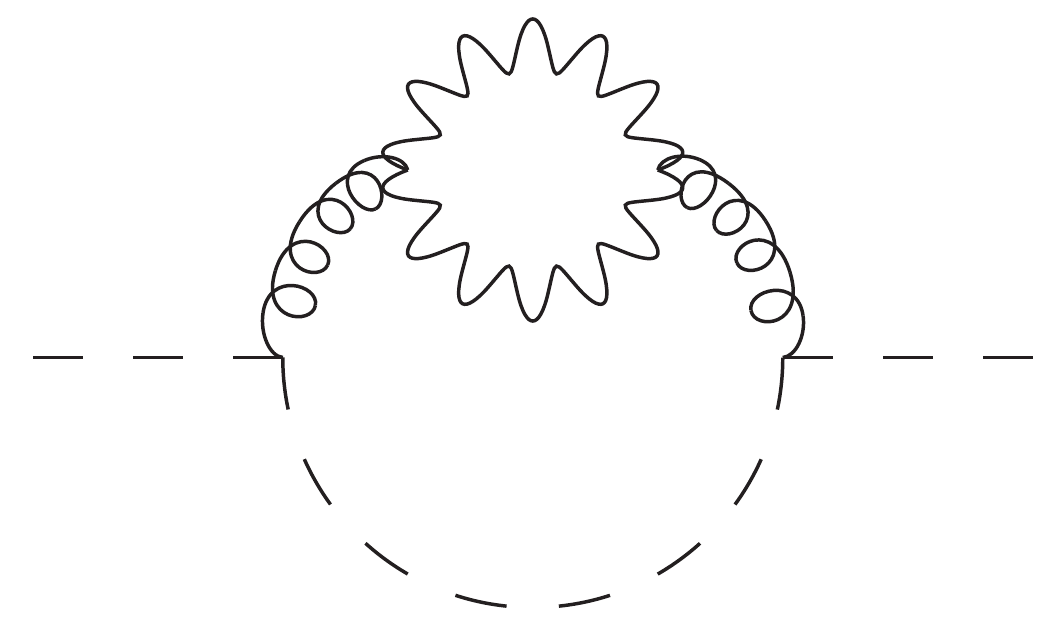} \\ & \phantom{ = } - \frac{1}{2} \graph{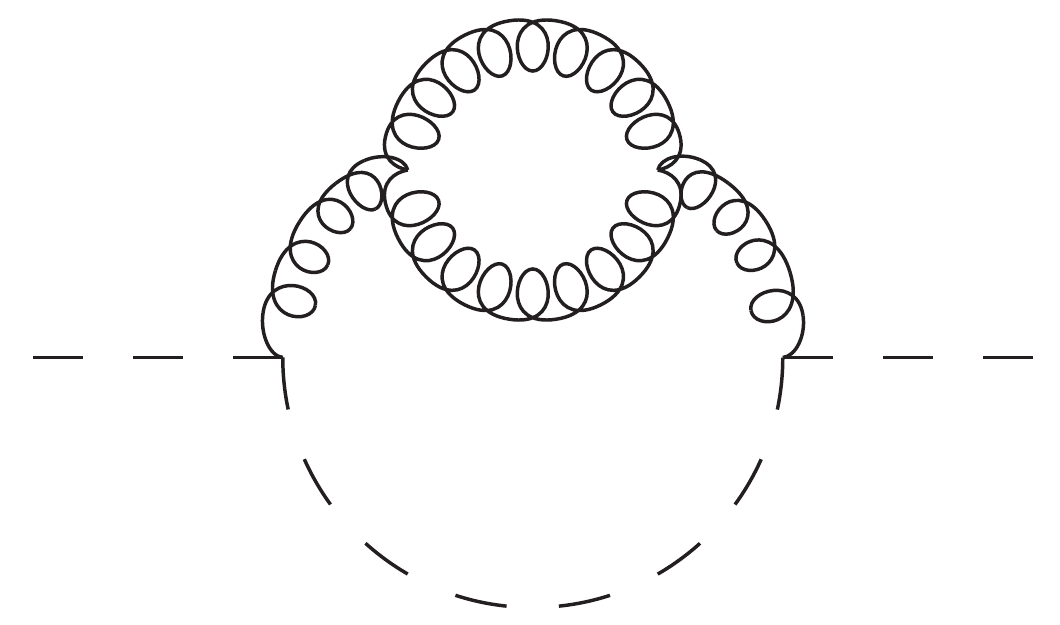} - \frac{1}{2} \graph{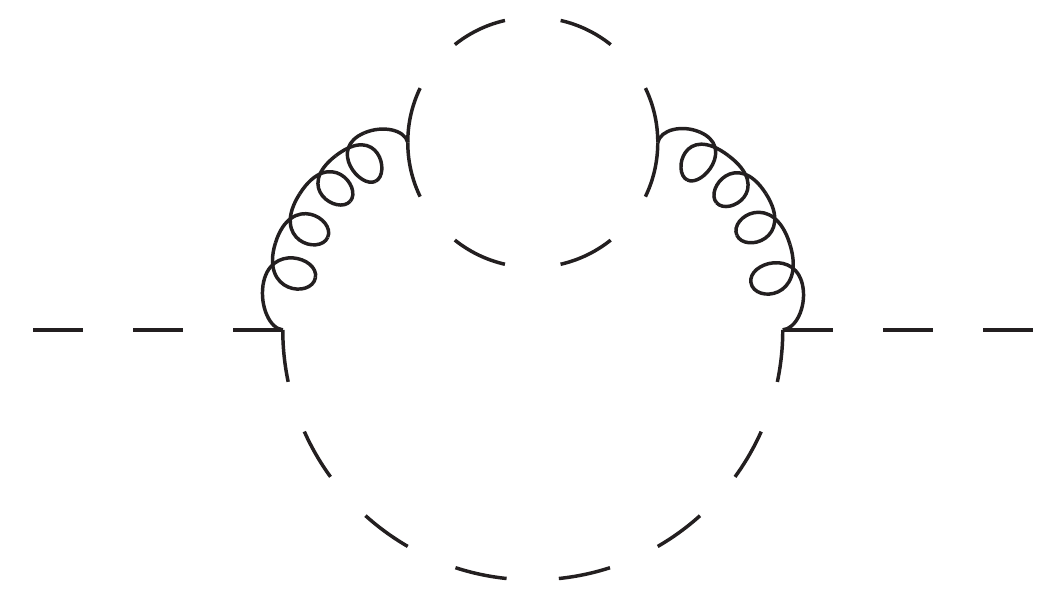} - \frac{1}{2} \graph{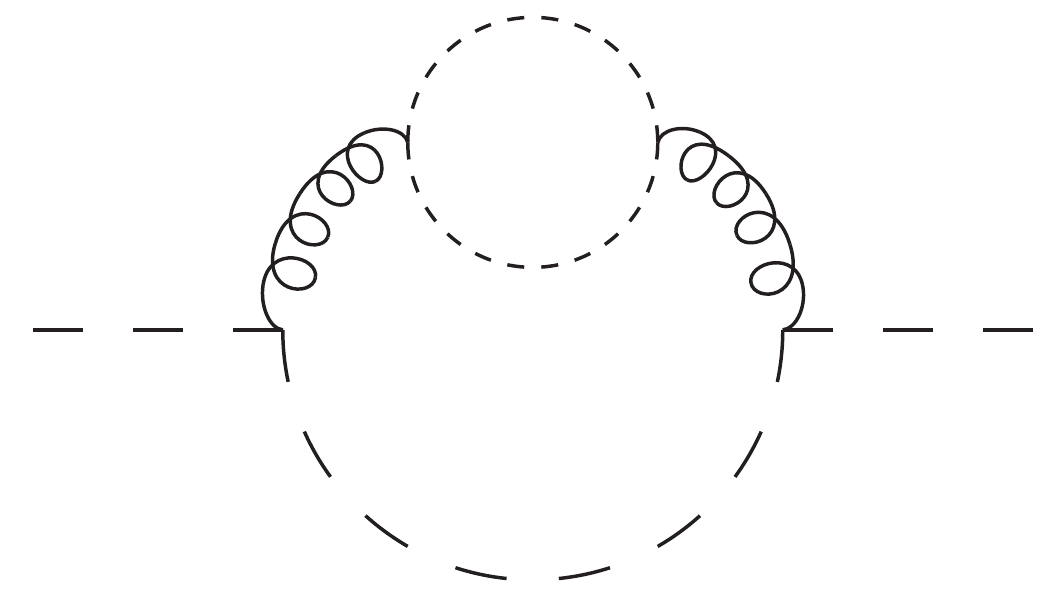} \\ & \phantom{ = } - \graph{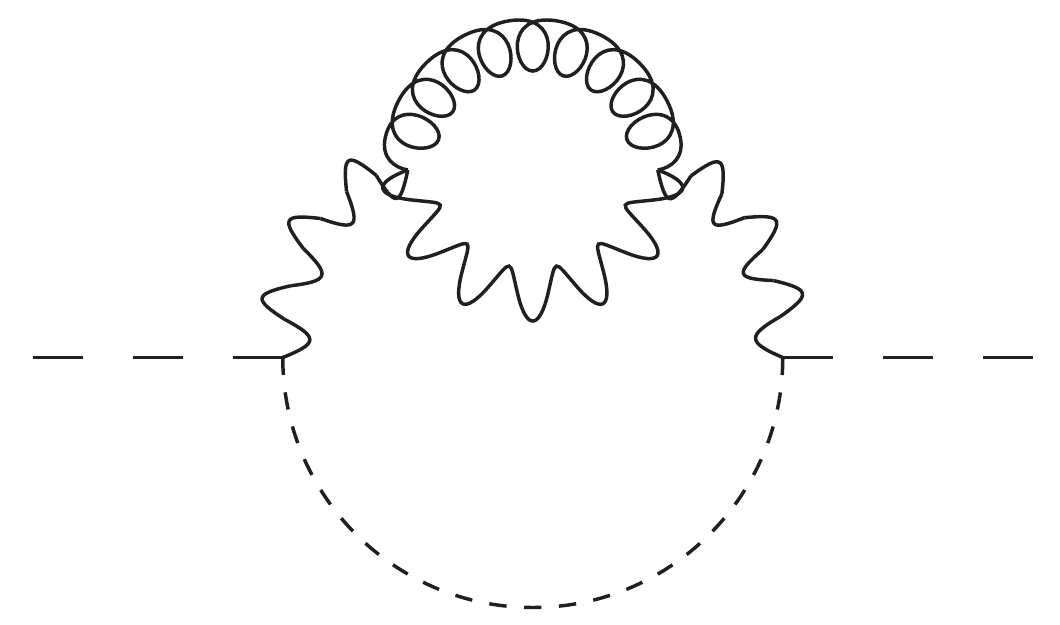} - \graph{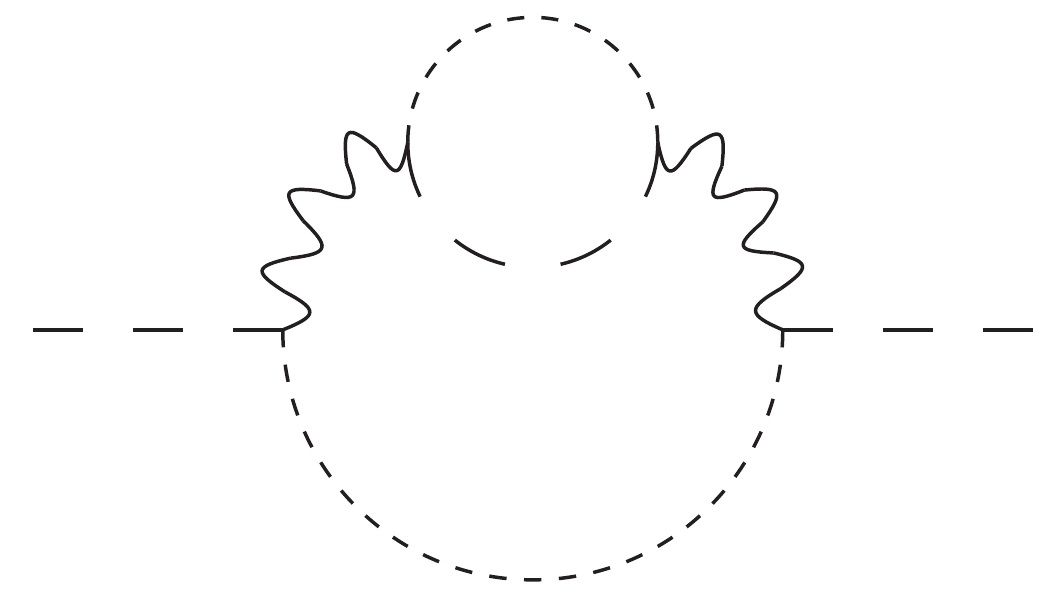} - 2 \ggraph{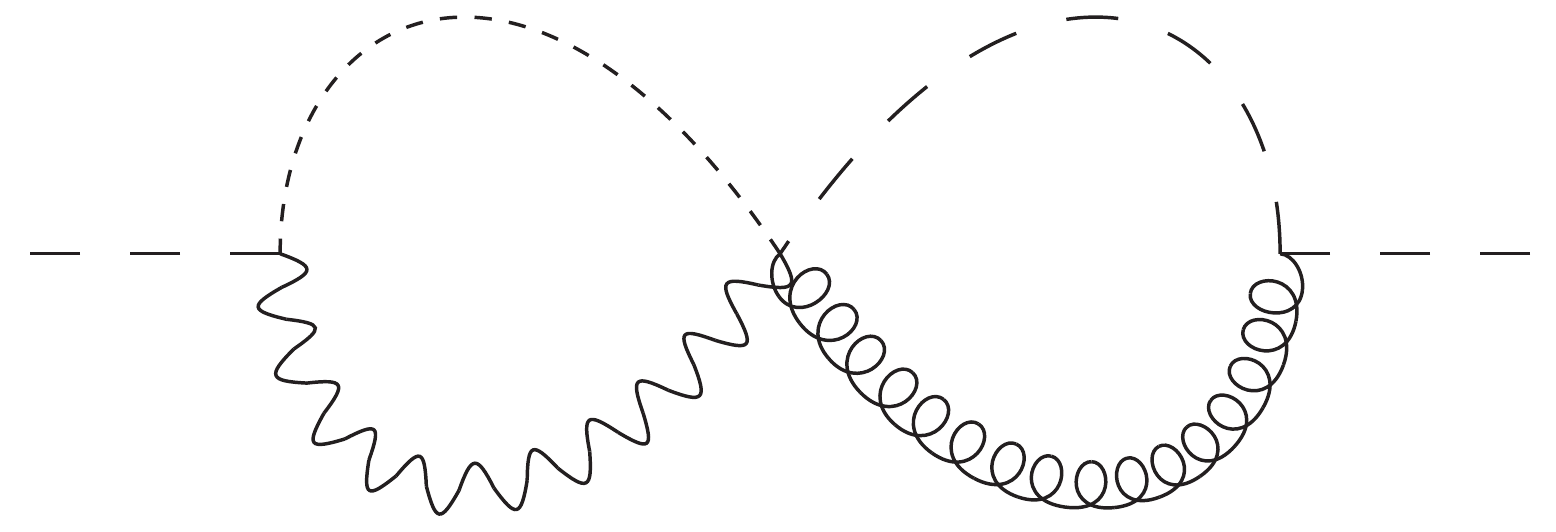} \\ & \phantom{ = } - \ggraph{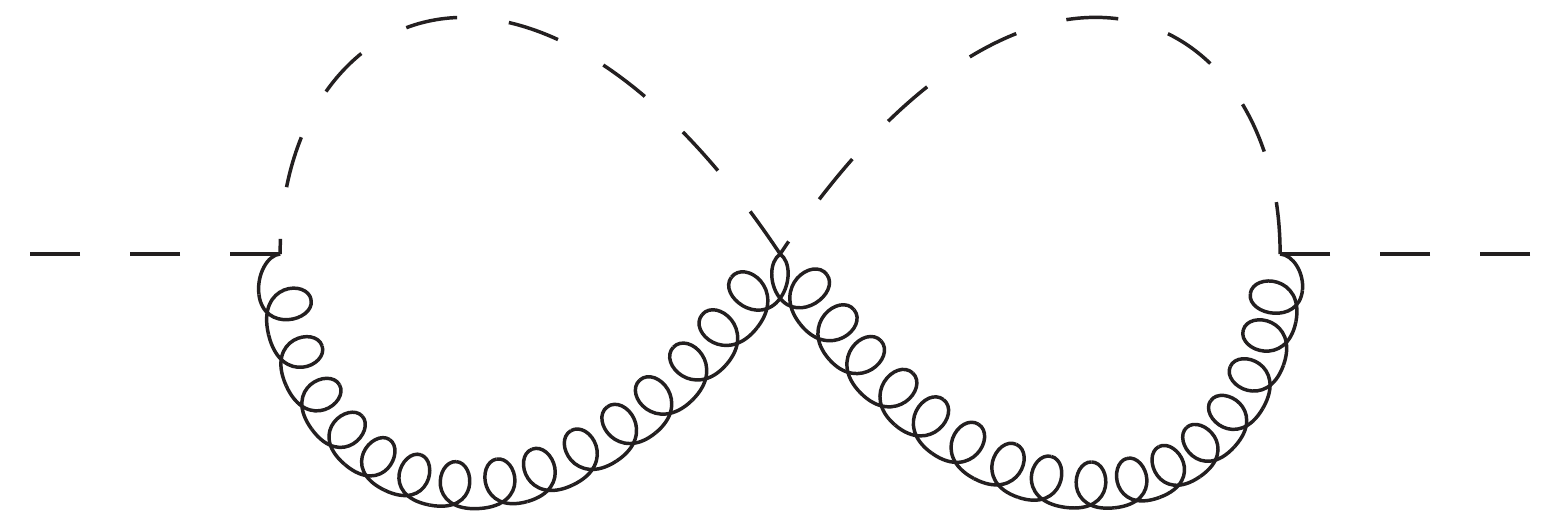} - 2 \graph{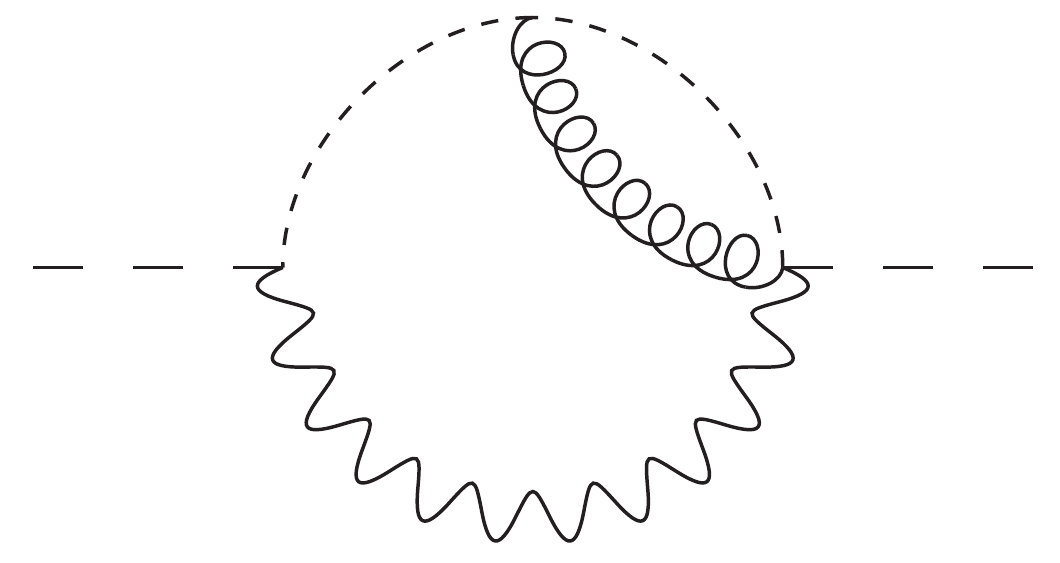} - 2 \graph{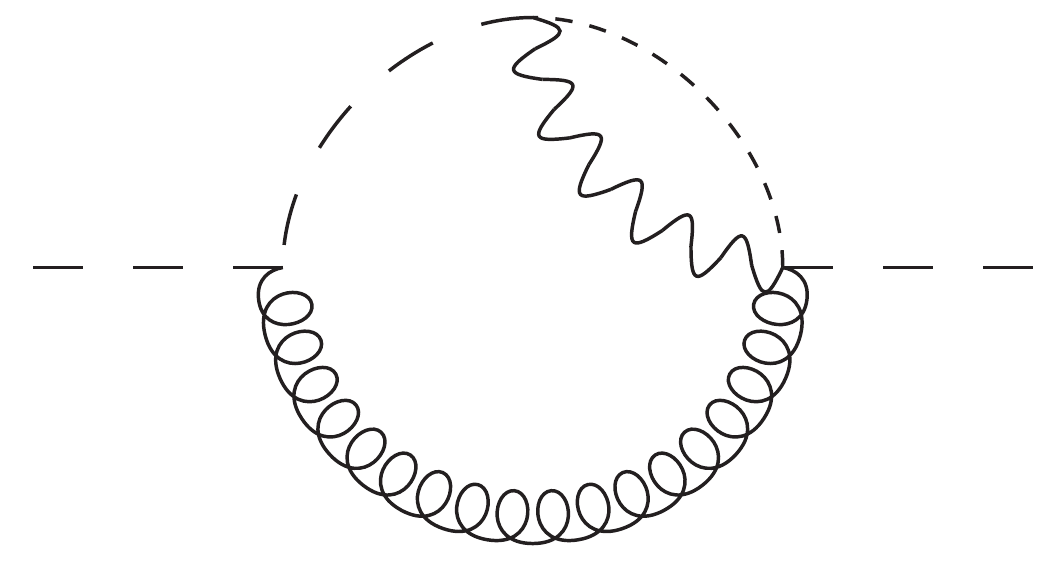} \\ & \phantom{ = } - 2 \graph{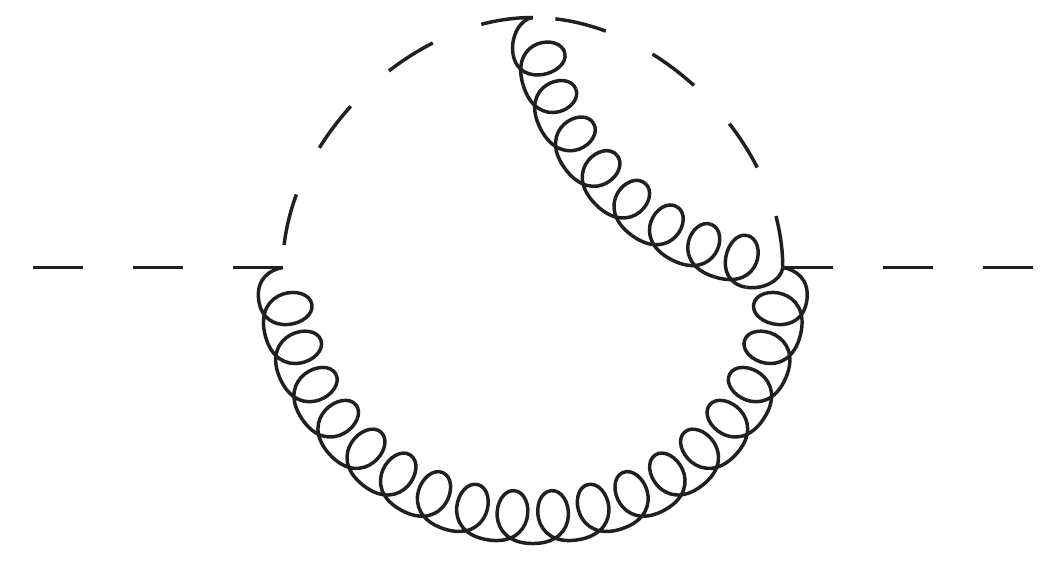} - \frac{2}{2} \graph{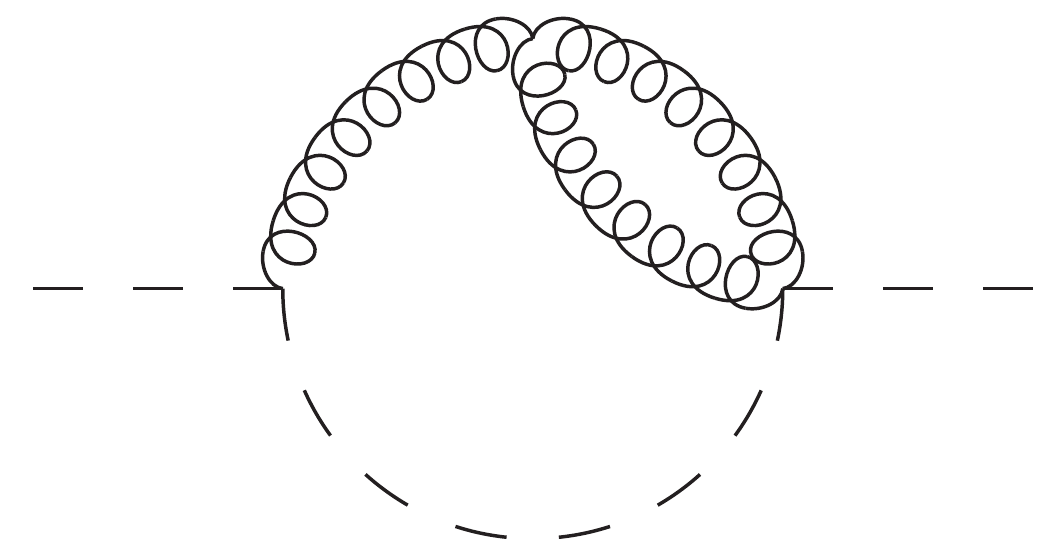} - 2 \graph{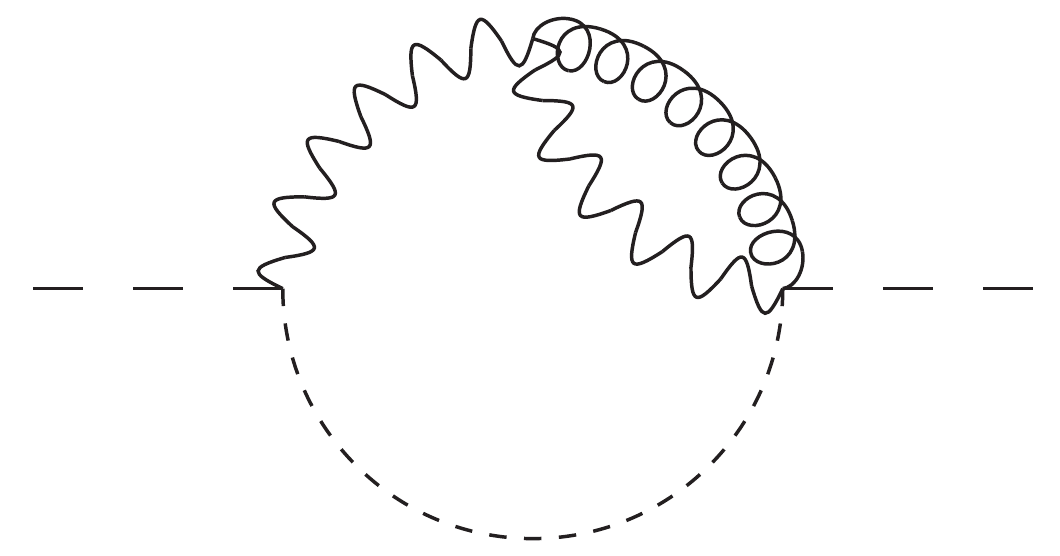}
\end{split}
\\
\rescombgreen^{\cgreen{c-hh}}_{(0,4)} & = 0
\\
\rescombgreen^{\cgreen{c-hh}}_{(2,2)} & = - \frac{1}{2} \graph{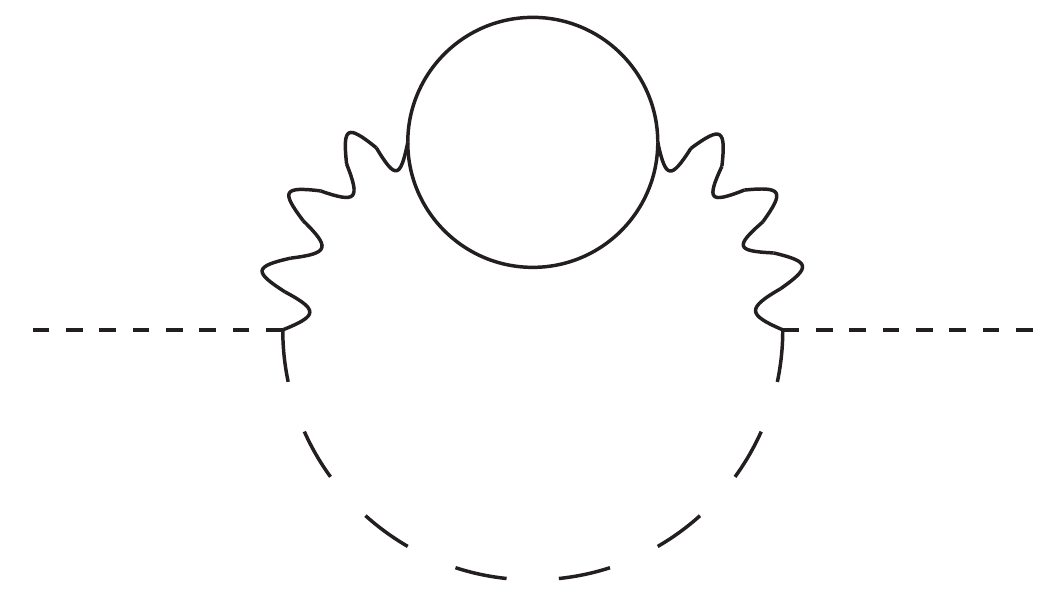}
\\
\begin{split}
\rescombgreen^{\cgreen{c-hh}}_{(4,0)} & = - \graph{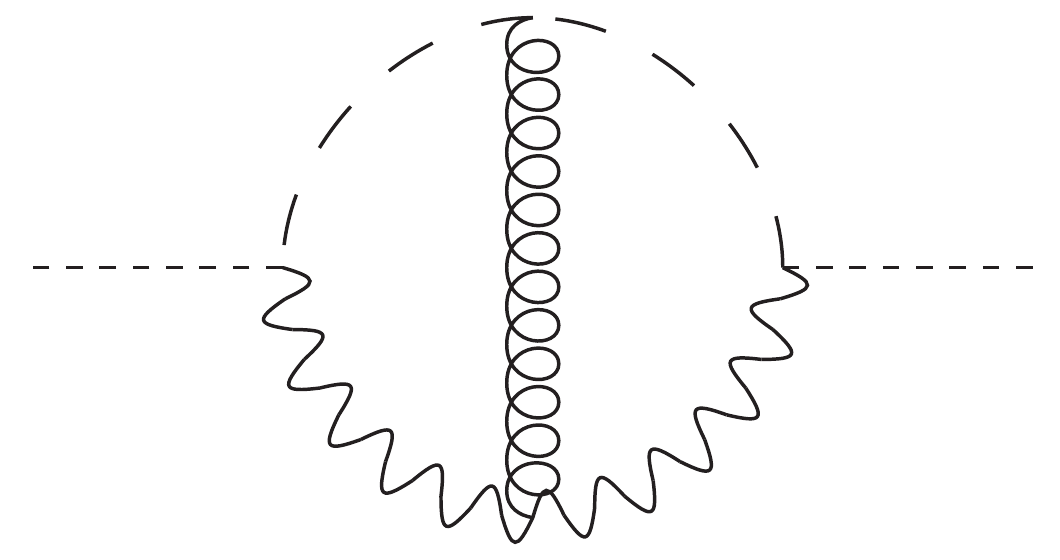} - 2 \graph{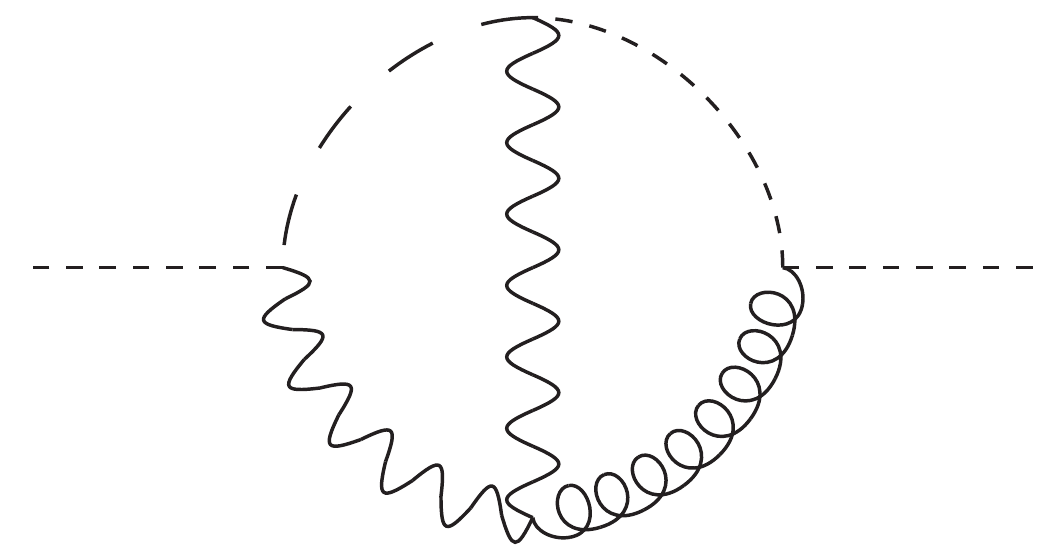} - \graph{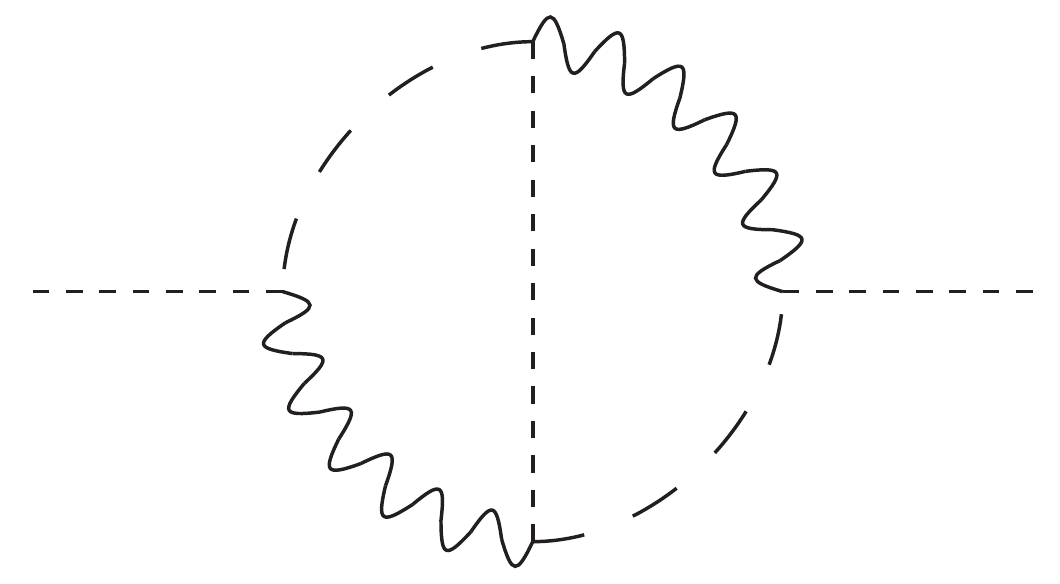} \\ & \phantom{ = } - 2 \graph{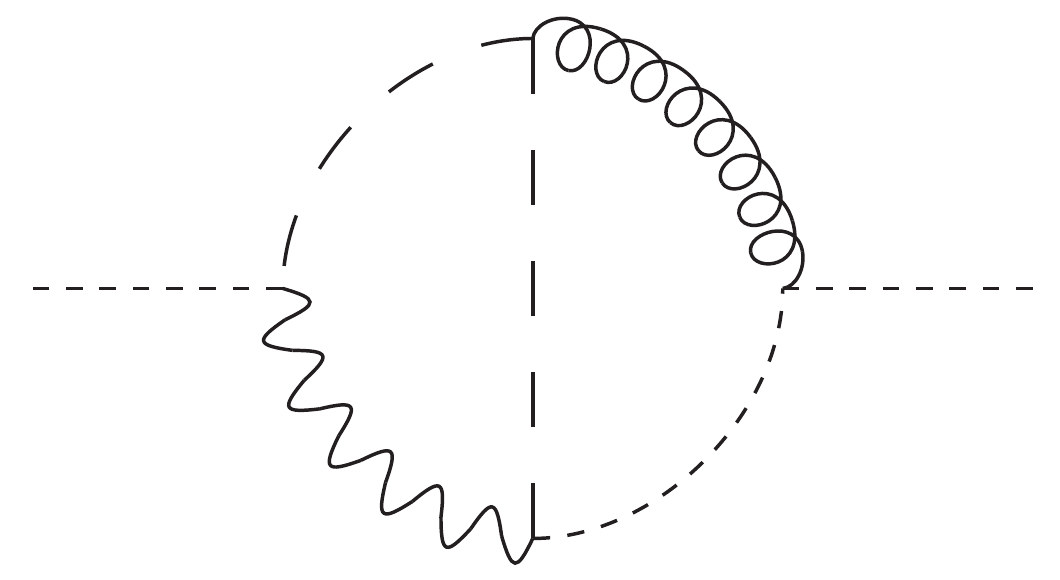} - \graph{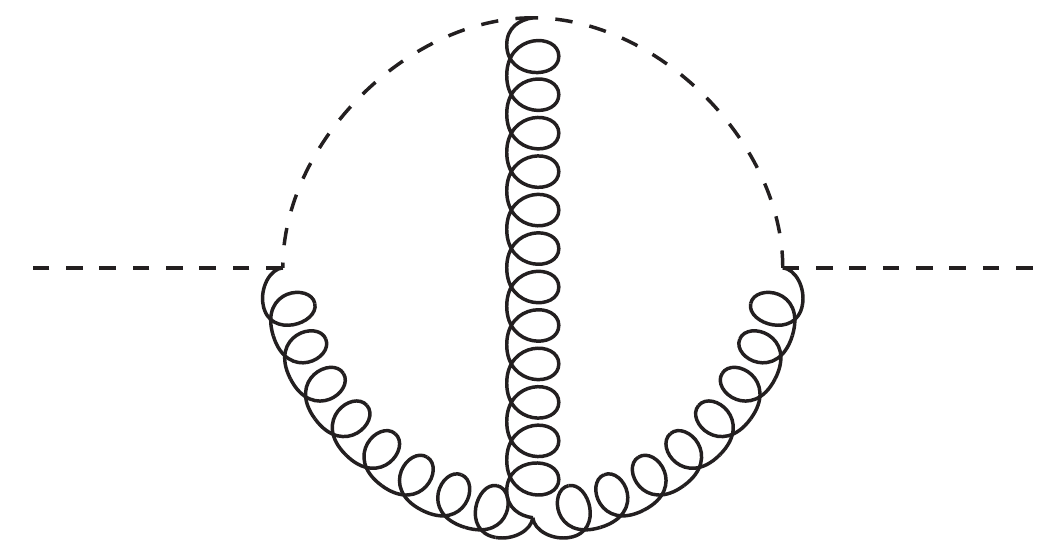} - \graph{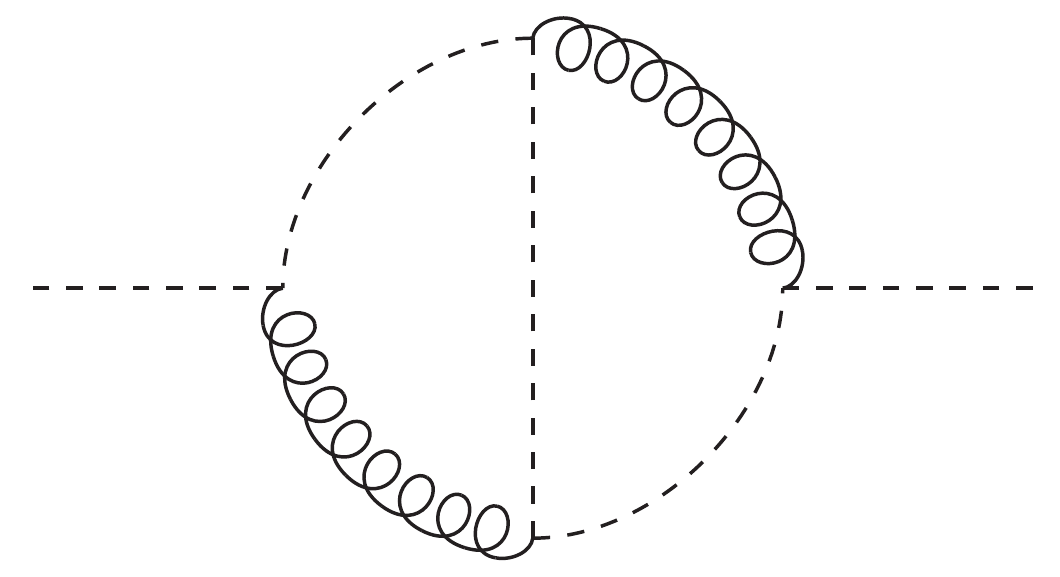} \\ & \phantom{ = } - \graph{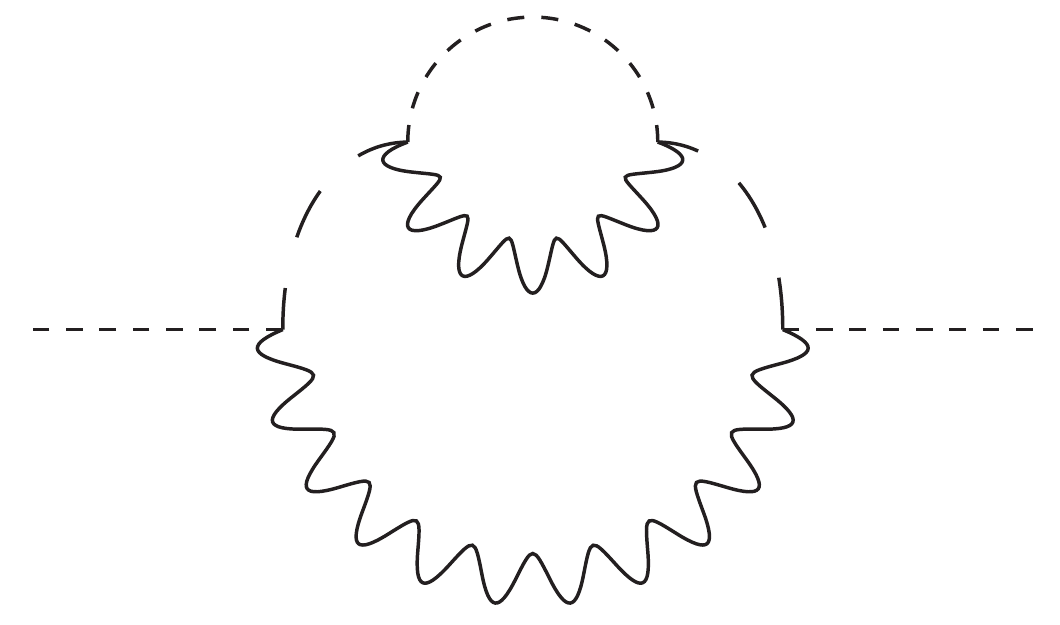} - \graph{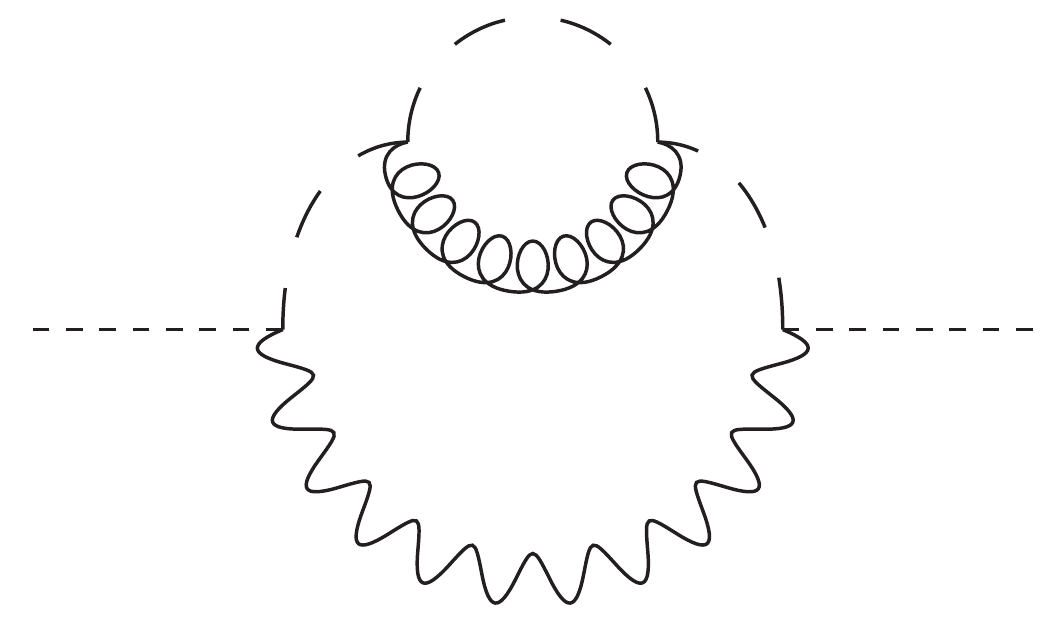} - \graph{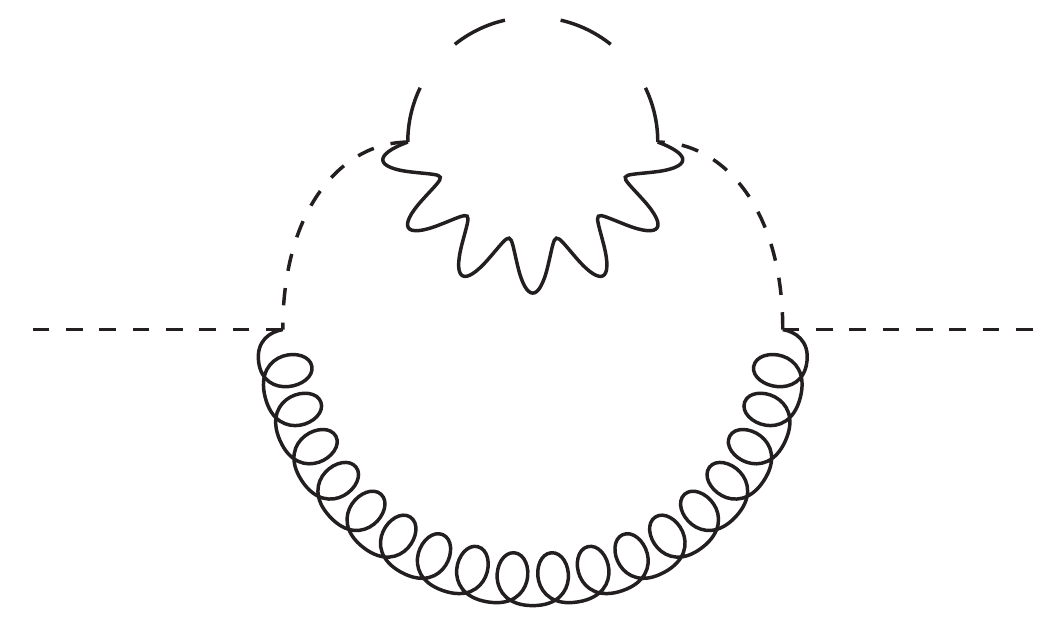} \\ & \phantom{ = } - \graph{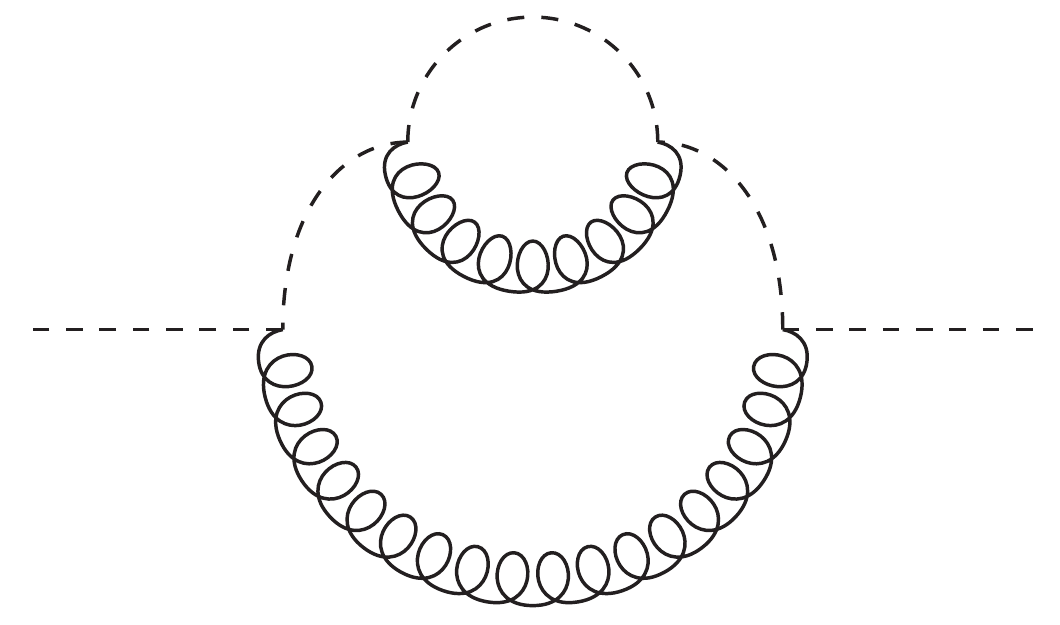} - \graph{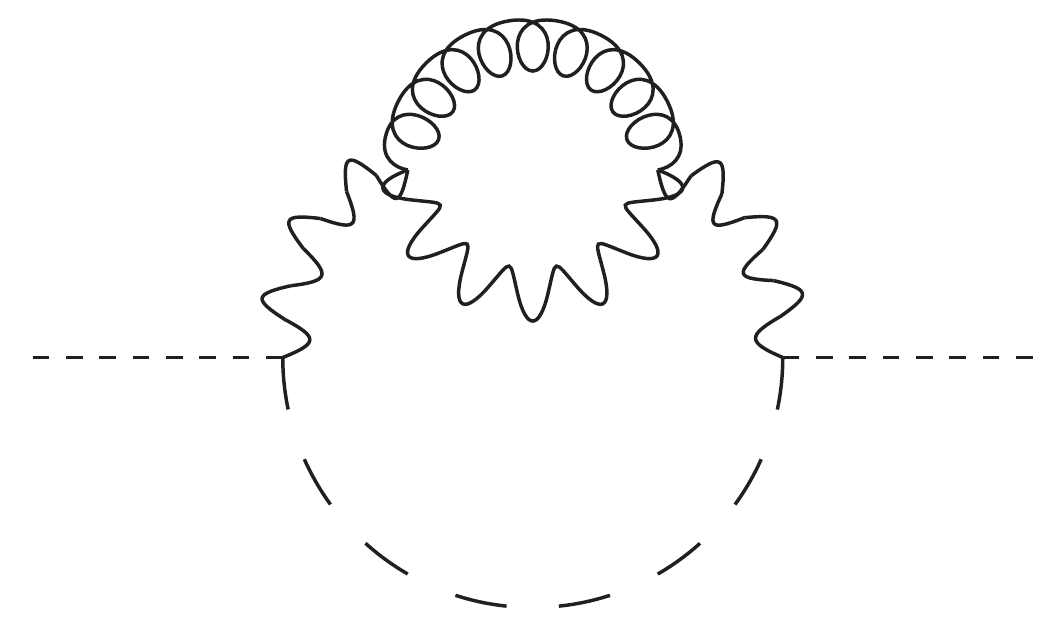} - \graph{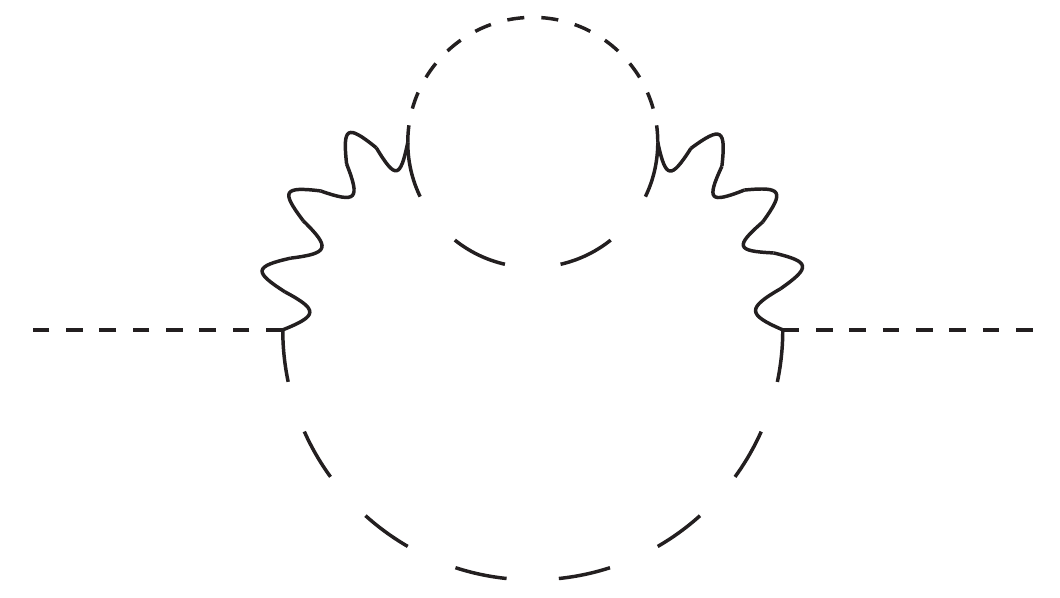} \\ & \phantom{ = } - \frac{1}{2} \graph{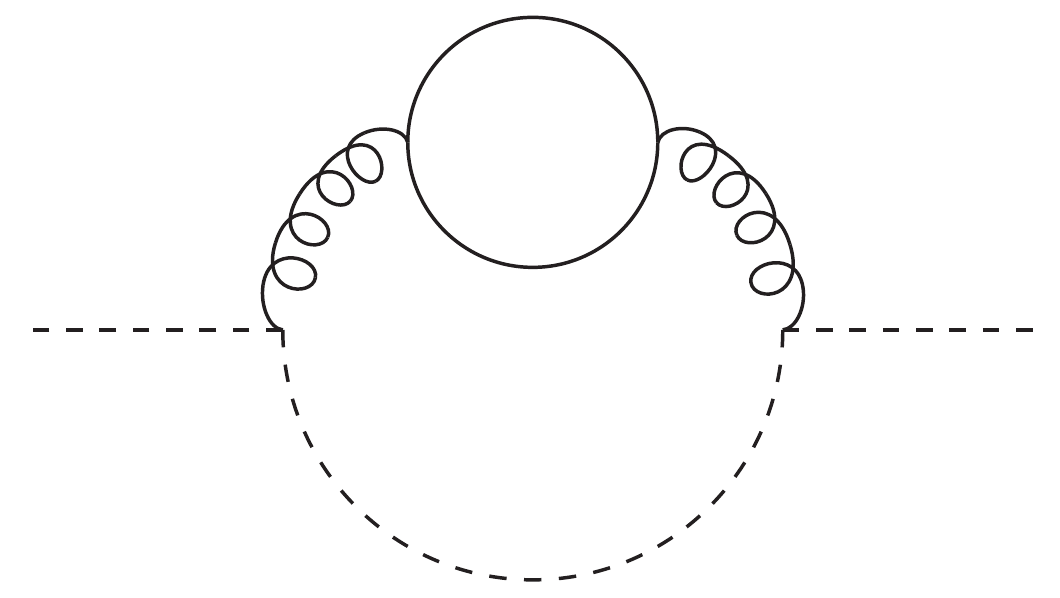} - \frac{1}{2} \graph{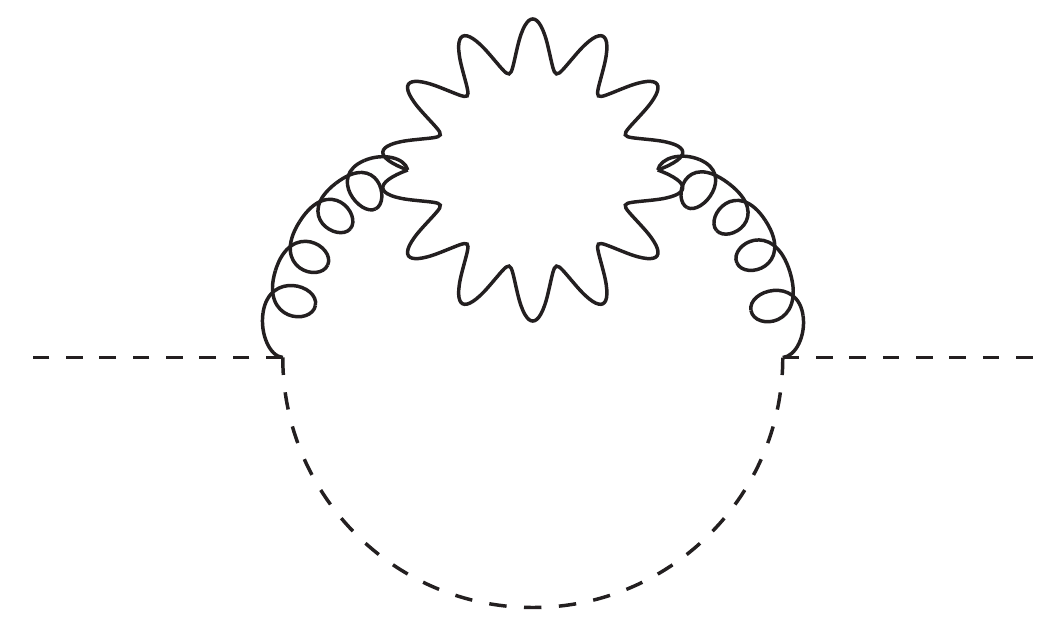} - \frac{1}{2} \graph{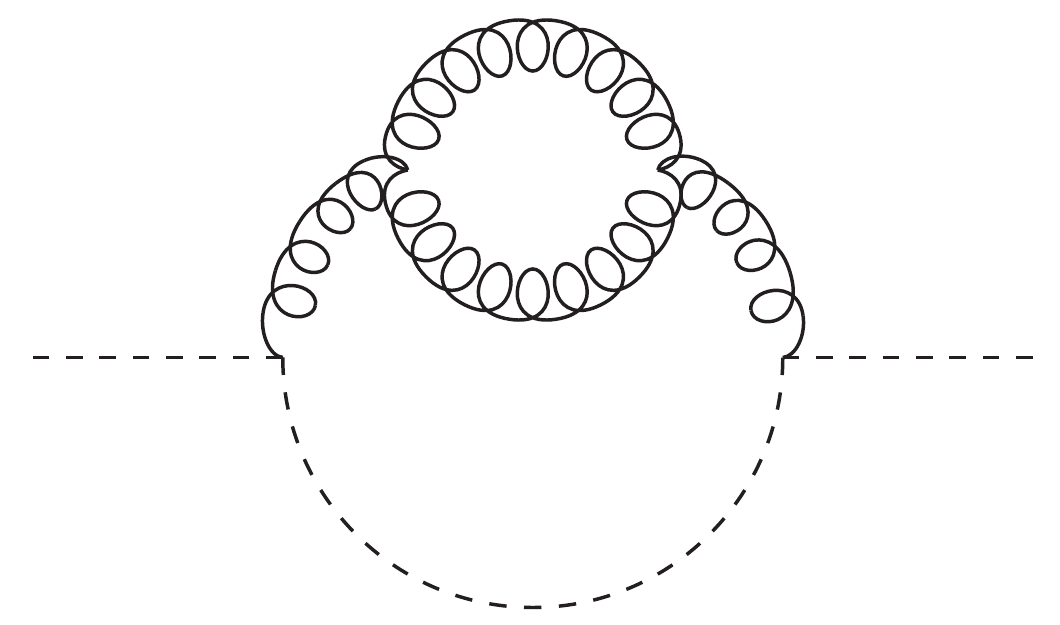} \\ & \phantom{ = } - \frac{1}{2} \graph{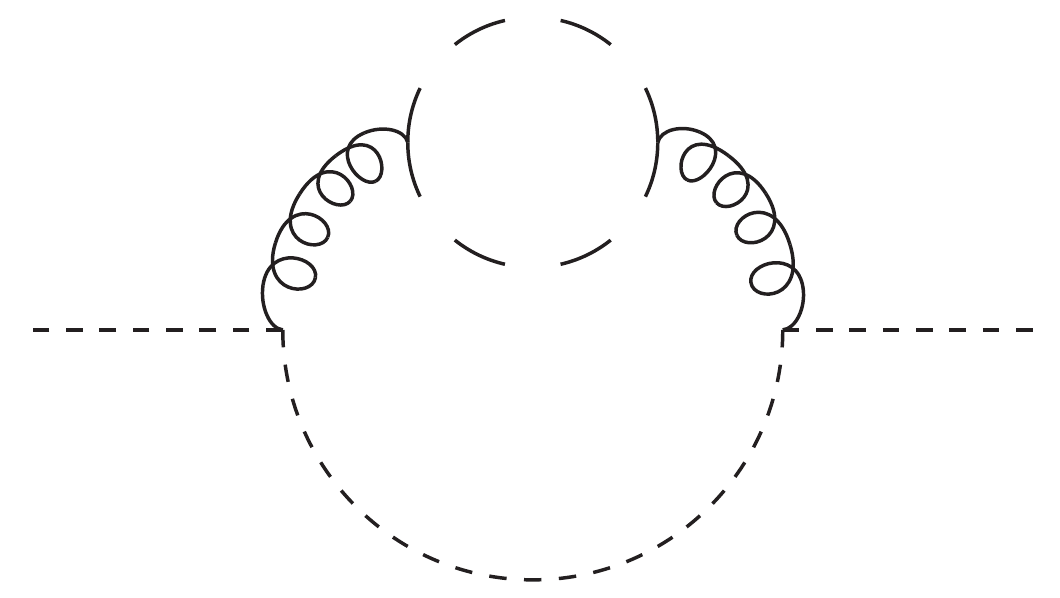} - \frac{1}{2} \graph{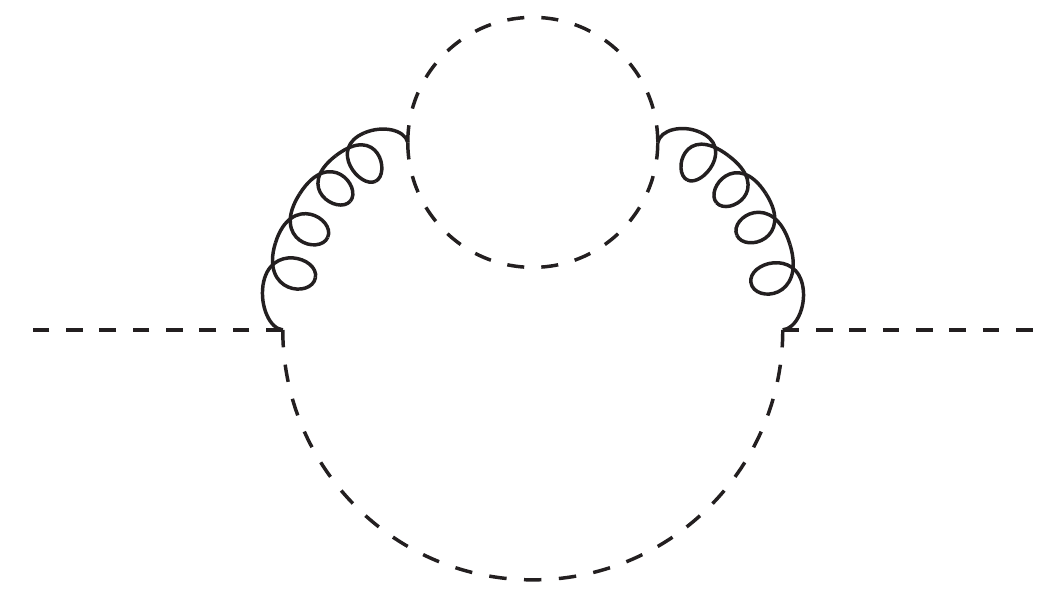} - 2 \ggraph{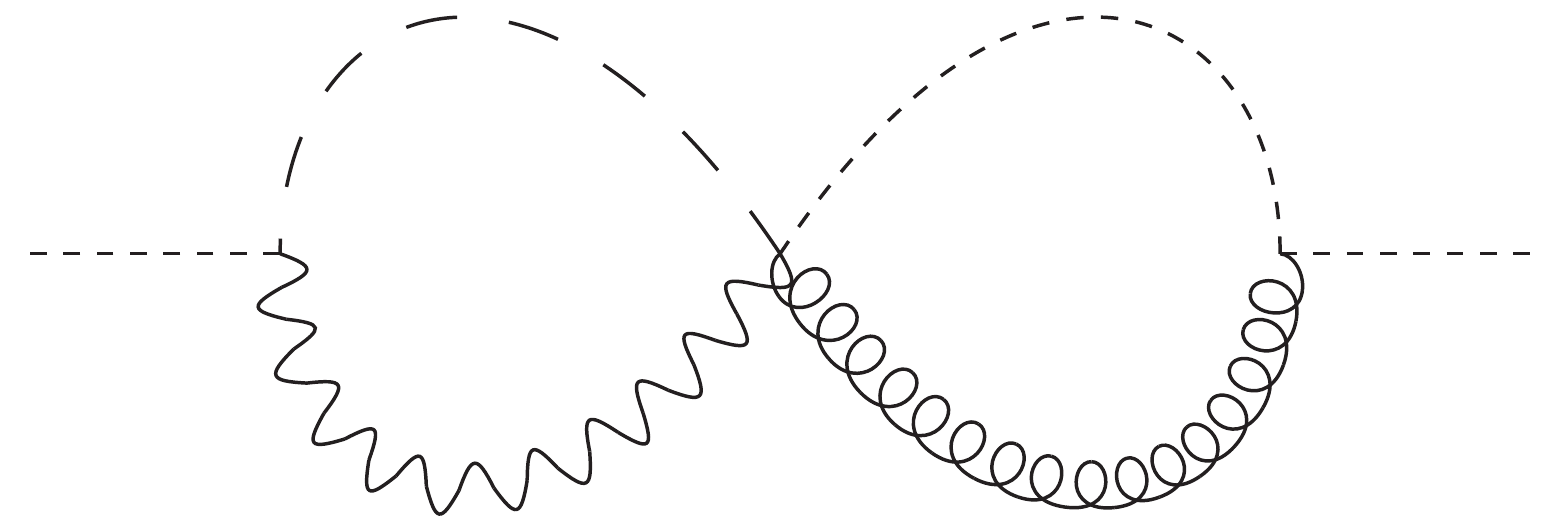} \\ & \phantom{ = } - \ggraph{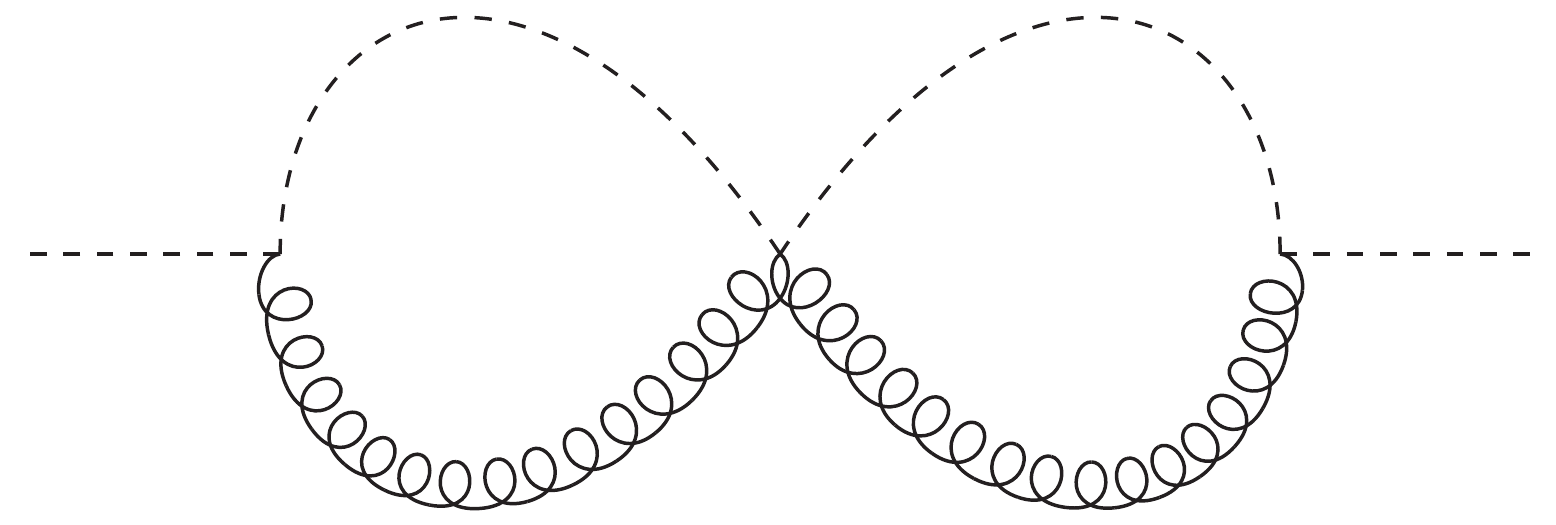} - 2 \graph{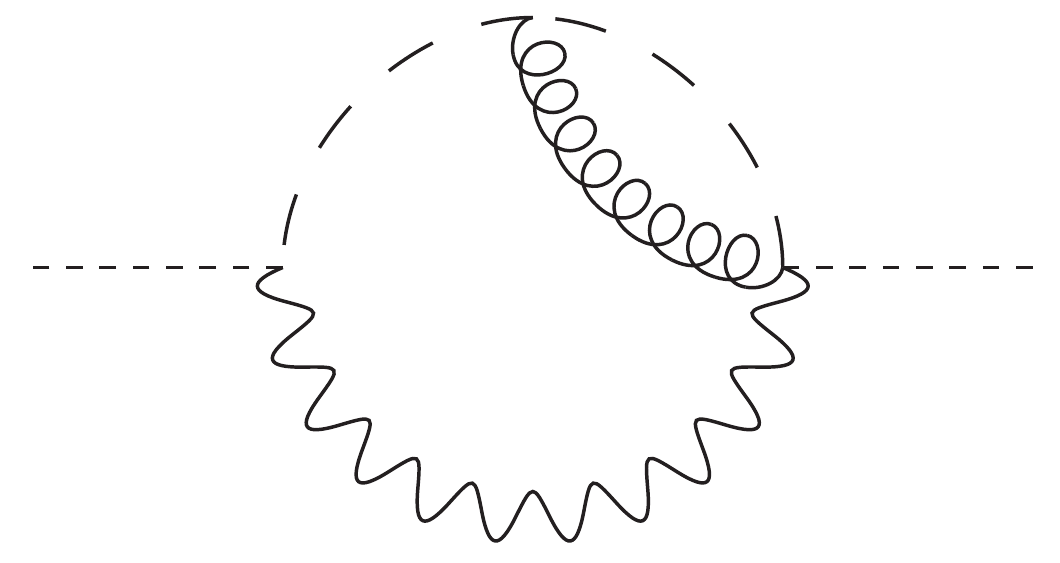} - 2 \graph{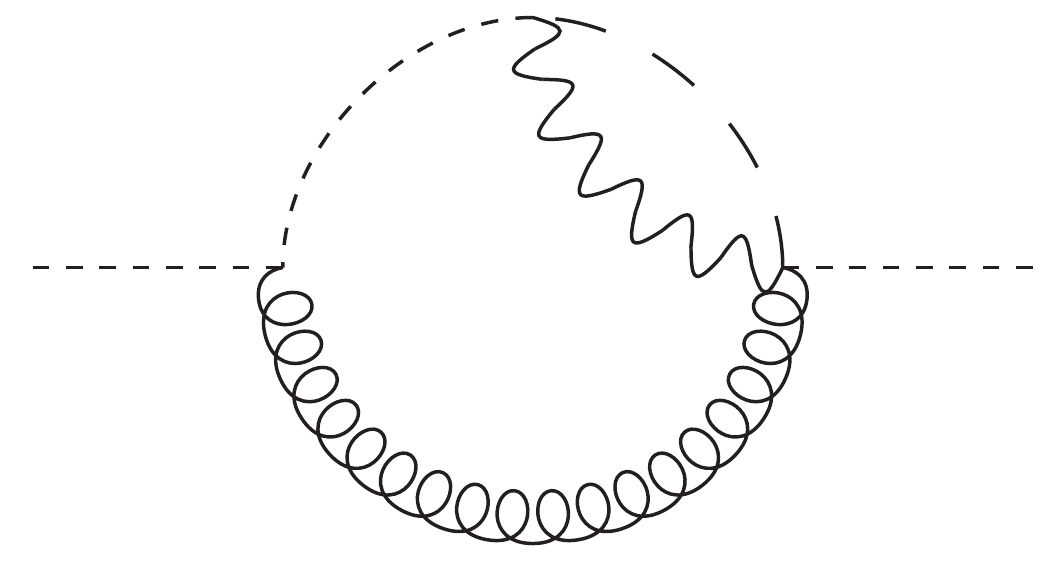} \\ & \phantom{ = } - 2 \graph{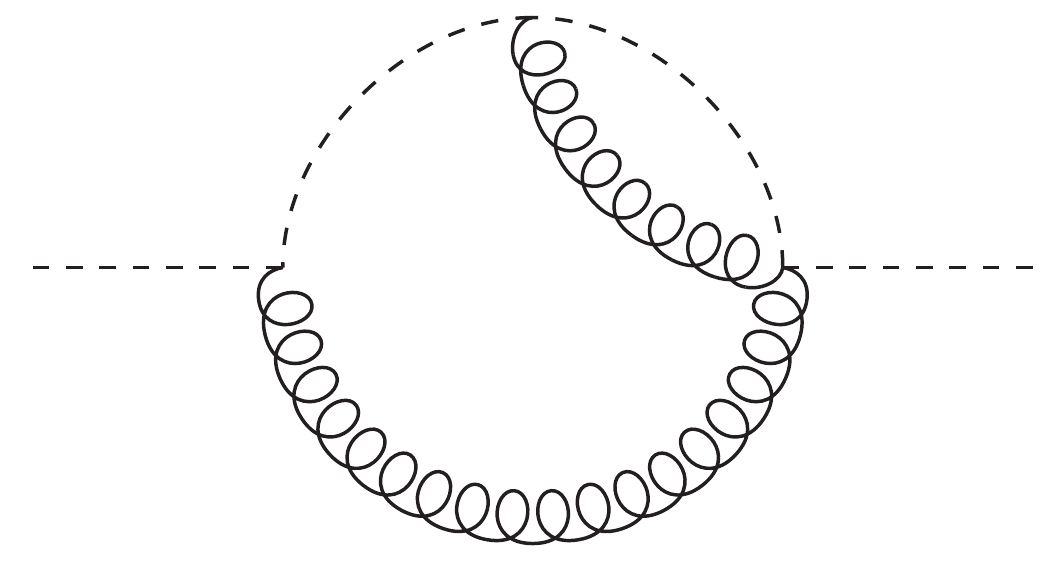} - 2 \graph{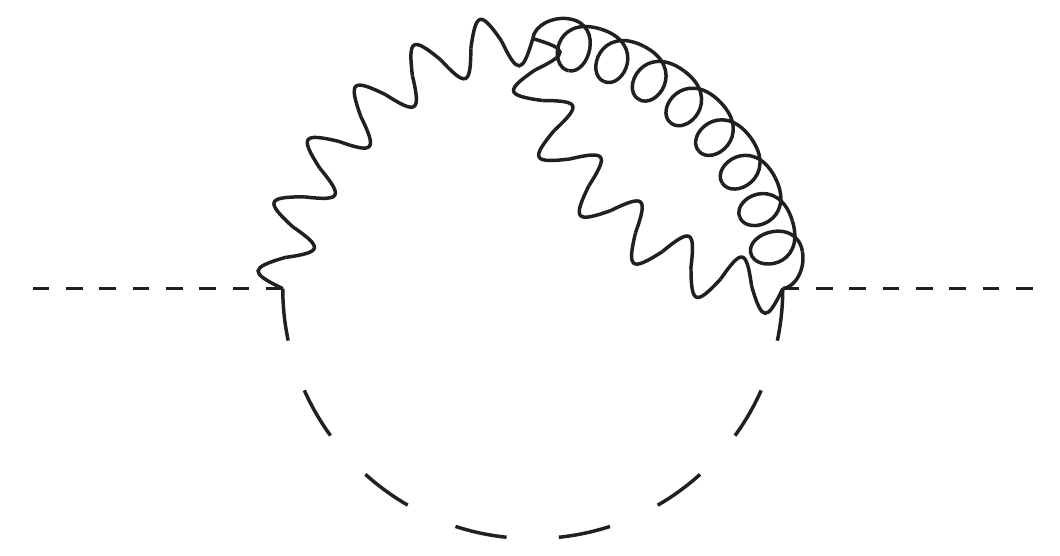} - \frac{2}{2} \graph{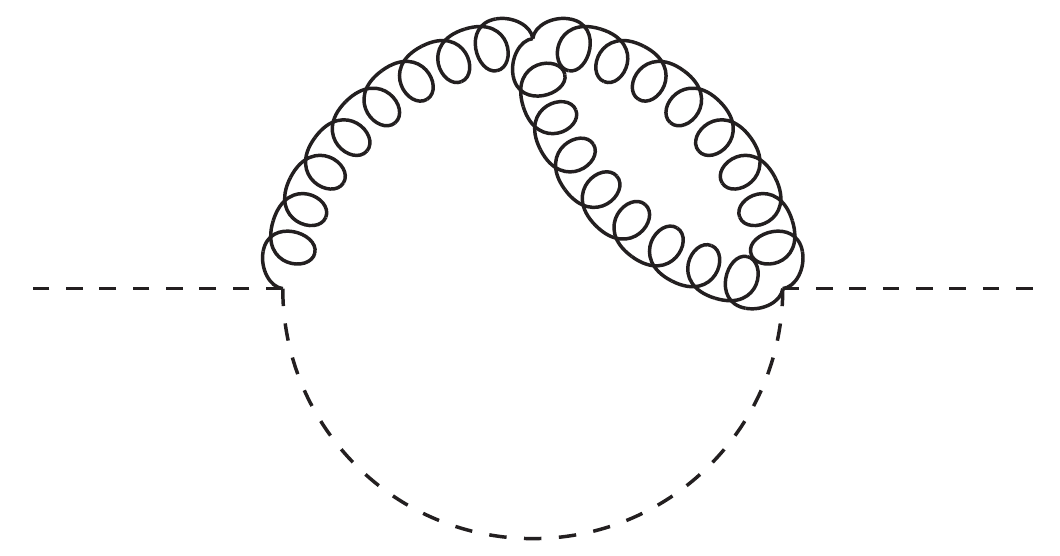}
\end{split}
\intertext{Additionally, we also present the one-loop two- and three-point amplitudes which were set to zero due to Generalized Furry's Theorem, stated in \thmref{thm:generalized_furry_theorem}:}
\rescombgreen^{\cgreen{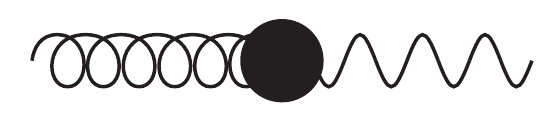}}_{(1,1)} & = - \ograph{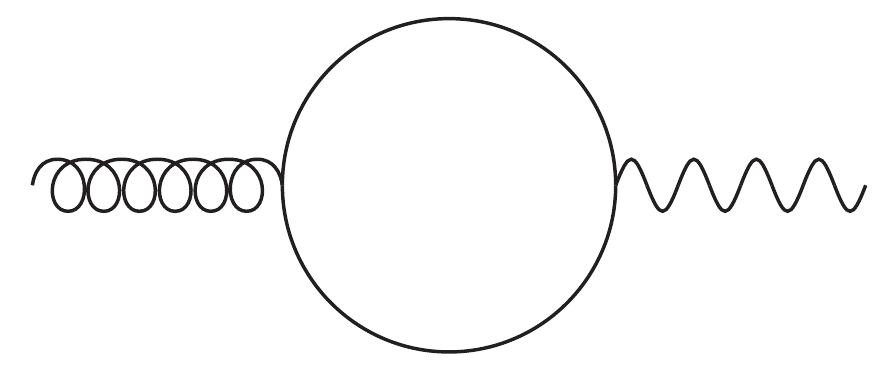}
\\
\rescombgreen^{\tcgreen{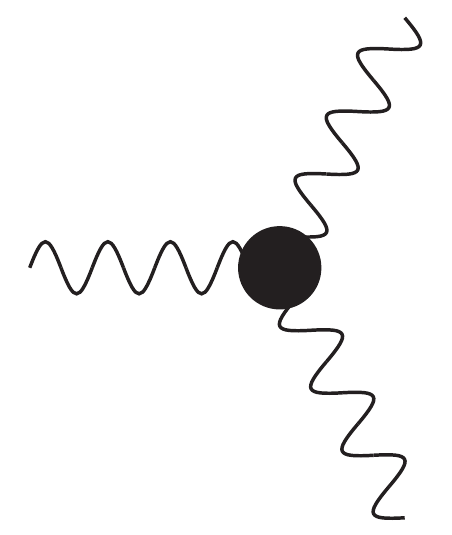}}_{(0,2)} & = \tgraph{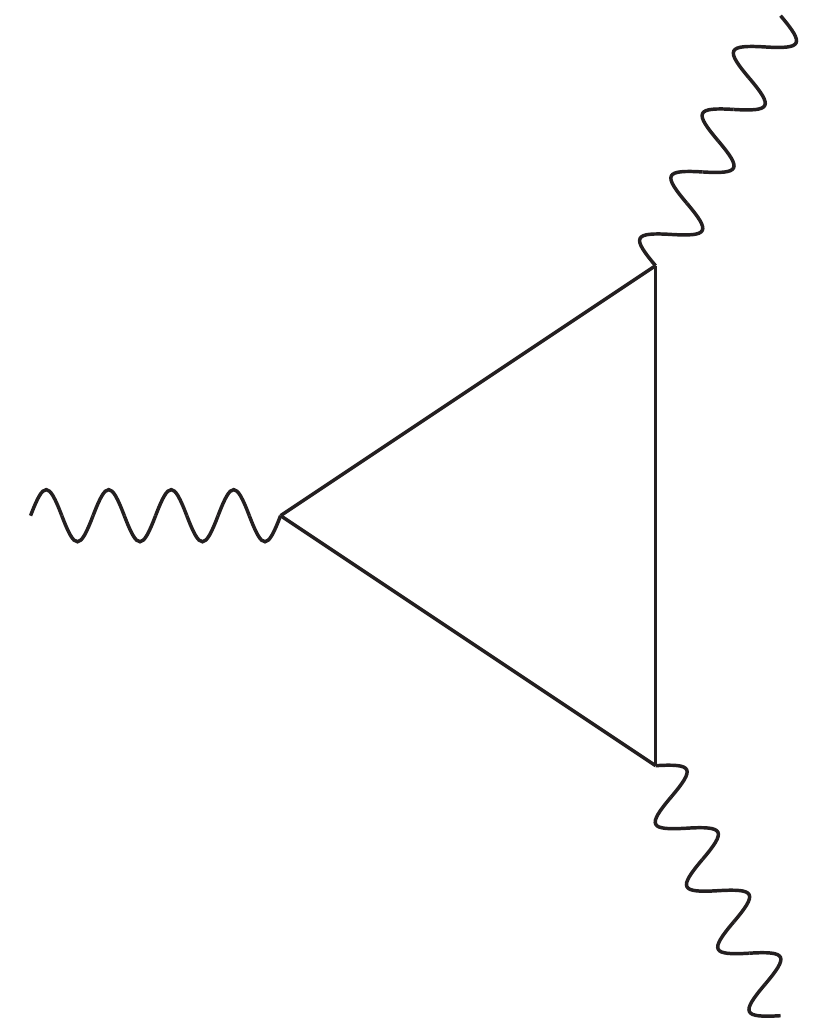}
\\
\rescombgreen^{\tcgreen{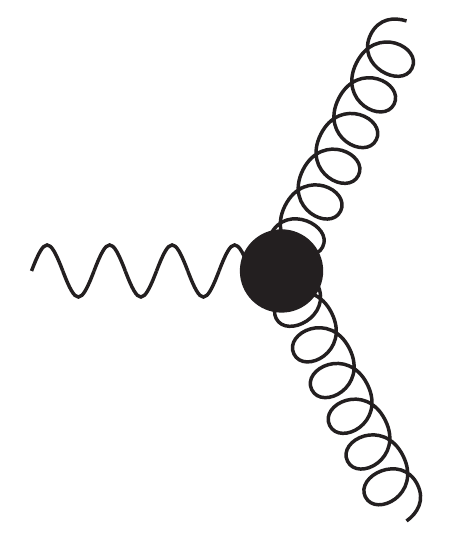}}_{(2,1)} & = \tgraph{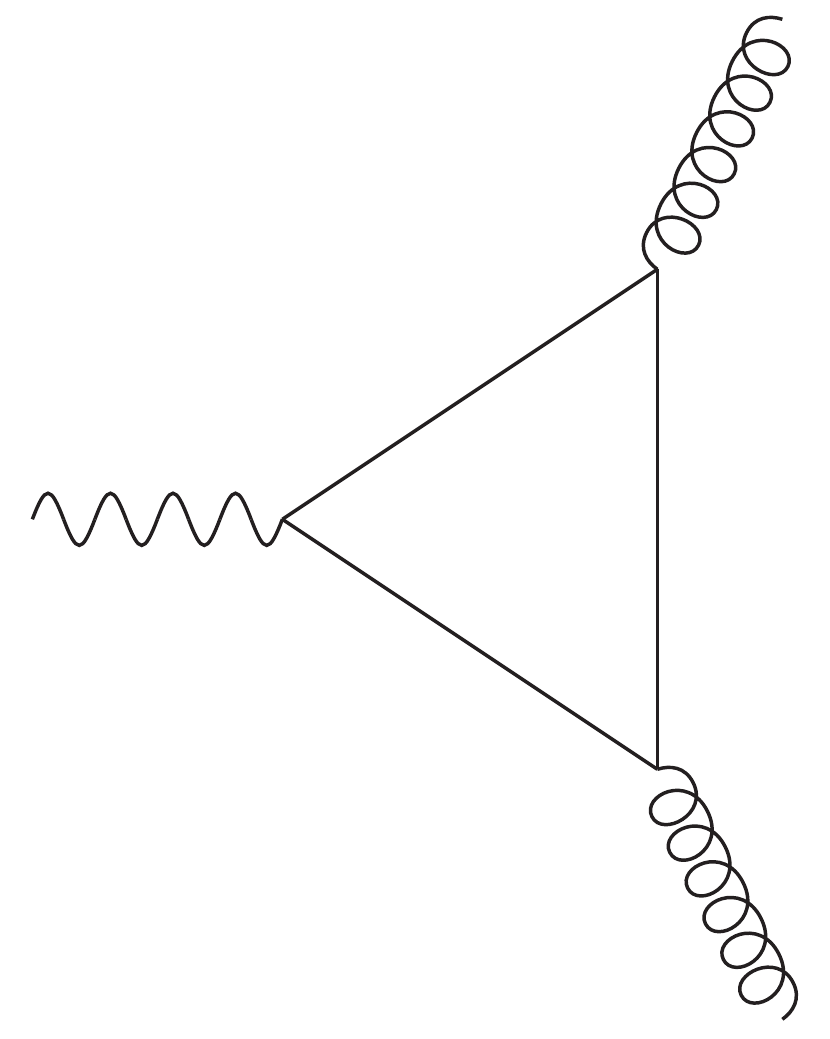}
\end{align}
}

\subsection{The coproduct structure in QGR-QED} \label{ssec:reduced_coproduct_structures_combinatorial_greens_functions}

We obtain the following reduced coproduct structure of the two-loop propagator combinatorial Green's functions:\footnote{We display only the non-vanishing restricted combinatorial Green's functions, cf.\ \ssecref{ssec:combinatorial_green_functions}.}

{\allowdisplaybreaks
\begin{align}
\Delta^\prime \left ( \rescombgreen^{\cgreen{c-ff}}_{(0,4)} \right ) & = \left ( - \rescombgreen^{\cgreen{c-ff}}_{(0,2)} - \rescombgreen^{\cgreen{c-pp}}_{(0,2)} + 2 \rescombgreen^{\tcgreen{c-pff}}_{(0,2)} \right ) \otimes \rescombgreen^{\cgreen{c-ff}}_{(0,2)}
\\
\begin{split}
\Delta^\prime \left ( \rescombgreen^{\cgreen{c-ff}}_{(2,2)} \right ) & = \left ( - \rescombgreen^{\cgreen{c-ff}}_{(2,0)} - \rescombgreen^{\cgreen{c-pp}}_{(2,0)} + 2 \rescombgreen^{\tcgreen{c-pff}}_{(2,0)} \right ) \otimes \rescombgreen^{\cgreen{c-ff}}_{(0,2)} \\ & \phantom{ = ( } + \left ( - \rescombgreen^{\cgreen{c-ff}}_{(0,2)} + 2 \rescombgreen^{\tcgreen{c-gff}}_{(0,2)} \right ) \otimes \rescombgreen^{\cgreen{c-ff}}_{(2,0)}
\end{split}
\\
\Delta^\prime \left ( \rescombgreen^{\cgreen{c-ff}}_{(4,0)} \right ) & = \left ( - \rescombgreen^{\cgreen{c-ff}}_{(2,0)} - \rescombgreen^{\cgreen{c-gg}}_{(2,0)} + 2 \rescombgreen^{\tcgreen{c-gff}}_{(2,0)} \right ) \otimes \rescombgreen^{\cgreen{c-ff}}_{(2,0)}
\\
\Delta^\prime \left ( \rescombgreen^{\cgreen{c-pp}}_{(0,4)} \right ) & = \left ( - 2 \rescombgreen^{\cgreen{c-ff}}_{(0,2)} + 2 \rescombgreen^{\tcgreen{c-pff}}_{(0,2)} \right ) \otimes \rescombgreen^{\cgreen{c-pp}}_{(0,2)}
\\
\begin{split}
\Delta^\prime \left ( \rescombgreen^{\cgreen{c-pp}}_{(2,2)} \right ) & = \left ( - 2 \rescombgreen^{\cgreen{c-ff}}_{(2,0)} + 2 \rescombgreen^{\tcgreen{c-pff}}_{(2,0)} \right ) \otimes \rescombgreen^{\cgreen{c-pp}}_{(0,2)} \\ & \phantom{ = ( } + \left ( - \rescombgreen^{\cgreen{c-pp}}_{(0,2)} + 2 \rescombgreen^{\tcgreen{c-gpp}}_{(0,2)} \right ) \otimes \rescombgreen^{\cgreen{c-pp}}_{(2,0)}
\end{split}
\\
\begin{split}
\Delta^\prime \left ( \rescombgreen^{\cgreen{c-pp}}_{(4,0)} \right ) & = \left ( - \rescombgreen^{\cgreen{c-pp}}_{(2,0)} - \rescombgreen^{\cgreen{c-gg}}_{(2,0)} + 2 \rescombgreen^{\tcgreen{c-gpp}}_{(2,0)} \right ) \otimes \ograph{c1pp_1_2}\\
& \phantom{ = ( } + \left ( - \rescombgreen^{\cgreen{c-aa}}_{(2,0)} - \rescombgreen^{\cgreen{c-hh}}_{(2,0)} + 2 \rescombgreen^{\tcgreen{c-pah}}_{(2,0)} \right ) \otimes \ograph{c1pp_1_3}
\end{split}
\\
\Delta^\prime \left ( \rescombgreen^{\cgreen{c-gg}}_{(0,4)} \right ) & = 0
\\
\begin{split}
\Delta^\prime \left ( \rescombgreen^{\cgreen{c-gg}}_{(2,2)} \right ) & = \left ( - 2 \rescombgreen^{\cgreen{c-ff}}_{(0,2)} + 2 \rescombgreen^{\tcgreen{c-gff}}_{(0,2)} \right ) \otimes \frac{1}{2} \ograph{c1gg_1_1} \\ & \phantom{ = } + \left ( - 2 \rescombgreen^{\cgreen{c-pp}}_{(0,2)} + 2 \rescombgreen^{\tcgreen{c-gpp}}_{(0,2)} \right ) \otimes \frac{1}{2} \ograph{c1gg_1_2}
\end{split}
\\
\begin{split}
\Delta^\prime \left ( \rescombgreen^{\cgreen{c-gg}}_{(4,0)} \right ) & = \left ( - 2 \rescombgreen^{\cgreen{c-ff}}_{(2,0)} + 2 \rescombgreen^{\tcgreen{c-gff}}_{(2,0)} \right ) \otimes \frac{1}{2} \ograph{c1gg_1_1} \\ & \phantom{ = } + \left ( - 2 \rescombgreen^{\cgreen{c-pp}}_{(2,0)} + 2 \rescombgreen^{\tcgreen{c-gpp}}_{(2,0)} \right ) \otimes \frac{1}{2} \ograph{c1gg_1_2} \\ & \phantom{ = } + \left ( - 2 \rescombgreen^{\cgreen{c-gg}}_{(2,0)} + 2 \rescombgreen^{\tcgreen{c-ggg}}_{(2,0)} \right ) \otimes \frac{1}{2} \ograph{c1gg_1_3} \\ & \phantom{ = } + \left ( - 2 \rescombgreen^{\cgreen{c-aa}}_{(2,0)} + 2 \rescombgreen^{\tcgreen{c-gaa}}_{(2,0)} \right ) \otimes \frac{1}{2} \ograph{c1gg_1_4} \\ & \phantom{ = } + \left ( - 2 \rescombgreen^{\cgreen{c-hh}}_{(2,0)} + 2 \rescombgreen^{\tcgreen{c-ghh}}_{(2,0)} \right ) \otimes \frac{1}{2} \ograph{c1gg_1_5}
\end{split}
\\
\Delta^\prime \left ( \rescombgreen^{\cgreen{c-aa}}_{(0,4)} \right ) & = 0
\\
\begin{split}
\Delta^\prime \left ( \rescombgreen^{\cgreen{c-aa}}_{(2,2)} \right ) & = - \rescombgreen^{\cgreen{c-pp}}_{(0,2)} \otimes \ograph{c1aa_1_1} \\
	& \phantom{ = } - 0 \otimes \ograph{c1aa_1_2}
\end{split}
\\
\begin{split}
\Delta^\prime \left ( \rescombgreen^{\cgreen{c-aa}}_{(4,0)} \right ) & = \left ( - \rescombgreen^{\cgreen{c-pp}}_{(2,0)} - \rescombgreen^{\cgreen{c-hh}}_{(2,0)} + 2 \rescombgreen^{\tcgreen{c-pah}}_{(2,0)} \right ) \otimes \ograph{c1aa_1_1} \\ & \phantom{ = } + \left ( - \rescombgreen^{\cgreen{c-gg}}_{(2,0)} - \rescombgreen^{\cgreen{c-aa}}_{(2,0)} + 2 \rescombgreen^{\tcgreen{c-gaa}}_{(2,0)} \right ) \otimes \ograph{c1aa_1_2}
\end{split}
\\
\Delta^\prime \left ( \rescombgreen^{\cgreen{c-hh}}_{(0,4)} \right ) & = 0
\\
\begin{split}
\Delta^\prime \left ( \rescombgreen^{\cgreen{c-hh}}_{(2,2)} \right ) & = - \rescombgreen^{\cgreen{c-pp}}_{(0,2)} \otimes \ograph{c1hh_1_1} \\
	& \phantom{ = } - 0 \otimes \ograph{c1hh_1_2}
\end{split}
\\
\begin{split}
\Delta^\prime \left ( \rescombgreen^{\cgreen{c-hh}}_{(4,0)} \right ) & = \left ( - \rescombgreen^{\cgreen{c-pp}}_{(2,0)} - \rescombgreen^{\cgreen{c-aa}}_{(2,0)} + 2 \rescombgreen^{\tcgreen{c-pah}}_{(2,0)} \right ) \otimes \ograph{c1hh_1_1} \\ & \phantom{ = } + \left ( - \rescombgreen^{\cgreen{c-gg}}_{(2,0)} - \rescombgreen^{\cgreen{c-hh}}_{(2,0)} + 2 \rescombgreen^{\tcgreen{c-ghh}}_{(2,0)} \right ) \otimes \ograph{c1hh_1_2}
\end{split}
\end{align}
}

\subsection{Obstructions to multiplicative renormalization in QGR-QED} \label{ssec:obstructions_to_multiplicative_renormalization}

Using \defnref{defn:hopf_subalgebras_renormalization_hopf_algebra}, we conclude that multiplicative renormalization is possible if the following generalized Ward-Takahashi and Slavnov-Taylor identities \cite{Ward,Takahashi,tHooft_YM,Taylor,Slavnov} hold on the level of Feynman rules, i.e.\ the \(\mathscr{R}\)-divergent contributions from each form factor of the corresponding integral expressions coincide (denoted via \(\simeq_\mathscr{R}\)), where \(\mathscr{R}\) is the chosen renormalization scheme:

{\allowdisplaybreaks
\begin{align}
	\left ( - \rescombgreen^{\cgreen{c-pp}}_{(2,0)} - \rescombgreen^{\cgreen{c-gg}}_{(2,0)} + 2 \rescombgreen^{\tcgreen{c-gpp}}_{(2,0)} \right ) & \simeq_\mathscr{R} \left ( - \rescombgreen^{\cgreen{c-aa}}_{(2,0)} - \rescombgreen^{\cgreen{c-hh}}_{(2,0)} + 2 \rescombgreen^{\tcgreen{c-pah}}_{(2,0)} \right ) \\
	\left ( - 2 \rescombgreen^{\cgreen{c-ff}}_{(0,2)} + 2 \rescombgreen^{\tcgreen{c-gff}}_{(0,2)} \right ) & \simeq_\mathscr{R} \left ( - 2 \rescombgreen^{\cgreen{c-pp}}_{(0,2)} + 2 \rescombgreen^{\tcgreen{c-gpp}}_{(0,2)} \right ) \\
\begin{split}
	\left ( - 2 \rescombgreen^{\cgreen{c-ff}}_{(2,0)} + 2 \rescombgreen^{\tcgreen{c-gff}}_{(2,0)} \right ) & \simeq_\mathscr{R} \left ( - 2 \rescombgreen^{\cgreen{c-pp}}_{(2,0)} + 2 \rescombgreen^{\tcgreen{c-gpp}}_{(2,0)} \right ) \\
	& \simeq_\mathscr{R} \left ( - 2 \rescombgreen^{\cgreen{c-gg}}_{(2,0)} + 2 \rescombgreen^{\tcgreen{c-ggg}}_{(2,0)} \right ) \\
	& \simeq_\mathscr{R} \left ( - 2 \rescombgreen^{\cgreen{c-aa}}_{(2,0)} + 2 \rescombgreen^{\tcgreen{c-gaa}}_{(2,0)} \right ) \\
	& \simeq_\mathscr{R} \left ( - 2 \rescombgreen^{\cgreen{c-hh}}_{(2,0)} + 2 \rescombgreen^{\tcgreen{c-ghh}}_{(2,0)} \right )
\end{split}
\\
\rescombgreen^{\cgreen{c-pp}}_{(0,2)} & \simeq_\mathscr{R} 0 \label{eqn:photon_propagator_convergent} \\
\rescombgreen^{\cgreen{c-hh}}_{(2,0)} & \simeq_\mathscr{R} \rescombgreen^{\cgreen{c-aa}}_{(2,0)} \\
\left ( - \rescombgreen^{\cgreen{c-pp}}_{(2,0)} + 2 \rescombgreen^{\tcgreen{c-pah}}_{(2,0)} \right ) & \simeq_\mathscr{R} \left ( - \rescombgreen^{\cgreen{c-gg}}_{(2,0)} + 2 \rescombgreen^{\tcgreen{c-ghh}}_{(2,0)} \right )
\end{align}
}

We refer for a study of these relations on a general level to \cite{Prinz_3}. We only remark here that \eqnref{eqn:photon_propagator_convergent} requires that the one-loop photon propagator with a fermion loop is convergent. This is not true, and thus shows that QGR-QED is a priori not multiplicative renormalizable using only two coupling constants (one for the electric charge and one for the gravitational coupling). This problem could be resolved by the introduction of two gravitational couplings --- one for the pure gravity part and one for the gravity-matter coupling. However, this is part of a more general theme and could be also resolved if additional symmetries are present, which will be studied in future work.

\section{Conclusion} \label{sec:conclusion}

In this article we considered Quantum General Relativity coupled to Quantum Electrodynamics (QGR-QED). First, we introduced the necessary differential geometric background and the Lagrange density of QGR-QED in \sectionref{sec:differential_geometric_notions_and_the_lagrange_density_of_qgr_qed}. Then, in \sectionref{sec:hopf_algebras_and_the_connes-kreimer_renormalization_hopf_algebra} we introduced Hopf algebras in general and the Connes-Kreimer renormalization Hopf algebra in particular. Furthermore, we discussed a problem which can occur when associating the renormalization Hopf algebra to a given local QFT and discuss possible solutions. Moreover, we examine Hopf ideals inside the renormalization Hopf algebra which represent the symmetries compatible with renormalization. Next, the application of these general results to QGR-QED is discussed. In particular, a generalization of Furry's Theorem including external gravitons and graviton ghosts is formulated and proved in \thmref{thm:generalized_furry_theorem}. This is in particular useful, since the calculations showed that, besides from pure self-loop Feynman graphs which vanish in kinematic renormalization schemes, these are the only graphs which need to be set to zero when constructing the renormalization Hopf algebra of QGR-QED for two-loop propagator graphs. Then, in \sectionref{sec:coproduct_structure_of_the_greens_function} we present all combinatorial Green's functions for the one- and two-loop propagator graphs and the one-loop three-point functions. Then, their coproduct structure is presented, for which the coproduct of 155 Feynman graphs has been computed. Using this result we presented and discussed the obstructions to multiplicative renormalization for QGR-QED.

We remark the following connections to further research in this direction: A detailed treatment of the Feynman rules for Quantum General Relativity coupled to the Standard Model is examined in \cite{Prinz_4}, as was mentioned in \remref{rem:Feynman_rules_1}. Furthermore, the effects of gauge symmetries on renormalization are studied in \cite{Prinz_3}, which results in generalizations of Ward-Takahashi and Slavnov-Taylor identities \cite{Ward,Takahashi,tHooft_YM,Taylor,Slavnov}, as was mentioned in \ssecref{ssec:obstructions_to_multiplicative_renormalization}. Finally, it will be studied if it is possible to define a Corolla polynomial in this setting \cite{Kreimer_Yeats,Kreimer_Sars_vSuijlekom,Sars,Prinz_1,Kreimer_Corolla,Berghoff_Knispel}, which would then create the amplitudes of QGR-QED from scalar \(\phi^3_4\)-theory.

\section*{Acknowledgments}
\addcontentsline{toc}{section}{Acknowledgments}
The author thanks Dirk Kreimer, Helga Baum and the rest of the Kreimer group for illuminating and helpful discussions! This research is supported by the International Max Planck Research School for Mathematical and Physical Aspects of Gravitation, Cosmology and Quantum Field Theory.

\bibliography{References}{}
\bibliographystyle{babunsrt}

\end{document}